\documentclass[11pt]{article}
\usepackage[utf8]{inputenc}
\usepackage[top=3cm,left=3cm,right=3cm,bottom=3cm]{geometry}
\usepackage{amsmath,amsthm,amssymb}
\usepackage[hyperindex,breaklinks]{hyperref}
\usepackage[sorting=nyt,giveninits=true,doi=false,url=false,isbn=false,maxbibnames=99]{biblatex}
\usepackage[titletoc,title]{appendix}
\usepackage{makecell}
\usepackage{caption}
\usepackage{mathrsfs}

\allowdisplaybreaks

\AtEveryBibitem{\clearlist{language}}
\AtEveryBibitem{\clearfield{month}}
\AtEveryBibitem{\clearfield{day}}

\makeatletter
\newcommand{\address}[1]{\gdef\@address{#1}}
\newcommand{\email}[1]{\gdef\@email{\url{#1}}}
\newcommand{\@endstuff}{\par\vspace{\baselineskip}\noindent\small
\begin{tabular}{@{}l}\scshape\@address\\\textit{E-mail address:} \@email\end{tabular}}
\AtEndDocument{\@endstuff}
\makeatother

\newtheorem{theorem}{Theorem}[section]
\newtheorem{lemma}[theorem]{Lemma}
\newtheorem{proposition}[theorem]{Proposition}
\newtheorem{corollary}[theorem]{Corollary}

\theoremstyle{definition}

\theoremstyle{definition}
\newtheorem{definition}[theorem]{Definition}

\theoremstyle{definition}
\newtheorem{remark}[theorem]{Remark}

\newcommand{\diver}{\mathrm{div}}
\newcommand{\ric}{\mathrm{Ric}}
\newcommand{\tr}{\mathrm{tr}}
\newcommand{\K}{\mathcal{K}}
\newcommand{\h}{\mathcal{H}}
\newcommand{\Ko}{\mathcal{\mathring{K}}}
\newcommand{\ho}{\mathcal{\mathring{H}}}
\newcommand{\Ric}{\mathfrak{Ric}}
\newcommand{\sric}{\overline{\ric}}
\newcommand{\lie}{\mathcal{L}}
\newcommand{\n}{\nabla}
\newcommand{\sn}{\overline{\n}}
\newcommand{\p}{\partial}
\newcommand{\s}{\varphi}
\newcommand{\phio}{\mathring{\Phi}}
\newcommand{\psio}{\mathring{\Psi}}
\newcommand{\R}{\mathbb{R}}
\newcommand{\ca}{\mathcal{A}}
\newcommand{\ce}{\mathcal{E}}
\newcommand{\cm}{\mathcal{M}}
\newcommand{\cd}{\mathcal{D}}
\newcommand{\adj}{\mathrm{adj}}
\newcommand{\deh}{\delta h}
\newcommand{\dk}{\delta K}
\newcommand{\dt}{\delta \theta}
\newcommand{\dep}{\delta \s}
\newcommand{\cs}{\mathcal{S}}
\newcommand{\cf}{\mathcal{F}}
\newcommand{\mfe}{\mathfrak E}
\newcommand{\mfx}{\mathfrak X}

\newcommand{\se}{\mathscr E}
\newcommand{\Ro}{\mathring{R}}

\title{Developments of initial data on big bang singularities for the Einstein--nonlinear scalar field equations}
\author{Andrés Franco-Grisales}
\date{}
\address{Department of Mathematics, KTH, 100 44 Stockholm, Sweden}
\email{anfg@kth.se}

\begin{document}

\maketitle

\begin{abstract}
    In a recent work, Ringström proposed a geometric notion of initial data on big bang singularities. Moreover, he conjectured that initial data on the singularity could be used to parameterize quiescent solutions to Einstein's equations; that is, roughly speaking, solutions whose leading order asymptotics are convergent. We prove that given initial data on the singularity for the Einstein--nonlinear scalar field equations in $4$ spacetime dimensions, as defined by Ringström, there is a corresponding unique development of the data. We do not assume any symmetry or analyticity, and we allow for arbitrary closed spatial topology. Our results thus present an important step towards resolving Ringström's conjecture. Furthermore, our results show that the Einstein--nonlinear scalar field equations have a geometric singular initial value problem formulation, which is analogous to the classical result by Choquet-Bruhat and Geroch for initial data on a Cauchy hypersurface.

    In the literature, there are two conditions which are expected to ensure that quiescent behavior occurs. The first one is an integrability condition on a special spatial frame. The second one is an algebraic condition on the eigenvalues of the expansion normalized Weingarten map associated with a foliation of the spacetime near the singularity. Our result is the first such result where both possibilities are allowed. That is, we allow for the first condition to ensure quiescence in one region of space and for the second condition to take over in the region where the first one is violated. This fact allows for our results to include the vacuum setting. We expect that all current results in the literature where the authors specify data on a big bang singularity and then construct a corresponding solution, to the Einstein vacuum or Einstein--scalar field equations in $4$ spacetime dimensions, with closed spatial topology, and where the limits of the eigenvalues of the expansion normalized Weingarten map are everywhere distinct, should be special cases of our results.  
\end{abstract}

\tableofcontents

\section{Introduction}

According to Hawking's singularity theorem, cosmological solutions to Einstein's equations typically present singularities in the form of incomplete timelike geodesics. However, the theorem gives no information about the nature of the singularity. One proposal for the behavior of generic big bang singularities is the so called BKL conjecture. It states that the singularity is spatially local, and either oscillatory or quiescent. Here we are concerned with the quiescent setting, which roughly means that the spacetime presents leading order asymptotics which are convergent near the singularity.

In \cite{ringstrom_initial_2022-1}, Ringström proposed a geometric notion of initial data on big bang singularities. The aim was to provide a unifying framework for several existing results in the literature, where the authors specify different notions of data on the singularity, and then prove existence of corresponding solutions to Einstein's equations. Moreover, he conjectured that initial data on the singularity could be used to parameterize quiescent solutions to Einstein's equations, which woluld potentially have applications to the study of oscillatory big bang singularities in the spatially inhomogeneous setting. However, for this to even be possible, one first has to show that quiescent solutions necessarily induce initial data on the singularity; and second, one has to show that given initial data on the singularity, there is a corresponding unique development. The former is addressed in \cite{ringstrom_initial_2022-1}. Here we are concerned with the latter. We prove that given initial data on the singularity for the Einstein--nonlinear scalar field equations, as in \cite{ringstrom_initial_2022-1}, with an appropriate potential $V$, there is a corresponding unique (up to isometry) solution to the Einstein--nonlinear scalar field equations with potential $V$. The spacetimes that we construct admit Gaussian foliations near the singularity, that is, the metric takes the form $g = -dt \otimes dt + h$ in a neighborhood of the singularity, where $h$ denotes the family of induced (Riemannian) metrics on the level sets of the $t$ coordinate. Our results have the following important consequence. The notion of initial data on the singularity, as in \cite{ringstrom_initial_2022-1}, leads to a singular initial value problem formulation for the Einstein--nonlinear scalar field equations. Here we prove results analogous to the fundamental results of Choquet-Bruhat and Geroch \cite{foures-bruhat_theoreme_1952,choquet-bruhat_global_1969} for initial data on a Cauchy hypersurface. That is, our results show that given initial data on the singularity, as in \cite{ringstrom_initial_2022-1}, for the Einstein--nonlinear scalar field equations, there is a corresponding unique (up to isometry) maximal globally hyperbolic development within the class of spacetimes under consideration.

For the construction of the solutions, we use methods similar to the ones developed by Fournodavlos and Luk in \cite{fournodavlos_asymptotically_2023}. In \cite{fournodavlos_asymptotically_2023} the authors construct solutions to the Einstein vacuum equations, on $(0,T] \times \mathbb{T}^3$, with prescribed initial data on the singularity. We remark that the notion of initial data used in \cite{fournodavlos_asymptotically_2023}, is a special case of the one introduced in \cite{ringstrom_initial_2022-1}. Our existence result can thus be seen as a generalization of \cite{fournodavlos_asymptotically_2023}, where we include a nonlinear scalar field and allow for arbitrary closed (compact without boundary) spatial topology. Moreover, for general developments of initial data on the singularity, we obtain detailed asymptotics for the eigenvalues and eigenspaces of the expansion normalized Weingarten map (the second fundamental form with an index raised, divided by the mean curvature) associated with the foliation of the spacetime near the singularity. Thanks to these detailed asymptotics, we are able to show that the metric, in a neighborhood of the singularity, can be put in the form
\[
g = -dt \otimes dt + \sum_{i,k} b_{ik}t^{2p_{\max\{i,k\}}} \omega^i \otimes \omega^k
\]
on $(0,T)\times \Sigma$. Here $p_i \in C^\infty(\Sigma)$ are time independent, $\{\omega^i\}$ is a (time independent) frame of one forms on $\Sigma$, and the functions $b_{ik} \in C^\infty((0,T)\times\Sigma)$ satisfy $b_{ik} \to \delta_{ik}$ as $t \to 0$. We note that this form of the metric is analogous to the ansatz used in \cite{fournodavlos_asymptotically_2023}. This fact is remarkable since we do not make any such assumptions for the definition of development of the data. Moreover, this turns out to be of essential importance to settle the question of uniqueness of developments.

\subsection{Developments of initial data on the singularity}

We are interested in solving the Einstein--nonlinear scalar field equations, with cosmological constant $\Lambda$ and a potential $V \in C^{\infty}(\R)$. Let $(M,g)$ be a 4-dimensional spacetime and let $\s \in C^{\infty}(M)$ denote the scalar field. Then the equations are
\begin{subequations}
\begin{align}
    \ric - \frac{1}{2} Sg + \Lambda g &= T, \label{einstein equation 1}\\
    \Box_g \s &= V' \circ \s, \label{matter equation}
\end{align} 
\end{subequations}
where $\ric$ and $S$ denote the Ricci and scalar curvature of $g$ respectively, $\Box_g = \tr_g \n^2$ is the wave operator associated with $g$, $\n$ is the Levi-Civita connection of $g$ and $T$ is the energy-momentum tensor of $\s$, which is given by
\[
T = d\s \otimes d\s - \bigg( \frac{1}{2}|d\s|^2_g + V \circ \s \bigg)g.
\]
Note that \eqref{einstein equation 1} may be reformulated as
\begin{equation} \label{einstein equation 2}
    \ric = d\s \otimes d\s + (V \circ \s)g + \Lambda g.
\end{equation}
It may be verified that \eqref{matter equation} implies that $\diver_g T = 0$, thus ensuring the compatibility of $T$ with the Einstein tensor $G = \ric - \frac{1}{2}Sg$. Note that the cosmological constant can be accounted for by adding a constant to $V$. Hence, there is no loss of generality in assuming $\Lambda = 0$, and we do so in what follows. We are only interested in potentials satisfying the following condition; cf. \cite[Definition~1]{oude_groeniger_formation_2023}.

\begin{definition} \label{admissible potential}
    We say that $V \in C^{\infty}(\R)$ is an \emph{admissible potential} if there are constants $C_m$ and $0 < a < \sqrt{6}$ such that
    \begin{equation*}
        |V^{(m)}(x)| \leq C_m e^{a|x|}
    \end{equation*}
    for every non-negative integer $m$. Also, define $\varepsilon_V := 1 - \frac{a}{2} \sqrt{\frac{2}{3}}$. Note that $\varepsilon_V > 0$.
\end{definition}

Now we introduce our notion of initial data on the singularity; cf. \cite[Definition 10]{ringstrom_initial_2022-1}. For our purposes, the definition of initial data on the singularity in the Einstein--scalar field setting translates directly to the Einstein--nonlinear scalar field setting.

\begin{definition} \label{initial data}
    Let $(\Sigma, \ho)$ be a closed 3-dimensional Riemannian manifold, $\Ko$ a (1,1)-tensor field on $\Sigma$ and $\phio, \psio \in C^{\infty}(\Sigma)$. Then $(\Sigma, \ho, \Ko, \phio, \psio)$ are \emph{non-degenerate quiescent initial data on the singularity for the Einstein--nonlinear scalar field equations} if the following holds:
    \begin{enumerate}
        \item $\tr \Ko = 1$ and $\Ko$ is symmetric with respect to $\ho$.
        \item $\tr \Ko^2 + \psio^2 = 1$ and $\diver_{\ho} \Ko = \psio d\phio$.
        \item The eigenvalues of $\Ko$ are everywhere distinct.
        \item $(\gamma_{23}^1)^2 = 0$ in a neighborhood of $x \in \Sigma$ if $p_1(x) \leq 0$, where $p_1 < p_2 < p_3$ are the eigenvalues of $\Ko$, the $e_i$ are orthonormal (with respect to $\ho$) eigenvector fields of $\Ko$ such that $\Ko(e_i) = p_i e_i$, and the  $\gamma_{ik}^{\ell}$ are defined by $[e_i,e_k] = \gamma_{ik}^{\ell}e_{\ell}$. 
    \end{enumerate}
\end{definition}

\begin{remark}
    Note that depending on the choice of eigenvectors, $\gamma_{ik}^{\ell}$ is well defined up to a sign. Hence the $(\gamma_{ik}^{\ell})^2$ are well defined.
\end{remark}

Henceforth, we shall refer to non-degenerate quiescent initial data on the singularity for the Einstein--nonlinear scalar field equations, simply as initial data on the singularity for short. Let us now fix some notation. Given initial data on the singularity, as in Definition~\ref{initial data}, we denote by $\{e_i\}$ a frame of eigenvectors of $\Ko$, with dual frame $\{\omega^i\}$, such that $\Ko(e_i) = p_i e_i$, where $p_1 < p_2 < p_3$ denote the eigenvalues of $\Ko$, and normalized so that $\ho(e_i,e_i)=1$. Also, denote by $\gamma_{ik}^\ell$ the structure coefficients of the frame, defined by $[e_i,e_k] = \gamma_{ik}^\ell e_\ell$. Note that such a frame always exists, at least locally. However, by an argument similar to \cite[Lemma A.1]{ringstrom_wave_2021}, there is a finite covering space of $\Sigma$ such that, if we pull back the initial data, the pullback of $\Ko$ by the covering map has a global frame of eigenvectors. That being the case, from now on we assume, if necessary, that we work on this finite covering space instead, so that the frame $\{e_i\}$ is global. Later we shall see why this assumption is not a restriction on our results; see Remark~\ref{about the global frame} below. Finally, denote by $D$ the Levi-Civita connection of $\ho$. 

Now, it is necessary to clarify what the correspondence between initial data on the singularity and a solution to Einstein's equations should be. For that purpose, let $(M,g)$ be a spacetime and let $\Sigma \subset M$ be a spacelike hypersurface. We denote by $h$ the induced metric on $\Sigma$ and by $k$ the second fundamental form, defined by
\[
k(X,Y) := g(\n_X U,Y),
\]
where $U$ is the future pointing unit normal of $\Sigma$ and $X,Y \in \mathfrak X(\Sigma)$. Define by $K := k^{\sharp}$ the Weingarten map (see Definition~\ref{raising an index} below for our conventions regarding the notation $\sharp$), and by $\theta := \tr_{h} k = \tr K$ the mean curvature. Now we introduce the expansion normalized quantities, which are the ones expected to converge to the initial data on the singularity along a suitable foliation of the spacetime.

\begin{definition} \label{expansion normalized}
    Let $(M,g)$ be a spacetime, $\s \in C^\infty(M)$ and $\Sigma \subset M$ a spacelike hypersurface with future pointing unit normal $U$. If $\theta > 0$, the \emph{expansion normalized Weingarten map} is defined by
    \[
    \K(X) := \frac{1}{\theta} K(X),
    \]
    and the \emph{expansion normalized induced metric} by
    \[
    \h(X,Y) := h(\theta^{\K}(X), \theta^{\K}(Y)), 
    \]
    where
    \[
    \theta^{\K}(X) := \sum_{n = 0}^{\infty} \frac{(\ln \theta)^n}{n!} \K^n(X)
    \]
    and $X, Y \in \mathfrak{X}(\Sigma)$. Furthermore, define the \emph{expansion normalized normal derivative of the scalar field} and the \emph{expansion normalized induced scalar field} by
    \[
    \Psi := \frac{1}{\theta}U \s, \qquad \Phi := \s + \Psi \ln \theta  
    \]
    respectively. 
\end{definition}

Before relating the expansion normalized quantities with the initial data on the singularity, we introduce some conventions regarding the type of foliation that we will use. Let $\Sigma$ be a manifold. In what follows, we will consider metrics of the form $g = -dt \otimes dt + h$ on $(0,T) \times \Sigma$, where $h$ denotes the family of induced metrics on the $\Sigma_t := \{t\} \times \Sigma$ hypersurfaces, which are assumed to be spacelike. Hence, $h_t$ for $t \in (0,T)$ is a smooth one parameter family of Riemannian metrics on $\Sigma$. Note that the future pointing unit normal of the $\Sigma_t$ hypersurfaces is given by $U = \p_t$. We are now ready to introduce the notion of development of initial data on the singularity.

\begin{definition} \label{development of the data}
    Let $(\Sigma,\ho,\Ko,\phio,\psio)$ be initial data on the singularity and let $V$ be an admissible potential. We say that a spacetime $(M,g)$ and a function $\s \in C^\infty(M)$ are a \emph{locally Gaussian development of the initial data} if the following holds. $(M,g,\s)$ is a solution to the Einstein--nonlinear scalar field equations with potential $V$. There is a $T > 0$ and a diffeomorphism $F$ from $(0,T) \times \Sigma$ to an open subset of $M$, such that $F^*g$ takes the form $F^*g = -dt \otimes dt + h$. The mean curvature $\theta$ of the $\Sigma_t$ hypersurfaces satisfies $\theta \to \infty$ as $t \to 0$. Finally, there are positive constants $\delta$ and $C_m$ such that the expansion normalized quantities, associated with the $\Sigma_t$ hypersurfaces, satisfy the estimates
    \begin{equation} \label{convergence estimates}
        |D^m(\h - \ho)|_\ho + |D^m(\K - \Ko)|_\ho + |D^m(\Phi - \phio)|_\ho + |D^m(\Psi - \psio)|_\ho \leq C_mt^\delta
    \end{equation}
    for all $t \in (0,T)$ and every non-negative integer $m$. Moreover, if $(M,g)$ is globally hyperbolic and the hypersurfaces $F(\Sigma_t)$ are Cauchy hypersurfaces in $(M,g)$, then $(M,g,\s)$ is called a \emph{locally Gaussian globally hyperbolic development of the initial data}.
\end{definition}

\begin{remark}
    In principle, one could consider developments of initial data on the singularity of finite regularity and consider convergence estimates such as \eqref{convergence estimates} with Sobolev norms instead. Nevertheless, here we choose to work in the smooth setting. The reason is that in the current situation, a loss of derivatives is expected to occur when going back and forth between regular initial data and data on the singularity. Moreover, the number of derivatives that one can expect to lose is unbounded; see \cite{li_scattering_2024}. 
\end{remark}

\begin{remark}
    One might wonder what the geometric significance of the time coordinate $t$ is in the definition above. Let us restrict ourselves to a neighborhood of the singularity where the metric takes the form $F^*g = -dt \otimes dt + h$ and consider a point $(t,x) \in (0,T) \times \Sigma$. It is clear that every past inextendible causal curve emanating from $(t,x)$ has length bounded above by $t$. Moreover, the segment of the integral curve of $\p_t$ which passes through $(t,x)$ and which lies to the past of $(t,x)$, is a past inextendible timelike geodesic with length equal to $t$. That is, the time coordinate $t$ is the \emph{cosmological time function} of the spacetime near the singularity; see \cite{andersson_cosmological_1998}.     
\end{remark}

\subsection{Main results}

In order to state our results, it is convenient to introduce some notation.

\begin{definition} \label{stable and unstable regions}
    Let $(\Sigma,\ho,\Ko,\phio,\psio)$ be initial data on the singularity. For every $x \in \Sigma$ such that $p_1(x) \leq 0$, there is an open neighborhood $U_x \ni x$ such that $\gamma_{23}^1|_{U_x} \equiv 0$. Define $D_- := \cup_x U_x$ and $D_+ := \Sigma \setminus D_-$. Thus $D_-$ is an open neighborhood of $\{x \in \Sigma : p_1(x) \leq 0 \}$ where $\gamma_{23}^1 \equiv 0$ and $D_+$ is closed, hence compact, and $p_1 > 0$ in $D_+$. Define
    \[
    \mathring{\varepsilon} := \min\Big\{ \min_{x \in D_+}\{ 2p_1(x) \}, \min_{x \in \Sigma}\{ 1-p_3(x) \} \Big\}.
    \]
    Note that $\mathring{\varepsilon} > 0$.
\end{definition}

\begin{definition}
    Let $(\Sigma,\ho,\Ko,\phio,\psio)$ be initial data on the singularity and let $V$ be an admissible potential. Define $\varepsilon := \min\{\mathring\varepsilon, \varepsilon_V\}$.
\end{definition}

Now we are ready to state our main existence result. 

\begin{theorem} \label{main existence theorem}
    Let $(\Sigma, \ho, \Ko, \phio, \psio)$ be initial data on the singularity and let $V$ be an admissible potential. Then there is a $T > 0$, depending only on the initial data and the potential, such that the following holds. There is a Lorentzian metric $g = -dt \otimes dt + h$ and a function $\s$ on $(0,T) \times \Sigma$, which solve the Einstein--nonlinear scalar field equations with potential $V$, such that the mean curvature $\theta$ of the $\Sigma_t$ hypersurfaces satisfies $\theta \to \infty$ as $t \to 0$. Moreover, the corresponding expansion normalized quantities satisfy the estimates
    \begin{equation*} 
        |D^m(\h - \ho)|_\ho + |D^m(\K - \Ko)|_\ho + |D^m(\Phi - \phio)|_\ho + |D^m(\Psi - \psio)|_\ho \leq C_mt^\varepsilon
    \end{equation*}
    for every non-negative integer $m$, where the constants $C_m$ depend only on the initial data and the potential.
\end{theorem}

The proof of Theorem~\ref{main existence theorem} is to be found at the end of Subsection~\ref{asymptotics for the expansion normalized metric}. In order to prove Theorem~\ref{main existence theorem}, it turns out to be useful to also define alternative versions of $\h$, $\Psi$ and $\Phi$.

\begin{definition}
    Let $(\Sigma,\ho,\Ko,\phio,\psio)$ be initial data on the singularity and consider a metric $g = -dt \otimes dt + h$ and a function $\s$ on $(0,T) \times \Sigma$. Define the \emph{approximate expansion normalized induced metric} by
    \[
    \bar \h(X,Y) := h(t^{-\Ko}(X),t^{-\Ko}(Y)),
    \]
    where $t^{-\Ko}$ is defined similarly as $\theta^{\K}$ and $X,Y \in \mfx(\Sigma)$. Also define $\bar \Psi$ and $\bar \Phi$ by
    \[
    \bar \Psi := t\p_t \s, \qquad \bar \Phi := \s - \bar \Psi \ln t.
    \]
\end{definition}

The motivation for defining $\bar\h$, $\bar\Phi$ and $\bar\Psi$ is the following. Even though this is not immediately obvious from Definition~\ref{development of the data}, the expectation is that the mean curvature $\theta$ should be asymptotic to $t^{-1}$. Since $\K$ should converge to $\Ko$, then $\bar\h$, $\bar\Phi$ and $\bar\Psi$ should be asymptotically the same as $\h$, $\Phi$ and $\Psi$ respectively. One advantage of $\bar\h$, $\bar\Phi$ and $\bar\Psi$ is then that it is easier to work with the time coordinate $t$ rather than with $\theta$. But more importantly, $\bar\h$ can be treated by using the frame $\{e_i\}$ of eigenvectors of $\Ko$, which turns out to be very useful for our construction. 

Throughout the proof of Theorem~\ref{main existence theorem}, we obtain much more detailed asymptotic information about the constructed spacetimes than what the convergence estimates \eqref{convergence estimates} immediately imply. It is thus of interest to know whether the more detailed asymptotics we obtain actually follow as a consequence of \eqref{convergence estimates} and Einstein's equations. The answer to this question turns out to be affirmative and is given by the following result. 

\begin{theorem} \label{asymptotics of the frame}
    Let $(\Sigma, \ho, \Ko, \phio, \psio)$ be initial data on the singularity and let $V$ be an admissible potential. Suppose that we have a locally Gaussian development $(M,g,\s)$ of the initial data. Then, by taking $T$ (as in Definition~\ref{development of the data}) smaller if necessary, there are constants $C_m$, for every non-negative integer $m$, such that the following holds. Define $\sigma := \min\{\varepsilon,\delta/2\}$. The mean curvature satisfies
    \[
    |D^m( t\theta - 1 )|_\ho + |D^m( \ln \theta + \ln t )|_\ho  \leq C_mt^\sigma.
    \]
    The eigenvalues of $\K$ are everywhere distinct. Let $q_1 < q_2 < q_3$ be the eigenvalues of $\K$. There is a frame $\{E_i\}$ on $\Sigma$, with $\K(E_i) = q_i E_i$, which is orthonormal with respect to $\h$, with dual frame $\{\eta^i\}$, such that
    \[
    |D^m(q_i - p_i)|_\ho + |D^m(E_i - e_i)|_\ho + |D^m(\eta^i - \omega^i)|_\ho \leq C_mt^\sigma.
    \]
    For $i\neq k$, the following off-diagonal improvements on the estimates hold,
    \[
    \big|D^m\big(\omega^k(E_i)\big)\big|_\ho + \big|D^m\big( \eta^k(e_i) \big)\big|_\ho \leq C_m t^\sigma \min\{1, t^{2(p_i - p_k)}\}.
    \]
    Moreover, for $x \in D_+$, the better improvements
    \[
    \big|D^m\big( \omega^k(E_i)\big)\big|_\ho(x) + \big|D^m\big( \eta^k(e_i)\big)\big|_\ho(x) \leq C_m t^{\sigma + 2(p_i-p_1)(x)}
    \]
    hold. In particular,
    \[
    |D^m(\bar \h - \ho)|_\ho + |D^m(tK - \Ko)|_\ho \leq C_mt^\sigma,
    \]
    and for $i \neq k$,
    \[
    \big|D^m\big( \bar\h(e_i,e_k) \big)\big|_\ho \leq C_mt^{\sigma + |p_i - p_k|}, \quad \big|D^m\big(tK(e_i,\omega^k)\big)\big|_\ho \leq C_m t^\sigma \min\{1,t^{2(p_i-p_k)}\}.
    \]
    Furthermore, for $x \in D_+$, the better improvements
    \[
    \big|D^m\big( \bar\h(e_i,e_k) \big)\big|_\ho(x) \leq C_mt^{\sigma + (p_i + p_k - 2p_1)(x)}, \quad \big|D^m\big(tK(e_i,\omega^k)\big)\big|_\ho(x) \leq C_m t^{\sigma + 2(p_i-p_1)(x)}
    \]
    hold. Finally, if $R$ denotes the Riemann curvature tensor of $g$, the Kretschmann scalar $|R|_g^2$ satisfies the estimate
    \[
    \big|D^m \big[ t^4|R|_g^2 - 4 \big( \textstyle\sum_i p_i^2(1-p_i)^2 + \textstyle\sum_{i<k} p_i^2p_k^2 \big) \big] \big|_\ho \leq C_mt^\sigma,
    \]
    so that the spacetime is $C^2$ past inextendible.
\end{theorem}

The proof of Theorem~\ref{asymptotics of the frame} is to be found in Subsection~\ref{proofs of detailed asymptotics and uniqueness}. As a consequence of Theorem~\ref{asymptotics of the frame}, it is also possible to obtain estimates for the time derivatives. 

\begin{corollary} \label{estimates for time derivatives}
    With the same setup as in Theorem~\ref{asymptotics of the frame}, there are constants $C_{m,r}$ such that
    \begin{align*}
        |D^m (t\p_t)^r (t\theta)|_\ho + |D^m (t\p_t)^r \bar\Psi|_\ho + |D^m (t\p_t)^r q_i|_\ho &\leq C_{m,r}t^\sigma,\\
        \big|D^m (t\p_t)^r\big(\omega^i(E_i)\big)\big|_\ho + \big|D^m (t\p_t)^r\big(\eta^i(e_i)\big)\big|_\ho &\leq C_{m,r}t^\sigma,\\
        \big|D^m (t\p_t)^r\big(\omega^k(E_i)\big)\big|_\ho(x) + \big|D^m (t\p_t)^r\big(\eta^k(e_i)\big)\big|_\ho(x) &\leq C_{m,r}t^{\sigma + 2(p_i-p_1)(x)},\\
        \big|D^m (t\p_t)^r\big(\omega^k(E_i)\big)\big|_\ho(y) + \big|D^m (t\p_t)^r\big(\eta^k(e_i)\big)\big|_\ho(y) &\leq C_{m,r}t^{\sigma} \min\{1, t^{2(p_i-p_k)(y)}\},
    \end{align*}
    for $i \neq k$ (no summation on $i$), $x \in D_+$ and $y \in D_-$, every non-negative integer $m$ and every integer $r \geq 1$. In particular,
    \[
    \big|D^m (t\p_t)^r\big(\bar \h(e_i,e_i) \big)\big|_\ho + \big|D^m (t\p_t)^r\big(tK(e_i,\omega^i)\big)\big|_\ho \leq C_{m,r}t^\sigma
    \]
    (no summation on $i$), and for $i \neq k$, $x \in D_+$ and $y \in D_-$,
    \begin{align*}
        \big|D^m (t\p_t)^r\big( \bar\h(e_i,e_k) \big)|_\ho(x) &\leq C_{m,r}t^{\sigma + (p_i + p_k - 2p_1)(x)},\\
        \big|D^m (t\p_t)^r\big( \bar\h(e_i,e_k) \big)\big|_\ho(y) &\leq C_{m,r}t^{\sigma + |p_i - p_k|(y)},\\
        \big|D^m (t\p_t)^r\big(tK(e_i,\omega^k)\big)\big|_\ho(x) &\leq C_{m,r} t^{\sigma + 2(p_i-p_1)(x)},\\
        \big|D^m (t\p_t)^r\big(tK(e_i,\omega^k)\big)\big|_\ho(y) &\leq C_{m,r} t^\sigma \min\{1,t^{2(p_i-p_k)(y)}\}.
    \end{align*}
\end{corollary}

\begin{remark}
    Since $\gamma_{23}^1 \equiv 0$ in $D_-$, it is of interest to know what the behavior of the corresponding structure coefficient of the frame $\{E_i\}$ is. Define the $\lambda_{ik}^\ell$ by $[E_i,E_k] = \lambda_{ik}^\ell E_\ell$. The convergence estimates for the $E_i$ then imply the decay
    \[
    |D^m(\lambda_{23}^1)|_\ho(y) \leq C_mt^{\sigma + 2(p_2-p_1)(y)},
    \]
    for $y \in D_-$.
\end{remark}

\begin{remark} \label{fournodavlos luk ansatz remark}
    A few comments are in order regarding the definition of $\h$. If we compare Definition~\ref{development of the data} with \cite[Definition 16]{ringstrom_initial_2022-1}, we note that the definitions of the expansion normalized induced metric differ. Specifically, in \cite{ringstrom_initial_2022-1}, the pulled-back metric is assumed to take the form 
    \begin{equation} \label{fournodavlos luk ansatz} 
        F^*g = -dt \otimes dt + \sum_{i,k} b_{ik}t^{2p_{\max\{i,k\}}} \omega^i \otimes \omega^k.
    \end{equation}
    Define $\check{h} := b_{ik} \omega^i \otimes \omega^k$. In \cite{ringstrom_initial_2022-1}, it is then required that $\check{h} \to \ho$ as $t \to 0$. We remark that this form of the metric arises from the ansatz used in \cite{fournodavlos_asymptotically_2023} for their construction, with the difference that \cite{fournodavlos_asymptotically_2023} uses a global coordinate frame on $\mathbb{T}^3$ instead of a frame of eigenvectors of $\Ko$. By contrast, our definition of $\h$ corresponds to that of \cite{oude_groeniger_formation_2023}, which has the advantage of being frame independent, and of being independent of the initial data on the singularity. The issue is that it is a priori unclear how convergence of $\h$ is related with convergence of $\check{h}$. The expectation is that both conditions should be equivalent, and our results show that this is indeed the case. In the course of our existence proof, we essentially show that $\check{h} \to \ho$ implies $\h \to \ho$ (see the proof of Theorem~\ref{main existence theorem} below). The other implication follows from Theorem~\ref{asymptotics of the frame} and Corollary~\ref{estimates for time derivatives} by defining $b_{ik} := t^{-|p_i-p_k|}\bar\h(e_i,e_k)$, so that $b_{ik} \to \delta_{ik}$ as $t \to 0$ with the estimates
    \[
    |D^m(t\p_t)^r (b_{ik} - \delta_{ik})|_\ho \leq C_{m,r}\langle \ln t \rangle^mt^\sigma,
    \]
    where $\langle \ln t \rangle := \sqrt{1 + (\ln t)^2}$. In fact, we can say a bit more. Since we have even better off-diagonal estimates for $\bar\h$ in $D_+$, we can define the functions $a_{ik} := t^{2p_1 - p_i-p_k}\bar\h(e_i,e_k)$ for $i \neq k$ and $a_{ii} := \bar\h(e_i,e_i)$. Then $F^*g$ can be written in the form
    \[
    F^*g = -dt \otimes dt + \sum_i a_{ii}t^{2p_i}\omega^i \otimes \omega^i + \sum_{i \neq k} a_{ik}t^{2(p_i+p_k-p_1)} \omega^i \otimes \omega^k,
    \]
    where $a_{ik}(x) \to \delta_{ik}$ as $t \to 0$ for $x \in D_+$ with the estimates 
    \[
    |D^m(t\p_t)^r (a_{ik} - \delta_{ik})|_\ho(x) \leq C_{m,r}\langle \ln t \rangle^mt^\sigma.
    \]
    Thus, given a locally Gaussian development, writing the metric in the form \eqref{fournodavlos luk ansatz} is not a restriction.
\end{remark}

\begin{remark}
    The off-diagonal estimates that we obtain for $K$, excluding the better ones that we obtain in $D_+$, are analogous to those obtained in \cite{fournodavlos_asymptotically_2023}. Similar estimates for the mean curvature and the off-diagonal components of the eigenvectors of $\K$, again excluding the better ones in $D_+$, are also obtained in \cite{ringstrom_initial_2022-1}.
\end{remark}

Now we move on to the question of uniqueness of developments of initial data on the singularity. We obtain the following geometric uniqueness result, which essentially states that if we have two locally Gaussian developments of the same initial data on the singularity, then the solutions are locally isometric in a neighborhood of the singularity. In order to prove this result, the detailed asymptotics of Theorem~\ref{asymptotics of the frame}, in particular the off-diagonal improvements on the estimates for $K$ and $\bar\h$, turn out to be of essential importance. 

\begin{theorem} \label{main uniqueness theorem}
    Let $(\Sigma, \ho, \Ko, \phio, \psio)$ be initial data on the singularity and let $V$ be an admissible potential. Suppose that there are locally Gaussian developments $(M_i,g_i,\s_i)$ of the data, with corresponding diffeomorphisms $F_i:(0,T_i) \times \Sigma \to U_i \subset M_i$ for $i = 1,2$. There is a sufficiently small $T>0$ such that, if $V_i := F_i((0,T) \times \Sigma)$ and $\Upsilon := F_2|_{(0,T) \times \Sigma} \circ F_1^{-1}|_{V_1}$, then the diffeomorphism $\Upsilon:V_1 \to V_2$ satisfies $\Upsilon^* g_2 = g_1$ and $\s_2 \circ \Upsilon = \s_1$. 
\end{theorem}

\begin{remark}
    Another way to interpret Theorem~\ref{main uniqueness theorem} is the following. For $i = 1,2$, let $(\Sigma_i, \ho_i, \Ko_i, \phio_i, \psio_i)$ be initial data on the singularity and let $((0,T_i)\times\Sigma_i,g_i,\s_i)$, with $g_i = -dt \otimes dt + h_i$, be corresponding locally Gaussian developments. If there is a diffeomorphism $\overline{\Upsilon}:\Sigma_1 \to \Sigma_2$ such that $\overline{\Upsilon}^*\ho_2 = \ho_1$, $\overline{\Upsilon}^*\Ko_2=\Ko_1$, $\phio_2 \circ \overline{\Upsilon} = \phio_1$ and $\psio_2 \circ \overline{\Upsilon} = \psio_1$ (that is, the data are isometric), then we can define $\Upsilon:(0,T)\times\Sigma_1 \to (0,T)\times\Sigma_2$ by $\Upsilon(t,x) := (t,\overline{\Upsilon}(x))$. By Theorem~\ref{main uniqueness theorem}, we see that $\Upsilon^*g_2 = g_1$ and $\s_2 \circ \Upsilon = \s_1$ for $T$ small enough. That is, isometric initial data on the singularity give rise to developments which are locally isometric in a neighborhood of the singularity, and the isometry preserves the foliation.
\end{remark}

The fundamental works \cite{foures-bruhat_theoreme_1952,choquet-bruhat_global_1969} of Choquet-Bruhat and Geroch ensure that given regular initial data (as opposed to initial data on the singularity) for the Einstein--nonlinear scalar field equations, there is a corresponding maximal globally hyperbolic development which is unique up to isometry; see \cite{ringstrom_cauchy_2009} for a detailed discussion. As a consequence of Theorem~\ref{main uniqueness theorem}, there is also a meaningful notion of a unique (up to isometry) \emph{maximal} locally Gaussian globally hyperbolic development.

\begin{definition}
    Let $(\Sigma,\ho,\Ko,\phio,\psio)$ be initial data on the singularity, $V$ an admissible potential, and consider two locally Gaussian developments of the data $(M_i,g_i,\s_i)$, for $i=1,2$. We say that the developments are \emph{isometric} if there is a diffeomorphism $\Upsilon:M_1 \to M_2$ such that $\Upsilon^*g_2 = g_1$ and $\s_2 \circ \Upsilon = \s_1$. 
\end{definition}

Note that for two developments to be isometric, we do not only require $\Upsilon$ to be an isometry in the usual sense, but also require that it preserves the scalar field.

\begin{definition} \label{maximal development}
    Let $(\Sigma,\ho,\Ko,\phio,\psio)$ be initial data on the singularity, $V$ an admissible potential, and let $(M,g,\s)$ be a locally Gaussian globally hyperbolic development of the data. We say that $(M,g,\s)$ is a \emph{maximal locally Gaussian globally hyperbolic development} of the initial data on the singularity if for any other locally Gaussian globally hyperbolic development $(\widetilde M,\widetilde g,\widetilde\s)$ of the same data, there is a map $\Upsilon:\widetilde M \to M$ which is a diffeomorphism onto its image such that $\Upsilon^*g = \widetilde g$, $\s \circ \Upsilon = \widetilde\s$ and $\Upsilon$ preserves time orientation.
\end{definition}

\begin{theorem} \label{uniqueness}
    Let $(\Sigma,\ho,\Ko,\phio,\psio)$ be initial data on the singularity and let $V$ be an admissible potential. Then there is a maximal locally Gaussian globally hyperbolic development of the data which is unique up to isometry.
\end{theorem}

\begin{remark} \label{preserves the foliation}
    Let $(M,g,\s)$ be a maximal locally Gaussian globally hyperbolic development with diffeomorphism $F: (0,T) \times \Sigma \to U \subset M$, and let $(\widetilde M,\widetilde g,\widetilde\s)$ be another locally Gaussian globally hyperbolic development of the same initial data with diffeomorphism ${\widetilde F:(0,\widetilde T) \times \Sigma \to \widetilde U \subset \widetilde M}$. The map $\Upsilon$ of Definition~\ref{maximal development} can be assumed to satisfy the following property. There is a sufficiently small $t_0 > 0$ such that $\Upsilon \circ \widetilde F|_{(0,t_0) \times \Sigma} = F|_{(0,t_0) \times \Sigma}$. That is, $\Upsilon$ can be assumed to preserve the foliation near the singularity. For details, see the proof of Theorem~\ref{uniqueness} below.
\end{remark}

\begin{proof}
    By Theorem~\ref{main existence theorem}, there exists a locally Gaussian globally hyperbolic development of the initial data $(M,g,\s)$, with diffeomorphism $F: (0,T) \times \Sigma \to U \subset M$. By extending $(M,g,\s)$ if necessary, we may assume that it is the maximal globally hyperbolic development of the regular initial data induced on the Cauchy hypersurface $\Sigma_t$ for some $t \in (0,T)$. We begin by showing that $(M,g,\s)$ is maximal in the sense of Definition~\ref{maximal development}. Let $(\widetilde M, \widetilde g, \widetilde \s)$ be another locally Gaussian globally hyperbolic development of the data with diffeomorphism $\widetilde F:(0,\widetilde T) \times \Sigma \to \widetilde U \subset \widetilde M$. Note that there is no loss of generality in assuming $\widetilde T = T$. By Theorem~\ref{main uniqueness theorem}, if we take $T$ smaller if necessary, we have $F^*g = \widetilde F^* \widetilde g$ and $\s \circ F = \widetilde\s \circ \widetilde F$. Hence both $(M,g,\s)$ and $(\widetilde M,\widetilde g,\widetilde\s)$ are globally hyperbolic developments of the same regular initial data induced on $\Sigma_t$. By maximality of $(M,g,\s)$, there is a map $\Upsilon:\widetilde M \to M$ which is a diffeomorphism onto its image such that $\Upsilon^* g = \widetilde g$, $\s \circ \Upsilon = \widetilde\s$, it preserves time orientation and $\Upsilon \circ \widetilde F|_{\Sigma_t} = F|_{\Sigma_t}$. This finishes the existence part. 
    
    Next, we prove that the property stated in Remark~\ref{preserves the foliation} holds for $(M,g,\s)$. That is, we show that in fact $\Upsilon \circ \widetilde F = F$. This works in exactly the same way as the proof that two maximal globally hyperbolic developments of the same regular initial data are isometric; cf. the comments made after Definition 16.5 in \cite{ringstrom_cauchy_2009}. Let $p \in J^+(\widetilde F(\Sigma_t),\widetilde U)$ (the causal future of $\widetilde F(\Sigma_t)$ in the spacetime $(\widetilde U,\widetilde g)$) and let $\gamma$ be an inextendible future directed timelike geodesic with $\gamma(0) = p$. Then there is an $s \leq 0$ such that $\gamma(s) \in \widetilde F(\Sigma_t)$. Hence $F \circ \widetilde F^{-1} \circ \gamma(s) = \Upsilon \circ \gamma(s)$. Moreover, $d(F \circ \widetilde F^{-1})$ and $d\Upsilon$ agree on $T \Sigma_t$ and since they are both time orientation preserving isometries, both send the future pointing unit normal of $\widetilde F(\Sigma_t)$ to the future pointing unit normal of $F(\Sigma_t)$. That is, $d(F \circ \widetilde F^{-1})$ and $d\Upsilon$ agree on $T_q \widetilde M$ for all $q \in \widetilde F(\Sigma_t)$. Consequently, $(F \circ \widetilde F^{-1} \circ \gamma)'(s) = (\Upsilon \circ \gamma)'(s)$, implying $F \circ \widetilde F^{-1} \circ \gamma = \Upsilon \circ \gamma$. In particular, $F \circ \widetilde F^{-1} (p) = \Upsilon(p)$. Similarly for $p \in J^-(\widetilde F(\Sigma_t,\widetilde U))$. We conclude that $F \circ \widetilde F^{-1} = \Upsilon$ as desired.
    
    For uniqueness, now assume $(\widetilde M, \widetilde g, \widetilde \s)$ to be maximal in the sense of Definition~\ref{maximal development} we show that it is isometric to $(M,g,\s)$. By maximality of $(\widetilde M, \widetilde g, \widetilde \s)$, there is a map $\widetilde\Upsilon: M \to \widetilde M$ which is a diffeomorphism onto its image such that $\widetilde\Upsilon^*\widetilde g = g$, $\widetilde\s \circ \widetilde\Upsilon = \s$ and it preserves time orientation. But now we can think of $\widetilde\Upsilon \circ F|_{\Sigma_t}$ as another embedding of the initial data induced on $\Sigma_t$ into $(\widetilde M, \widetilde g, \widetilde \s)$. Therefore, by maximality of $(M,g,\s)$, as a development of the data induced on $\Sigma_t$, there is a map $\Upsilon: \widetilde M \to M$ (not necessarily the same as above) which is a diffeomorphism onto its image such that $\Upsilon^*g = \widetilde g$, $\s \circ \Upsilon = \widetilde \s$, it preserves time orientation and $\Upsilon \circ \widetilde\Upsilon \circ F|_{\Sigma_t} = F|_{\Sigma_t}$. In particular, $\Upsilon \circ \widetilde\Upsilon$ is the identity map on $F(\Sigma_t)$. The same argument as in the paragraph above now shows that $\Upsilon$ is the inverse of $\widetilde\Upsilon$.  
\end{proof}

\begin{remark}
    One limitation of our uniqueness result is that it is only applicable to spacetimes with Gaussian foliations inducing data on the singularity. Ideally, one would like to have a similar result which allows for more general foliations, like, for instance, asymptotically CMC (constant mean curvature) foliations with zero shift vector field. Note that, by Theorem~\ref{asymptotics of the frame}, Gaussian foliations fall under this category. However, we do not pursue this line of thought here.

    Other generalizations of our results which would be desirable to have are dropping the requirement that the spatial manifold be closed and allowing for higher dimensions. We do note that, in order to generalize to higher dimensions, one would need a way to deal with the spatial Weyl tensor during the construction of the solution in Section~\ref{construction of the solution}. That is, if one were to use the same methods that we do.
\end{remark}

Now we take some time to discuss the conditions of Definition~\ref{initial data}. We note that the necessity of these conditions is discussed at length in \cite{ringstrom_initial_2022-1}. The necessity of Condition 1 is clear, since by the definition of $\h$ and $\K$, we see that $\tr\K = 1$ and $\K$ is symmetric with respect to $\h$. Condition 2 consists of the limits of the constraint equations. The Hamiltonian constraint along the $\Sigma_t$ hypersurfaces reads
    \[
    \bar S + \theta^2 - \tr K^2 = (\p_t\s)^2 + |d\s|_h^2 + 2V \circ \s,
    \]
    where $\bar S$ denotes the scalar curvature of $h$. We can expansion normalize this equation, by multiplying it by $\theta^{-2}$, to obtain
    \[
    \theta^{-2}\bar S + 1 - \tr\K^2 = \Psi^2 + \theta^{-2} |d\s|_h^2 + 2\theta^{-2} V \circ \s.
    \]
The expectation is that $\theta^{-2}\bar S$ and $\theta^{-2}|d\s|_h^2$ should converge to zero as $t \to 0$. Here we are only interested in potentials which do not yield a contribution to the leading order asymptotics, hence we want to ensure that $\theta^{-2}V \circ \s$ converges to zero as $t \to 0$ as well. Note that $\s$ should be asymptotic to $-\psio \ln\theta + \phio$. Moreover, since $\tr\Ko = 1$, then $\tr \Ko^2 \geq 1/3$, implying $|\psio| = \sqrt{1-\tr \Ko^2} \leq \sqrt{2/3}$. Hence we can ensure that $\theta^{-2}V \circ \s$ converges to zero if we impose a bound $|V(x)| \leq Ce^{a|x|}$ with $a < \sqrt{6}$. This motivates the introduction of Definition~\ref{admissible potential} and the first equation in Condition 2. The limit of the momentum constraint is, however, not as straightforward; we refer the reader to Subsection~\ref{the velocity dominated solution} below, in particular Lemma~\ref{velocity dominated solution} below, for a motivation. 

Condition 3 is required to ensure the existence of the frame of eigenvectors $\{e_i\}$, which is extensively used in our construction. It is also important for the formulation of Condition 4, although there is a frame independent way of formulating this condition; see \cite[Remark~27]{ringstrom_initial_2022-1}. Moreover, Condition 3 allows us to construct the frame of eigenvectors of $\K$ near the singularity. Finally, Condition 4 is the one ensuring that the corresponding solutions are indeed quiescent. The necessity of this condition is shown in \cite[Theorem~49]{ringstrom_initial_2022-1}. This condition is discussed in greater detail in Subsection~\ref{related works} below.

\subsection{Related works} \label{related works}

\paragraph{BKL conjecture and AVTD singularities.} By Hawking's singularity theorem (see \cite[Theorem 55A, p. 431]{oneill_semi-riemannian_1983}), we know that cosmological solutions to Einstein's equations typically have singularities. However, the theorem only asserts the existence of incomplete timelike geodesics. One proposal for the generic behavior of solutions near the singularity was provided in a series of papers by Belinskii, Khalatnikov and Lifschitz (BKL); see \cite{lifshitz_investigations_1963,belinskii_oscillatory_1970,belinskii_effect_1973,belinskii_general_1982}, and \cite{damour_cosmological_2003,heinzle_cosmological_2009} for more recent improvements and generalizations. They propose that spatial derivatives should be negligible near the singularity. Moreover, the behavior is expected to be either quiescent, which in the present context can be taken to mean convergence of the eigenvalues of the expansion normalized Weingarten map, or oscillatory. Since our focus is in the quiescent setting, we shall not go into details about the oscillatory setting. The generic behavior in $3+1$-dimensions is expected to be oscillatory. However, there are some situations in which quiescent behavior is to be expected. One possibility is the vacuum setting with spacetime dimensions $\geq 11$. This was first observed in \cite{demaret_non-oscillatory_1985}. Another possibility is the presence of a scalar field or a stiff fluid, see \cite{belinskii_effect_1973,barrow_quiescent_1978}. From these and related works, see e.g. \cite{damour_kasner-like_2002,ringstrom_geometry_2021,fournodavlos_stable_2023,oude_groeniger_formation_2023}, it is expected that the following condition on the eigenvalues of $\Ko$, in the $n+1$-dimensional setting, is sufficient to ensure stable quiescent behavior,
\begin{equation} \label{eigenvalue condition}
    1 + p_1 - p_{n-1} - p_n > 0
\end{equation}
where $p_1 \leq \cdots \leq p_n$. Note that in the $3+1$-dimensional setting this reduces to $p_1 > 0$; cf. Condition 4 in Definition~\ref{initial data}. It turns out that \eqref{eigenvalue condition} is incompatible with vacuum in the $3+1$-dimensional setting, since in that case the $p_i$ have to satisfy the \emph{Kasner relations} $\sum p_i = \sum p_i^2 = 1$, which forces $p_1 \leq 0$. This is where the scalar field enters the picture, since in the scalar field setting the condition on the sum of the squares changes to $\sum p_i^2 + \psio^2 = 1$, which now allows for all of the $p_i$ to be positive. In the regions of the spatial manifold where \eqref{eigenvalue condition} is violated, where $p_1 \leq 0$ in our setting, we need to impose another condition ensuring quiescence, this is why we introduce the condition on the vanishing of $\gamma_{23}^1$. That this condition should be sufficient is suggested by \cite{ringstrom_geometry_2021}, necessity is discussed in \cite{ringstrom_initial_2022-1}. Note that $\gamma_{23}^1 = 0$ may be reformulated as $\omega^1 \wedge d\omega^1 = 0$, which dates back to \cite{lifshitz_investigations_1963}. As opposed to \eqref{eigenvalue condition}, quiescent behavior arising from this condition is not expected to be stable in general. In fact, it is because of this condition that it was proposed in \cite{belinskii_oscillatory_1970} that quiescent behavior is non-generic in the $3+1$-dimensional vacuum setting.

Another terminology that is usually used in the context of quiescent behavior is \emph{asymptotically velocity term dominated} (AVTD) behavior. Essentially what this means is that the asymptotics of a solution to Einstein's equations near the singularity are dictated by the behavior of a solution to a system of simplified equations, the \emph{velocity term dominated} (VTD) equations, where some of the spatial derivative terms have been dropped from the evolution equations. This terminology originates in \cite{eardley_velocitydominated_1972,isenberg_asymptotic_1990}. The solutions that we construct are AVTD in this sense; see Subsection~\ref{the velocity dominated solution} below. There is an extensive literature regarding construction of AVTD solutions to Einstein's equations. In symmetric settings this has been done in \cite{ames_quasilinear_2013,ames_class_2017,choquet-bruhat_topologically_2004,isenberg_asymptotic_1999,isenberg_asymptotic_2002,isenberg_asymptotic_1990,kichenassamy_analytic_1998,rendall_fuchsian_2000,stahl_fuchsian_2002} for the Einstein vacuum equations. The first construction of AVTD solutions without symmetry assumptions was obtained for the Einstein--scalar field and Einstein--stiff fluid equations in $3+1$-dimensions by Andersson and Rendall in \cite{andersson_quiescent_2001}, in the real analytic setting. This was later generalized to include more matter models and higher dimensions in \cite{damour_kasner-like_2002}. A related result is \cite{klinger_new_2015}, where analytic AVTD solutions to the Einstein vacuum equations without symmetry are constructed. Finally, there is \cite{fournodavlos_asymptotically_2023} in the $3+1$-dimensional vacuum setting, which is the first such result without symmetry or analyticity assumptions; a localized version of this result was later introduced in \cite{athanasiou2024localizedconstructionkasnerlikesingularities}. See also \cite{nutzi_semiglobal_2020} for a related result for the Einstein--scalar field equations. We note that, as a consequence of \cite[Propositions 5 and 22]{ringstrom_initial_2022-1}, the results of \cite{fournodavlos_asymptotically_2023} and \cite{andersson_quiescent_2001}, in the scalar field setting, are special cases of our result.

\paragraph{Stability of big bang singularities.} Recently, there has been a lot of progress regarding stability of big bang singularities. The first results came from a series of papers by Rodnianski and Speck \cite{rodnianski_regime_2018,rodnianski_stable_2018,speck_maximal_2018,rodnianski_nature_2021}. They were later joined by Fournodavlos in \cite{fournodavlos_stable_2023}, where they proved that the singularity of the family of spatially homogeneous and spatially flat solutions to the Einstein--scalar field equations satisfying the condition \eqref{eigenvalue condition} is nonlinearly stable. For other recent results on stability of big bang singularities see \cite{beyer_localized_2023,beyer_past_2023,beyer2025localizedpaststabilitysubcritical,fajman_cosmic_2023,fajman_pastmaximal_2024}. We note that in all of these results, the constructed spacetimes exhibit convergent behavior near the singularity.

\paragraph{From regular initial data to data on the singularity.} Related to \cite{fournodavlos_stable_2023} there is the recent work \cite{oude_groeniger_formation_2023} of Oude Groeniger, Petersen and Ringström, where they identify conditions on initial data for the Einstein--nonlinear scalar field equations which lead to the formation of a quiescent big bang singularity. What is remarkable about \cite{oude_groeniger_formation_2023} is that it does not make reference to any background solution. In \cite{oude_groeniger_formation_2023}, the condition \eqref{eigenvalue condition} also plays an important role. Recall that one of the motivations for introducing the notion of initial data on the singularity is to parameterize quiescent solutions. From that point of view, \cite{oude_groeniger_formation_2023} is complementary to our results, since they obtain part of the data on the singularity that one would like to prescribe. They obtain the data on the singularity induced by the scalar field and the limits of the eigenvalues of $\K$. However, they do not manage to get full convergence of $\K$ and $\h$. This is due to the fact that they do not get much information about the frame that they use to express the components of the tensors. Another issue is that in \cite{oude_groeniger_formation_2023}, the authors use a CMC foliation, which has the advantage of synchronizing the singularity. On the other hand, we use a Gaussian foliation. Hence, in order to construct a suitable map between initial data on the singularity and regular initial data, one would presumably need an analog to our result, but which uses a CMC foliation instead. 

\paragraph{Strong cosmic censorship.} AVTD behavior has also been important in the context of the strong cosmic censorship conjecture. Strong cosmic censorship was first proved for polarized Gowdy spacetimes in \cite{chrusciel_strong_1990}. In the general $\mathbb{T}^3$-Gowdy case, it was proved by Ringström in \cite{ringstrom_existence_2006,ringstrom_strong_2009} (see also \cite{ringstrom_strong_2008}). In both cases, the inextendibility of the spacetimes past the singularity is due to curvature blow up, which in turn is obtained through the AVTD behavior. 

\paragraph{Linear wave equations on cosmological spacetimes.} There are also a number of results regarding the asymptotics of solutions to linear wave equations on big bang backgrounds near the singularity; see, for instance, \cite{alho_wave_2019,ringstrom_linear_2020,franco_grisales_asymptotics_2023,ringstrom_unified_2019,allen_asymptotics_2010,li_scattering_2024}.

\subsection{Strategy for the proof}

\paragraph{Existence of developments.}

For the proof of Theorem~\ref{main existence theorem}, we follow the methods of \cite{fournodavlos_asymptotically_2023}. First, we construct a sequence of approximate solutions to Einstein's equations, such that the corresponding approximate expansion normalized quantities converge to the initial data on the singularity. Second, we prove existence of an actual solution to Einstein's equations, by using energy estimates to control the difference between the solution and an appropriately chosen approximate solution. At the end of the construction, convergence of $\K$, $\Phi$ and $\Psi$ to the initial data is already obtained. The final step is then to obtain convergence of $\h$.

For the construction of the approximate solutions, we start with the construction of a \emph{velocity dominated solution} from the initial data. Assuming the metric to take the form ${g = -dt \otimes dt + h}$, one can deduce a system of evolution and constraint equations for the induced metric, the Weingarten map and the scalar field. The velocity dominated solution is then a solution to the associated VTD equations and it has the property that the expansion normalized quantities are constant in time; cf. the concept of the \emph{scaffold} in \cite[Subsection~1.7]{oude_groeniger_formation_2023}. Using the velocity dominated solution as a starting point, we then inductively construct the sequence as follows. For simplicity, we illustrate the idea in the vacuum setting. In that case the evolution equation for the Weingarten map can be written as
\begin{equation} \label{vacuum equation for k}
    \lie_{\p_t} K + \sric^\sharp + \theta K = 0,
\end{equation}
where $\sric$ denotes the Ricci curvature of $h$ and $\lie_{\p_t}K$ is introduced in Definition~\ref{normal lie derivative} below. Assume we are given a one parameter family of $(1,1)$-tensors $\bar K_{n-1}$ and a Lorenzian metric $g_{n-1} = -dt \otimes dt + h_{n-1}$, such that $\bar K_{n-1}$ approximates the Weingarten map of the $\Sigma_t$ hypersurfaces with respect to $g_{n-1}$. Since the spatial derivative terms in the equations are expected to be negligible, we replace $\sric$ in the equation above by the Ricci curvature of $h_{n-1}$ and solve the resulting equation to define $\bar K_n$. Of course, for this idea to be consistent, we need to have an appropriate bound for the spatial Ricci curvature of the approximate solutions. The relevant bound is $t^2 \sric^\sharp = O(t^\delta)$ for some $\delta>0$. Condition 4 in Definition~\ref{initial data} is what allows us to obtain the required bound for the velocity dominated solution, thus ensuring that we can carry on with the construction. Comparing with \cite{fournodavlos_asymptotically_2023}, our situation is more complicated since we include the scalar field. Additionally, in \cite{fournodavlos_asymptotically_2023} they use a global coordinate frame on $\mathbb{T}^3$, whereas the frame $\{e_i\}$ that we use does not, in general, commute, so one has to deal with the structure coefficients. Nonetheless, our construction of the approximate solutions is more streamlined in the sense that we identify a common structure on the equations satisfied by the expansion normalized quantities, which allows us to apply general existence and uniqueness results throughout. 

The construction of the actual solution to Einstein's equations is based on performing energy estimates by using a wave equation for the Weingarten map $K$. In order to obtain this equation we time differentiate the evolution equation for $K$, Equation~\eqref{vacuum equation for k} in the vacuum case, and then use the first variation formula for the Ricci tensor. This leads to a second order equation for $K$ such that the higher order derivative terms are
\[
\lie_{\p_t}^2 K - \Delta_h K + \sn^2(\tr K)^\sharp,
\]
where $\sn$ is the Levi-Civita connection of $h$. The resulting equation is not a wave equation because of the term $\sn^2(\tr K)^\sharp$. In order to get around this, we consider $\theta = \tr K$ as a separate variable and add an evolution equation for it. The system that we solve is then a wave transport system where the variables are $h$, $K$, $\theta$ and $\s$. There is one issue with this system, which is that it presents a potential loss of derivatives. In order to deal with this difficulty we introduce, besides the main energy, a modified energy. The modified energy is designed in such a way that the loss of derivatives is avoided. Moreover, by using elliptic estimates, it can be shown that both energies are in fact equivalent. We still have not mentioned the role that the approximate solutions play here. In the current situation, the expectation is that energy estimates take the form
\[
\frac{d}{dt}\ce(t) \leq \frac{C}{t} \ce(t).
\]
The issue is that we want uniform control of the energy for all $t$ in an interval of the form $(0,T]$, but this cannot be done with Grönwall's inequality, since $Ct^{-1}$ is not integrable all the way to $t=0$. This is where the approximate solutions come in. Instead of controlling the energy of the solution to the system, we control the energy of the difference between the solution and an appropriately chosen approximate solution. The advantage of taking this approach is that this difference can be made to decay as an arbitrarily large power of $t$ as $t \to 0$. Moreover, an inhomogeneous term is introduced to the energy estimate, which can also be made to decay as an arbitrarily large power of $t$ as $t \to 0$. This is what allows us to close the energy estimate. Existence of solutions to the system on an interval of the form $(0,T]$ then follows as a consequence of the energy estimate. The next step is then to prove that in fact $K$ is the Weingarten map, $\theta = \tr K$ and that the constructed solution is indeed a solution to Einstein's equations. All of these things can be accomplished by using that the difference between the solution and the approximate solution decays as a large power of $t$. The constructed solutions are in principle of finite regularity. In order to prove that there is a smooth solution, we use a uniqueness statement for the wave transport system. So far, this part of the proof follows closely the arguments of \cite{fournodavlos_asymptotically_2023}, with the added difficulties of dealing with the scalar field. 

After constructing the solution, we already obtain convergence of $\K$, $\Phi$ and $\Psi$ to the initial data. The only thing left to finish the proof of Theorem~\ref{main existence theorem}, is to prove convergence of $\h$. In order to do so, we obtain convergence estimates for the eigenvalues and the eigenprojections of $\K$. This is done through an application of the implicit function theorem and some perturbation theory. It is then possible to construct a frame $\{E_i\}$ of eigenvectors of $\K$, which is orthonormal with respect to $\h$ and converges to the frame $\{e_i\}$ as $t \to 0$, thus showing convergence of $\h$. We stress that being able to construct the frame $\{E_i\}$ depends crucially on the improved off-diagonal estimates for $\K$ in terms of the frame $\{e_i\}$, which are inherited by the $E_i$.

\paragraph{Detailed asymptotics and uniqueness.}

The proof of Theorem~\ref{asymptotics of the frame} consists of two steps. First, starting from a locally Gaussian development of the initial data, we need to show that the mean curvature satisfies an estimate of the form $t\theta - 1 = O(t^\delta)$. This is accomplished by using Einstein's equations, with arguments resembling those of \cite{ringstrom_initial_2022-1}. As mentioned before, the differences are due to our definitions of developments of initial data not coinciding; see Remark~\ref{fournodavlos luk ansatz remark}. Moreover, the presence of the potential introduces further difficulties in our setting. The remaining step is to construct the frame $\{E_i\}$ of eigenvectors of $\K$. Showing existence of the frame works similarly as in the proof of Theorem~\ref{main existence theorem}, but it is not immediately clear whether the off-diagonal improvements also hold in this case. This is due to the fact that, a priori, we do not have any additional information about the off-diagonal components of $\K$ in terms of $\{e_i\}$. In order to get the improvements, we use Einstein's equations to deduce evolution equations for the frame $\{E_i\}$. It is then possible to use the evolution equations to iteratively improve on the estimates, until we obtain the desired ones. One thing to point out is that the system of equations for the $E_i$ presents a potential loss of derivatives, but this is not an issue for us since from the beginning we already have estimates for all derivatives of the $E_i$. 

Finally, Theorem~\ref{main uniqueness theorem} follows from Theorem~\ref{asymptotics of the frame}. The idea is that the detailed asymptotics allow us to show that the uniqueness statement for the wave transport system, used in the proof of Theorem~\ref{main existence theorem}, is applicable to any locally Gaussian development. One thing to point out is that we are not able to obtain a uniqueness statement analog to Theorem~\ref{main uniqueness theorem} in finite regularity. This is due to the following fact. Given one set of initial data on the singularity, we need to have the detailed asymptotics up to some finite degree of regularity to prove uniqueness. However, how many derivatives are needed depends on the initial data under consideration. Hence, the only clean uniqueness result that we can prove is in the smooth setting.

\subsection{Outline}

The paper is organized as follows. In Section~\ref{sequence of approximate solutions} we construct the sequence of approximate solutions to Einstein's equations. The main result of this section is Theorem~\ref{approximate solutions}. In Subsections~\ref{the velocity dominated solution}--\ref{general results for odes} we introduce the necessary setup for the iteration. The sequence is defined in Subsections~\ref{the approximate weingarten map}--\ref{the scalar field}. Then Subsections~\ref{comparing with the weingarten map} and \ref{the sequence satisfies the eqs asymptotically} are dedicated to obtaining the necessary estimates for the error terms of the sequence.

Proceeding with the proof of Theorem~\ref{main existence theorem}, in Section~\ref{construction of the solution} we construct the required solution to Einstein's equations. Existence of a solution to the wave transport system is proven in Subsections~\ref{preliminary estimates}--\ref{finishing the construction}. In Subsection~\ref{the constructed solution solves einsteins eqs} we prove that the solution we construct indeed solves Einstein's equations. Smoothness of the solution is then proved in Subsection~\ref{the constructed solution is smooth}. Finally, Subsection~\ref{asymptotics for the expansion normalized metric} is dedicated to obtaining convergence of $\h$, thus finishing the proof of Theorem~\ref{main existence theorem}.

In Section~\ref{detailed asymptotics and uniqueness} we prove Theorems~\ref{asymptotics of the frame} and \ref{main uniqueness theorem}. Subsections~\ref{asymptotics for the mean curvature section} and \ref{asymptotics for the frame of eigenvectors} are dedicated to obtaining the required estimates for the mean curvature and constructing the frame of eigenvectors of $\K$. In Subsection~\ref{proofs of detailed asymptotics and uniqueness} we conclude the proofs of Theorems~\ref{asymptotics of the frame} and \ref{main uniqueness theorem}.

Finally, in Appendix~\ref{appendix}, we introduce our conventions regarding notation for constants, norms of tensors, The notation $\lie_{\p_t}T$ for $T$ a one parameter family of tensors on $\Sigma$, and our use of the notation $\sharp$ for raising indices.

\subsection*{Acknowledgements}

The author would like to thank Hans Ringström for suggesting the topic, the helpful discussions and the comments on the manuscript; and Oliver Petersen for the helpful discussions. This research was funded by the Swedish Research Council, dnr. 2017-03863 and 2022-03053; and supported by foundations managed by The Royal Swedish Academy of Sciences.

\section{Sequence of approximate solutions} \label{sequence of approximate solutions}

Throughout, we use the notation $\langle \xi \rangle := \sqrt{1 + \xi^2}$ for $\xi \in \R$. In this section, we set out to prove the following result.

\begin{theorem} \label{approximate solutions}
    Let $(\Sigma, \ho, \Ko, \phio, \psio)$ be initial data on the singularity and $V$ an admissible potential. For every non-negative integer $n$, there is a $t_n > 0$, depending only on the initial data and the potential, and for $t \in (0,t_n]$ a smooth one parameter family of Riemannian metrics $h_n(t)$ on $\Sigma$ and a smooth function $\s_n$ on $(0,t_n] \times \Sigma$ such that the following holds.
    
    \paragraph{Convergence to initial data:} Define $g_n := -dt \otimes dt + h_n$. Then $g_n$ is a Lorentzian metric on $(0,t_n] \times \Sigma$. Let $K_n$ denote the Weingarten map of the $\Sigma_t$ hypersurfaces with respect to the metric $g_n$. Also define 
    \[
    \bar\h_n := h_n(t^{-\Ko}(\,\cdot\,),t^{-\Ko}(\,\cdot\,)), \quad \bar\Psi_n := t\p_t \s_n, \quad \bar\Phi_n := \s_n - \bar\Psi_n \ln t. 
    \]
    Then there are constants $C_{m,r,n}$ and $C_{m,n}$ such that
    \begin{align*}
        |D^m (t\lie_{\p_t})^r( \bar\h_n - \ho )|_{\ho} &\leq C_{m,r,n} \langle \ln t \rangle^{m+2} t^{2\varepsilon},\\
        |D^m (t\lie_{\p_t})^r( tK_n - \Ko )|_{\ho} &\leq C_{m,r,n} \langle \ln t \rangle^{m+2} t^{2\varepsilon},\\
        |D^m (t\p_t)^r( \bar\Psi_n - \psio )|_{\ho} &\leq C_{m,r,n} \langle \ln t \rangle^{m+2} t^{2\varepsilon},\\
        |D^m ( \bar\Phi_n - \phio )|_{\ho} &\leq C_{m,n} \langle \ln t \rangle^{m+3} t^{2\varepsilon}.
    \end{align*}
    Moreover, for $i \neq k$, $x \in D_+$ and $y \in D_-$,
    \begin{align*}
        \big|D^m (t\p_t)^r\big( \bar\h_n(e_i,e_k) \big)\big|_{\ho}(x) &\leq C_{m,r,n} \langle \ln t \rangle^{m+2} t^{2\varepsilon + (p_i + p_k - 2p_1)(x)},\\
        \big|D^m (t\p_t)^r\big( \bar\h_n(e_i,e_k) \big)\big|_{\ho}(y) &\leq C_{m,r,n} \langle \ln t \rangle^{m+2} t^{2\varepsilon + |p_i - p_k|(y)},\\
        \big|D^m (t\p_t)^r\big( tK_n(e_i,\omega^k) \big)\big|_{\ho}(x) &\leq C_{m,r,n} \langle \ln t \rangle^{m+2} t^{2\varepsilon + 2(p_i - p_1)(x)},\\
        \big|D^m (t\p_t)^r\big( tK_n(e_i,\omega^k) \big)\big|_{\ho}(y) &\leq C_{m,r,n} \langle \ln t \rangle^{m+2} t^{2\varepsilon} \min\{ 1, t^{2(p_i - p_k)(y)} \}.
    \end{align*}

    \paragraph{Einstein's equations approximately satisfied:} Define
    \[
    E_n := \ric_n - d\s_n \otimes d\s_n - (V \circ \s_n)g_n,
    \]
    where $\ric_n$ denotes the Ricci tensor of $g_n$. Also, for $X, Y \in \mathfrak{X}(\Sigma)$, define the one parameter family of $(1,1)$-tensors $\ce_n$ and the one parameter family of one forms $\cm_n$ on $\Sigma$ by
    \[
    h_n(\ce_n(X),Y) := E_n(X,Y), \qquad \cm_n(X) := E_n(\p_t,X).
    \]
    Then there are non-negative integers $N_n$, depending only on $n$, such that
    \begin{align*}
        t^2|D^m( t\lie_{\partial_t})^r\ce_n |_{\ho} &\leq C_{m,r,n} \langle \ln t \rangle^{m+N_n} t^{2(n+1)\varepsilon},\\
        t^2\big|D^m (t\p_t)^r\big(E_n(\partial_t,\partial_t)\big)\big|_{\ho} + t|D^m (t\lie_{\p_t})^r \cm_n|_{\ho} &\leq C_{m,r,n} \langle \ln t \rangle^{m+N_n} t^{2(n+1)\varepsilon},\\
        t^2|D^m (t\p_t)^r( \Box_{g_n} \s_n - V' \circ \s_n )|_{\ho} &\leq C_{m,r,n} \langle \ln t \rangle^{m+N_n} t^{2(n+1)\varepsilon}.
    \end{align*}
\end{theorem}

The conclusion of the proof of Theorem~\ref{approximate solutions} is to be found at the end of this section. We will begin by defining the velocity dominated solution associated with the initial data in Subsection~\ref{the velocity dominated solution}. Then, in Subsection~\ref{estimates for ricci etc}, we obtain some general estimates for the spatial Ricci curvature, the Levi-Civita connection of the spatial metric and some quantities related with the scalar field. Next, in Subsection~\ref{general results for odes}, we establish two general results for ODEs which will be the main tools used for the construction. After that, the remaining subsections will be devoted to constructing the sequence, Subsections~\ref{the approximate weingarten map}--\ref{the scalar field}, and estimating the error terms, Subsections~\ref{comparing with the weingarten map} and \ref{the sequence satisfies the eqs asymptotically}.

\subsection{The velocity dominated solution} \label{the velocity dominated solution}

Before starting, we need to express the Ricci tensor in terms of the foliation that we will use. Consider a metric $g = -dt \otimes dt + h$. The Levi-Civita connection, Riemann tensor, Ricci tensor and scalar curvature of the family of induced metrics $h$ will be denoted by $\sn$, $\bar R$, $\sric$ and $\bar S$ respectively, as opposed to the corresponding objects associated with the spacetime metric $g$, which are denoted without the bars.

\begin{proposition} \label{ricci}
    Let $\Sigma$ be a manifold and consider a metric $g = -dt \otimes dt + h$ on $(0,T] \times \Sigma$. Then
    \begin{align*}
        \ric(\partial_t,\partial_t) &= -\partial_t \theta - |k|_h^2,\\
        \ric(\partial_t,X) &= \diver_h k (X) - d\theta(X),\\
        \Ric(X) &= \mathcal{L}_{\partial_t} K(X) + \overline{\ric}^{\sharp}(X) + \theta K(X),
    \end{align*}
    where $\Ric$ is the one parameter family of (1,1)-tensors on $\Sigma$ defined by
    \[
    h(\Ric(X),Y) := \ric(X,Y)
    \]
    and $X, Y \in \mathfrak{X}(\Sigma)$.
\end{proposition}

\begin{proof}
    This is a special case of \cite[Chapter 6, (3.20)--(3.22)]{choquet-bruhat_general_2009}, by setting the lapse function to $1$ and the shift vector field to zero.
\end{proof}

\begin{remark}
    Given a metric $g = -dt \otimes dt + h$ and a function $\s$ on $(0,T] \times \Sigma$, the notation $d\s$ could either mean the spacetime differential of $\s$ or the ``spatial" differential on a $\Sigma_t$ hypersurface. It will be clear from the context which one of the two is meant. $\n \s$ and $\sn \s$ will denote the gradient of $\s$ with respect to $g$ and $h$ respectively. There is the possibility of confusion, since when extending the connection $\n$ to tensors one defines $\n \s := d\s$. However, which one is meant should be clear from the context.
\end{remark}

Using Proposition~\ref{ricci}, we can formulate the Einstein--nonlinear scalar field equations in terms of the foliation $g = -dt \otimes dt + h$ as follows. We have the \emph{constraint equations},
\begin{subequations}
\begin{align}
    \bar{S} + \theta^2 - \tr K^2 &= (\p_t \s)^2 + |d\s|_h^2 + 2V \circ \s,\label{hamiltonian constraint}\\
    \diver_h K - d\theta &= (\p_t \s)d\s.\label{momentum constraint}
\end{align}
\end{subequations}  
Equation \eqref{hamiltonian constraint} is called the \emph{Hamiltonian constraint equation} and \eqref{momentum constraint} is called the \emph{momentum constraint equation}. And we have the \emph{evolution equations for $h$, $K$ and $\s$},
\begin{subequations} \label{adm equations}
\begin{align}
    \lie_{\p_t} h &= 2k,\label{evolution equation for h}\\
    \lie_{\p_t} K + \sric^{\sharp} + \theta K &= d\s \otimes  \sn\s + (V \circ \s)I,\label{evolution equation for k}\\
    -\p_t^2 \s + \Delta_h \s - \theta \p_t \s &= V' \circ \s,\label{evolution equation for phi}
\end{align}
\end{subequations}
where $I$ denotes the identity $(1,1)$-tensor field on $\Sigma$.

We are now ready to introduce the velocity dominated solution. These spacetimes can be thought of as the $0$-th order approximation of the corresponding solutions to Einstein's equations. 

\begin{definition}
    Given initial data on the singularity $(\Sigma, \ho, \Ko, \phio, \psio)$, the associated \emph{velocity dominated solution} is the triple $((0,\infty) \times \Sigma,g_0,\s_0)$, where the metric $g_0 = -dt \otimes dt + h_0$ and the function $\s_0$ are defined by
    \[
    h_0(X,Y) := \ho(t^{\Ko}(X),t^{\Ko}(Y)), \qquad \s_0 := \psio \ln t + \phio,
    \]
    for $X,Y \in \mfx(\Sigma)$.
\end{definition}

Let $(\Sigma, \ho, \Ko, \phio, \psio)$ be initial data on the singularity and consider the associated velocity dominated solution $((0,\infty) \times \Sigma,g_0,\s_0)$. In terms of the frame $\{e_i\}$, we have $t^{\Ko}(e_i) = t^{p_i}e_i$. Hence we can write the metric $g_0$ as
\[
g_0 = -dt \otimes dt + \sum_i t^{2p_i} \omega^i \otimes \omega^i;
\]
cf. \cite{lifshitz_investigations_1963}. If $k_0$ is the second fundamental form of the $\Sigma_t$ hypersurfaces, then
\[
k_0(e_i,e_k) = \frac{1}{2}\partial_t \big( g_0(e_i,e_k)\big) = \frac{p_i}{t} t^{2p_i} \delta_{ik}.
\]
We see that the Weingarten map $K_0$ takes the form
\[
K_0= \sum_i \frac{p_i}{t}\omega^i \otimes e_i,
\]
and the mean curvature $\theta_0$ is equal to $t^{-1}$. Thus if $\K_0$ is the expansion normalized Weingarten map, we immediately see that $\K_0 = \Ko$. Moreover, by definition of $h_0$, the expansion normalized induced metric $\h_0$ is given by $\h_0 = \ho$. Turning our attention to the scalar field, it is clear that $\Psi_0 = \psio$ and $\Phi_0 = \phio$. So we see that all the expansion normalized quantities associated with the velocity dominated solution, are constant in time and equal to the corresponding initial data on the singularity.

Note that $K_0$ and $\s_0$ satisfy the following equations,
\[
\theta_0^2 - \tr K_0^2 = (\p_t\s_0)^2, \qquad \lie_{\p_t} K_0 + \theta_0 K_0 = 0, \qquad \p_t^2 \s_0 + \theta_0 \p_t \s_0 = 0.
\]
These correspond to the VTD equations associated with \eqref{hamiltonian constraint}, \eqref{evolution equation for k} and \eqref{evolution equation for phi}; if, in addition to the spatial derivative terms, we also drop the potential terms. We do so since we are interested only in the situation where the potential yields a negligible contribution to the asymptotics. In Lemma~\ref{velocity dominated solution} below we verify that, as a consequence of Condition 2 in Definition~\ref{initial data}, the momentum constraint \eqref{momentum constraint} is satisfied by the velocity dominated solution.

Of course, in general, the velocity dominated solution is not going to be a solution to the Einstein--nonlinear scalar field equations with potential $V$. So in order to justify why it is a reasonable model spacetime for the situation that we are interested in, we need to verify that it is an approximate solution to Einstein's equations as $t \to 0$, in an appropriately normalized sense. 

\begin{lemma} \label{velocity dominated solution}
    Let $(\Sigma, \ho, \Ko, \phio, \psio)$ be initial data on the singularity, let $V$ be an admissible potential and let $((0,\infty) \times \Sigma,g_0,\s_0)$ be the associated velocity dominated solution. Let $\ric_0$ be the Ricci tensor of $g_0$ and define
    \[
    E_0 := \ric_0 - d\s_0 \otimes d\s_0 - (V \circ \s_0)g_0.
    \]
    Also define the one parameter family of $(1,1)$-tensors $\ce_0$ and the one parameter family of one forms $\cm_0$ on $\Sigma$ by  
    \[
    h(\ce_0(X),Y) := E_0(X,Y), \qquad \cm_0(X) := E_0(\p_t,X),
    \]
    for $X,Y \in \mathfrak{X}(\Sigma)$. Then $\cm_0 = 0$, and there is a $T > 0$ and constants $C_{m,r}$ such that
    \[
    \begin{split}
        t^2\big|D^m (t\p_t)^r\big(E_0(\p_t,\p_t)\big)\big|_\ho + t^2|D^m (t\lie_{\p_t})^r \ce_0|_\ho &\leq C_{m,r} \langle \ln t \rangle^{m+2}t^{2\varepsilon},\\
        t^2|D^m (t\p_t)^r(\Box_{g_0} \s_0 - V' \circ \s_0)|_\ho &\leq C_{m,r} \langle \ln t \rangle^{m+2}t^{2\varepsilon},
    \end{split}
    \]
    for $t \in (0,T]$.
\end{lemma}

\begin{proof}
    We begin with the momentum constraint $\cm_0$. If $\sn^{(0)}$ is the Levi-Civita connection of $h_0$, then
    \[
    \begin{split}
        \omega^{\ell}(\sn^{(0)}_{e_i} e_k) &= (\ln t) e_i(p_k) \delta_{k\ell} + (\ln t) e_k(p_i) \delta_{i\ell} - (\ln t) t^{2(p_i - p_{\ell})} e_{\ell}(p_i) \delta_{ik}\\
        &\phantom{=} - \frac{1}{2}t^{2(p_i-p_{\ell})} \gamma_{k\ell}^i - \frac{1}{2} t^{2(p_k - p_{\ell})} \gamma_{i\ell}^k + \frac{1}{2} \gamma_{ik}^{\ell}.
    \end{split}
    \]
    Hence, by recalling that $\tr \Ko^2 + \psio^2 = 1$,
    \[
    \begin{split}
        \ric_0(\partial_t,e_i) = \diver_{h_0} k_0 (e_i) &= \frac{1}{t} e_i(p_i) + \sum_{k \neq i} \frac{1}{t} \big((\ln t) e_i(p_k) + \gamma_{ki}^k \big)(p_i - p_k)\\
        &= \frac{1}{t}\bigg( \diver_{\ho} \Ko (e_i) - \frac{1}{2} (\ln t) e_i(\tr \Ko^2) \bigg)\\
        &= \frac{1}{t} \bigg( \diver_{\ho} \Ko (e_i) + (\ln t) \psio e_i(\psio) \bigg)\\
        &= \frac{1}{t} \bigg( (\diver_{\ho} \Ko - \psio d\phio) (e_i) + t\p_t \s_0 d\s_0(e_i) \bigg).
    \end{split} 
    \]
    Implying $t\cm_0 = \diver_\ho \Ko - \psio d\phio = 0$, by Definition~\ref{initial data}.
    
    Next, consider $E_0(\p_t,\p_t)$. First note that
    \[
    \ric_0(\partial_t,\partial_t) = -\partial_t \theta_0 - |k_0|_{h_0}^2 = \frac{1}{t^2}(1 - \tr \Ko^2) = \frac{1}{t^2} \psio^2 = (\p_t \s_0)^2.
    \]
    Hence $E_0(\p_t,\p_t) = V \circ \s_0$. So we just need to derive an appropriate bound for $V \circ \s_0$ to obtain the result. For $\ce_0$, recall that $\mathcal{L}_{\partial_t} K_0 + \theta_0 K_0 = 0$, thus $\Ric_0 = \overline{\ric}_0^{\sharp}$. We then have
    \[
    \ce_0 = \sric^{\sharp}_0 - d\s_0 \otimes \sn \s_0 - (V \circ \s_0)I,
    \]
    where $\sn \s_0$ denotes the gradient of $\s_0$ with respect to $h_0$. Hence, to obtain the result for $\ce_0$, we need appropriate estimates for $\sric_0^\sharp$, $d\s_0 \otimes \sn \s_0$ and $V \circ \s_0$. Finally, $\s_0$. Since it satisfies $\p_t^2 \s_0 + \theta_0 \p_t \s_0 = 0$, it is enough to obtain estimates for $\Delta_{h_0}\s_0$ and $V' \circ \s_0$. All of the estimates needed to obtain the conclusions follow from Lemma~\ref{decay of ricci} below.
\end{proof}

\subsection{Estimates for the spatial Ricci curvature, the spatial connection and the scalar field} \label{estimates for ricci etc}

Now we derive some general estimates which shall be used extensively in what follows.

\begin{remark}
    Since we are interested in the asymptotic behavior as $t \to 0$, when we consider manifolds of the form $(0,T] \times \Sigma$, there is no loss of generality in assuming that $T \leq 1$, and we do so in what follows.
\end{remark}

\begin{lemma} \label{general estimate for ricci lemma}
    Let $(\Sigma,\ho,\Ko,\phio,\psio)$ be initial data on the singularity and let $h(t)$, for $t \in (0,T]$, be a smooth one parameter family of Riemannian metrics on $\Sigma$. Suppose there is a frame $\{E_i(t)\}$, with dual frame $\{\eta^i(t)\}$, and constants $C_{m,r}$ such that
    \[
    |D^m (t\lie_{\p_t})^r E_i|_\ho + |D^m (t\lie_{\p_t})^r \eta^i|_\ho \leq C_{m,r} \langle \ln t \rangle^m. 
    \]
    Moreover, assume that $h^{ik} := h^{-1}(\eta^i,\eta^k)$ satisfy
    \begin{equation*}
        |D^m(t\p_t)^r (h^{ik})|_\ho \leq C_{m,r}\langle \ln t \rangle^m t^{-2p_{\min\{i,k\}}};
    \end{equation*}
    and that $\Gamma_{ik}^\ell := \eta^\ell(\sn_{E_i} E_k)$ satisfy
    \begin{equation} \label{hypothesis on the connection}
    \begin{split}
        |D^m(t\p_t)^r (\Gamma_{ii}^\ell) |_\ho &\leq C_{m,r}\langle \ln t \rangle^{m+1} t^{2(p_i - p_\ell)},\\
        |D^m(t\p_t)^r (\Gamma_{ik}^i) |_\ho + |D^m(t\p_t)^r (\Gamma_{ik}^k) |_\ho &\leq C_{m,r}\langle \ln t \rangle^{m+1},\\
        |D^m (t\p_t)^r(h^{ab} E_\alpha \Gamma_{ab}^k) |_\ho &\leq C_{m,r}\langle \ln t \rangle^{m+|\alpha|+1}t^{-2p_k},
    \end{split}
    \end{equation}
    (no summation on $i$ or $k$), where $\sn$ is the Levi-Civita connection of $h$ and $\alpha$ is a multiindex (see Definition~\ref{multiindex definition} below for our conventions regarding multiindices) with $|\alpha| \leq 1$. Then, if $\sric$ is the Ricci tensor of $h$ and the $\lambda_{ik}^\ell$ are defined by $[E_i,E_k] = \lambda_{ik}^\ell E_\ell$, there are constants $C_{m,r}$ such that
    \begin{subequations} \label{general estimates for ricci}
    \begin{align}
        \big|D^m (t\p_t)^r \big(\sric^\sharp(E_i,\eta^i) + \Lambda_{ik\ell}\big)\big|_\ho &\leq C_{m,r} \langle \ln t \rangle^{m+2} t^{-2p_3},\\
        \big|D^m (t\p_t)^r \big(\sric^\sharp(E_i,\eta^k) + \Upsilon_{ik\ell}\big)\big|_\ho &\leq C_{m,r} \langle \ln t \rangle^{m+2} t^{-2p_3} \min\{1,t^{2(p_i - p_k)}\}, \label{general off diagonal estimate for ricci}
    \end{align}
    \end{subequations}
    where $i$, $k$ and $\ell$ are distinct (no summation on $i$), and
    \begin{align*}
        \Lambda_{ik\ell} &:= h^{\ell i} \Gamma_{ii}^k \Gamma_{\ell k}^i + h^{\ell k} \Gamma_{ik}^k \Gamma_{\ell k}^i + h^{a\ell} \Gamma_{i\ell}^k \Gamma_{ak}^i + h^{ki} \Gamma_{ii}^\ell \Gamma_{k\ell}^i\\
        &\quad + h^{ak} \Gamma_{ik}^\ell \Gamma_{a\ell}^i + h^{k\ell} \Gamma_{i\ell}^\ell \Gamma_{k\ell}^i + h^{a\ell} \lambda_{ia}^k \Gamma_{k\ell}^i + h^{ak} \lambda_{ia}^\ell \Gamma_{\ell k}^i,\\
        \Upsilon_{ik\ell} &:= h^{a\ell} E_a \Gamma_{i\ell}^k - h^{ab} \Gamma_{ab}^\ell \Gamma_{i\ell}^k + h^{\ell a} \Gamma_{ia}^i \Gamma_{\ell i}^k + h^{a\ell}\Gamma_{i\ell}^k \Gamma_{ak}^k + h^{ii} \Gamma_{ii}^\ell \Gamma_{i\ell}^k\\
        &\quad + h^{ak}\Gamma_{ik}^\ell \Gamma_{a\ell}^k + h^{i\ell} \Gamma_{i\ell}^\ell \Gamma_{i\ell}^k + h^{\ell i} \lambda_{i\ell}^k \Gamma_{ki}^k + h^{ai} \lambda_{ia}^\ell \Gamma_{\ell i}^k + h^{a\ell} \lambda_{ia}^i \Gamma_{i\ell}^k + h^{\ell\ell} \lambda_{i\ell}^k \Gamma_{k\ell}^k.
    \end{align*}
\end{lemma}

\begin{remark}
    The assumptions on the time derivatives are only necessary to obtain the conclusions for the time derivatives. So, if we only knew the assumptions to hold for all $m$ and $r \leq R$, then we would still obtain the conclusions for all $m$ and $r \leq R$.
\end{remark}

\begin{proof}
    For this proof, let $i$, $k$ and $\ell$ denote fixed indices, so that there is no summation over them when they are repeated. First, note that our assumptions on the $E_i$ and the $\eta^i$ imply
    \[
    |D^m(t\p_t)^r(\lambda_{ik}^\ell)|_\ho \leq C_{m,r}\langle \ln t \rangle^{m+1}.
    \]
    Now we move on to $\sric$. We have
    \[
    \begin{split}
        \sric^{\sharp}(E_i) &= h^{ab} \bar R(E_i,E_{a})E_b\\
        &= h^{ab}( \sn_{E_i} \sn_{E_{a}} E_b - \sn_{E_{a}} \sn_{E_i} E_b - \sn_{[E_i,E_{a}]} E_b )\\
        &= h^{ab}\big( (E_i \Gamma_{ab}^c)E_c - (E_{a}\Gamma_{ib}^c) E_c + \Gamma_{ab}^c \Gamma_{ic}^d E_d - \Gamma_{ib}^c \Gamma_{a c}^d E_d - \lambda_{i a}^c \Gamma_{cb}^d E_d \big).
    \end{split}
    \]
    Thus
    \begin{equation} \label{ricci in terms of the frame}
        \sric^{\sharp}(E_i,\eta^k) = \underbrace{h^{ab} E_i \Gamma_{a b}^k}_{\mathrm{I}} - \underbrace{h^{ab} E_{a} \Gamma_{ib}^k}_{\mathrm{II}} + \underbrace{ h^{ab} \Gamma_{ab}^c \Gamma_{ic}^k}_{\mathrm{III}} - \underbrace{h^{ab} \Gamma_{ib}^c \Gamma_{a c}^k}_{\mathrm{IV}} - \underbrace{h^{ab} \lambda_{i a}^c \Gamma_{cb}^k}_{\mathrm{V}}.
    \end{equation}
    Note that our assumptions on the $\Gamma_{ab}^c$ directly give us control over $\mathrm{I}$ and the first two factors in $\mathrm{III}$. Forgetting about $\mathrm{V}$ for the moment, the idea is that we want to single out all the cases where there is a $\Gamma_{ab}^c$ with all three indices being distinct in $\mathrm{II}$ and $\mathrm{IV}$, and in the third factor of $\mathrm{III}$, while we estimate the rest. To that end, we first focus on the following terms: those with repeated indices in the $\Gamma_{ab}^c$ appearing in $\mathrm{II}$, in the third factor of $\mathrm{III}$, and in the second and third factors of $\mathrm{IV}$ (here we look for repeated indices in both factors at the same time). In the case $i=k$, we see that $D^m(t\p_t)^r$ of the corresponding terms is bounded by
    \[
    C_{m,r}\langle \ln t \rangle^{m+2} t^{-2p_3}.
    \]
    Furthermore, if $i \neq k$, then $D^m(t\p_t)^r$ of the terms under consideration, except one, are bounded by
    \[
    C_{m,r}\langle \ln t \rangle^{m+2} t^{-2p_k}.
    \]
    Whereas the remaining term, which is $h^{ki} \Gamma_{ii}^\ell \Gamma_{k \ell}^k$ for $i$, $k$ and $\ell$ distinct (this comes from $\mathrm{IV}$), satisfies
    \[
    |D^m(t\p_t)^r( h^{ki} \Gamma_{ii}^\ell \Gamma_{k \ell}^k )|_\ho \leq C_{m,r} \langle \ln t \rangle^{m+2} t^{2( p_i - p_\ell - p_{\min\{i,k\}} )}.
    \]
    So, altogether, the terms under consideration are bounded by
    \[
    C_{m,r} \langle \ln t \rangle^{m+2} t^{-2p_3} \min\{1,t^{2(p_i - p_k)}\}.
    \]
    Now we look at $\mathrm{V}$. Returning to the case $i=k$, we see that the terms in $\mathrm{V}$ which present repeated indices in the $\Gamma_{ab}^c$ are also bounded by the same expression. Hence, if $i$, $k$ and $\ell$ are distinct,
    \[
    \begin{split}
        &\big|D^m(t\p_t)^r\big( \sric^\sharp(E_i,\eta^i) + h^{\ell i} \Gamma_{ii}^k \Gamma_{\ell k}^i + h^{\ell k} \Gamma_{ik}^k \Gamma_{\ell k}^i + h^{a\ell} \Gamma_{i\ell}^k \Gamma_{ak}^i\\
        &+ h^{ki} \Gamma_{ii}^\ell \Gamma_{k\ell}^i + h^{ak} \Gamma_{ik}^\ell \Gamma_{a\ell}^i + h^{k\ell} \Gamma_{i\ell}^\ell \Gamma_{k\ell}^i + h^{a\ell} \lambda_{ia}^k \Gamma_{k\ell}^i + h^{ak} \lambda_{ia}^\ell \Gamma_{\ell k}^i \big)\big|_\ho \leq C_{m,r}\langle \ln t \rangle^{m+2} t^{-2p_3},
    \end{split}
    \]
    which is what we wanted to prove. On the other hand, if $i \neq k$ in \eqref{ricci in terms of the frame}, we get the desired bound for the terms in $\mathrm{V}$ when $b=k$ or $c=b$. Hence,
    \[
    \begin{split}
        &\big|D^m(t\p_t)^r \big(\mathrm{V} - h^{ai} \lambda_{ia}^k \Gamma_{ki}^k - h^{ai} \lambda_{ia}^\ell \Gamma_{\ell i}^k - h^{a\ell} \lambda_{ia}^i \Gamma_{i\ell}^k - h^{a\ell} \lambda_{ia}^k \Gamma_{k\ell}^k  \big)\big|_\ho\\
        &\hspace{5cm} \leq C_{m,r} \langle \ln t \rangle^{m+2} t^{-2p_3}\min\{1,t^{2(p_i-p_k)}\}.
    \end{split}
    \]
    Note that the first and the last terms to the right of $\mathrm{V}$ vanish when $a=i$ and satisfy the desired bound when $a=k$. Thus, we can go back to \eqref{ricci in terms of the frame} to obtain
    \[
    \begin{split}
        &\big|D^m(t\p_t)^r\big( \sric^\sharp(E_i,\eta^k) + h^{a\ell} E_a \Gamma_{i\ell}^k - h^{ab} \Gamma_{ab}^\ell \Gamma_{i\ell}^k + h^{\ell a} \Gamma_{ia}^i \Gamma_{\ell i}^k + h^{a\ell}\Gamma_{i\ell}^k \Gamma_{ak}^k + h^{ii} \Gamma_{ii}^\ell \Gamma_{i\ell}^k\\
        &+ h^{ak}\Gamma_{ik}^\ell \Gamma_{a\ell}^k + h^{i\ell} \Gamma_{i\ell}^\ell \Gamma_{i\ell}^k + h^{\ell i} \lambda_{i\ell}^k \Gamma_{ki}^k + h^{ai} \lambda_{ia}^\ell \Gamma_{\ell i}^k + h^{a\ell} \lambda_{ia}^i \Gamma_{i\ell}^k + h^{\ell\ell} \lambda_{i\ell}^k \Gamma_{k\ell}^k \big)\big|_\ho\\
        &\hspace{6cm} \leq C_{m,r}\langle \ln t \rangle^{m+2} t^{-2p_3} \min\{1,t^{2(p_i - p_k)}\}.
    \end{split}
    \]
    The lemma follows.
\end{proof}

\begin{lemma} \label{decay of ricci}
    Let $(\Sigma, \ho, \Ko, \phio, \psio)$ be initial data on the singularity and consider an admissible potential $V$, a metric $g = -dt \otimes dt + h$ and a function $\s$ on $(0,T] \times \Sigma$. Assume that there are constants $C$, $C_m$, $C_{m,r}$, $\eta > 0$ and $\delta > 0$ such that
    \begin{align*}
        |D^m (t\mathcal{L}_{\partial_t})^r \bar \h|_{\ho} &\leq C_{m,r} \langle \ln t \rangle^m, & |D^m (t\p_t)^r \bar \Psi |_{\ho} &\leq C_{m,r}, & |D^m\bar \Phi|_{\ho} &\leq C_{m},\\
        |\bar \h(e_i,e_i)| &\geq \eta, & |\bar\Psi - \psio| &\leq Ct^\delta.
    \end{align*}
    Moreover, if $i \neq k$, $x \in D_+$ and $y \in D_-$, assume that the following off-diagonal improvements of the estimates hold,
    \begin{equation*}
    \begin{split}
        \big|D^m (t\p_t)^r \big(\bar\h(e_i,e_k) \big)\big|_{\ho}(x) &\leq C_{m,r} \langle \ln t \rangle^m t^{(p_i + p_k - 2p_1)(x)},\\
        \big|D^m (t\p_t)^r \big(\bar\h(e_i,e_k) \big)\big|_{\ho}(y) &\leq C_{m,r} \langle \ln t \rangle^m t^{|p_i - p_k|(y)}.\\
    \end{split}
    \end{equation*}
    Then, by taking $T$ smaller if necessary, there are constants $C_{m,r}$, depending only on $\eta$, $\delta$, the initial data and the potential, such that the following holds. Let $h_{ik} := h(e_i,e_k)$ and $h^{ik} := h^{-1}(\omega^i,\omega^k)$. Then
    \[
    \begin{split}
        |D^m(t\p_t)^r (h_{ik})|_\ho &\leq C_{m,r}\langle \ln t \rangle^m t^{2p_{\max\{i,k\}}},\\
        |D^m(t\p_t)^r (h^{ik})|_\ho &\leq C_{m,r}\langle \ln t \rangle^m t^{-2p_{\min\{i,k\}}}.
    \end{split}
    \]
    Moreover, for $x \in D_+$ the following improvements hold for $i \neq k$,
    \[
    \begin{split}
        |D^m(t\p_t)^r (h_{ik})|_\ho(x) &\leq C_{m,r}\langle \ln t \rangle^m t^{2(p_i+p_k-p_1)(x)},\\
        |D^m(t\p_t)^r (h^{ik})|_\ho(x) &\leq C_{m,r}\langle \ln t \rangle^m t^{-2p_1(x)}.
    \end{split}
    \]
    Let $\bar{\Gamma}_{ik}^{\ell} := \omega^{\ell}(\sn_{e_i} e_k)$, where $\sn$ is the Levi-Civita connection of $h$. If $x \in D_+$ and $y \in D_-$, then
    \begin{equation} \label{connection estimates}
    \begin{split}
        |D^m (t\p_t)^r (\bar{\Gamma}_{ii}^{\ell})|_\ho &\leq C_{m,r} \langle \ln t \rangle^{m+1} t^{2(p_i - p_{\ell})},\\
        |D^m (t\p_t)^r (\bar{\Gamma}_{ik}^i)| + |D^m (t\p_t)^r (\bar{\Gamma}_{ik}^k)|_\ho &\leq C_{m,r} \langle \ln t \rangle^{m+1},\\
        |D^m (t\p_t)^r (\bar{\Gamma}_{ik}^{\ell})|_\ho(x) &\leq C_{m,r} \langle \ln t \rangle^{m+1} t^{2(p_1 - p_{\ell})(x)},\\
        |D^m (t\p_t)^r (\bar{\Gamma}_{ik}^{\ell})|_\ho(y) &\leq C_{m,r} \langle \ln t \rangle^{m+1} t^{2(p_2 - p_{\ell})(y)}, 
    \end{split}
    \end{equation}
    where $i$, $k$ and $\ell$ are distinct in the last two inequalities (no summation on $i$ or $k$). Furthermore,
    \[
    t^2|D^m (t\mathcal{L}_{\partial_t})^r \sric^{\sharp}|_{\ho} \leq C_{m,r} \langle \ln t \rangle^{m+2} t^{2\mathring{\varepsilon}}
    \]
    and for $i \neq k$, $x \in D_+$ and $y \in D_-$, the following off-diagonal improvements hold,
    \begin{equation*}
    \begin{split}
        t^2\big|D^m (t\p_t)^r \big(\sric^{\sharp}(e_i,\omega^k)\big)\big|_{\ho}(x) &\leq C_{m,r} \langle \ln t \rangle^{m+2} t^{2\mathring{\varepsilon} + 2(p_i - p_1)(x)},\\
        t^2\big|D^m (t\p_t)^r \big(\sric^{\sharp}(e_i,\omega^k)\big)\big|_{\ho}(y) &\leq C_{m,r} \langle \ln t \rangle^{m+2} t^{2\mathring{\varepsilon}} \min\{ 1, t^{2(p_i - p_k)(y)} \}. 
    \end{split}
    \end{equation*}
    Finally,
    \[
    \begin{split}
        t^2 |D^m (t\lie_{\p_t})^r( d\s \otimes \sn \s )|_{\ho} + t^2 |D^m (t\p_t)^r ( \Delta_h \s )|_{\ho} &\leq C_{m,r} \langle \ln t \rangle^{m+2} t^{2\mathring{\varepsilon}},\\
        t^2 |D^m (t\p_t)^r( V\circ\s )|_\ho + t^2 |D^m (t\p_t)^r( V'\circ\s )|_\ho &\leq C_{m,r} \langle \ln t \rangle^m t^{2\varepsilon_V},
    \end{split}
    \]
    and for $i \neq k$, $x \in D_+$ and $y \in D_-$,
    \[
    \begin{split}
        t^2\big|D^m (t\p_t)^r \big((d\s \otimes \sn \s) (e_i,\omega^k)\big)\big|_{\ho}(x) &\leq C_{m,r} \langle \ln t \rangle^{m+2} t^{2\mathring{\varepsilon} + 2(p_i - p_1)(x)},\\
        t^2\big|D^m (t\p_t)^r \big(( d\s \otimes \sn \s) (e_i,\omega^k)\big)\big|_{\ho}(y) &\leq C_{m,r} \langle \ln t \rangle^{m+2} t^{2\mathring{\varepsilon}} \min\{ 1, t^{2(p_i - p_k)(y)} \}. 
    \end{split}
    \]
\end{lemma}

\begin{remark} \label{finite regularity}
    If we only had the assumptions for all $m \leq M$ and all $r \leq R$, then we would still obtain the conclusions for all $m \leq M-2$ and all $r \leq R$. Finally, the conclusions about the potential $V$ depend only on the assumptions on $\s$.
\end{remark}

\begin{proof}
    As in the proof of Lemma~\ref{general estimate for ricci lemma}, for this proof $i$, $k$ and $\ell$ denote fixed indices, so there is no summation over them when repeated. Also, let $x$ denote an element of $D_+$ and $y$ denote an element of $D_-$.
    
    \paragraph{The metric and the connection coefficients:} Note that $h_{ik} = t^{p_i + p_k} \bar\h(e_i,e_k)$, hence for every multiindex $\alpha$ of order $m$, 
    \[
    |e_{\alpha} (t\p_t)^r h_{ii}| \leq C_{m,r} \langle \ln t \rangle^m t^{2p_i}.
    \]
    For $i \neq k$, we consider each case separately. We have
    \[
    |e_{\alpha} (t\p_t)^r h_{ik}|(x) \leq C_{m,r} \langle \ln t \rangle^m t^{2(p_i + p_k - p_1)(x)}.
    \]
    On the other hand, 
    \[
    |e_{\alpha} (t\p_t)^r h_{ik}|(y) \leq C_{m,r} \langle \ln t \rangle^m t^{(p_i + p_k + |p_i - p_k|)(y)} \leq C_{m,r} \langle \ln t \rangle^m t^{2p_{\max\{i,k\}}(y)}.
    \]
    Note that, in particular, the estimates
    \[
    |e_\alpha (t\p_t)^r h_{ik}| \leq C_{m,r}\langle \ln t \rangle^m t^{2p_{\max\{i,k\}}}
    \]
    always hold. This in turn implies
    \[
    |e_\alpha (t\p_t)^r \det h| \leq C_{m,r}\langle \ln t \rangle^m t^2, \quad |\det h| \geq t^2 \Big( \textstyle \prod_i |\bar\h(e_i,e_i)| - Ct^{2\min\{ p_2 - p_1, p_3 - p_2 \}}\Big),
    \]
    which, along with the lower bounds on $|\bar\h(e_i,e_i)|$, and after taking $T$ smaller if necessary, yields
    \[
    t^r\bigg| e_{\alpha} (t\p_t)^r \bigg( \frac{1}{\det h} \bigg) \bigg| \leq C_{m,r}\langle \ln t \rangle^m t^{-2}.
    \]
    To estimate the dual metric, note that 
    \[
    h^{ik} = \frac{1}{\det h} \adj(h)_{ik},
    \]
    where $\adj(h)$ is the adjugate of the matrix with components $h_{ik}$. It can then be computed that
    \[
    |e_{\alpha} (t\p_t)^r h^{ii}| \leq C_{m,r}\langle \ln t \rangle^m t^{-2p_i}
    \]
    and, for $i \neq k$,
    \[
    |e_{\alpha} (t\p_t)^r h^{ik}|(x) \leq C_{m,r}\langle \ln t \rangle^m t^{-2p_1(x)}, \quad  |e_{\alpha} (t\p_t)^r h^{ik}|(y) \leq C_{m,r}\langle \ln t \rangle^m t^{-2p_{\min\{i,k\}}(y)}.
    \]
    In particular, the estimates
    \[
    |e_\alpha (t\p_t)^r h^{ik}| \leq C_{m,r}\langle \ln t \rangle^m t^{-2p_{\min\{i,k\}}}
    \]
    always hold. 
    
    Now we move on to $\sn$. By the Koszul formula, we have
    \[
    2\bar{\Gamma}_{ik}^{\ell} = h^{\ell a} ( e_i h_{ka} + e_k h_{ia} - e_a h_{ik} - \gamma_{ka}^b h_{ib} - \gamma_{ia}^b h_{kb} ) + \gamma_{ik}^{\ell}.
    \]
    For simplicity, let us focus on the case with no derivatives. First we make the following basic observation,
    \begin{equation*}
        |h^{\ell a} e_i h_{ka}| \leq C\langle \ln t \rangle \min\{1, t^{2(p_k - p_{\ell})}\} = C\langle \ln t \rangle t^{2(p_{\max\{k,\ell\}} - p_{\ell})}.
    \end{equation*}
    Consider first the case $i = k$. Then
    \[
    2\bar{\Gamma}_{ii}^{\ell} = h^{\ell a} ( 2e_i h_{ia} - e_a h_{ii} - 2\gamma_{ia}^b h_{ib} ),
    \]
    and thus
    \[
    |\bar{\Gamma}_{ii}^{\ell}| \leq C \langle \ln t \rangle t^{2(p_i - p_{\ell})}.
    \]
    Now consider $\bar{\Gamma}_{ik}^i$. Note that
    \[
    |h^{ia} h_{kb} \gamma_{ia}^b|(y) \leq C,
    \]
    since the only way for this term to grow is if $k < i$, but for $k = 2$ it does not happen by antisymmetry of the $\gamma_{ik}^{\ell}$, and for $k = 1$ it does not happen because $\gamma_{23}^1(y) = 0$. On the other hand, 
    \[
    |h^{ia} h_{kb} \gamma_{ia}^b|(x) \leq Ct^{2(p_k - p_1)(x)} \leq C.
    \]
    The rest of the terms can then be estimated to obtain
    \[
    |\bar{\Gamma}_{ik}^i| \leq C\langle \ln t \rangle.
    \]
    For $\bar \Gamma_{ik}^k$ just note that $\bar \Gamma_{ik}^k = \bar \Gamma_{ki}^k + \gamma_{ik}^k$. Now take $\bar{\Gamma}_{ik}^{\ell}$ with $i$, $k$ and $\ell$ distinct. Then
    \[
    t^{2p_{\ell}}|h^{\ell a}( e_i h_{ka} + e_k h_{ia} - e_a h_{ik} )| \leq C\langle \ln t \rangle(t^{2p_{\max\{k,\ell\}}} + t^{2p_{\max\{i,\ell\}}} + t^{2p_{\max\{i,k\}}}),
    \]
    and since $i, k$ and $\ell$ are distinct, the worst power of $t$ on the right-hand side is $t^{2p_2}$. Now for the terms with structure coefficients,
    \[
    \begin{split}
        t^{2p_{\ell}}|h^{\ell a} h_{ib} \gamma_{ka}^b| &\leq C( t^{2p_i} |\gamma_{k \ell}^i|  + t^{2p_{\max\{i,k\}}} + t^{2p_{\max\{i,\ell\}}} ),\\
        t^{2p_{\ell}} |h^{\ell a} h_{kb} \gamma_{ia}^b| &\leq C( t^{2p_{\max\{i,k\}}} + t^{2p_k} |\gamma_{i \ell}^k| + t^{2p_{\max\{k,\ell\}}} ).
    \end{split}
    \]
    We conclude that
    \[
    |\bar{\Gamma}_{ik}^{\ell}|(x) \leq C \langle \ln t \rangle t^{2(p_1 - p_{\ell})(x)}.
    \]
    For $D_-$, the condition $\gamma_{23}^1(y) = 0$ ensures that the terms with structure coefficients are not worse than the other terms, thus
    \[
    |\bar{\Gamma}_{ik}^{\ell}|(y) \leq C\langle \ln t \rangle t^{2(p_2 - p_{\ell})(y)}.
    \]
    For the derivatives, note that for every derivative of $\bar{\Gamma}_{ik}^{\ell}$, the resulting expression can be estimated in exactly the same way, the only difference being that an additional power of $\langle \ln t \rangle$ is introduced for every spatial derivative. Hence \eqref{connection estimates} follows.
    
    \paragraph{The Ricci tensor:} We use Lemma~\ref{general estimate for ricci lemma}. Note that \eqref{connection estimates} implies
    \[
    |D^m (t\p_t)^r (h^{ab} e_{\alpha} \bar{\Gamma}_{ab}^k)| \leq C_{m,r}\langle \ln t \rangle^{m+|\alpha|+1} t^{-2p_k}
    \]
    for $|\alpha| \leq 1$. Moreover, since $E_i = e_i$ in this case, the conditions on the frame are trivially satisfied. We start by verifying that the expression on the right-hand side of \eqref{general off diagonal estimate for ricci} satisfies the desired bounds in $D_+$ (note that it already satisfies what we want in $D_-$). Indeed, note that for $i \neq 3$,
    \[
    t^{-2p_3}\min\{1,t^{2(p_i-p_k)}\} \leq t^{2(p_1 - p_2 - p_3) + 2(p_i - p_1)} = t^{-2 + 4p_1 + 2(p_i-p_1)};
    \]
    and for $i=3$, 
    \[
    t^{-2p_3} \min\{1,t^{2(p_3-p_k)}\} = t^{2(p_1-p_k-p_3) + 2(p_3-p_1)} \leq t^{-2+4p_1 + 2(p_3-p_1)}.
    \]
    Hence, we only need to estimate $\Lambda_{ik\ell}$ and $\Upsilon_{ik\ell}$. Starting with $\Lambda_{ik\ell}$, for clarity, we focus on the case with no derivatives. We have
    \[
    \begin{split}
        \Lambda_{ik\ell}(x) &\leq  |h^{\ell\ell} \bar \Gamma_{i\ell}^k \bar \Gamma_{\ell k}^i|(x) + |h^{kk}\bar \Gamma_{ik}^\ell \bar \Gamma_{k\ell}^i|(x)\\
        &\quad + |h^{\ell \ell}\gamma_{i\ell}^k \bar \Gamma_{k\ell}^i|(x) + |h^{kk} \gamma_{ik}^\ell \bar \Gamma_{\ell k}^i|(x) + C \langle \ln t \rangle^{2}t^{-2p_3(x)}\\
        &\leq C \langle \ln t \rangle^{2}t^{2(p_1 - p_2 - p_3)(x)}.
    \end{split}
    \]
    On the other hand, for $D_-$, if we look at the fifth and seventh terms in $\Lambda_{ik\ell}$,
    \[
    \begin{split}
        |h^{ak}\bar \Gamma_{ik}^\ell \bar \Gamma_{a\ell}^i|(y) + |h^{a\ell}\gamma_{ia}^k \bar \Gamma_{k\ell}^i|(y) &\leq |h^{ik}\bar \Gamma_{ik}^\ell \bar \Gamma_{i\ell}^i|(y) + |h^{kk}\bar \Gamma_{ik}^\ell \bar \Gamma_{k\ell}^i|(y) + |h^{\ell k}\bar \Gamma_{ik}^\ell \bar \Gamma_{\ell\ell}^i|(y)\\
        &\quad + C\langle \ln t \rangle^{2}t^{2(p_2 - p_i - p_\ell)(y)}\\
        &\leq C\langle \ln t \rangle^2 \big( t^{2(p_2 - p_\ell - p_{\min\{i,k\}})(y)} + t^{2(2p_2 - 1)(y)} + t^{2(p_2 - p_i - p_\ell )(y)}\big)\\
        &\leq C\langle \ln t \rangle^2 t^{-2p_3(y)}.
    \end{split}
    \]
    The rest of the terms are similar, hence $\Lambda_{ik\ell}(y) \leq C\langle \ln t \rangle^2 t^{-2p_3(y)}$. Since a derivative introduces at worst a factor of $\langle \ln t \rangle$ in the estimates, if it is spatial, altogether we obtain
    \[
    |D^m (t\p_t)^r(\Lambda_{ik\ell})|_\ho \leq C_{m,r}\langle \ln t \rangle^{m+2} t^{-2+2\mathring\varepsilon}.
    \]
    Moving on to $\Upsilon_{ik\ell}$, for $D_+$ we see that
    \[
    \begin{split}
        |D^m (t\p_t)^r(\Upsilon_{ik\ell})|_\ho(x) &\leq C_{m,r}\langle \ln t \rangle^{m+2} t^{2(p_1 - p_k - p_\ell)(x)}\\
        &= C_{m,r}\langle \ln t \rangle^{m+2} t^{2(p_1 + p_i-1)(x)}\\
        &\leq C_{m,r}\langle \ln t \rangle^{m+2} t^{-2 + 4p_1(x) + 2(p_i - p_1)(x)},
    \end{split}
    \]
    which gives the desired estimate. And for $D_-$,
    \[
    \begin{split}
        &|D^m (t\p_t)^r(\Upsilon_{ik\ell})|_\ho(y)\\
        &\hspace{2cm} \leq C_{m,r}\langle \ln t \rangle^{m+2} \Big( t^{2(p_2-p_\ell-p_k)(y)} + t^{-2p_k(y)} + t^{-2p_\ell(y)}\sum_{s=0}^m|D^s \gamma_{i\ell}^k|_\ho(y) \Big)\\
        &\hspace{2cm} \leq C_{m,r}\langle \ln t \rangle^{m+2} t^{-2p_3(y)}.
    \end{split}
    \]
    Moreover, if $i>k$,
    \[
    \begin{split}
        &|D^m (t\p_t)^r(\Upsilon_{ik\ell})|_\ho(y)\\
        &\qquad \leq C_{m,r}\langle \ln t \rangle^{m+2}t^{2(p_i-p_k)(y)} \Big( t^{2(p_2-p_\ell-p_i)(y)} + t^{-2p_i(y)} + t^{2(p_k-p_i-p_\ell)(y)}\sum_{s=0}^m|D^s \gamma_{i\ell}^k|_\ho(y) \Big).
    \end{split}
    \]
    The only way for $t^{2(p_k-p_i-p_\ell)}$ to be worse than $t^{-2p_3}$ is if $k = 1$, but in that case $\gamma_{i\ell}^k = \pm \gamma_{23}^1$, which vanishes in a neighborhood of $y$. Thus, we obtain
    \[
    |D^m (t\p_t)^r(\Upsilon_{ik\ell})|_\ho(y) \leq C_{m,r}\langle \ln t \rangle^{m+2} t^{-2p_3(y)}\min\{1,t^{2(p_i-p_k)(y)}\}.
    \]
    This finishes the proof of the estimates for $\sric^\sharp$.

    \paragraph{The scalar field:} Finally, the scalar field $\s$. We have
    \[
    (d\s \otimes \sn \s)(e_i, \omega^k) = (e_i \s) h^{k a} (e_a \s).
    \]
    Moreover, we can write $\s = \bar\Psi \ln t + \bar\Phi$. Hence,
    \[
    \big|e_\alpha (t\p_t)^r\big((d\s \otimes \sn \s)(e_i,\omega^k)\big)\big| \leq C_{m,r} \langle \ln t \rangle^{m+2} t^{-2p_k}.
    \]
    For $i \neq k$, we obtain the improvements by noting that 
    \[
    -2p_k = 2(p_1 - p_i - p_k) + 2(p_i - p_1) = -2p_i + 2(p_i - p_k).
    \]
    Moving on to the potential,
    \[
    e_{\alpha} (t\p_t)^r( V \circ \s ) = \sum ( V^{(q)} \circ \s ) ( e_{\beta_1} (t\p_t)^{r_1} \s ) \cdots ( e_{\beta_q} (t\p_t)^{r_q} \s ),
    \]
    where the sum is over appropriate tuples $(r_1,\ldots,r_q)$ and multiindices $\beta_i$ such that $|\beta_1| + \cdots + |\beta_q| = |\alpha|$ and $r_1 + \cdots + r_q = r$. Then
    \[
    t^2|e_{\alpha}(t\p_t)^r( V \circ \s )| \leq C_{m,r} \langle \ln t \rangle^m t^2 e^{a|\s|} \leq C_{m,r} \langle \ln t \rangle^m t^{2 -a |\psio|} \leq C_{m,r} \langle \ln t \rangle^m t^{2\varepsilon_V}.
    \]
    The estimate for $V' \circ \s$ follows similarly. And the Laplacian,
    \[
    |D^m (t\p_t)^r (\Delta_h \s)|_\ho = | D^m (t\p_t)^r (h^{ab} e_a e_b \s - h^{ab} \bar{\Gamma}_{ab}^c e_c \s ) |_\ho \leq C_{m,r} \langle \ln t \rangle^{m+2} t^{-2p_3}.
    \]
    The result follows.
\end{proof}

Often we will need to estimate the difference between two corresponding objects which are derived from two different metrics and scalar fields. The necessary estimates are obtained in the following result, which is a consequence of the proof of Lemma~\ref{decay of ricci}.

\begin{corollary} \label{estimates of differences}
    Suppose that we have metrics $g_1$ and $g_2$, and functions $\s_1$ and $\s_2$, satisfying the assumptions of Lemma~\ref{decay of ricci}. Moreover, assume that there is a function $f:(0,T] \to \R$ such that
    \[
    \begin{split}
        |D^m (t\lie_{\p_t})^r ( \bar\h_1 - \bar\h_2 )|_{\ho} + |D^m (t\p_t)^r( \bar\Psi_1 - \bar\Psi_2 )|_{\ho} &\leq C_{m,r} \langle \ln t \rangle^{m} f(t),\\
        |D^m(\bar\Phi_1 - \bar\Phi_2)|_\ho &\leq C_m\langle \ln t \rangle^m f(t),
    \end{split}
    \]
    and for $i \neq k$, $x \in D_+$ and $y \in D_-$,
    \[
    \begin{split}
        \big|D^m (t\p_t)^r \big(( \bar\h_1 - \bar\h_2 )(e_i,e_k)\big)\big|_{\ho}(x) &\leq C_{m,r} \langle \ln t \rangle^{m} t^{(p_i + p_k - 2p_1)(x)} f(t),\\
        \big|D^m (t\p_t)^r \big(( \bar\h_1 - \bar\h_2 )(e_i,e_k)\big)\big|_{\ho}(y) &\leq C_{m,r} \langle \ln t \rangle^{m} t^{|p_i - p_k|(y)} f(t).
    \end{split}
    \]
    Define $\cd := \sn^1 - \sn^2$, where $\sn^1$ and $\sn^2$ denote the Levi-Civita connections of $h_1$ and $h_2$ respectively. Then there are constants $C_{m,r}$, depending only on $\eta$, $\delta$, the initial data and the potential, such that, if $x \in D_+$ and $y \in D_-$, then
    \begin{equation*}
    \begin{split}
        |D^m(t\p_t)^r (\cd_{ii}^{\ell})|_\ho &\leq C_{m,r} \langle \ln t \rangle^{m+1} t^{2(p_i - p_{\ell})} f(t),\\
        |D^m( t\p_t^r) (\cd_{ik}^i)|_\ho + |D^m( t\p_t)^r (\cd_{ik}^k)|_\ho &\leq C_{m,r} \langle \ln t \rangle^{m+1} f(t),\\
        |D^m(t\p_t)^r (\cd_{ik}^{\ell})|_\ho(x) &\leq C_{m,r} \langle \ln t \rangle^{m+1} t^{2(p_1 - p_{\ell})(x)} f(t),\\
        |D^m(t\p_t)^r (\cd_{ik}^{\ell})|_\ho(y) &\leq C_{m,r} \langle \ln t \rangle^{m+1} t^{2(p_2 - p_{\ell})(y)} f(t), 
    \end{split}
    \end{equation*}
    where $i$, $k$ and $\ell$ are distinct in the last two inequalities (no summation on $i$ or $k$). Furthermore,
    \[
    t^2|D^m (t\mathcal{L}_{\partial_t})^r (\sric_1^{\sharp} - \sric_2^{\sharp})|_{\ho} \leq C_{m,r} \langle \ln t \rangle^{m+2} t^{2\mathring{\varepsilon}} f(t)
    \]
    and for $i \neq k$, $x \in D_+$ and $y \in D_-$, the following off-diagonal improvements hold,
    \begin{equation*}
    \begin{split}
        t^2\big|D^m (t\p_t)^r \big((\sric_1^{\sharp} - \sric_2^{\sharp})(e_i,\omega^k)\big)\big|_{\ho}(x) &\leq C_{m,r} \langle \ln t \rangle^{m+2} t^{2\mathring{\varepsilon} + 2(p_i - p_1)(x)} f(t),\\
        t^2\big|D^m (t\p_t)^r \big((\sric_1^{\sharp} - \sric_2^{\sharp})(e_i,\omega^k)\big)\big|_{\ho}(y) &\leq C_{m,r} \langle \ln t \rangle^{m+2} t^{2\mathring{\varepsilon}} \min\{ 1, t^{2(p_i - p_k)(y)} \} f(t). 
    \end{split}
    \end{equation*}
    Finally,
    \[
    \begin{split}
        t^2 |D^m (t\lie_{\p_t})^r( d\s_1 \otimes \sn \s_1 - d\s_2 \otimes \sn \s_2 )|_{\ho} &\leq C_{m,r} \langle \ln t \rangle^{m+3} t^{2\mathring{\varepsilon}} f(t),\\
        t^2 |D^m (t\p_t)^r ( \Delta_{h_1} \s_1 - \Delta_{h_2} \s_2 )|_{\ho} &\leq C_{m,r} \langle \ln t \rangle^{m+3} t^{2\mathring{\varepsilon}} f(t),\\
        t^2 |D^m (t\p_t)^r ( V \circ \s_1 - V \circ \s_2 )|_{\ho} &\leq C_{m,r} \langle \ln t \rangle^{m+1} e^{a|\s_1 - \s_2|} t^{2\varepsilon_V} f(t),\\
        t^2 |D^m (t\p_t)^r ( V' \circ \s_1 - V' \circ \s_2 )|_{\ho} &\leq C_{m,r} \langle \ln t \rangle^{m+1} e^{a|\s_1 - \s_2|} t^{2\varepsilon_V} f(t);
    \end{split}
    \]
    and for $i \neq k$, $x \in D_+$ and $y \in D_-$,
    \[
    \begin{split}
        &t^2\big|D^m (t\p_t)^r \big((d\s_1 \otimes \sn \s_1 - d\s_2 \otimes \sn \s_2)(e_i,\omega^k)\big)\big|_{\ho}(x)\\
        &\hspace{5cm}\leq C_{m,r} \langle \ln t \rangle^{m+3} t^{2\mathring{\varepsilon} + 2(p_i - p_1)(x)} f(t),\\
        &t^2\big|D^m (t\p_t)^r \big((d\s_1 \otimes \sn \s_1 - d\s_2 \otimes \sn \s_2) (e_i,\omega^k)\big)\big|_{\ho}(y)\\
        &\hspace{5cm}\leq C_{m,r} \langle \ln t \rangle^{m+3} t^{2\mathring{\varepsilon}} \min\{ 1, t^{2(p_i - p_k)(y)} \} f(t). 
    \end{split}
    \]
\end{corollary}

\begin{remark}
    Remark~\ref{finite regularity} equally applies here.
\end{remark}

\begin{proof}
    We begin by making the following basic observation. Let $M(x)$ denote a monomial in the variables $x = (x_1, \ldots, x_n)$, so that there are non-negative integers $b_1, \ldots, b_n$ and a real number $c$ such that
    \[
    M(x) = c x_1^{b_1} \cdots x_n^{b_n}.
    \]
    If $y = (y_1, \ldots, y_n)$ denotes another set of variables, then
    \begin{equation} \label{difference of monomials}
        M(x) - M(y) = c \sum_{i=1}^n \sum_{r = 0}^{b_i - 1} y_1^{b_1} \cdots y_i^{r} (x_i - y_i) x_i^{b_i - r - 1} \cdots x_n^{b_n}.
    \end{equation}

    The assumptions on $h_1$ and $h_2$ imply
    \begin{equation} \label{difference between metrics}
        \begin{split}
        |e_{\alpha} (t\p_t)^r(h_1 - h_2)_{ii}| &\leq C_{m,r} \langle \ln t \rangle^{m} t^{2p_i} f(t),\\
        |e_{\alpha} (t\p_t)^r( h_1 -h_2 )_{ik}|(x) &\leq C_{m,r} \langle \ln t \rangle^{m} t^{2(p_i + p_k - p_1)(x)} f(t),\\
        |e_{\alpha} (t\p_t)^r( h_1 -h_2 )_{ik}|(y) &\leq C_{m,r} \langle \ln t \rangle^{m} t^{2p_{\max\{i,k\}}(y)} f(t),
    \end{split}
    \end{equation}
    for $i \neq k$, $x \in D_+$ and $y \in D_-$. Since $\det h_a$ is a polynomial in the components $(h_a)_{ik}$, by using \eqref{difference of monomials} on the differences of the corresponding terms, we see that
    \[
    |e_{\alpha} (t\p_t)^r( \det h_1 - \det h_2 )| \leq C_{m,r} \langle \ln t \rangle^{m} t^2 f(t).
    \]
    Now to estimate the difference between the dual metrics,
    \[
    \begin{split}
        (h_1 - h_2)^{ik} &= \frac{1}{\det h_1} \adj(h_1)_{ik} - \frac{1}{\det h_2} \adj(h_2)_{ik}\\
        &= \bigg( \frac{\det h_2 - \det h_1}{(\det h_1)(\det h_2)} \bigg) \adj(h_1)_{ik} + \frac{1}{\det h_2}\big( \adj(h_1)_{ik} - \adj(h_2)_{ik} \big).
    \end{split}
    \]
    Again, since the entries of $\adj(h_a)$ are polynomial on the $(h_a)_{ik}$, we can use \eqref{difference of monomials} and \eqref{difference between metrics} to estimate this expression in the same way as in Lemma~\ref{decay of ricci}, the only difference being that an additional $ f(t)$ factor is introduced. We conclude that
    \[
    \begin{split}
        |e_{\alpha} (t\p_t)^r( h_1 - h_2 )^{ii}| &\leq C_{m,r} \langle \ln t \rangle^{m} t^{-2p_i} f(t),\\
        |e_{\alpha} (t\p_t)^r( h_1 - h_2 )^{ik}|(x) &\leq C_{m,r} \langle \ln t \rangle^{m} t^{-2p_1(x)} f(t),\\
        |e_{\alpha} (t\p_t)^r( h_1 - h_2 )^{ik}|(y) &\leq C_{m,r} \langle \ln t \rangle^{m} t^{-2p_{\min\{i,k\}}(y)} f(t),\\
    \end{split}
    \]
    for $i \neq k$, $x \in D_+$ and $y \in D_-$. All of the estimates, except the ones involving the potential $V$, then follow by using \eqref{difference of monomials} and following the same steps as in the proofs of Lemmas~\ref{decay of ricci} and \ref{general estimate for ricci lemma}.

    For the remaining estimates, write
    \[
    V \circ \s_1 - V \circ \s_2 = \int_0^1 V'\big( s\s_1 + (1-s)\s_2 \big)ds( \s_1 - \s_2 ),
    \]
    implying
    \[
    \begin{split}
        |e_{\alpha} (t\p_t)^r( V \circ \s_1 - V \circ \s_2 )| &\leq C_{m,r} \langle \ln t \rangle^{m+1} e^{a( |\s_1 - \s_2| + |\s_2| )} f(t)\\
        &\leq C_{m,r} \langle \ln t \rangle^{m+1} e^{a |\s_1 - \s_2|} t^{2(\varepsilon_V - 1)} f(t).
    \end{split}
    \]
    Similarly for $V' \circ \s_1 - V' \circ \s_2$.
\end{proof}

\subsection{General results for ODEs} \label{general results for odes}

Here we prove the general existence and uniqueness results for ODEs which will be used to construct the sequence of approximate solutions.

\begin{lemma} \label{nonlinear ode lemma}
    Let $(\Sigma,\ho)$ be a closed Riemannian manifold and $D$ the Levi-Civita connection of $\ho$. Consider the equation 
    \[
    \partial_{\tau} u = u^2 + u + f
    \]
    on $M_{\tau_0} = [\tau_0,\infty) \times \Sigma$, where $\tau_0 > 0$, $f:M_{\tau_0} \to \mathbb{R}$ is smooth, and there are constants $C_m$, $N_m$ and $\delta > 0$ such that
    \[
    | D^m f |_{\ho} \leq C_m \langle \tau \rangle^{N_m} e^{-\delta\tau}.
    \]
    Then, by taking $\tau_0$ large enough (independent of $m$), there are constants $C_m$ and a unique smooth solution ${u:M_{\tau_0} \to \mathbb{R}}$ to the equation such that
    \[
    | D^m u |_{\ho} \leq C_m \langle \tau \rangle^{N_m} e^{-\delta\tau}.
    \]
\end{lemma}

\begin{proof}
    Define the map
    \[
    \varphi(u) := -\int_{\tau}^{\infty} e^{\tau - s} \big( u(s)^2 + f(s) \big)ds
    \]
    and the space $C_{\delta,B}^m(M_{\tau_0}) := \{ u \in C( [\tau_0,\infty), C^m(\Sigma) ) : \; \|u\|_m \leq B \; \}$, where
    \[
    \|u\|_m := \sup_{\tau \geq \tau_0} \langle \tau \rangle^{-N_m} e^{\delta\tau} \|u\|_{C^m(\Sigma)}.
    \]
    Clearly $C_{\delta,B}^m(M_{\tau_0})$ is a complete metric space with distance $d(u,v) := \|u - v\|_m$. We verify the conditions of the Banach fixed point theorem for the map $\varphi$. If $u \in C_{\delta,B}^m(M_{\tau_0})$ and $|\alpha| \leq m$,
    \[
    \begin{split}
        &|e_{\alpha} \varphi(u)|\\
        &\leq e^{\tau} \int_{\tau}^{\infty} e^{-s}\big( B^2 \langle s \rangle^{2N_m} e^{-2\delta s} + \|f\|_m \langle s \rangle^{N_m} e^{-\delta s} \big)ds\\
        &\leq C_m B^2 \langle \tau \rangle^{2N_m} e^{-2\delta \tau} \int_{\tau}^{\infty} \langle s-\tau \rangle^{2N_m} e^{-(1+2\delta)(s-\tau)}ds\\
        &\phantom{\leq} \hspace{4cm} + C_m \|f\|_m \langle \tau \rangle^{N_m} e^{-\delta \tau} \int_{\tau}^{\infty} \langle s-\tau \rangle^{N_m} e^{-(1+\delta)(s-\tau)}ds\\
        &\leq C_m\bigg( B^2 \langle \tau \rangle^{N_m} e^{-\delta \tau} \int_0^{\infty} \langle r \rangle^{2N_m} e^{-(1+2\delta)r}dr + \|f\|_m \int_0^{\infty} \langle r \rangle^{N_m} e^{-(1+\delta)r}dr \bigg) \langle \tau \rangle^{N_m} e^{-\delta \tau}.
    \end{split}
    \]
    By choosing $B$ large enough, it is possible to take $\tau_0$ large enough such that $\varphi$ maps $C_{\delta,B}^m(M_{\tau_0})$ to itself. Now we verify that it is a contraction. Let $u,v \in C_{\delta,B}^m(M_{\tau_0})$, then
    \[
    \begin{split}
        \big|e_{\alpha} \big( \varphi(u) - \varphi(v) \big)\big| &\leq e^{\tau} \int_{\tau}^{\infty} e^{-s} \sum |e_{\beta}(u) e_{\gamma}(u) - e_{\beta}(v) e_{\gamma}(v)|ds\\
        &\leq e^{\tau} \int_{\tau}^{\infty} e^{-s} \sum \Big( |e_{\beta}(u)||e_{\gamma}(u-v)| + |e_{\gamma}(v)||e_{\beta}(u-v)| \Big)ds\\
        &\leq C_m B e^{\tau} \int_{\tau}^{\infty} e^{-s} \langle s \rangle^{2N_m} e^{-2\delta s}ds \|u-v\|_m\\
        &\leq C_m B \langle \tau \rangle^{2N_m} e^{-2\delta \tau} \int_{\tau}^{\infty} \langle s-\tau \rangle^{2N_m} e^{-(1+2\delta)(s-\tau)}ds \|u-v\|_m,
    \end{split}
    \]
    where the sum is over appropriate multiindices $\beta$ and $\gamma$ such that $|\beta| + |\gamma| = |\alpha|$, and $C_m$ is a constant depending only on $m$. Hence
    \[
    \langle \tau \rangle^{-N_m} e^{\delta \tau} \| \varphi(u) - \varphi(v) \|_{C^m(\Sigma)} \leq \Big( C_m B \langle \tau \rangle^{N_m} e^{-\delta \tau} \int_0^{\infty} \langle r \rangle^{2N_m} e^{-(1+2\delta)r}dr \Big) \|u-v\|_m.
    \]
    By taking $\tau_0$ large enough, we can ensure that the expression inside the parentheses is smaller than $1/2$, so that $\varphi$ is a contraction. We conclude, from Banach's fixed point theorem, that there is a unique $u \in C^m_{\delta,B}(M_{\tau_0})$ such that 
    \[
    u = \varphi(u) = -\int_{\tau}^{\infty} e^{\tau-s} \big( u(s)^2 + f(s) \big)ds.
    \]
    In particular, $u$ solves the differential equation. Note that $u$ is independent of $B$, since for two choices $B_1 < B_2$ with corresponding solutions $u_1$ and $u_2$, we have $u_1 \in C^m_{\delta,B_2}(M_{\tau_0})$. Thus by uniqueness $u_1 = u_2$ on the intersection of their domains. Moving on to the regularity of $u$. Given any degree of regularity $C^m(\Sigma)$, there is a corresponding solution defined for $\tau \in [\tau_0(m),\infty)$, with $\tau_0(m)$ in principle increasing with $m$. If we have two solutions, $u_1$ and $u_2$, which are $C^{m_1}(\Sigma)$ and $C^{m_2}(\Sigma)$ respectively, then they are both $C^{\min\{m_1,m_2\}}(\Sigma)$ and thus by uniqueness agree on the intersection of their domains. The regularity can then be transported to some $\tau_0$ independent of $m$, thus there is a $C^\infty(\Sigma)$ solution. Regularity in time then follows from differentiating the equation.   
\end{proof}

\begin{lemma} \label{linear ode lemma}
    Let $(\Sigma,\ho)$ be a closed Riemannian manifold and $D$ the Levi-Civita connection of $\ho$. Consider the equation
    \[
    \partial_{\tau} v = A_{\mathrm{rem}} v + F
    \]
    on $M_{\tau_0} = [\tau_0,\infty) \times \Sigma$, where $v, F: M_{\tau_0} \to \mathbb{R}^k$ with $F$ smooth, and $A_{\mathrm{rem}}: M_{\tau_0} \to \mathbb{M}_k(\mathbb{R})$ smooth. If there are constants $C_m$, $N_m$ and $\varepsilon, \delta > 0$ such that
    \[
    |D^m A_{\mathrm{rem}}|_{\ho} \leq C_m e^{-\varepsilon \tau}, \quad |D^m F|_{\ho} \leq C_m \langle \tau \rangle^{N_m} e^{-\delta \tau},
    \]
    then there are constants $C_m$ and a unique smooth solution $v$ to the equation such that
    \[
    |D^m v|_{\ho} \leq C_m \langle \tau \rangle^{N_m} e^{-\delta \tau}.
    \]
    In particular, if $F = 0$ the only solution which decays as $\tau \to \infty$ is the trivial solution $v = 0$.
\end{lemma}

\begin{remark}
    Here $D^m F$, $D^m v$ and $D^m A_{\mathrm{rem}}$ are taken componentwise and the $|\cdot|_{\ho}$ norms can be defined as the sum of the $|\cdot|_{\ho}$ norms of the components.
\end{remark}

\begin{proof}
    This is very similar to the proof of Lemma~\ref{nonlinear ode lemma} but simpler, since the equation in question is linear.  
\end{proof}

\subsection{The approximate Weingarten map} \label{the approximate weingarten map}

We now begin the construction of the sequence of approximate solutions, by constructing an approximate Weingarten map $\bar K_n$ which satisfies the convergence estimates to the initial data on the singularity that we want. Later we will prove that the difference between $\bar K_n$ and the actual Weingarten map of the $\Sigma_t$ hypersurfaces, with respect to the $n$-th approximate solution, decays as $t \to 0$ at a rate which is increasing in $n$. This will then imply the desired convergence for the Weingarten map. Before proceeding, we set some conventions regarding notation.

\begin{remark}
    For the remainder of this section, $N_n$ will denote a positive integer whose value may change from line to line, which is only allowed to depend on $n$.
\end{remark}

\begin{remark}
    Below, we will work with metrics $g_n = -dt \otimes dt + h_n$ and functions $\s_n$ indexed by $n$. We will denote by $\sn\s_n$ the gradient of $\s_n$ with respect to $h_n$. Tensors which are derived from the metrics $g_n$ and $h_n$ will be indexed by $n$ accordingly. So, for instance, $\sric_n$ denotes the Ricci tensor of $h_n$. Moreover, indices will be raised and lowered with the corresponding metric, so when we write $\sric_n^\sharp$, we mean $\sharp$ to be taken with respect to $h_n$. Finally, denote by $\sn^{(n)}$ the Levi-Civita connection of $h_n$.
\end{remark}

\begin{lemma} \label{definition of kn}
    Let $(\Sigma,\ho,\Ko,\phio,\psio)$ be initial data on the singularity and let $V$ be an admissible potential. Assume that we have a metric $g_{n-1} = -dt \otimes dt + h_{n-1}$ and a function $\s_{n-1}$ on $(0,t_{n-1}] \times \Sigma$ satisfying the assumptions of Lemma~\ref{decay of ricci}. Then there is a $t_n > 0$ and a unique $\bar K_n$ on $(0,t_n] \times \Sigma$ solving
    \begin{equation} \label{approximate weingarten map equation}
        \mathcal{L}_{\partial_t} \bar K_n + \overline{\ric}_{n-1}^{\sharp} + \bar \theta_n \bar K_n = d\s_{n-1} \otimes \sn \s_{n-1} + (V \circ \s_{n-1})I,
    \end{equation}
    where $\bar \theta_n = \tr \bar K_n$, such that
    \[
    |D^m (t\lie_{\p_t})^r(t\bar K_n - \Ko)|_{\ho} \leq C_{m,r,n} \langle \ln t \rangle^{m+2} t^{2\varepsilon}.
    \]
    Moreover, for $i \neq k$, $x \in D_+$ and $y \in D_-$,
    \[
    \begin{split}
        \big|D^m (t\p_t)^r\big( t\bar K_n(e_i,\omega^k) \big)\big|_{\ho}(x) &\leq C_{m,r,n} \langle \ln t \rangle^{m+2} t^{2\varepsilon +2(p_i - p_1)(x)},\\
        \big|D^m (t\p_t)^r\big( t\bar K_n(e_i,\omega^k) \big)\big|_{\ho}(y) &\leq C_{m,r,n} \langle \ln t \rangle^{m+2} t^{2\varepsilon} \min\{ 1, t^{2(p_i - p_k)(y)} \}.
    \end{split}
    \]
\end{lemma}

\begin{proof}
    We begin by defining $\bar \theta_n$. Introduce the time coordinate $\tau = -\ln t$. By taking the trace of Equation~\eqref{approximate weingarten map equation}, we see that $\bar \theta_n$ should satisfy
    \[
    \partial_t \bar \theta_n + \bar{S}_{n-1} + \bar \theta_n^2 = |d\s_{n-1}|_{h_{n-1}}^2 + 3V \circ \s_{n-1},
    \]
    which we can rewrite as
    \begin{equation} \label{thetan equation normalized}
        \partial_{\tau}(e^{-\tau} \bar \theta_n - 1) = (e^{-\tau} \bar \theta_n - 1)^2 + (e^{-\tau} \bar \theta_n - 1) + e^{-2\tau} \bar{S}_{n-1} - e^{-2\tau} |d\s_{n-1}|_{h_{n-1}}^2 - 3e^{-2\tau} V \circ \s_{n-1}.
    \end{equation}
    Note that by Lemma~\ref{decay of ricci}, since $|d\s_{n-1}|_{h_{n-1}}^2 = \tr( d\s_{n-1} \otimes \sn \s_{n-1} )$, the assumptions on $g_{n-1}$ and $\s_{n-1}$ imply
    \[
    t^2 |D^m (t\p_t)^r (\bar{S}_{n-1} - |d\s_{n-1}|_{h_{n-1}}^2 - 3 V \circ \s_{n-1})|_{\ho}  \leq C_{m,r,n} \langle \ln t \rangle^{m+2} t^{2\varepsilon}.
    \]
    Thus, by Lemma~\ref{nonlinear ode lemma}, there is a sufficiently large $\tau_n > 0$ such that we can define $\bar\theta_n$ as the unique solution to \eqref{thetan equation normalized} such that
    \[
    |D^m(e^{-\tau} \bar \theta_n - 1)|_{\ho} \leq C_{m,n} \langle \tau \rangle^{m+2} e^{-2\varepsilon \tau}
    \]
    on $[\tau_n,\infty) \times \Sigma$. Moving on to $\bar K_n$, we can rewrite \eqref{approximate weingarten map equation} as
    \begin{equation} \label{kn equation normalized}
    \begin{split}
        \lie_{\partial_{\tau}} ( e^{-\tau}\bar K_n - \Ko ) &= (e^{-\tau}\bar \theta_n - 1) ( e^{-\tau} K_n - \Ko ) + (e^{-\tau} \bar \theta_n - 1) \Ko + e^{-2\tau} \sric_{n-1}^{\sharp}\\
        &\quad -e^{-2\tau} d\s_{n-1} \otimes \sn \s_{n-1} - e^{-2\tau}(V \circ \s_{n-1})I. 
    \end{split}
    \end{equation}
    So, by Lemmas~\ref{decay of ricci} and \ref{linear ode lemma}, we can define $\bar K_n$ as the unique solution of \eqref{kn equation normalized} such that
    \begin{equation} \label{convergence of kn bar}
        |D^m (e^{-\tau} \bar K_n - \Ko)|_{\ho} \leq C_{m,n} \langle \tau \rangle^{m+2} e^{-2\varepsilon \tau}.
    \end{equation}
    Now need to show that $\bar \theta_n = \tr \bar K_n$. By taking the trace of \eqref{kn equation normalized}, we get
    \[
    \begin{split}
        \partial_{\tau}( e^{-\tau} \tr \bar K_n - 1) &= (e^{-\tau} \bar \theta_n - 1)(e^{-\tau} \tr \bar K_n - 1) + (e^{-\tau} \bar \theta_n - 1) + e^{-2\tau} \bar{S}_{n-1}\\
        &\quad - e^{-2\tau} |d\s_{n-1}|_{h_{n-1}}^2 - 3e^{-2\tau} V \circ \s_{n-1}.
    \end{split}
    \]
    Since $\bar \theta_n$ also solves this equation, by uniqueness in Lemma~\ref{linear ode lemma}, we conclude that $\tr \bar K_n = \bar \theta_n$. Finally, we need to prove that the off-diagonal improvements hold. From \eqref{kn equation normalized} and \eqref{convergence of kn bar}, it follows that for $i \neq k$,
    \[
    \begin{split}
        e^{-\tau} \bar K_n(e_i, \omega^k)(\tau) = - \int_{\tau}^{\infty} &\big( (e^{-s}\bar \theta_n - 1)e^{-s} \bar K_n\\
        &+ e^{-2s} \sric^{\sharp}_{n-1} - e^{-2s} d\s_{n-1} \otimes \sn \s_{n-1} \big)(e_i,\omega^k)(s)ds. 
    \end{split}
    \]
    Using this, we can successively improve on the estimates until we obtain the desired ones. Finally, the estimates for the time derivatives come directly from differentiating \eqref{kn equation normalized}.
\end{proof}

\begin{lemma} \label{kn become closer}
    Let $(\Sigma,\ho,\Ko,\phio,\psio)$ be initial data on the singularity and let $V$ be an admissible potential. Assume that we have metrics $g_{n-1} = -dt \otimes dt + h_{n-1}$ and $g_{n-2} = -dt \otimes dt + h_{n-2}$, and functions $\s_{n-1}$ and $\s_{n-2}$ on $(0,t_{n-1}] \times \Sigma$, for $n \geq 2$, satisfying the assumptions of Lemma~\ref{decay of ricci}. Moreover, assume that 
    \[
    \begin{split}
        | D^m (t\lie_{\p_t})^r(\bar\h_{n-1} - \bar\h_{n-2}) |_{\ho} &\leq C_{m,r,n-1} \langle \ln t \rangle^{m+N_{n-1}} t^{2(n-1)\varepsilon},\\
        |D^m (t\p_t)^r(\bar\Psi_{n-1} - \bar\Psi_{n-2})|_{\ho} &\leq C_{m,r,n-1} \langle \ln t \rangle^{m+N_{n-1}} t^{2(n-1)\varepsilon},\\
        |D^m (\bar\Phi_{n-1} - \bar\Phi_{n-2})|_{\ho} &\leq C_{m,n-1} \langle \ln t \rangle^{m+N_{n-1}} t^{2(n-1)\varepsilon}, 
    \end{split}
    \]
    and for $i \neq k$, $x \in D_+$ and $y \in D_-$,
    \[
    \begin{split}
        \big|D^m (t\p_t)^r\big( (\bar\h_{n-1} - \bar\h_{n-2})(e_i,e_k) \big)\big|_{\ho}(x) &\leq C_{m,r,n-1} \langle \ln t  \rangle^{m+N_{n-1}} t^{2(n-1)\varepsilon + (p_i + p_k - 2p_1)(x)},\\
        \big|D^m (t\p_t)^r\big( (\bar\h_{n-1} - \bar\h_{n-2})(e_i,e_k) \big)\big|_{\ho}(y) &\leq C_{m,r,n-1} \langle \ln t  \rangle^{m+N_{n-1}} t^{2(n-1)\varepsilon + |p_i - p_k|(y)}.
    \end{split}
    \]
    Then, if $\bar K_n$ and $\bar K_{n-1}$ are defined as in Lemma~\ref{definition of kn},
    \[
    |D^m (t\lie_{\p_t})^r(t\bar K_n - t\bar K_{n-1})|_{\ho} \leq C_{m,r,n} \langle \ln t \rangle^{m+N_n} t^{2n\varepsilon},
    \]
    and for $i \neq k$, $x \in D_+$ and $y \in D_-$,
    \[
    \begin{split}
        \big|D^m (t\p_t)^r\big( (t\bar K_n - t\bar K_{n-1})(e_i,\omega^k) \big)\big|_{\ho}(x) &\leq C_{m,r,n} \langle \ln t \rangle^{m+N_n} t^{2n\varepsilon +2(p_i - p_1)(x)},\\
        \big|D^m (t\p_t)^r\big( (t\bar K_n - t\bar K_{n-1})(e_i,\omega^k) \big)\big|_{\ho}(y) &\leq C_{m,r,n} \langle \ln t \rangle^{m+N_n} t^{2n\varepsilon} \min\{ 1, t^{2(p_i - p_k)(y)} \}.
    \end{split}
    \]
\end{lemma}

\begin{proof}
    From Corollary~\ref{estimates of differences} we conclude that
    \[
    \begin{split}
        t^2| D^m (t\lie_{\p_t})^r ( \sric_{n-1}^{\sharp} - \sric_{n-2}^{\sharp} ) |_{\ho} &\leq C_{m,r,n} \langle \ln t \rangle^{m+2+N_{n-1}} t^{2n\varepsilon},\\
        t^2|D^m (t\lie_{\p_t})^r ( d\s_{n-1} \otimes \sn \s_{n-1} - d\s_{n-2} \otimes \sn \s_{n-2} )|_{\ho} &\leq C_{m,r,n} \langle \ln t \rangle^{m+3+N_{n-1}} t^{2n\varepsilon},\\
        t^2|D^m (t\p_t)^r( V \circ \s_{n-1} - V \circ \s_{n-2} )|_{\ho} &\leq C_{m,r,n} \langle \ln t \rangle^{m + 1+ N_{n-1}} t^{2n\varepsilon},
    \end{split}
    \]
    along with the corresponding off-diagonal improvements. We deal with $\bar \theta_n - \bar \theta_{n-1}$ first. It satisfies the equation
    \[
    \begin{split}
        \partial_{\tau}\big( e^{-2\tau}(\bar \theta_n - \bar \theta_{n-1}) \big) &= (e^{-\tau}\bar \theta_n + e^{-\tau}\bar \theta_{n-1} - 2)e^{-2\tau}(\bar \theta_n - \bar \theta_{n-1}) + e^{-3\tau}(\bar{S}_{n-1} - \bar{S}_{n-2})\\
        &\quad -e^{-3\tau}( |d\s_{n-1}|_{h_{n-1}}^2 - |d\s_{n-2}|_{h_{n-2}}^2 ) - 3e^{-3\tau}( V \circ \s_{n-1} - V \circ \s_{n-2}).
    \end{split}
    \]
    By Lemma~\ref{linear ode lemma}, we see that $e^{-2\tau}(\bar \theta_n - \bar \theta_{n-1})$ is the unique decaying solution, and the decay is given by the inhomogeneous term; hence,
    \[
    e^{-\tau}| D^m(\bar \theta_n - \bar \theta_{n-1}) |_{\ho} \leq C_{m,n} \langle \tau \rangle^{m+N_{n}} e^{-2n\varepsilon\tau},
    \]
    for a suitable integer $N_n$. Now for $\bar K_n - \bar K_{n-1}$. We have
    \[
    \begin{split}
        \lie_{\partial_{\tau}}\big( e^{-\tau}(\bar K_n - \bar K_{n-1}) \big) &= (e^{-\tau}\bar \theta_{n-1} - 1) e^{-\tau}(\bar K_n - \bar K_{n-1}) + e^{-\tau}( \bar \theta_n - \bar \theta_{n-1} ) e^{-\tau}\bar K_n \\
        &\quad + e^{-2\tau}( \sric_{n-1}^{\sharp} - \sric_{n-2}^{\sharp} )\\
        &\quad -e^{-2\tau}(d\s_{n-1} \otimes \sn \s_{n-1} - d\s_{n-2} \otimes \sn \s_{n-2})\\
        &\quad - e^{-2\tau}( V \circ \s_{n-2} - V \circ \s_{n-2} )I.
    \end{split}
    \]
    Thus similarly as above, by Lemma~\ref{linear ode lemma}, we obtain
    \[
    e^{-\tau}| D^m (\bar K_n - \bar K_{n-1}) |_{\ho} \leq C_{m,n} \langle \tau \rangle^{m+N_{n}} e^{-2n\varepsilon\tau}.
    \]
    The off-diagonal improvements follow similarly as in Lemma~\ref{definition of kn}. The estimates for the time derivatives then follow from repeatedly differentiating the evolution equations.
\end{proof}

\subsection{The induced metric} \label{the induced metric}

\begin{lemma} \label{definition of hn}
    Let $(\Sigma,\ho,\Ko,\phio,\psio)$ be initial data on the singularity. Suppose we have a one parameter family of $(1,1)$-tensors $\bar K_n$ on $(0,t_n] \times \Sigma$ such that
    \[
    | D^m (t\lie_{\p_t})^r(t\bar K_n - \Ko) |_{\ho} \leq C_{m,r,n} \langle \ln t \rangle^{m+2} t^{2\varepsilon},
    \]
    and for $i \neq k$, $x \in D_+$ and $y \in D_-$,
    \[
    \begin{split}
        \big|D^m (t\p_t)^r\big( t\bar K_n(e_i,\omega^k) \big)\big|_{\ho}(x) &\leq C_{m,r,n} \langle \ln t \rangle^{m+2} t^{2\varepsilon +2(p_i - p_1)(x)},\\
        \big|D^m (t\p_t)^r\big( t\bar K_n(e_i,\omega^k) \big)\big|_{\ho}(y) &\leq C_{m,r,n} \langle \ln t \rangle^{m+2} t^{2\varepsilon} \min\{ 1, t^{2(p_i - p_k)(y)} \}.
    \end{split}
    \]
    Then there exists a unique $h_n$ on $(0,t_n] \times \Sigma$ solving
    \begin{equation} \label{approximate metric equation}
        \lie_{\partial_t} h_n(X,Y) = h_n(\bar K_n(X),Y) + h_n(X,\bar K_n(Y)),
    \end{equation}
    for $X, Y \in \mathfrak X(\Sigma)$, such that $h_n$ is symmetric and $\bar\h_n = h_n(t^{-\Ko}(\,\cdot\,),t^{-\Ko}(\,\cdot\,))$ satisfies
    \[
    |D^m (t\lie_{\p_t})^r(\bar\h_n - \ho)|_{\ho} \leq C_{m,r,n} \langle \ln t \rangle^{m+2} t^{2\varepsilon}.
    \]
    Moreover, for $i \neq k$, $x \in D_+$ and $y \in D_-$,
    \[
    \begin{split}
        \big|D^m (t\p_t)^r \big(\bar\h_n(e_i,e_k) \big)\big|_{\ho}(x) &\leq C_{m,r,n} \langle \ln t  \rangle^{m+2} t^{2\varepsilon + (p_i + p_k - 2p_1)(x)},\\
        \big|D^m (t\p_t)^r \big(\bar\h_n(e_i,e_k) \big)\big|_{\ho}(y) &\leq C_{m,r,n} \langle \ln t  \rangle^{m+2} t^{2\varepsilon + |p_i - p_k|(y)}.
    \end{split}
    \]
\end{lemma}

\begin{proof}
    By rewriting the equation for $h_n$ in terms of $\bar\h_n$ and $\tau = -\ln t$, we obtain
    \begin{equation*}
    \begin{split}
        \lie_{\partial_{\tau}} \bar\h_n(X,Y) &= -\bar\h_n( e^{-\tau \Ko} \circ ( e^{-\tau}\bar K_n - \Ko ) \circ e^{\tau\Ko}(X), Y )\\
        &\phantom{=} \hspace{3cm} - \bar\h_n(X, e^{-\tau \Ko} \circ ( e^{-\tau}\bar K_n - \Ko ) \circ e^{\tau\Ko}(Y)).
    \end{split} 
    \end{equation*}
    So in terms of the frame $\{e_i\}$,
    \begin{equation} \label{hn equation normalized}
    \begin{split}
        \lie_{\partial_{\tau}}(\bar\h_n - \ho)(e_i,e_k) &= -\sum_\ell e^{(p_i - p_\ell)\tau}( e^{-\tau}\bar K_n - \Ko )(e_i,\omega^\ell)(\bar\h_n - \ho)(e_\ell,e_k)\\
        &\quad -\sum_\ell e^{(p_k - p_\ell)\tau}(e^{-\tau}\bar K_n - \Ko)(e_k,\omega^\ell)(\bar\h_n - \ho)(e_i,e_\ell)\\
        &\quad -e^{(p_i - p_k)\tau}(e^{-\tau} \bar K_n - \Ko)(e_i,\omega^k) - e^{(p_k - p_i)\tau}(e^{-\tau}\bar K_n - \Ko)(e_k,\omega^i).
    \end{split}
    \end{equation}
    Thus, we can use Lemma~\ref{linear ode lemma} to define $\bar\h_n$ as the unique solution to the system such that
    \[
    |D^m(\bar\h_n - \ho)|_{\ho} \leq C_{m,n} \langle \tau \rangle^{m+2} e^{-2\varepsilon\tau}.
    \]
    To verify that $\bar\h_n$ is symmetric, note that
    \[
    \begin{split}
        \partial_{\tau}\big( \bar\h_n(e_i,e_k) - \bar\h_n(e_k,e_i) \big) &= \sum_\ell e^{(p_i - p_\ell)\tau}(e^{-\tau}\bar K_n - \Ko)(e_i,\omega^\ell)\big( \bar\h_n(e_k,e_\ell) - \bar\h_n(e_\ell,e_k) \big)\\
        &\quad + \sum_\ell e^{(p_k - p_\ell)\tau}(e^{-\tau}\bar K_n - \Ko)(e_k,\omega^\ell)\big( \bar\h_n(e_\ell,e_i) - \bar\h_n(e_i,e_\ell) \big).
    \end{split}
    \]
    That is, the antisymmetric part of $\bar\h_n$ satisfies an equation as in Lemma~\ref{linear ode lemma} with $F = 0$, so it must vanish. We conclude that $\bar\h_n$ is symmetric. Now we obtain the improved estimates for the off-diagonal components of $\bar\h_n$. For $i$, $k$ and $\ell$ distinct (no summation over any of them),
    \[
    \begin{split}
        \lie_{\partial_{\tau}} \bar\h_n(e_i,e_k) &= \big( (\Ko - e^{-\tau}\bar K_n)(e_i,\omega^i) + (\Ko - e^{-\tau}\bar K_n)(e_k,\omega^k) \big)\bar\h_n(e_i,e_k)\\
        &\quad - e^{(p_i - p_k)\tau}e^{-\tau}\bar K_n(e_i,\omega^k)\bar\h_n(e_k,e_k) - e^{(p_k - p_i)\tau}e^{-\tau}\bar K_n(e_k,\omega^i)\bar\h_n(e_i,e_i)\\
        &\quad - e^{(p_i - p_{\ell})\tau}e^{-\tau}\bar K_n(e_i,\omega^{\ell})\bar\h_n(e_{\ell},e_k) - e^{(p_k - p_{\ell})\tau}e^{-\tau}\bar K_n(e_k,\omega^{\ell})\bar\h_n(e_i,e_{\ell}).
    \end{split}
    \]
    We make two observations regarding the equation above. First, the terms involving $\bar\h_n(e_i,e_i)$ and $\bar\h_n(e_k,e_k)$, present the decay that we want for $\bar\h_n(e_i,e_k)$. Second, if we already knew the desired estimates to hold for $\bar\h_n(e_\ell,e_k)$ or $\bar\h_n(e_i,e_\ell)$, then the corresponding terms in the equation present better decay than what we want for $\bar\h_n(e_i,e_k)$. Keeping this in mind, similarly as in Lemma~\ref{definition of kn}, we can integrate these individual equations from $\tau$ to $\infty$ to start successively improving on the estimates for the $\bar\h_n(e_i,e_k)$. Note that, given a decay estimate for the $\bar\h_n(e_i,e_k)$, all the terms on the right-hand side of the equation decay faster than the given estimate, except for the ones involving $\bar\h_n(e_i,e_i)$ and $\bar\h_n(e_k,e_k)$. This means that we can continue iterating the improvement process until we achieve decay as in these terms for some $\bar\h_n(e_i,e_k)$. At that point, by our first observation, we are done with the improvements for that particular $\bar\h_n(e_i,e_k)$. Moreover, by our second observation, we can continue the improvement process for the remaining components of $\bar\h_n$ until we achieve the desired result. 
    Finally, the estimates for the time derivatives can be deduced directly from Equation~\eqref{hn equation normalized}.
\end{proof}

\begin{lemma} \label{hn become closer}
    Let $(\Sigma,\ho,\Ko,\phio,\psio)$ be initial data on the singularity. Suppose we have two one parameter families of $(1,1)$-tensors $\bar K_n$ and $\bar K_{n-1}$ satisfying the assumptions of Lemma~\ref{definition of hn} on $(0,t_n] \times \Sigma$. Moreover, assume that
    \[
    |D^m (t\lie_{\p_t})^r(t\bar K_n - t\bar K_{n-1})|_{\ho} \leq C_{m,r,n} \langle \ln t \rangle^{m+N_{n}} t^{2n\varepsilon},
    \]
    and for $i \neq k$, $x \in D_+$ and $y \in D_-$,
    \[
    \begin{split}
        \big|D^m (t\p_t)^r\big( (t\bar K_n - t\bar K_{n-1})(e_i,\omega^k) \big)|_{\ho}(x) &\leq C_{m,r,n} \langle \ln t \rangle^{m+N_{n}} t^{2n\varepsilon +2(p_i - p_1)(x)},\\
        \big|D^m (t\p_t)^r\big( (t\bar K_n - t\bar K_{n-1})(e_i,\omega^k) \big)|_{\ho}(y) &\leq C_{m,r,n} \langle \ln t \rangle^{m+N_{n}} t^{2n\varepsilon} \min\{ 1, t^{2(p_i - p_k)(y)} \}.
    \end{split}
    \]
    Then, if $h_n$ and $h_{n-1}$ are defined as in Lemma~\ref{definition of hn},
    \[
    | D^m (t\lie_{\partial_t})^r(\bar\h_n - \bar\h_{n-1}) |_{\ho} \leq C_{m,r,n} \langle \ln t \rangle^{m+N_{n}} t^{2n\varepsilon},
    \]
    and for $i \neq k$, $x \in D_+$ and $y \in D_-$,
    \[
    \begin{split}
        \big|D^m (t\p_t)^r \big((\bar\h_n - \bar\h_{n-1})(e_i,e_k) \big)\big|_{\ho}(x) &\leq C_{m,r,n} \langle \ln t  \rangle^{m+N_{n}} t^{2n\varepsilon + (p_i + p_k - 2p_1)(x)},\\
        \big|D^m (t\p_t)^r \big((\bar\h_n - \bar\h_{n-1})(e_i,e_k) \big)\big|_{\ho}(y) &\leq C_{m,r,n} \langle \ln t  \rangle^{m+N_{n}} t^{2n\varepsilon + |p_i - p_k|(y)}.
    \end{split}
    \]
\end{lemma}

\begin{proof}
    From Equation~\eqref{hn equation normalized},
    \[
    \begin{split}
        \lie_{\partial_{\tau}}(\bar\h_n - \bar\h_{n-1})(e_i,e_k) &= \sum_\ell e^{(p_i-p_{\ell})\tau}(\Ko - e^{-\tau}\bar K_n)(e_i,\omega^\ell)(\bar\h_n - \bar\h_{n-1})(e_\ell,e_k)\\
        &\quad + \sum_\ell e^{(p_k - p_\ell)\tau}(\Ko - e^{-\tau}\bar K_n)(e_k,\omega^\ell)(\bar\h_n - \bar\h_{n-1})(e_i,e_\ell)\\
        &\quad -\sum_\ell e^{(p_i - p_\ell)\tau}(e^{-\tau} \bar K_n - e^{-\tau}\bar K_{n-1})(e_i,\omega^\ell) \bar\h_{n-1}(e_\ell,e_k)\\
        &\quad -\sum_\ell e^{(p_k - p_\ell)\tau}(e^{-\tau}\bar K_n - e^{-\tau}\bar K_{n-1})(e_k,\omega^\ell)\bar\h_{n-1}(e_i,e_\ell).
    \end{split}
    \]
    So, by Lemma~\ref{linear ode lemma}, $\bar\h_n - \bar\h_{n-1}$ has to satisfy
    \[
    |D^m(\bar\h_n - \bar\h_{n-1})|_{\ho} \leq C_{m,n} \langle \tau \rangle^{m+N_{n}} e^{-2n\varepsilon\tau}.
    \]
    The improved estimates for the off-diagonal components are obtained similarly as in the proof of Lemma~\ref{definition of hn}. Then the estimates for the time derivatives follow directly by differentiating the equation above.
\end{proof}

\subsection{The scalar field} \label{the scalar field}

\begin{lemma} \label{definition of phin}
    Let $(\Sigma,\ho,\Ko,\phio,\psio)$ be initial data on the singularity and let $V$ be an admissible potential. Suppose we have a metric $g_{n} = -dt \otimes dt + h_{n}$, and functions $\bar \theta_n$ and $\s_{n-1}$ on $(0,t_n] \times \Sigma$, such that $h_n$ and $\s_{n-1}$ satisfy the assumptions of Lemma~\ref{decay of ricci}, and $\bar\theta_n$ satisfies the estimates
    \[
    |D^m (t\p_t)^r(t\bar \theta_n - 1)|_{\ho} \leq C_{m,r,n} \langle \ln t \rangle^{m+2} t^{2\varepsilon}.
    \]
    Then there is a unique $\s_n$ on $(0,t_n] \times \Sigma$ solving
    \begin{equation} \label{approximate scalar field equation}
        -\p_t^2 \s_n + \Delta_{h_{n}} \s_{n-1} - \bar \theta_n \p_t \s_n = V' \circ \s_{n-1},
    \end{equation}
    such that $\bar\Psi_n = t\p_t \s_n$ and $\bar\Phi_n = \s_n - \bar\Psi_n \ln t$ satisfy
    \[
    \begin{split}
        |D^m (t\p_t)^r(\bar\Psi_{n} - \psio)|_{\ho} &\leq C_{m,r,n} \langle \ln t \rangle^{m+2} t^{2\varepsilon},\\
        |D^m ( \bar\Phi_{n} - \phio )|_{\ho} &\leq C_{m,n} \langle \ln t \rangle^{m+3} t^{2\varepsilon}.
    \end{split}
    \]
\end{lemma}

\begin{proof}
    We can rewrite \eqref{approximate scalar field equation} in terms of $\bar\Psi_n$ and $\tau = -\ln t$ to obtain
    \begin{equation} \label{phin equation normalized}
        \p_{\tau} (\bar\Psi_n - \psio ) = (e^{-\tau}\bar \theta_n - 1)( \bar\Psi_n - \psio ) + (e^{-\tau} \bar \theta_n - 1 )\psio - e^{-2\tau} \Delta_{h_{n}} \s_{n-1} + e^{-2\tau} V' \circ \s_{n-1}.
    \end{equation}
    Note that by Lemma~\ref{decay of ricci},
    \[
    t^2 |D^m (t\p_t)^r ( V' \circ \s_{n-1} - \Delta_{h_{n}} \s_{n-1} )|_{\ho} \leq C_{m,r,n}\langle \ln t \rangle^{m+2} t^{2\varepsilon}.
    \]
    By Lemma~\ref{linear ode lemma}, we can thus define $\bar\Psi_n$ as the unique solution of \eqref{phin equation normalized} such that
    \[
    |D^m(\bar\Psi_n - \psio)|_{\ho} \leq C_{m,n} \langle \tau \rangle^{m+2} e^{-2\varepsilon\tau}.
    \]
    Now define
    \[
    \bar\Phi_n := \phio - \int_{\tau}^{\infty} s\p_s \bar\Psi_n(s)ds. 
    \]
    Then
    \[
    |D^m(\bar\Phi_n - \phio)|_{\ho} \leq C_{m,n} \langle \tau \rangle^{m+3} e^{-2\varepsilon\tau},
    \]
    and we can define $\s_n := \bar\Psi_n \ln t + \bar\Phi_n$. Finally, the estimates for the time derivatives follow directly from \eqref{phin equation normalized}.
\end{proof}

\begin{lemma} \label{phin become closer}
    Let $(\Sigma,\ho,\Ko,\phio,\psio)$ be initial data on the singularity and let $V$ be an admissible potential. Suppose we have metrics $g_n = -dt \otimes dt + h_n$ and $g_{n-1} = -dt \otimes dt + h_{n-1}$, and functions $\bar \theta_n$, $\bar \theta_{n-1}$, $\s_{n-1}$ and $\s_{n-2}$ on $(0,t_n] \times \Sigma$, such that $h_n$, $h_{n-1}$, $\s_{n-1}$ and $\s_{n-2}$ satisfy the assumptions of Lemma~\ref{decay of ricci}, and $\bar\theta_n$ and $\bar\theta_{n-1}$ satisfy the estimates
    \[
    |D^m (t\p_t)^r(t\bar \theta_a - 1)|_{\ho} \leq C_{m,r,n} \langle \ln t \rangle^{m+2} t^{2\varepsilon}, 
    \]
    for $a = n-1,n$. Moreover, assume that
    \[
    \begin{split}
        |D^m (t\lie_{\p_t})^r( \bar\h_{n} - \bar\h_{n-1} ) |_{\ho} &\leq C_{m,r,n} \langle \ln t \rangle^{m+N_{n}} t^{2n\varepsilon},\\
        |D^m (t\p_t)^r(\bar\Psi_{n-1} - \bar\Psi_{n-2})|_{\ho} &\leq C_{m,r,n-1} \langle \ln t \rangle^{m+N_{n-1}} t^{2(n-1)\varepsilon},\\
        |D^m (\bar\Phi_{n-1} - \bar\Phi_{n-2})|_{\ho} &\leq C_{m,n-1} \langle \ln t \rangle^{m+N_{n-1}} t^{2(n-1)\varepsilon},\\
        |D^m (t\p_t)^r(t\bar \theta_n - t\bar \theta_{n-1}) |_{\ho} &\leq C_{m,r,n} \langle \ln t \rangle^{m+N_{n}} t^{2n\varepsilon},
    \end{split}
    \]
    and for $i \neq k$, $x \in D_+$ and $y \in D_-$,
    \[
    \begin{split}
        \big|D^m (t\p_t)^r\big( (\bar\h_n - \bar\h_{n-1})(e_i,e_k) \big)\big|_{\ho}(x) &\leq C_{m,r,n} \langle \ln t  \rangle^{m+N_{n}} t^{2n\varepsilon + (p_i + p_k - 2p_1)(x)},\\
        \big|D^m (t\p_t)^r\big( (\bar\h_n - \bar\h_{n-1})(e_i,e_k) \big)\big|_{\ho}(y) &\leq C_{m,r,n} \langle \ln t  \rangle^{m+N_{n}} t^{2n\varepsilon + |p_i - p_k|(y)}.
    \end{split}
    \]
    Then, if $\s_n$ and $\s_{n-1}$ are defined as in Lemma~\ref{definition of phin},
    \[
    \begin{split}
        |D^m (t\p_t)^r(\bar\Psi_n - \bar\Psi_{n-1})|_{\ho} &\leq C_{m,r,n} \langle \ln t \rangle^{m+N_n} t^{2n\varepsilon},\\
        |D^m (\bar\Phi_n - \bar\Phi_{n-1})|_{\ho} &\leq C_{m,n} \langle \ln t \rangle^{m+N_n} t^{2n\varepsilon}.
    \end{split}
    \]
\end{lemma}

\begin{proof}
    From Equation~\eqref{approximate scalar field equation},
    \begin{equation} \label{psin difference equation}
    \begin{split}
        \p_{\tau}(\bar\Psi_n - \bar\Psi_{n-1}) &= (e^{-\tau} \bar \theta_{n-1} - 1)(\bar\Psi_n - \bar\Psi_{n-1}) + e^{-\tau}(\bar \theta_n - \bar \theta_{n-1}) \bar\Psi_n\\
        &\phantom{=} - e^{-2\tau}(\Delta_{h_n} \s_{n-1} - \Delta_{h_{n-1}} \s_{n-2}) +e^{-2\tau}( V' \circ \s_{n-1} - V' \circ \s_{n-2} ).
    \end{split}
    \end{equation}
    Furthermore, Corollary~\ref{estimates of differences} gives
    \[
    t^2|D^m (t\p_t)^r(\Delta_{h_n} \s_{n-1} - \Delta_{h_{n-1}} \s_{n-2} + V' \circ \s_{n-1} - V' \
    \circ \s_{n-2})|_{\ho} \leq C_{m,r,n} \langle \tau \rangle^{m+N_{n}} t^{2n\varepsilon}.
    \]
    By Lemma~\ref{linear ode lemma}, $\bar\Psi_n - \bar\Psi_{n-1}$ is then the unique decaying solution of \eqref{psin difference equation}, so it has to satisfy
    \[
    |D^m(\bar\Psi_n - \bar\Psi_{n-1})|_{\ho} \leq C_{m,n} \langle \tau \rangle^{m+N_{n}} e^{-2n\varepsilon\tau}.
    \]
    The estimate for the time derivatives follow directly from \eqref{psin difference equation}. Moreover, by definition of $\bar\Phi_n$ and $\bar\Phi_{n-1}$,
    \[
    \bar\Phi_n - \bar\Phi_{n-1} = \int_{\tau}^{\infty} s\p_s(\bar\Psi_{n-1} - \bar\Psi_n)(s)ds,
    \]
    from which the desired estimate follows.
\end{proof}

We are now in a position to construct the sequence of approximate solutions, and we do so in the following proposition. The following two subsections are then devoted to obtaining the remaining estimates required to finish the proof of Theorem~\ref{approximate solutions}.

\begin{proposition} \label{construction of the approximate solutions}
    Let $(\Sigma,\ho,\Ko,\phio,\psio)$ be initial data on the singularity and let $V$ be an admissible potential. Then for every non-negative integer $n$ there is a $t_n > 0$, a one parameter family of Riemannian metrics $h_n$, a one parameter family of $(1,1)$-tensors $\bar K_n$ and a function $\s_n$ on $(0,t_n] \times \Sigma$ such that the following holds. For $n \geq 1$, $h_n$, $\bar K_n$ and $\s_n$ satisfy the equations~\eqref{approximate metric equation}, \eqref{approximate weingarten map equation} and \eqref{approximate scalar field equation} respectively. Moreover, the convergence to initial data estimates of Theorem~\ref{approximate solutions} are satisfied with $K_n$ replaced by $\bar K_n$. Finally, the conclusions of Lemmas~\ref{kn become closer}, \ref{hn become closer} and \ref{phin become closer} hold. 
\end{proposition}

\begin{proof}
    Set $t_0 = 1$ and consider the velocity dominated solution, $g_0 = -dt \otimes dt + h_0$ and $\s_0$, associated with the initial data. Set $\bar K_0 := K_0$. We can now construct $(h_n,\bar K_n,\s_n)$ for every $n$ inductively, by using Lemmas~\ref{definition of kn}, \ref{definition of hn} and \ref{definition of phin} with $(h_0,\bar K_0,\s_0)$ as a starting point.
\end{proof}

\subsection{Estimating the error of the approximate Weingarten map} \label{comparing with the weingarten map}

In this subsection we show that $\bar K_n$, as in Proposition~\ref{construction of the approximate solutions}, is indeed an approximate Weingarten map for the $\Sigma_t$ hypersurfaces as $t \to 0$. Furthermore, the approximation becomes better as $n$ increases. The desired convergence to initial data for the actual Weingarten map then follows as a consequence. 

\begin{lemma} \label{error in the approximate weingarten map}
    Let $(\Sigma,\ho,\Ko,\phio,\psio)$ be initial data on the singularity and let $V$ be an admissible potential. Consider the sequence $(h_n, \bar K_n, \s_n)$ given by Proposition~\ref{construction of the approximate solutions}. If $K_n$ is the Weingarten map of the $\Sigma_t$ hypersurfaces with respect to $g_n := -dt \otimes dt + h_n$, then
    \[
    t| D^m (t\lie_{\partial_t})^r(K_n - \bar{K}_n)  |_{\ho} \leq C_{m,r,n} \langle \ln t \rangle^{m+N_{n}} t^{2(n+1)\varepsilon},
    \]
    and for $i \neq k$, $x \in D_+$ and $y \in D_-$,
    \[
    \begin{split}
        t \big|D^m (t\p_t)^r\big((K_n - \bar{K}_n)(e_i,\omega^k) \big)\big|_{\ho}(x) &\leq C_{m,r,n} \langle \ln t \rangle^{m+N_{n}} t^{2(n+1)\varepsilon +2(p_i - p_1)(x)},\\
        t \big|D^m (t\p_t)^r\big((K_n - \bar{K}_n)(e_i,\omega^k) \big)\big|_{\ho}(y) &\leq C_{m,r,n} \langle \ln t \rangle^{m+N_{n}} t^{2(n+1)\varepsilon} \min\{ 1, t^{2(p_i - p_k)(y)} \}.
    \end{split}
    \]
\end{lemma}

\begin{proof}
    For $X,Y \in \mfx(\Sigma)$, define $A_n$ by
    \[
    A_n(X,Y) := h_n(\bar K_n(X),Y) - h_n(X,\bar K_n(Y)),
    \]
    so that $A_n$ is the antisymmetric part of $\bar K_n$ with respect to $h_n$. Also define $\ca_n$, the expansion normalized version of $A_n$, by
    \[
    \begin{split}
        \mathcal{A}_n(X,Y) &:= tA_n(t^{-\Ko}(X),t^{-\Ko}(Y))\\
        &= h_n( t\bar K_n \circ t^{-\Ko}(X),t^{-\Ko}(Y) ) - h_n(t^{-\Ko}(X), t\bar K_n \circ t^{-\Ko}(Y)).
    \end{split}
    \]
    Then $\ca_n$ satisfies the equation
    \begin{equation} \label{equation for an}
        \begin{split}
        \lie_{\partial_{\tau}} \mathcal{A}_n(X,Y) &= (e^{-\tau} \bar \theta_n - 1) \mathcal{A}_n(X,Y) + \mathcal{A}_n( e^{-\tau\Ko} \circ ( \Ko - e^{-\tau}\bar K_n ) \circ e^{\tau\Ko}(X), Y )\\
        &\quad + \mathcal{A}_n( X, e^{-\tau\Ko} \circ (\Ko - e^{-\tau}\bar K_n) \circ e^{\tau\Ko}(Y) )\\
        &\quad + (\bar\h_n - \bar\h_{n-1})( e^{-\tau\Ko} \circ e^{-2\tau}( \sric_{n-1}^{\sharp} - d\s_{n-1} \otimes \sn \s_{n-1} ) \circ e^{\tau\Ko}(X), Y )\\
        &\quad - (\bar\h_n - \bar\h_{n-1})( X, e^{-\tau\Ko} \circ e^{-2\tau}( \sric_{n-1}^{\sharp} - d\s_{n-1} \otimes \sn \s_{n-1} ) \circ e^{\tau\Ko}(Y) ).
        \end{split}
    \end{equation}
    We want to apply Lemma~\ref{linear ode lemma}, so we need to show that $\mathcal{A}_n$ decays in $\tau$. Indeed, note that
    \[
    \begin{split}
        \mathcal{A}_n(X,Y) &= \bar\h_n( t^{\Ko} \circ ( t\bar K_n - \Ko ) \circ t^{-\Ko}(X), Y ) - \bar\h_n( X, t^{\Ko} \circ ( t\bar K_n - \Ko ) \circ t^{-\Ko}(Y) )\\
        &\quad + (\bar\h_n - \ho)(\Ko(X),Y) - (\bar\h_n - \ho)(X,\Ko(Y))
    \end{split}
    \]
    which clearly decays as we want. Thus, by Lemma~\ref{linear ode lemma}, we obtain
    \[
    |D^m (t\lie_{\p_t})^r \mathcal{A}_n|_{\ho} \leq C_{m,r,n}  \langle \ln t \rangle^{m+N_{n}} t^{2(n+1)\varepsilon},
    \]
    where the estimates for the time derivatives follow directly from \eqref{equation for an}. Moreover, for $i \neq k$, $x \in D_+$ and $y \in D_-$,
    \[
    \begin{split}
        \big|D^m (t\p_t)^r \big(\mathcal{A}_n(e_i,e_k) \big)\big|_{\ho}(x) &\leq C_{m,r,n} \langle \ln t \rangle^{m+N_{n}} t^{2(n+1)\varepsilon +(p_i + p_k - 2p_1)(x)},\\
        \big|D^m (t\p_t)^r \big(\mathcal{A}_n(e_i,e_k) \big)\big|_{\ho}(y) &\leq C_{m,r,n} \langle \ln t \rangle^{m+N_{n}} t^{2(n+1)\varepsilon+ |p_i - p_k|(y)}, 
    \end{split}
    \]
    which follow similarly as in the proof of Lemma~\ref{definition of hn}. Going back to $A_n$, we see that
    \[
    t A_n(e_i,e_k) = t^{p_i + p_k} \mathcal{A}_n(e_i,e_k).
    \]
    By repeatedly taking $t\p_t$ of this equality, we obtain, for $i \neq k$, $x \in D_+$ and $y \in D_-$,
    \[
    \begin{split}
        t\big| D^m (t\p_t)^r \big(A_n(e_i,e_k) \big)\big|_{\ho}(x) &\leq C_{m,r,n} \langle \ln t \rangle^{m+N_{n}} t^{2(n+1)\varepsilon + 2(p_i + p_k - p_1)(x)},\\
        t\big| D^m (t\p_t)^r \big(A_n(e_i,e_k) \big)\big|_{\ho}(y) &\leq C_{m,r,n} \langle \ln t \rangle^{m+N_{n}} t^{2(n+1)\varepsilon + 2p_{\max\{i,k\}}(x)}.
    \end{split}
    \]
    Now we raise an index with $h_n$,
    \[
    \begin{split}
        t|A_n^{\sharp}(e_i,\omega^i)| = t|(h_n)^{i\ell}A_n(e_i,e_\ell)| &\leq C_{n} \langle \ln t \rangle^{N_{n}} t^{2(n+1)\varepsilon} \sum_\ell t^{-2p_{\min\{i,\ell\}} + 2p_{\max\{i,\ell\}}}\\
        &\leq C_{n} \langle \ln t \rangle^{N_{n}} t^{2(n+1)\varepsilon}
    \end{split}
    \]
    (no summation over $i$). Moreover, for $i \neq k$, $x \in D_+$ and $y \in D_-$,
    \[
    t|A_n^{\sharp}(e_i,\omega^k)|(x) = t|(h_n)^{k\ell}A_n(e_i,e_\ell)|(x) \leq C_{n} \langle \ln t \rangle^{N_{n}} t^{2(n+1)\varepsilon + 2(p_i - p_1)(x)}
    \]
    and
    \[
    \begin{split}
        t|A_n^{\sharp}(e_i,\omega^k)|(y) &= t|(h_n)^{k\ell}A_n(e_i,e_\ell)|(y)\\
        &\leq C_{n} \langle \ln t \rangle^{N_{n}} t^{2(n+1)\varepsilon} \sum_\ell t^{-2p_{\min\{k,\ell\}}(y) + 2p_{\max\{i,\ell\}}(y)}\\
        &\leq C_{n} \langle \ln t \rangle^{N_{n}} t^{2(n+1)\varepsilon} \min\{1,t^{2(p_i - p_k)(y)}\}.
    \end{split}
    \]
    As usual, derivatives just introduce factors of $\langle \ln t \rangle$, if they are spatial. Thus,
    \[
    \begin{split}
        t\big|D^m (t\partial_t)^r\big( A_n^{\sharp}(e_i,\omega^i)\big)\big| &\leq C_{m,r,n} \langle \ln t \rangle^{m+N_{n}} t^{2(n+1)\varepsilon},\\
        t\big|D^m (t\partial_t)^r \big(A_n^{\sharp}(e_i,\omega^k)\big)\big|(x) &\leq C_{m,r,n} \langle \ln t \rangle^{m+N_{n}} t^{2(n+1)\varepsilon + 2(p_i - p_1)(x)},\\ 
        t\big|D^m (t\partial_t)^r \big(A_n^{\sharp}(e_i,\omega^k)\big)\big|(y) &\leq C_{m,r,n} \langle \ln t \rangle^{m+N_{n}} t^{2(n+1)\varepsilon} \min\{1,t^{2(p_i-p_k)(y)}\},
    \end{split}
    \]
    for $i \neq k$ (no summation over $i$), $x \in D_+$ and $y \in D_-$.

    Moving on to the second fundamental form,
    \[
    \begin{split}
        h_n( (\bar K_n - K_n)(X), Y ) &= h_n(\bar K_n(X),Y) - \frac{1}{2} \lie_{\partial_t} h_n(X,Y)\\
        &= \frac{1}{2} \Big( h_n(\bar K_n(X),Y) - h_n(X,\bar K_n(Y)) \Big)\\
        &= \frac{1}{2} A_n(X,Y),
    \end{split}
    \]
    hence $\bar K_n - K_n = \frac{1}{2} A_n^{\sharp}$. The lemma follows from the estimates for $A_n^{\sharp}$.
\end{proof}

\subsection{Estimating the error in Einstein's equations} \label{the sequence satisfies the eqs asymptotically}

In this subsection, we show that the sequence constructed in Proposition~\ref{construction of the approximate solutions} indeed consists of approximate solutions to the Einstein--nonlinear scalar field equations with potential $V$ as $t \to 0$. Thus, we finish the proof of Theorem~\ref{approximate solutions}.

\begin{lemma}
    Let $(\Sigma,\ho,\Ko,\phio,\psio)$ be initial data on the singularity and let $V$ be an admissible potential. The sequence $(h_n, \bar K_n, \s_n)$, given by Proposition~\ref{construction of the approximate solutions}, satisfies
    \begin{align} 
    t^2|D^m(\bar{S}_n - \tr \bar K_n^2 + \bar \theta_n^2 - (\p_t \s_n)^2 - |d\s_n|_{h_n}^2 - 2V \circ \s_n)|_{\ho} &\leq C_{m,n} \langle \ln t \rangle^{m+2} t^{2\varepsilon}\label{approximate hamiltonian constraint},\\
    t^2|D^m(\partial_t \bar \theta_n + \tr \bar K_n^2 + (\p_t \s_n)^2 - V \circ \s_n)|_{\ho} &\leq C_{m,n} \langle \ln t \rangle^{m+2} t^{2\varepsilon}\label{approximate theta equation},\\ 
    t|D^m(\diver_{h_n} \bar K_n - d\bar \theta_n - (\p_t \s_n)d\s_n)|_{\ho} &\leq C_{m,n} \langle \ln t \rangle^{m+5} t^{2\varepsilon} \label{approximate momentum constraint}.
    \end{align}
    Moreover, if we define
    \[
    \bar\ce_n := \lie_{\p_t} \bar K_n + \sric_n^\sharp + \bar\theta_n \bar K_n - d\s_n \otimes \sn\s_n - (V \circ \s_n)I,
    \]
    then
    \begin{equation} \label{approximate equations 2}
        t^2|D^m(t\lie_{\partial_t})^r \bar\ce_n|_{\ho} \leq C_{m,r,n} \langle \ln t \rangle^{m+N_{n}} t^{2(n+1)\varepsilon},
    \end{equation}
    and for $i \neq k$, $x \in D_+$ and $y \in D_-$,
    \[
    \begin{split}
        t^2\big|D^m (t\partial_t)^r\big( \bar\ce_n(e_i,\omega^k)\big)\big|_{\ho}(x) &\leq C_{m,r,n} \langle \ln t \rangle^{m+N_{n}} t^{2(n+1)\varepsilon + 2(p_i - p_1)(x)},\\
        t^2\big|D^m (t\partial_t)^r\big( \bar\ce_n(e_i,\omega^k)\big)\big|_{\ho}(y) &\leq C_{m,r,n} \langle \ln t \rangle^{m+N_{n}} t^{2(n+1)\varepsilon} \min\{ 1, t^{2(p_i - p_k)(y)} \}.
    \end{split}
    \]
\end{lemma}

\begin{proof}
    The estimates for $\bar\ce_n$ follow immediately from Corollary~\ref{estimates of differences}, by noting that
    \[
    \bar\ce_n = \sric_n^{\sharp} - \sric_{n-1}^{\sharp} + d\s_{n-1} \otimes \sn \s_{n-1} - d\s_n \otimes \sn \s_n + (V \circ \s_{n-1} - V \circ \s_n)I,
    \]
    which follows from Equation~\eqref{approximate weingarten map equation}.

    For \eqref{approximate hamiltonian constraint}, from Lemma~\ref{decay of ricci}, we already know that
    \[
    t^2|D^m (\bar{S}_n - |d\s_n|_{h_n}^2 - 2V\circ \s_n )|_{\ho} \leq C_{m,n} \langle \ln t \rangle^{m+2} t^{2\varepsilon}.
    \]
    Moreover,
    \[
    \begin{split}
        t^2(\bar \theta_n^2 - \tr \bar K_n^2 - (\p_t \s_n)^2) &= (t \bar \theta_n - 1)^2 + 2(t\bar \theta_n - 1) - \tr(t\bar K_n - \Ko)^2\\
        &\quad - \tr\big( \Ko \circ (t\bar K_n - \Ko) \big) - \tr\big( (t\bar K_n - \Ko) \circ \Ko \big) + (\psio + \bar\Psi_n)( \psio - \bar\Psi_n ),
    \end{split}
    \]
    which can be estimated similarly. 

    For \eqref{approximate theta equation}, by estimating the trace of $\bar\ce_n$, we obtain
    \[
    t^2|D^m( \partial_t \bar \theta_n + \bar{S}_n + \bar \theta_n^2 - |d\s_n|_{h_n}^2 - 3V \circ \s_n )|_{\ho} \leq C_{m,n} \langle \ln t \rangle^{m+N_{n}} t^{2(n+1)\varepsilon}.
    \]
    This together with \eqref{approximate hamiltonian constraint} yields \eqref{approximate theta equation}.

    Finally, \eqref{approximate momentum constraint}. Define $\cd_n := \sn^{(n)} - \sn^{(0)}$, where $\sn^{(n)}$ and $\sn^{(0)}$ are the Levi-Civita connections of $h_n$ and $h_0$ respectively. By Corollary~\ref{estimates of differences}, if $|\alpha| \leq m$, $x \in D_+$ and $y \in D_-$, we obtain
    \begin{equation} \label{difference of connections estimates}
    \begin{split}
        |e_{\alpha} (\cd_n)_{ii}^{\ell}| &\leq C_{m,n} \langle \ln t \rangle^{m+3} t^{2(p_i - p_{\ell}) + 2\varepsilon},\\
        |e_{\alpha} (\cd_n)_{ik}^i| + |e_{\alpha} (\cd_n)_{ik}^k| &\leq C_{m,n} \langle \ln t \rangle^{m+3} t^{2\varepsilon},\\
        |e_{\alpha} (\cd_n)_{ik}^{\ell}|(x) &\leq C_{m,n} \langle \ln t \rangle^{m+3} t^{2(p_1 - p_{\ell})(x) + 2\varepsilon},\\
        |e_{\alpha} (\cd_n)_{ik}^{\ell}|(y) &\leq C_{m,n} \langle \ln t \rangle^{m+3} t^{2(p_2 - p_{\ell})(y) + 2\varepsilon}, 
    \end{split}
    \end{equation}
    where $i$, $k$ and $\ell$ are distinct in the last two inequalities, and there is no summation over $i$ or $k$. Now we compute,
    \begin{align*}
        t \,\diver_{h_n} \bar K_n(e_i) &= t \sn^{(n)}_{e_k} \bar K_n(e_i,\omega^k)\\
        &= e_k\big( t\bar K_n(e_i,\omega^k) \big) - t\bar K_n( \sn^{(n)}_{e_k} e_i, \omega^k ) - t\bar K_n(e_i, \sn^{(n)}_{e_k} \omega^k)\\
        &= e_k\big( t\bar K_n(e_i,\omega^k) \big) - t\bar K_n( \sn^{(0)}_{e_k} e_i, \omega^k ) - t\bar K_n(e_i, \sn^{(0)}_{e_k} \omega^k)\\
        &\quad - t\bar K_n( \cd_n(e_k,e_i), \omega^k ) - t\bar K_n\big( e_i, (\sn^{(n)}_{e_k} - \sn^{(0)}_{e_k}) \omega^k \big)\\
        &= \sn^{(0)}_{e_k} \Ko(e_i,\omega^k) + e_k\big( (t\bar K_n - \Ko)(e_i,\omega^k) \big) - (t\bar K_n - \Ko)(\sn^{(0)}_{e_k} e_i,\omega^k)\\
        &\quad - (t\bar K_n - \Ko)(e_i, \sn^{(0)}_{e_k} \omega^k) - (t\bar K_n - \Ko)(\cd_n(e_k,e_i), \omega^k)\\
        &\quad - (t\bar K_n - \Ko)\big(e_i,(\sn^{(n)}_{e_k} - \sn^{(0)}_{e_k}) \omega^k\big) - \Ko( \cd_n(e_k,e_i), \omega^k )\\
        &\quad - \Ko\big(e_i, (\sn^{(n)}_{e_k} - \sn^{(0)}_{e_k}) \omega^k\big)\\
        &= t \,\diver_{h_0} K_0(e_i) + e_k\big( (t\bar K_n - \Ko)(e_i,\omega^k) \big) - (\bar{\Gamma}_0)_{k i}^\ell (t\bar K_n - \Ko)(e_\ell,\omega^k)\\
        &\quad + (\bar{\Gamma}_0)_{k \ell}^k (t\bar K_n - \Ko)(e_i,\omega^\ell) - (\cd_n)_{k i}^\ell (t\bar K_n - \Ko)(e_\ell, \omega^k)\\
        &\quad + (\cd_n)_{k \ell}^k(t\bar K_n - \Ko)(e_i,\omega^\ell) - \sum_k \Big( p_k (\cd_n)_{k i}^k - p_i (\cd_n)_{k i}^k \Big),
    \end{align*} 
    where $(\bar\Gamma_0)_{ik}^\ell := \omega^\ell(\sn^{(0)}_{e_i} e_k)$. Moreover,
    \[
    (t\p_t \s_n)d\s_n = (t\p_t \s_0)d\s_0 + t\p_t(\s_n - \s_0)d\s_n + t\p_t \s_0 d(\s_n - \s_0).
    \]
    Hence we can use the fact that $t(\diver_{h_0}K_0 - (\p_t\s_0)d\s_0) = \cm_0 = 0$ (See Lemma~\ref{velocity dominated solution} and Proposition~\ref{ricci}), in addition to \eqref{difference of connections estimates} and the estimates for $(\bar\Gamma_0)_{ik}^\ell$ coming from Lemma~\ref{decay of ricci}, to obtain
    \[
    t|D^m( \diver_{h_n} \bar K_n - (\p_t \s_n)d\s_n )|_{\ho} \leq C_{m,n} \langle \ln t \rangle^{m+5} t^{2\varepsilon}.
    \]
    This together with what we know about $\bar\theta_n$ yields \eqref{approximate momentum constraint}.
\end{proof}

The above lemma already implies that $E_n$, as in Theorem~\ref{approximate solutions} and with the appropriate normalization, decays as $t \to 0$. So we only need a way to obtain improvements. In order to do so, we now obtain evolution equations for the relevant quantities by using the fact that the Einstein tensor is divergence free.  

\begin{lemma} \label{evolution equations for the constraint equations}
    Let $\Sigma$ be a manifold, consider a metric $g = -dt \otimes dt + h$ and a function $\s$ on $(0,T] \times \Sigma$ and let $V \in C^\infty(\R)$. Define
    \[
    E := \ric - d\s \otimes d\s - (V \circ \s)g.
    \]
    Also define the one parameter family of $(1,1)$-tensors $\ce$ and the one parameter family of one forms $\cm$ on $\Sigma$ by
    \[
    h(\ce(X),Y) := E(X,Y), \qquad \cm(X) := E(\p_t,X),
    \]
    for $X,Y \in \mfx(\Sigma)$. Then
    \begin{subequations}
    \begin{align}
        \p_t\big( E(\p_t,\p_t) \big) &= -2\theta E(\p_t,\p_t) - 2\tr( \ce \circ K ) + 2\diver_{h} \cm -\p_t( \tr \ce ) + 2( \Box_{g} \s - V' \circ \s ) \p_t \s,\label{einstein tt evolution equation}\\
        \lie_{\p_t} \cm &= -\theta \cm + \diver_{h} \ce + \frac{1}{2}  d\big( E(\p_t,\p_t) - \tr \ce \big) + ( \Box_{g} \s - V' \circ \s ) d\s.\label{momentum constraint evolution equation}
    \end{align}
    \end{subequations}
\end{lemma}

\begin{proof}
    Note that
    \[
    \diver_{g} E = \frac{1}{2} dS - (\Box_{g} \s) d\s - \n_{\n \s} d\s - (V' \circ \s)d\s.
    \]
    Moreover,
    \[
    \tr_{g} E = S - |d\s|_{g}^2 - 4(V \circ \s)
    \]
    and $d( |d\s|_{g}^2 ) = 2\n_{\n \s} d\s$. Thus
    \begin{equation} \label{divergence of einstein scalar field tensor}
        \diver_{g} E = \frac{1}{2} d(\tr_{g}E) + ( V' \circ \s - \Box_{g} \s )d\s.
    \end{equation}
    We compute,
    \[
    \begin{split}
        (V' \circ \s - \Box_g \s) \p_t\s &= \bigg( \diver_{g} E - \frac{1}{2}d(\tr_{g}E) \bigg)(\p_t)\\
        &= - \n_{\p_t} E(\p_t,\p_t) + h^{\ell m} \n_{e_\ell} E(\p_t,e_m) - \frac{1}{2} \p_t \big( - E(\p_t,\p_t) + \tr \ce \big),
    \end{split}
    \]
    and
    \[
    \begin{split}
        h^{\ell m}\n_{e_\ell} E(\p_t,e_m) &= h^{\ell m} \Big( e_\ell\big( \cm(e_m) \big) - E(\n_{e_\ell} \p_t, e_m) - E(\p_t, \n_{e_\ell} e_m) \Big)\\
        &= h^{\ell m} \Big( \sn_{e_\ell} \cm(e_m) - E(K(e_\ell), e_m) - k(e_\ell,e_m) E(\p_t,\p_t) \Big)\\
        &= \diver_h \cm - \tr( \ce \circ K ) - \theta E(\p_t,\p_t).
    \end{split}
    \]
    This yields \eqref{einstein tt evolution equation}. Moreover, if $X \in \mathfrak{X}(\Sigma)$,
    \[
    \begin{split}
        \bigg( \diver_{g} E - \frac{1}{2}d(\tr_{g}E) \bigg)(X) &= -\n_{\p_t} E(\p_t, X) + h^{\ell m}\n_{e_\ell} E(e_m,X)-\frac{1}{2} X\big( -E(\p_t,\p_t) + \tr \ce \big)\\
        &= -\p_t\big( \cm(X) \big) + E(\p_t, \n_{\p_t} X) + h^{\ell m}e_\ell\big( E(e_m,X) \big)\\
        &\quad - h^{\ell m}E( \n_{e_\ell} e_m, X) - h^{\ell m} E( e_m, \n_{e_\ell} X ) + \frac{1}{2} X\big( E(\p_t,\p_t)  - \tr \ce \big)\\
        &= -\p_t\big( \cm(X) \big) + \diver_{h} \ce(X) - \theta \cm(X) + \frac{1}{2} X\big( E(\p_t,\p_t)  - \tr \ce \big),
    \end{split}
    \]
    where we have used that $h^{\ell m}\n_{e_\ell} e_m = \theta \p_t + h^{\ell m}\sn_{e_\ell} e_m$ and that $\n_{\p_t} X = K(X)$. Equation~\eqref{momentum constraint evolution equation} now follows from \eqref{divergence of einstein scalar field tensor}.
\end{proof}

\begin{lemma} \label{the Einstein scalar field tensor decays very fast}
    Let $(\Sigma,\ho,\Ko,\phio,\psio)$ be initial data on the singularity and let $V$ be an admissible potential. The sequence $(h_n, \bar K_n, \s_n)$, given by Proposition~\ref{construction of the approximate solutions}, satisfies
    \begin{equation*}
    \begin{split}
        t^2|D^m( t\lie_{\partial_t})^r\ce_n |_{\ho} &\leq C_{m,r,n} \langle \ln t \rangle^{m+N_n} t^{2(n+1)\varepsilon},\\
        t^2\big|D^m (t\p_t)^r\big(E_n(\partial_t,\partial_t)\big)\big|_{\ho} + t|D^m (t\lie_{\p_t})^r \cm_n|_{\ho} &\leq C_{m,r,n} \langle \ln t \rangle^{m+N_n} t^{2(n+1)\varepsilon},\\
        t^2|D^m (t\p_t)^r( \Box_{g_n} \s_n - V' \circ \s_n )|_{\ho} &\leq C_{m,r,n} \langle \ln t \rangle^{m+N_{n}} t^{2(n+1)\varepsilon},
    \end{split}
    \end{equation*}
    where $g_n = -dt \otimes dt + h_n$. Moreover, for $i \neq k$, $x \in D_+$ and $y \in D_-$,
    \[
    \begin{split}
        t^2 \big|D^m (t\p_t)^r \big(\ce_n(e_i,\omega^k) \big)\big|_{\ho}(x) &\leq C_{m,r,n} \langle \ln t \rangle^{m+N_n} t^{2(n+1)\varepsilon + 2(p_i - p_1)(x)},\\
        t^2 \big|D^m (t\p_t)^r \big(\ce_n(e_i,\omega^k) \big)\big|_{\ho}(y) &\leq C_{m,r,n} \langle \ln t \rangle^{m+N_n} t^{2(n+1)\varepsilon} \min\{1, t^{2(p_i - p_k)(y)}\}.
    \end{split}
    \]
\end{lemma}

\begin{proof}
    We begin with $\ce_n$. Note that
    \[
    \begin{split}
        \ce_n &= \bar\ce_n + \lie_{\partial_t}(K_n - \bar{K}_n) + (\theta_n - \bar{\theta}_n)K_n + \bar \theta_n(K_n - \bar{K}_n),
    \end{split}
    \]
    from which the estimates immediately follow.

    For $\s_n$, just note that
    \[
    \begin{split}
        \Box_{g_n} \s_n - V' \circ \s_n &= -\p_t^2 \s_n + \Delta_{h_n} \s_n - \theta_n \p_t \s_n - V' \circ \s_n\\
        &= \Delta_{h_n}( \s_n - \s_{n-1} ) + (\bar\theta_n - \theta_n)\p_t \s_n + V' \circ \s_{n-1} - V' \circ \s_n, 
    \end{split}
    \]
    and
    \[
    \Delta_{h_n}(\s_n - \s_{n-1}) = (h_n)^{\ell m}( e_\ell e_m - (\bar{\Gamma}_n)_{\ell m}^a e_a )(\s_n - \s_{n-1}),
    \]
    where $(\bar\Gamma_n)_{ik}^\ell := \omega^\ell(\sn^{(n)}_{e_i} e_k)$.

    For the remaining estimates, we write expansion normalized versions of the evolution equations in Lemma~\ref{evolution equations for the constraint equations};
    \begin{subequations} \label{einstein scalar field tensor normalized equations}
    \begin{align}
        \begin{split}
            \p_{\tau} \big(E_n(\p_{\tau},\p_{\tau})\big) &= 2(e^{-\tau} \theta_n - 1) E_n(\p_{\tau},\p_{\tau}) + 2e^{-3\tau} \tr( \ce_n \circ K_n ) - 2e^{-3\tau} \diver_{h_n} \cm_n\\
            &\quad - e^{-2\tau} \p_{\tau}( \tr \ce_n ) + 2e^{-2\tau}( \Box_{g_n} \s_n - V' \circ \s_n ) \p_{\tau} \s_n,
        \end{split}\label{einstein scalar field tensor normalized equations 1}\\
        \begin{split}
            \lie_{\p_{\tau}}( e^{-\tau}\cm_n ) &= (e^{-\tau} \theta_n - 1) e^{-\tau} \cm_n - e^{-2\tau} \diver_{h_n} \ce_n + \frac{1}{2}d\big( 
            e^{-2\tau} \tr \ce_n - E_n(\p_{\tau},\p_{\tau}) \big)\\
            &\quad - e^{-2\tau}( \Box_{g_n} \s_n - V' \circ \s_n ) d\s_n.
        \end{split} \label{einstein scalar field tensor normalized equations 2}
    \end{align}
    \end{subequations}
    Note that from Proposition~\ref{ricci}, Lemma~\ref{error in the approximate weingarten map}, \eqref{approximate theta equation} and \eqref{approximate momentum constraint}, it follows that
    \[
    t^2\big|D^m\big(E_n(\p_t,\p_t)\big)\big|_{\ho} \leq C_{m,n} \langle \ln t \rangle^{m+2} t^{2\varepsilon}, \qquad t|D^m \cm_n|_{\ho} \leq C_{m,n} \langle \ln t \rangle^{m+5} t^{2\varepsilon}.
    \]
    So in order to improve on these estimates, we just need to ensure that the inhomogeneous terms in \eqref{einstein scalar field tensor normalized equations} decay as desired. Indeed,
    \[
    \begin{split}
        t^2 \diver_{h_n} \ce_n(e_i) &= t^2 \sn_{e_\ell}^{(n)} \ce_n(e_i,\omega^\ell)\\
        &= t^2\big( e_\ell( \ce_n(e_i,\omega^\ell)) - (\bar{\Gamma}_n)_{\ell i}^m \ce_n(e_m,\omega^\ell) + (\bar{\Gamma}_n)_{\ell m}^\ell \ce_n(e_i,\omega^m) \big).
    \end{split}
    \]
    Therefore, the off-diagonal improvements on the estimates for $\ce_n$ ensure that 
    \[
    t^2|D^m (t\lie_{\p_t})^r \diver_{h_n} \ce_n|_{\ho} \leq C_{m,r,n} \langle \ln t \rangle^{m+N_n} t^{2(n+1)\varepsilon}. 
    \]
    The rest of the terms are clear. 
    
    Now we are ready to improve on the estimates for $E_n(\p_\tau,\p_\tau)$ and $\cm_n$. First note that there is a potential loss of derivatives in the system \eqref{einstein scalar field tensor normalized equations}, since the right-hand side contains $\diver_{h_n} \cm_n$ and $d( E_n(\p_\tau,\p_\tau) )$. But this is not an issue for us, since we already have estimates for spatial derivatives of all orders for $E_n(\p_\tau,\p_\tau)$ and $\cm_n$. We start by plugging in the estimate for $\cm_n$ into \eqref{einstein scalar field tensor normalized equations 1}, which yields an improvement for $E_n(\p_\tau,\p_\tau)$. Then we can plug the improved estimate into \eqref{einstein scalar field tensor normalized equations 2}, which now yields an improvement for $\cm_n$. We can iterate this process to obtain further improvements, until we are stopped by the terms involving $\ce_n$ and $\s_n$, at which point we will have obtained the desired result without time derivatives. The estimates for the time derivatives then follow directly from \eqref{einstein scalar field tensor normalized equations}. 
\end{proof}

\begin{proof}[Proof of Theorem~\ref{approximate solutions}]
    Take the sequence $(h_n,\bar K_n,\s_n)$ as given by Proposition~\ref{construction of the approximate solutions}. It only remains to show that the estimates for $K_n$, $E_n$ and $\Box_{g_n}\s_n - V' \circ \s_n$ hold. The convergence estimates for $K_n$ follow from Lemma~\ref{error in the approximate weingarten map} and those for $\bar K_n$. Finally, the estimates for $E_n$ and $\Box_{g_n}\s_n - V' \circ \s_n$ follow from Lemma~\ref{the Einstein scalar field tensor decays very fast}. 
\end{proof}

\section{Existence of developments} \label{construction of the solution}

In order to construct the desired solution to Einstein's equations, we want to derive a second order equation for the second fundamental form. Consider a metric $g = -dt \otimes dt + h$ and a function $\s$ on $(0,T] \times \Sigma$ which do not necessarily satisfy the Einstein--nonlinear scalar field equations with potential $V$. As before, define
\[
    E := \ric - d\s \otimes d\s - (V \circ \s)g, \qquad
    h(\ce(X),Y) := E(X,Y), \qquad \cm(X) := E(\p_t,X).
\]
Then Proposition~\ref{ricci} implies that the Weingarten map $K$ satisfies the equation
\begin{equation} \label{k first order equation}
    \lie_{\p_t} K + \sric^{\sharp} + (\tr K) K = d\s \otimes \sn \s + (V \circ \s)I + \ce.
\end{equation}
Now we want to take $\lie_{\p_t}$ of this equation. To that end, first note that
\[
[\p_t, \sn \s] = \sn( \p_t \s ) - 2K(\sn \s).
\]
Next, by the first variation formula for the Ricci tensor (see, for instance, \cite[Equation~(2.31)]{chow_hamiltons_2006}),
\[
\lie_{\p_t} \sric = -\Delta_h k - 2\mathring{\bar{R}} k + \lie_{(\diver_h k)^{\sharp}} h - \sn^2 (\tr K) + \sric \circ k + k \circ \sric,
\]
where we use the notation $h \circ k := (h^{\sharp} \circ k^{\sharp})^{\flat}$ for $h$ and $k$ symmetric covariant 2-tensors, and $\mathring{\bar{R}} k( X,Y ) := \tr_h( k( \bar{R}(\,\cdot\,,X)Y,\,\cdot\, ) )$. To deal with $\mathring{\bar{R}}$, we use that in dimension 3, the Riemann tensor is completely determined by the Ricci tensor by the formula
\begin{equation} \label{curvature in terms of ricci}
    \begin{split}
    h(\bar{R}(X,Y)Z,W) &= -\sric(X,Z)h(Y,W) + \sric(X,W)h(Y,Z)\\
    &\quad -\sric(Y,W)h(X,Z) + \sric(Y,Z)h(X,W)\\
    &\quad -\frac{1}{2}\bar{S}\big( h(X,W)h(Y,Z) - h(X,Z)h(Y,W) \big)
\end{split}
\end{equation}
(see \cite[Equation (1.62)]{chow_hamiltons_2006}). Using this, we conclude that
\[
    \mathring{\bar{R}} k = -k \circ \sric -\sric \circ k + h( \sric, k ) h + (\tr K) \sric -\frac{1}{2} \bar{S}\big( (\tr K) h - k \big).
\]
Therefore,
\[
\begin{split}
    \lie_{\p_t} \sric^{\sharp} &= -\Delta_h K + 3\sric^{\sharp} \circ K + K \circ \sric^{\sharp} - 2\tr( \sric^\sharp \circ K )I\\
    &\quad -2(\tr K) \sric^{\sharp} + \bar{S}\big( (\tr K) I - K \big) + \big( \lie_{(\diver_h k)^{\sharp}} h \big)^{\sharp} - \sn^2 (\tr K)^{\sharp}.
\end{split}
\]
Moving on, note that
\[
\lie_{\cm^{\sharp}} h = \lie_{(\diver_h k)^{\sharp}}h - 2\sn^2 (\tr K) - d(\p_t \s) \otimes d\s - d\s \otimes d(\p_t \s) - 2(\p_t \s) \sn^2 \s.
\]
Now, from \eqref{k first order equation},
\[
\lie_{\p_t}^2 K + \lie_{\p_t} \sric^{\sharp} + \lie_{\p_t}\big( (\tr K) K \big) = \lie_{\p_t} d\s \otimes \sn \s + d\s \otimes [\p_t, \sn \s] + (V' \circ \s)\p_t \s I + \lie_{\p_t} \ce.
\]
By putting together all our previous observations, and using \eqref{k first order equation} again to express $\sric^{\sharp}$ in terms of $K$, $\s$ and $\ce$, we conclude that $K$ satisfies the equation
\begin{equation} \label{k second order equation}
    \lie_{\p_t}^2 K - \Delta_h K + \sn^2 (\tr K)^{\sharp} + F_1(K) + F_2(K) + F_3(\s) = G(E),
\end{equation}
where
\begin{align*}
    F_1(K) &:= -4(\tr K) K^2 + 2(\tr K) (\tr K^2) I + 3 (\tr K)^2 K - (\tr K)^3 I,\\
    \begin{split}
        F_2(K) &:= -3\lie_{\p_t} K \circ K - K \circ \lie_{\p_t} K + 2\tr( \lie_{\p_t} K \circ K ) I + 2 (\tr K) \lie_{\p_t} K\\
        &\quad - (\p_t (\tr K)) (\tr K) I + (\p_t (\tr K)) K + \lie_{\p_t}\big( (\tr K) K \big),
    \end{split}\\
        F_3(\s) &:= H(\s) + 2(\p_t \s) \sn^2 \s^{\sharp} + (V \circ \s)\big( K - (\tr K) I \big) - (V' \circ \s) \p_t \s I,\\
    \begin{split}
        H(\s) &:= 3 (d\s \circ K) \otimes \sn \s + 3d\s \otimes K(\sn \s) - 2\tr\big( (d\s \circ K) \otimes \sn \s \big)I\\
        &\quad - 2(\tr K) d\s \otimes \sn \s + |d\s|_h^2\big( (\tr K) I - K \big),
    \end{split}\\
    \begin{split}
        G(E) &:= \lie_{\p_t} \ce - 3 \ce \circ K - K \circ \ce + 2\tr( \ce \circ K )I\\
        &\quad + 2(\tr K) \ce + (\tr \ce)\big( (\tr K) I - K \big) - ( \lie_{\cm^{\sharp}} h )^{\sharp}.
    \end{split}
\end{align*}
The motivation for grouping the terms like this is the following. $F_1(K)$ consists of sums of contractions of $K \otimes K \otimes K$; $F_2(K)$ consists of sums of contractions of tensor products of $\lie_{\p_t} K$ and $K$; $F_3(\s)$ contains all of the scalar field terms, so $F_3(\s) = 0$ if there is no scalar field and the potential vanishes; $H(\s)$ consists of sums of contractions of tensor products of $h^{-1}$, $d\s \otimes d\s$ and $K$; and $G(E)$ contains all the terms involving $\ce$ and $\cm$, so that if Einstein's equations are satisfied, then $G(E) = 0$.

The problem with \eqref{k second order equation}, is that it is not a wave equation because of the term $\sn^2 (\tr K)^{\sharp}$. In order to deal with this, we think of $\theta = \tr K$ as a variable and introduce an additional evolution equation for it. Therefore, we set out to construct the solution to Einstein's equations by solving the system
\begin{subequations} \label{the system}
    \begin{align}
    \p_t \theta &= -\tr K^2 - (\p_t \s)^2 + V \circ \s,\label{main equation for theta}\\
    \lie_{\p_t} h(X,Y) &= h(K(X),Y) + h(X,K(Y)),\label{main equation for h}\\
    -\lie_{\p_t}^2 K + \Delta_h K &= \sn^2 \theta^{\sharp} + F_1(K) + F_2(K) + F_3(\s),\label{main equation for k}\\
    -\p_t^2 \s + \Delta_h \s &= \theta \p_t \s + V' \circ \s.\label{main equation for phi}
\end{align}
\end{subequations}
Note that as a consequence of Proposition~\ref{ricci}, Equation~\eqref{main equation for theta} is the evolution equation that the mean curvature $\theta$ would satisfy if the Einstein--nonlinear scalar field equations were satisfied. 

\begin{remark}
    In what follows, we will mostly use $L^p$ and Sobolev norms associated with the metric $h$. Therefore, for simplicity in the notation, we will write $\|\cdot\|_{L^p} = \|\cdot\|_{L^p(\Sigma_t,h)}$ and similarly for Sobolev norms. See Appendix~\ref{appendix} below for our conventions regarding norms of tensors.
\end{remark}

Given initial data on the singularity $(\Sigma,\ho,\Ko,\phio,\psio)$ and an admissible potential $V$, the plan is to control the difference between the solution to \eqref{the system} and an appropriate approximate solution as in Theorem~\ref{approximate solutions}. For that purpose, introduce
\[
\deh := h - h_n, \quad \deh^{-1} := h^{-1} - h_n^{-1}, \quad \dk := K - K_n, \quad \dt := \theta - \tr K_n, \quad \dep := \s - \s_n.
\]
The main energy to be controlled is
\begin{equation*}
    \begin{split}
        \se_s(t) &:= \sum_{m = 0}^{s-1} t^{2(m+1)} \| \sn^m \lie_{\p_t} \dk \|_{L^2}^2 + \sum_{m = 0}^s t^{2m} \| \sn^m \dk \|_{L^2}^2 + \sum_{m = 0}^{s+1} t^{2m} \| \sn^m \dt \|_{L^2}^2\\
        &\quad + \sum_{m = 0}^s t^{2m} \| \sn^m \p_t \dep \|_{L^2}^2 + \sum_{m = 0}^{s+1} t^{2(m-1)} \| \sn^m \dep \|_{L^2}^2\\
        &\quad + \sum_{m = 0}^{s+1} t^{2(m-1)} \big( \| \sn^m \deh \|_{L^2}^2 + \| \sn^m \deh^{-1} \|_{L^2}^2 \big)
    \end{split}
\end{equation*}
for $s \geq 5$. Now we state the main existence result that we will obtain for solutions to \eqref{the system}.

\begin{theorem} \label{global existence}
    For every $s \geq 5$ and every sufficiently large positive integer $N$ (depending only on $s$, the initial data and the potential), there is an $n_{N,s}$ such that for every $n \geq n_{N,s}$ there is a $T_{N,s,n} > 0$ such that the following holds. There is a $C^3 \times C^2 \times C^2 \times C^3$ solution $(h,K,\theta,\varphi)$ to \eqref{the system} on $(0,T_{N,s,n}] \times \Sigma$ satisfying the estimate
    \begin{equation} \label{main main energy estimate}
        \se_s(t) \leq t^{2N + 2s}
    \end{equation}
    for $t \in (0,T_{N,s,n}]$. Moreover, if $g := -dt \otimes dt + h$, then $K$ is the Weingarten map of the $\Sigma_t$ hypersurfaces, $\theta = \tr K$ and $(g,\s)$ is a solution to the Einstein--nonlinear scalar field equations with potential $V$.
\end{theorem}

Theorem~\ref{global existence} is a direct consequence of Propositions~\ref{global existence proposition}--\ref{vanishing of the einstein scalar field tensor} below. In order to establish Theorem~\ref{global existence}, the first step is to obtain local solutions to \eqref{the system} by setting as initial data the one induced by the approximate solutions. This is given by the following lemma, which follows from \cite[Corollary 4 and Section 5]{ringstrom_local_2024}.

\begin{lemma} \label{local existence}
    Let $n$ be a non-negative integer and $t_n$ as in Theorem~\ref{approximate solutions}. For every $t_0 \in (0,t_n]$, there is a $\delta > 0$ such that there is a unique smooth solution $(h,K,\theta,\varphi)$ to \eqref{the system} on the interval $[t_0,t_0 + \delta]$ satisfying
\[
\begin{split}
    &\theta(t_0) = \tr K_n(t_0), \quad h(t_0) = h_n(t_0), \quad K(t_0) = K_n(t_0), \quad \lie_{\p_t} K(t_0) = \lie_{\p_t} K_n(t_0),\\
    &\s(t_0) = \s_n(t_0), \quad \p_t \s(t_0) = \p_t \s_n(t_0).
\end{split}
\]
Moreover, $h$ is symmetric.
\end{lemma}

Now, by a bootstrap argument, we must show that the local solutions given by Lemma~\ref{local existence} can be extended to a uniform existence time which is independent of $t_0$. For that purpose, we now introduce the \emph{bootstrap assumptions},
\begin{subequations} \label{bootstrap assumptions}
    \begin{align} 
        |\bar\h(e_i,e_k) - \delta_{ik}| &\leq t^{\varepsilon + |p_i - p_k|},\label{first bootstrap assumption}\\
        \|\deh\|_{H^{s+1}} + \|\deh^{-1}\|_{H^{s+1}} + \|\dk\|_{H^s} + \|\lie_{\p_t} \dk\|_{H^{s-1}} &\leq t^{5/2},\label{second bootstrap assumption}\\
        \|\dt\|_{H^{s+1}} + \|\dep\|_{H^{s+1}} + \|\p_t \dep\|_{H^s} &\leq t^{5/2},\label{third bootstrap assumption}
    \end{align}
\end{subequations}
for a fixed $s \geq 5$. Now we state our bootstrap improvement result.

\begin{theorem} \label{bootstrap theorem}
    For every $s \geq 5$ and every sufficiently large positive integer $N$ (depending only on $s$, the initial data and the potential), there is an $n_{N,s}$ such that for every $n \geq n_{N,s}$ there is a $T_{N,s,n} \in (0,t_n]$ such that the following holds. Let $(h,K,\theta,\s)$ be the solution to \eqref{the system} in $[t_0,t_b]$, with $t_b \leq T_{N,s,n}$, with initial data at $t_0$ as in Lemma~\ref{local existence}. Furthermore, assume that the bootstrap assumptions \eqref{bootstrap assumptions} hold on $[t_0,t_b]$. Then
    \begin{equation} \label{main energy estimate}
        \se_s(t) \leq t^{2N + 2s}
    \end{equation}
    for $t \in [t_0,t_b]$. In particular, the bootstrap assumptions are improved.
\end{theorem}

\begin{corollary} \label{global existence corollary}
    Let $N$, $s$ and $n$ be as in Theorem~\ref{bootstrap theorem}. Then the solution to \eqref{the system} given by Lemma~\ref{local existence} can be extended to all of $[t_0,T_{N,s,n}] \times \Sigma$. Moreover, \eqref{main energy estimate} holds for all $t \in [t_0,T_{N,s,n}]$.
\end{corollary}

The continuation criterion required to obtain Corollary~\ref{global existence corollary} comes from \cite[Corollary~4]{ringstrom_local_2024}. Once we have established Corollary~\ref{global existence corollary}, we proceed to prove Theorem~\ref{global existence} as follows. We fix a sequence of positive times $t_i \to 0$ with corresponding sequence of solutions to \eqref{the system} on $[t_i,T_{N,s,n}] \times \Sigma$. Then, by the Arzelà-Ascoli theorem and by passing to a subsequence if necessary, we obtain convergence to a solution to \eqref{the system} which is defined on $(0,T_{N,s,n}] \times \Sigma$. Finally, we show that by an appropriate choice of the parameters $N$ and $n$, the limit solution is in fact a solution to the Einstein--nonlinear scalar field equations with potential $V$. 

In order to prove Theorem~\ref{bootstrap theorem}, we will also make use of a modified energy. But before we can introduce it, we need some definitions.

\begin{definition}
    Let $T$ be a one parameter family of tensors on $\Sigma$. Then we define the \emph{basic energy} by
    \[
    \mathbb E[T] := \int_{\Sigma_t} |\lie_{\p_t} T|_h^2 + |\sn T|_h^2 + t^{-2}|T|_h^2 \mu,
    \]
    where $\mu$ denotes the volume form of $h$. Furthermore, define the \emph{$m$-th order energy} by ${\mathbb E_m[T] := \mathbb E[\sn^m T]}$.
\end{definition}

\begin{definition} \label{weird symmetrization}
    Let $T$ be a $(1,r)$-tensor on $\Sigma$ with $r \geq 2$. For $X,Y \in \mfx(\Sigma)$, define
    \[
    \begin{split}
        \cs[T](X,Y) &:= \frac{1}{2}\Big( T(X,Y) + T(Y,X) + h^{\ell m}h(T(X,e_\ell),Y)e_m\\
        &\quad + h^{\ell m} h(T(Y,e_\ell),X)e_m - h^{\ell m} h(T(e_\ell,X),Y)e_m - h^{\ell m}h(T(e_\ell,Y),X)e_m \Big),
    \end{split}
    \]
    where the vector fields are inserted into the \emph{last two} covariant entries of $T$. 
\end{definition}

\begin{definition}
    Let $T$ be a $(1,r)$-tensor on $\Sigma$ with $r \geq 1$. Define $\tr_i T$ to be the contraction of $T$ obtained by contracting with the $i$-th covariant entry.
\end{definition}

We define the \emph{modified top order quantities} for $\dt$, $\deh$ and $\deh^{-1}$ by
\begin{align*}
    \begin{split}
        \widetilde{\mathbb E}_{s+1}[\dt] &:= \int_{\Sigma_t} |\Delta_h \sn^{s-1} \dt + 2\tr_s\big( (\dk + K_n) \circ \lie_{\p_t} \sn^{s-1} \dk \big)\\
        &\hspace{1cm} + 2\p_t(\dep + \s_n) \Delta_h \sn^{s-1} \dep|_h^2 \mu,
    \end{split}\\
    \begin{split}
        \widetilde{\mathbb E}_{s+1}[\deh] &:=  \int_{\Sigma_t} |\Delta_h \sn^{s-1} \deh - h(\lie_{\p_t} \sn^{s-1} \dk,\,\cdot\,) - h(\,\cdot\,,\lie_{\p_t} \sn^{s-1} \dk)\\
        &\hspace{1cm} + \deh( \cs[\lie_{\p_t} \sn^{s-1}\dk],\,\cdot\, ) + \deh(\,\cdot\,, \cs[\lie_{\p_t} \sn^{s-1}\dk] )|_h^2 \mu,
    \end{split}\\
    \begin{split}
        \widetilde{\mathbb E}_{s+1}[\deh^{-1}] &:= \int_{\Sigma_t} | \Delta_h \sn^{s-1} \deh^{-1} + (\lie_{\p_t} \sn^{s-1} \dk)(e_a)\otimes\big( h^{-1}(\omega^a,\,\cdot\,) + h^{-1}(\,\cdot\,,\omega^a) \big)\\
        &\hspace{1cm} - \cs[\lie_{\p_t} \sn^{s-1}\dk](e_a) \otimes \big(\deh^{-1}(\omega^a,\,\cdot\,) + \deh^{-1}(\,\cdot\,,\omega^a)\big)|_h^2 \mu.
    \end{split}
\end{align*}
Also, define the \emph{modified energy} $\widetilde \se_s(t)$ by
\begin{equation*}
    \begin{split}
        \widetilde \se_s(t) &:= \sum_{m=0}^{s-1} t^{2(m+1)} \mathbb E_m[\dk] + \sum_{m=0}^s t^{2m} \mathbb E_m[\dep]\\
        &\quad + \sum_{m=0}^s t^{2m}\|\sn^m \dt\|_{L^2}^2 + \sum_{m=0}^s t^{2(m-1)} \big( \|\sn^m \deh\|_{L^2}^2 + \|\sn^m \deh^{-1}\|_{L^2}^2 \big)\\
        &\quad + t^{2(s+1)} \widetilde{\mathbb E}_{s+1}[\dt] + t^{2s}\big( \widetilde{\mathbb E}_{s+1}[\deh] + \widetilde{\mathbb E}_{s+1}[\deh^{-1}] \big).
    \end{split}
\end{equation*}
The purpose of the modified energy is to deal with the fact that Equation~\eqref{main equation for k} seems to lead to a loss of derivatives. To be precise, assume that we want to estimate $m$ derivatives of $K$. Regarding $\theta$, we would need control of $m+1$ derivatives. But Equation~\eqref{main equation for theta} does not seem to give control of $m+1$ derivatives of $\theta$ given control of $m$ derivatives of $K$. Turning our attention to $\s$, due to the second term in $F_3(\s)$ we would need control of $m+1$ derivatives of $\s$. Looking at \eqref{main equation for phi}, we see that commuting derivatives with $\Delta_h \s$ gives rise to terms involving $m+1$ derivatives of $h$. But turning to Equation~\eqref{main equation for h}, we see that we have the same issue as with $\theta$. The issue between $h$ and $K$ also arises from commuting derivatives with $\Delta_h K$, since it contains $2$ derivatives of $h$. The modified top order quantities are designed to avoid this issue, making it possible to estimate the modified energy. Then, by using elliptic estimates, we show that the modified energy in fact controls the main energy.

Subsections~\ref{preliminary estimates}--\ref{controlling the energy with the modified energy} are devoted to the proof of Theorem~\ref{bootstrap theorem}. In Subsection~\ref{preliminary estimates}, we obtain all the estimates that are required in preparation for the energy estimates. In Subsection~\ref{energy estimates}, we obtain the necessary energy estimates for \eqref{the system}. The conclusion of the proof of Theorem~\ref{bootstrap theorem} is then found in Subsection~\ref{controlling the energy with the modified energy}. Subsections~\ref{finishing the construction} and \ref{the constructed solution solves einsteins eqs} then comprise the proof of Theorem~\ref{global existence}. In Subsection~\ref{the constructed solution is smooth} we prove that there is a \emph{smooth} solution such that $\K$, $\Phi$ and $\Psi$ converge as required. Finally, in Subsection~\ref{asymptotics for the expansion normalized metric}, we obtain convergence of $\h$, thus finishing the proof of Theorem~\ref{main existence theorem}.

\subsection{Preliminary estimates} \label{preliminary estimates}

From now on, and until the end of Subsection~\ref{controlling the energy with the modified energy}, we assume that the hypotheses of Theorem~\ref{bootstrap theorem} hold. Therefore, we have a solution $(h,K,\theta,\s)$ to \eqref{the system}, as in Lemma~\ref{local existence}, defined on $[t_0,t_b] \times \Sigma$, which satisfies the bootstrap assumptions \eqref{bootstrap assumptions}. Note that \eqref{first bootstrap assumption} implies that
\begin{equation*} 
    |(h - h_0)_{ik}| \leq t^{2p_{\max\{i,k\}} + \varepsilon}, \quad |\det h - t^2| \leq Ct^{2+\varepsilon}, \quad |(h - h_0)^{ik}| \leq Ct^{-2p_{\min\{i,k\}} + \varepsilon}
\end{equation*}
for $t \in [t_0,t_b]$; see the proof of Corollary~\ref{estimates of differences}. Finally, fix $s \geq 5$.

\begin{remark}
    Until the end of Subsection~\ref{the constructed solution solves einsteins eqs}, the constants $C$ and $C_n$ will be allowed to depend on $s$, in addition to the initial data, the potential, and $n$ in the case of $C_n$. Importantly, $C$ is not allowed to depend on $n$. 
\end{remark}

\begin{remark}
    We work with local solutions to \eqref{the system} as given by Lemma~\ref{local existence}. In particular, for the time being, we do not know $K$ to be the Weingarten map of the $\Sigma_t$ hypersurfaces, nor do we know $\theta$ to coincide with either $\tr K$ or the mean curvature. 
\end{remark}

\begin{lemma} \label{dual metric equivalence}
    There is a $T > 0$ such that if \eqref{first bootstrap assumption} holds for $t_b < T$, then there is a constant $C$, depending only on $\varepsilon$, such that for $\alpha \in T_p^* \Sigma$
    \[
    C^{-1} t^{-2p_1} |\alpha|_{\ho}^2 \leq |\alpha|_h^2 \leq C t^{-2p_3} |\alpha|_{\ho}^2.
    \]
\end{lemma}

\begin{proof}
    Let $\alpha_i := \alpha(e_i)$ and $\alpha \neq 0$, so that $\alpha_i \neq 0$ for some $i$, then
    \[
    \begin{split}
        |\alpha|_h^2 = h^{\ell m} \alpha_\ell \alpha_m &= (h_0)^{\ell m}\alpha_\ell \alpha_m + (h - h_0)^{\ell m} \alpha_\ell \alpha_m\\
        &= \sum_k t^{-2p_k} \alpha_k^2 \left( 1 + \frac{(h-h_0)^{\ell m} \alpha_\ell \alpha_m}{\sum_k t^{-2p_k} \alpha_k^2} \right). 
    \end{split}
    \]
    Moreover, by using
    \[
    \frac{|\alpha_i|}{\sqrt{\sum_k t^{-2p_k}\alpha_k^2}} \leq t^{p_i},
    \]
    we have
    \[
    \begin{split}
        \left| \frac{(h-h_0)^{\ell m} \alpha_\ell \alpha_m}{\sum_k t^{-2p_k} \alpha_k^2} \right| \leq \sum_{\ell,m} t^{p_\ell + p_m} |(h - h_0)^{\ell m}| \leq Ct^\varepsilon. 
    \end{split}
    \]
    The result follows by choosing $T$ such that $C T^\varepsilon < 1$.
\end{proof}

\begin{lemma}[Sobolev embedding]
    Let $T$ be a tensor on $\Sigma$, then
    \[
    \|T\|_{L^\infty} \leq C t^{-5/4} \|T\|_{W^{1,4}}, \qquad \|T\|_{L^4} \leq C t^{-5/4} \|T\|_{H^1}.
    \]
    In particular, 
    \[
    \|T\|_{L^\infty} \leq C t^{-5/2} \|T\|_{H^2}.
    \]
\end{lemma}

\begin{proof}
    First consider a function $f$. By Sobolev embedding,
    \[
    \sup_{\Sigma} |f| \leq C \|f\|_{W^{1,4}(\Sigma,\ho)} = C \Big( \int_{\Sigma} f^4 \mathring{\mu} \Big)^{1/4} + C \Big( \int_{\Sigma} |df|_{\ho}^4 \mathring{\mu}, \Big)^{1/4}
    \]
    where $\mathring{\mu}$ is the volume form of $\ho$. Also
    \[
    |df|_{\ho}^4 \leq Ct^{-4} |df|_h^4, \quad C^{-1} t^2 \leq \det h \leq C t^2.
    \]
    Hence, since $\mu = \sqrt{\det h} \mathring{\mu}$, where $\mu$ denotes the volume form of $h$,
    \[
    \sup_{\Sigma} |f| \leq C \Big( \int_{\Sigma} f^4 t^{-1} \mu \Big)^{1/4} + C \Big( \int_{\Sigma} t^{-5} |df|_h^4 \mu \Big)^{1/4} \leq C t^{-5/4} \|f\|_{W^{1,4}}.
    \]
    Similarly,
    \[
    \begin{split}
        \|f\|_{L^4} \leq C t^{1/4} \|f\|_{L^4(\Sigma,\ho)} &\leq Ct^{1/4} \|f\|_{H^1(\Sigma,\ho)}\\
        &\leq Ct^{1/4} \left[ \Big( \int_{\Sigma} f^2 t^{-1}\mu \Big)^{1/2} + \Big( \int_{\Sigma} t^{-2}|df|_h^2 t^{-1}\mu \Big)^{1/2} \right]\\
        &\leq t^{-5/4} \|f\|_{H^1}.
    \end{split}
    \]
    For a tensor $T$, we apply the already obtained inequalities to $f_\delta = \sqrt{|T|_h^2 + \delta}$ for $\delta > 0$ and then let $\delta \to 0$.
\end{proof}

\begin{lemma} \label{estimates in tensor components}
    There is a $T_n > 0$ small enough and a constant $C_{n,m}$, depending only on the initial data and the potential, such that if $t_b \leq T_n$ and $T$ is a $(q,r)$-tensor on $\Sigma$, then
    \[
    | (\sn^{(n)})^m T|_{h_n}^2 \leq C_{n,m} \langle \ln t \rangle^{2m} t^{2m(-1 + \varepsilon)} \sum t^{-2(p_{i_1} + \cdots + p_{i_r})} t^{2(p_{k_1} + \cdots + p_{k_q})} \big(e_{\alpha} T^{k_1 \cdots k_q}_{i_1 \cdots i_r} \big)^2
    \]
    and
    \[
    \begin{split}
        &\sum t^{-2(p_{\alpha_1} + \cdots + p_{\alpha_\ell})} t^{-2(p_{i_1} + \cdots + p_{i_r})} t^{2(p_{k_1} + \cdots + p_{k_q})}\big(e_{\alpha} T_{i_1 \cdots i_r}^{k_1 \dots k_q}\big)^2\\
        &\hspace{4cm}\leq C_{n,m} \sum_{a=0}^m \langle \ln t \rangle^{2(m-a)} t^{2(m-a)(-1+\varepsilon)} |(\sn^{(n)})^a T|_{h_n}^2,
    \end{split}
    \]
    where the sums are over all indices and all $\alpha$ with $|\alpha| = \ell \leq m$, and the indices refer to the components of $T$ in terms of the frame $\{e_i\}$.
\end{lemma}

\begin{proof}
    For notational simplicity, let us for this proof drop the $n$ when referring to $\sn^{(n)}$. So we write $\sn$ and $\bar{\Gamma}_{ik}^{\ell}$ instead of $\sn^{(n)}$ and $(\bar{\Gamma}_n)_{ik}^{\ell}$. Also, we focus on the case when $T$ is a $(1,1)$-tensor, since the general case works in the same way but requires more notation.

    First consider the case with no derivatives,
    \[
    \begin{split}
        |T|_{h_n}^2 &= (h_n)^{ik}(h_n)_{\ell m} T_i^\ell T_k^m\\
        &= (h_0)^{ik}(h_0)_{\ell m} T_i^\ell T_k^m + (h_n-h_0)^{ik}(h_0)_{\ell m} T_i^\ell T_k^m + (h_n)^{ik}(h_n-h_0)_{\ell m} T_i^\ell T_k^m\\
        &= \sum_{i,\ell} t^{-2p_i} t^{2p_\ell}(T_i^\ell)^2 \left( 1 + \frac{(h_n -h_0)^{ik}(h_0)_{\ell m} T_i^\ell T_k^m + (h_n)^{ik}(h_n-h_0)_{\ell m} T_i^\ell T_k^m}{\sum_{i,\ell} t^{-2p_i} t^{2p_\ell}(T_i^\ell)^2} \right)
    \end{split}
    \]
    and
    \[
    \left|\frac{(h_n -h_0)^{ik}(h_n)_{\ell m} T_i^\ell T_k^m + (h_n)^{ik}(h_n-h_0)_{\ell m} T_i^\ell T_k^m}{\sum_{i,\ell} t^{-2p_i} t^{2p_\ell}(T_i^\ell)^2} \right| \leq C_nt^\varepsilon.
    \]
    Hence there is a small enough $T_n$, depending only on the initial data, the potential and $n$, such that
    \[
    C^{-1} \sum_{i,\ell} t^{-2p_i} t^{2p_\ell}(T_i^\ell)^2 \leq |T|_{h_n}^2 \leq C \sum_{i,\ell} t^{-2p_i} t^{2p_\ell}(T_i^\ell)^2.
    \]
    This proves the case with no derivatives.

    Now with derivatives. For the rest of this proof, we allow the constants $C$ and $C_n$ to depend on $m$. We have
    \[
    \sn^m_{e_{i_1}, \ldots, e_{i_m}} T(e_k,\omega^{\ell}) = \sum \pm (e_{\alpha_1} \bar{\Gamma}) \cdots (e_{\alpha_r} \bar{\Gamma})(e_{\beta} T),
    \]
    where the sum is over appropriate multiindices such that $|\alpha_1| + \cdots + |\alpha_r| + |\beta| + r = m$, and we omit the indices since the exact contractions are not important. Then, by using the case with no derivatives,
    \[
    \begin{split}
        |\sn^m T|_{h_n}^2 &\leq C \sum t^{-2(p_{i_1} + \cdots + p_{i_m})} t^{-2p_k} t^{2p_{\ell}} \big( \sn^m_{e_{i_1}, \ldots , e_{i_m}} T(e_k,\omega^{\ell}) \big)^2\\
        &\leq C \sum t^{-2(p_{i_1} + \cdots + p_{i_m})} t^{-2p_k} t^{2p_{\ell}} \sum (e_{\alpha_1} \bar{\Gamma})^2 \cdots (e_{\alpha_r} \bar{\Gamma})^2(e_{\beta} T)^2.
    \end{split}
    \]
    Now, regarding the indices, note that every non-contracted index corresponds to a $t^{\pm 2p_i}$ factor, with $+$ if it is an upper index and with $-$ if it is a lower index. Moreover, if an index is contracted, we can multiply by $1 = t^{2p_i} t^{-2p_i}$ to introduce a corresponding power of $t$ for each contracted index. We conclude that it is enough to estimate objects of the form
    \[
    t^{-(p_{\alpha_1} + \cdots + p_{\alpha_q})} t^{p_{\ell} - p_i - p_k} |e_{\alpha} \bar{\Gamma}_{ik}^{\ell}|,
    \]
    where $|\alpha| = q$. But then Lemma~\ref{decay of ricci} implies that $t^{p_{\ell} - p_i - p_k} |e_{\alpha} \bar{\Gamma}_{ik}^{\ell}| \leq C_n \langle \ln t \rangle^{q+1} t^{-1+\varepsilon}$, thus
    \[
    t^{-(p_{\alpha_1} + \cdots + p_{\alpha_q})} t^{p_{\ell} - p_i - p_k} |e_{\alpha} \bar{\Gamma}_{ik}^{\ell}| \leq C_n \langle \ln t \rangle^{q+1} t^{(q+1)(-1+\varepsilon)}.
    \]
    Hence
    \[
    |\sn^m T|_{h_n}^2 \leq C_n \langle \ln t \rangle^{2m} t^{2m(-1+\varepsilon)} \sum t^{-2p_i}t^{2p_k}\big( e_{\alpha} T_i^k \big)^2,
    \]
    where the sum is over $i$, $k$ and every multiindex $\alpha$ with $|\alpha| = q \leq m$.

    For the second inequality, note that for $|\alpha| = m$,
    \[
    \begin{split}
        &\sum t^{-2(p_{\alpha_1} + \cdots + p_{\alpha_m)}} t^{-2p_i} t^{2p_k} \big(e_{\alpha} T_i^k\big)^2\\
        &\leq C \sum t^{-2(p_{\alpha_1} + \cdots + p_{\alpha_m)}} t^{-2p_i} t^{2p_k} \left(\sn^m_{e_{\alpha_1},\ldots,e_{\alpha_m}} T(e_i,\omega^k)^2 + \sum (e_{\beta_1} \bar \Gamma)^2 \cdots (e_{\beta_r} \bar \Gamma)^2 (e_{\gamma}T)^2 \right),
    \end{split}
    \]
    where the inner sum is over appropriate multiindices such that $|\beta_1| + \cdots + |\beta_r| + |\gamma| + r = m$ and $r \geq 1$, so that $|\gamma| \leq m-1$. Note that the first term, taken with the sum outside, is equal to $|\sn^m T|_{h_0}^2$ and the second term can be treated similarly as we did before. Hence
    \[
    \begin{split}
        &\sum t^{-2(p_{\alpha_1} + \cdots + p_{\alpha_m})} t^{-2p_i} t^{2p_k} \big(e_{\alpha} T_i^k\big)^2\\
        &\hspace{0.5cm}\leq C|\sn^m T|_{h_0}^2 + C_n\sum_{|\gamma| \leq m-1} \langle \ln t \rangle^{2(m-|\gamma|)} t^{2(m-|\gamma|)(-1+\varepsilon)} t^{-2(p_{\gamma_1} + \cdots + p_{\gamma_q})} t^{-2p_i} t^{2p_k} \big(e_{\gamma} T_i^k\big)^2\\
        &\hspace{0.5cm}\leq C|\sn^m T|_{h_n}^2 + C_n\sum_{r=0}^{m-1} \Big( \langle \ln t \rangle^{2(m-r)} t^{2(m-r)(-1+\varepsilon)} \sum_{|\gamma| = r} t^{-2(p_{\gamma_1} + \cdots + p_{\gamma_r})} t^{-2p_i} t^{2p_k} \big(e_\gamma T_i^k\big)^2 \Big) 
    \end{split}
    \]
    and we have reduced it to the case with $m-1$ derivatives. Since we have it for zero derivatives, the result follows by induction.
\end{proof}

\begin{lemma} \label{sup norm bootstrap estimates}
    The following estimates hold,
    \[
    \begin{split}
        \|\deh\|_{W^{s-1,\infty}} + \|\deh^{-1}\|_{W^{s-1,\infty}} + \|\dk\|_{W^{s-2,\infty}} + \|\lie_{\p_t} \dk\|_{W^{s-3,\infty}} &\leq C,\\
        \|\dt\|_{W^{s-1,\infty}} + \|\dep\|_{W^{s-1,\infty}} + \|\p_t \dep\|_{W^{s-2,\infty}} &\leq C.
    \end{split}
    \]
\end{lemma}

\begin{proof}
    This follows from \eqref{bootstrap assumptions} and Sobolev embedding.
\end{proof}

\begin{lemma} \label{difference of connections bound}
    If $\cd := \sn - \sn^{(n)}$, then for $m \leq s-2$,
    \[
    \|\cd\|_{W^{m,\infty}} \leq C\big( \|\deh\|_{W^{m+1,\infty}} + \|\deh^{-1}\|_{W^{m,\infty}} \big),
    \]
    and for $m \leq s$,
    \[
    \|\cd\|_{H^m} \leq C\big( \|\deh\|_{H^{m+1}} + \|\deh^{-1}\|_{H^m} \big).
    \]
\end{lemma}

\begin{proof}
    For $X,Y \in \mfx(\Sigma)$ and $\alpha \in \Omega^1(\Sigma)$, we have
    \[
    \begin{split}
        \cd( X, Y, \alpha ) &= -\frac{1}{2}(h_n^{-1})(\alpha,\omega^{\ell})\big( \sn_X h_n(Y,e_{\ell}) + \sn_Y h_n(X,e_{\ell}) - \sn_{e_{\ell}} h_n(X,Y) \big)\\
        &= \frac{1}{2}( h^{-1} - \deh^{-1} )(\alpha,\omega^{\ell})\big( \sn_X \deh(Y,e_{\ell}) + \sn_Y \deh(X,e_{\ell}) - \sn_{e_{\ell}} \deh(X,Y) \big),
    \end{split}
    \]
    implying
    \[
    |\sn^m \cd|_h \leq C \sum_{r = 0}^m |\sn^r( h^{-1} - \deh^{-1} )|_h |\sn^{m-r}( \sn \deh )|_h.
    \]
    The result follows by Lemma~\ref{sup norm bootstrap estimates} and \eqref{second bootstrap assumption}.
\end{proof}

\begin{lemma} \label{comparison of h and hn norms}
    Let $T$ be a tensor on $\Sigma$. For $m \leq s-1$,
    \[
    C^{-1} \|T\|_{W^{m,\infty}(\Sigma,h_n)} \leq \|T\|_{W^{m,\infty}} \leq C \|T\|_{W^{m,\infty}(\Sigma,h_n)},
    \]
    and for $m \leq s+1$,
    \[
    C^{-1} \|T\|_{H^m(\Sigma,h_n)} \leq \|T\|_{H^m} \leq C \|T\|_{H^m(\Sigma,h_n)}.
    \]
\end{lemma}

\begin{proof}
    Note that Lemma~\ref{sup norm bootstrap estimates} implies $|\deh|_h + |\deh^{-1}|_h \leq C$. Hence there is a constant $C$ such that
    \[
    C^{-1} |T|_{h_n} \leq |T|_h \leq C |T|_{h_n},
    \]
    which deals with the case with no derivatives.

    Now with derivatives. As in the proof of Lemma~\ref{estimates in tensor components}, for notational clarity, we focus on the case when $T$ is a $(1,1)$-tensor. Let $X,Y \in \mfx(\Sigma)$ and $\alpha \in \Omega^1(\Sigma)$, then
    \[
    \big( \sn_X - \sn_X^{(n)} \big) T(Y,\alpha) = -\cd( X,Y,\omega^{\ell} ) T(e_{\ell},\alpha) + \cd(X,e_{\ell},\alpha)T(Y,\omega^{\ell}),
    \]
    implying
    \[
    |\sn^m T|_h \leq |\sn^{m-1}( \sn^{(n)} T )|_h + \sum_{a+b=m-1} 2|\sn^a \cd|_h |\sn^b T|_h.
    \]
    There is a similar inequality for $|(\sn^{(n)})^m T|_{h_n}$. For the $L^{\infty}$ norms, this reduces the result to the case with $m-1$ derivatives, as long as $m-1 \leq s-2$ by Lemma~\ref{difference of connections bound}. Since we already proved the case with no derivatives, by induction we conclude that
    \[
    C^{-1} \|T\|_{W^{m,\infty}(\Sigma,h_n)} \leq \|T\|_{W^{m,\infty}} \leq C \|T\|_{W^{m,\infty}(\Sigma,h_n)}
    \]
    for $m\leq s-1$.

    Now for the $L^2$ norms. We only focus on one direction of the inequality, since the other is similar. We have
    \[
    |\sn^m T|_h^2 \leq C|\sn^{m-1}( \sn^{(n)} T )|_h^2 +  \sum_{a+b=m-1} C|\sn^a \cd|_h^2 |\sn^b T|_h^2.
    \]
    Focusing on the terms in the sum, we have two options,
    \[
    \begin{split}
        \int_{\Sigma_t} |\sn^a \cd|_h^2 |\sn^b T|_h^2 \mu &\leq \Big(\sup_{\Sigma} |\sn^a \cd|_h \Big)^2 \|T\|_{H^b}^2,\\
        \int_{\Sigma_t} |\sn^a \cd|_h^2 |\sn^b T|_h^2 \mu &\leq \|\cd\|_{H^a}^2 \Big( \sup_{\Sigma} |\sn^b T|_h \Big)^2 \leq C t^{-5} \|\cd\|_{H^a}^2 \|T\|_{H^{b+2}}^2.
    \end{split}
    \]
    As for the $L^{\infty}$ norms, if $a \leq s-2$, we can use the first inequality to reduce the corresponding terms to the case with $m-1$ derivatives. On the other hand, if $a \geq s-1$ we use the second inequality. In that case, $b = m-1-a \leq m-s$. Since $s \geq 5$, we see that $b+2 \leq m-1$, and hence $\|T\|_{H^{b+2}} \leq \|T\|_{H^{m-1}}$. But we still need to control $\cd$, which requires $a \leq s$ by Lemma~\ref{difference of connections bound}. Hence, as long as $m \leq s+1$, we have reduced it to the case with $m-1$ derivatives. Since we already have it for no derivatives, the result follows.
\end{proof}

\begin{lemma} \label{technical estimate 2}
    If $T_1$ and $T_2$ are tensors on $\Sigma$, then
    \[
    \| \sn^s T_1 \otimes  T_2 \|_{L^2} \leq C\big(\| \sn^{s-1}(\sn^{(n)} T_1)\|_{L^{\infty}} + \|T_1\|_{W^{s-1,\infty}} \big) \|T_2\|_{L^2} + C\|T_1\|_{L^{\infty}} \|T_2\|_{H^2}.
    \]
    Moreover,
    \[
    \begin{split}
        \|\sn^{s+1}T_1 \otimes T_2\|_{L^2} &\leq C\big( \|\sn^{s-1}(\sn^{(n)} \sn^{(n)} T_1)\|_{L^\infty} + \|\sn^{(n)} T_1\|_{W^{s-1,\infty}} \big)\|T_2\|_{L^2}\\
        &\quad + C\|\sn^{(n)} T_1\|_{L^\infty} \|T_2\|_{H^2} + \|T_1\|_{W^{s-1,\infty}} \|T_2\|_{H^2}.
    \end{split}
    \]
\end{lemma}

\begin{proof}
    We compute,
    \[
    \begin{split}
        &\int_{\Sigma_t} |\sn^s T_1|_h^2 |T_2|_h^2 \mu\\
        &\qquad \leq C\int_{\Sigma_t} |\sn^{s-1}( \sn^{(n)} T_1 )|_h^2 |T_2|_h^2 \mu + C\sum_{c+d = s-1} \int_{\Sigma_t} |\sn^c \cd|_h^2 |\sn^d T_1|_h^2 |T_2|_h^2 \mu\\
        &\qquad \leq C\big(\| \sn^{s-1}(\sn^{(n)} T_1)\|_{L^{\infty}}^2 + \|T_1\|_{W^{s-1,\infty}}^2 \big) \|T_2\|_{L^2}^2 + C \int_{\Sigma_t} |\sn^{s-1} \cd|_h^2 |T_1|_h^2 |T_2|_h^2 \mu\\
        &\qquad \leq C\big(\| \sn^{s-1}(\sn^{(n)} T_1)\|_{L^{\infty}}^2 + \|T_1\|_{W^{s-1,\infty}}^2 \big) \|T_2\|_{L^2}^2 + C\|\sn^{s-1} \cd\|_{L^2}^2 \|T_1\|_{L^{\infty}}^2 \|T_2\|_{L^{\infty}}^2\\
        &\qquad \leq C\big(\| \sn^{s-1}(\sn^{(n)} T_1)\|_{L^{\infty}}^2 + \|T_1\|_{W^{s-1,\infty}}^2 \big) \|T_2\|_{L^2}^2 + C\|T_1\|_{L^{\infty}}^2 \|T_2\|_{H^2}^2.
    \end{split}
    \]
    The other inequality is similar.
\end{proof}

\begin{lemma} \label{product sobolev estimate}
    If $T_1$ and $T_2$ are tensors on $\Sigma$, then for $m \geq 3$,
    \[
    \|\sn^m( T_1 \otimes T_2 )\|_{L^2} \leq Ct^{-5/2} \|T_1\|_{H^m} \|T_2\|_{H^m}.
    \]
\end{lemma}

\begin{proof}
    This follows from Sobolev embedding.
\end{proof}

\begin{lemma} \label{estimates for inhomogeneous terms}
    Define
    \[
    \begin{split}
        I_{\theta_n} &:= \p_t \theta_n + \tr K_n^2 + (\p_t \s_n)^2 - V \circ \s_n\\
        I_{K_n} &:= \lie_{\p_t}^2 K_n - \Delta_{h_n} K_n + (\sn^{(n)})^2 \theta_n^{\sharp} + F_1(K_n) + F_2(K_n) + F_3(\s_n).
    \end{split}
    \]
    Given a positive integer $N$, there is an $n_{N,s}$ such that if $n \geq n_{N,s}$, then
    \[
    \sum_{m = 0}^{s+1} \sum_{r=0}^2 t^m \|\p_t^rI_{\theta_n}\|_{H^m} + \sum_{m = 0}^{s-1} \sum_{r=0}^1 t^m \|\lie_{\p_t}^rI_{K_n}\|_{H^m} \leq C_n t^{N+s}.
    \]
\end{lemma}

\begin{proof}
    Note that
    \[
    I_{\theta_n} = -E_n(\p_t,\p_t), \qquad I_{K_n} = G(E_n).
    \]
    The result follows from Lemmas~\ref{comparison of h and hn norms}, \ref{estimates in tensor components} and Theorem~\ref{approximate solutions}.
\end{proof}

\begin{lemma} \label{kn estimates}
    We have
    \[
    t^{r+1}\|\lie_{\p_t}^r K_n\|_{L^{\infty}} \leq C + C_nt^{\varepsilon}, \qquad \sum_{m=1}^{s-1} t^{m+r+1} \|\sn^m \lie_{\p_t}^r K_n\|_{L^{\infty}} \leq C_n \langle \ln t \rangle t^{\varepsilon}
    \]
    and
    \[
    \sum_{m=0}^{s-1}\big( t^{m+2} \|\sn^m( \sn^{(n)} K_n )\|_{L^{\infty}} + t^{m+3}\|\sn^m( \sn^{(n)} \sn^{(n)} K_n )\|_{L^{\infty}} \big) \leq C_n \langle \ln t \rangle t^{\varepsilon},
    \]
    for $r = 0,1$.
\end{lemma}

\begin{proof}
    We have,
    \[
    t|K_n|_h \leq |\Ko|_h + |tK_n - \Ko|_h.
    \]
    Note that, by Lemmas~\ref{comparison of h and hn norms} and \ref{estimates in tensor components}, $|\cdot|_h$ is equivalent to $|\cdot|_{h_0}$, thus $|\Ko|_h \leq C$. Moreover, the off-diagonal improvements ensure that the second term on the right-hand side is bounded by $C_nt^\varepsilon$. The remaining estimates for $K_n$ follow from Lemmas~\ref{comparison of h and hn norms} and \ref{estimates in tensor components}. 

    For $\lie_{\p_t} K_n$, we have
    \[
    \lie_{\p_t} K_n = -\sric_n^{\sharp} - \theta_n K_n + d\s_n \otimes \sn \s_n + (V \circ \s_n)I + \ce_n.
    \]
    For the estimate without derivatives, note that the only term that is of order $t^{-2}$ is $\theta_n K_n$; as above, the off-diagonal improvements allow us to conclude that the rest are of order $t^{-2+\varepsilon}$. The estimates with derivatives follow again by Lemmas~\ref{comparison of h and hn norms} and \ref{estimates in tensor components}.
\end{proof}

\begin{lemma} \label{phin estimates}
    We have
    \[
    \begin{split}
        \sum_{m=0}^{s-1} t^{m+1}\|\sn^m d\s_n\|_{L^{\infty}} &\leq C_n \langle \ln t \rangle t^\varepsilon,\\
        \sum_{m=0}^{s-1} \big( t^{m+2}\|\sn^m( \sn^{(n)} d\s_n )\|_{L^{\infty}} + t^{m+3}\|\sn^m( \sn^{(n)} \sn^{(n)} d\s_n )\|_{L^\infty} \big) &\leq C_nt^\varepsilon
    \end{split}
    \]
    and
    \[
    t|\p_t \s_n| \leq C + C_nt^\varepsilon, \qquad \sum_{m=0}^{s-1} \big( t^{m+2}\|\sn^m d(\p_t \s_n)\|_{L^{\infty}} + t^{m+3}\|\sn^m( \sn^{(n)} d(\p_t\s_n) )\|_{L^\infty} \big) \leq C_n t^\varepsilon.
    \]
\end{lemma}

\begin{proof}
    For $\p_t \s_n$ without derivatives, just note that $t\p_t \s_n = \psio + (\bar\Psi_n - \psio)$. The estimates with derivatives follow from Lemmas~\ref{comparison of h and hn norms} and \ref{estimates in tensor components}.
\end{proof}

\begin{lemma} \label{estimates for the curvature}
    For $m \leq s-1$, 
    \[
    t^{m+2}\|\sn^m \sric_n\|_{L^{\infty}} \leq C_n t^\varepsilon.
    \]
    Moreover,
    \[
    \sum_{m=0}^{s-3} \|\sn^m(\bar R - \bar R_n)\|_{L^\infty} \leq C, \qquad \sum_{m=0}^{s-1} \|\sn^m(\bar R - \bar R_n)\|_{L^2} \leq C t^{5/2}.
    \]
\end{lemma}

\begin{proof}
    The estimate for $\sric_n$ follows from Lemmas~\ref{comparison of h and hn norms} and \ref{estimates in tensor components}. In order to control $\bar R - \bar R_n$, note that
    \[
    \bar R(X,Y)Z = \bar R_n(X,Y)Z + \sn_X \cd(Y,Z) - \sn_Y \cd(X,Z) + \cd( Y,\cd(X,Z) ) - \cd( X, \cd(Y,Z) ).
    \]
    The result follows from Lemma~\ref{difference of connections bound}.
\end{proof}

\begin{definition}
    Let $X$ and $Y$ be one parameter families of vector fields on $\Sigma$. Define $\lie_{\p_t} \sn$ by
    \[
    (\lie_{\p_t} \sn)(X,Y) := [\p_t,\sn_X Y] - \sn_{[\p_t,X]} Y - \sn_X [\p_t,Y].
    \]
    Note that $\lie_{\p_t} \sn$ defines a one parameter family of $(1,2)$-tensors on $\Sigma$.
\end{definition}

\begin{lemma} \label{commutator formula}
    We have $\lie_{\p_t} \sn = \cs[\sn K]$ (recall Definition~\ref{weird symmetrization}).
\end{lemma}

\begin{proof}
    Let $X, Y, Z \in \mathfrak X(\Sigma)$. By taking $\p_t$ of the Koszul formula for $\sn$, we obtain
    \[
    \begin{split}
        &2\lie_{\p_t} h(\sn_X Y,Z) + 2h(( \lie_{\p_t} \sn )(X,Y),Z)\\
        &\hspace{3cm} = X( \lie_{\p_t} h(Y,Z) ) + Y(\lie_{\p_t} h(X,Z)) - Z(\lie_{\p_t} h(X,Y))\\
        &\hspace{3cm} \quad -\lie_{\p_t} h(X,[Y,Z]) - \lie_{\p_t} h(Y,[X,Z]) + \lie_{\p_t} h(Z,[X,Y]),
    \end{split}
    \]
    implying
    \[
    2h(( \lie_{\p_t} \sn )(X,Y),Z) = \sn_X \lie_{\p_t} h(Y,Z) + \sn_Y \lie_{\p_t} h(X,Z) - \sn_Z \lie_{\p_t} h(X,Y).
    \]
    The result follows by using the equation for $h$ and raising an index.
\end{proof}

\begin{lemma} \label{commutator formula 2}
    Let $T$ be a one parameter family of $(q,r)$-tensors on $\Sigma$, then
    \[
    \begin{split}
        &([\lie_{\p_t}, \sn]T)(X,Y_1,\ldots,Y_r,\alpha_1,\ldots,\alpha_q)\\
        &\hspace{3cm}= -\sum_i T(Y_1,\ldots,(\lie_{\p_t} \sn)(X,Y_i),\ldots,Y_r,\alpha_1,\ldots,\alpha_q)\\
        &\hspace{3cm}\quad + \sum_k T(Y_1,\ldots,Y_r,\alpha_1,\ldots,\alpha_k((\lie_{\p_t} \sn)(X,\,\cdot\,)),\ldots,\alpha_q).
    \end{split}
    \]
\end{lemma}

\begin{proof}
    This follows by a direct computation.
\end{proof}

\begin{lemma} \label{commutator estimates}
    Let $T$ be a one parameter family of tensors on $\Sigma$. Then
    \[
    t^{m+2}\|\sn^m([\lie_{\p_t},\sn]T)\|_{L^2} \leq C_n\langle \ln t \rangle t^\varepsilon \sum_{a=0}^m t^{a} \|\sn^aT\|_{L^2}
    \]
    for $m\leq s-1$. Moreover,
    \[
    t^{m+1}\|[\lie_{\p_t},\sn^m]T\|_{L^2} \leq C_n\langle \ln t \rangle t^\varepsilon \sum_{a=0}^{m-1} t^a\|\sn^a T\|_{L^2}
    \]
    for $m \leq s$.
\end{lemma}

\begin{proof}
    First consider the case $m \leq s-3$. Then
    \[
    |\sn^m( [\lie_{\p_t}, \sn]T )|_h \leq C\sum_{a+b=m}|\sn^{a+1} K|_h |\sn^b T|_h.
    \]
    Therefore, we can write $K = K_n + \dk$ and estimate the corresponding terms in $L^{\infty}$ to obtain the result. If $m = s-2, s-1$, we need to proceed differently. In that case
    \[
    |\sn^m( [\lie_{\p_t},\sn]T )|_h \leq C|\sn^m(\sn^{(n)} K_n \otimes T)|_h + C|\sn^m( \cd \otimes K_n \otimes T )|_h + C|\sn^m( \sn \dk \otimes T )|_h,
    \]
    hence
    \[
    \begin{split}
        \|\sn^m([\lie_{\p_t},\sn]T)\|_{L^2} &\leq C_n \langle \ln t \rangle \sum_{a+b=m}t^{-a-2+\varepsilon}\|\sn^b T\|_{L^2}\\
        &\quad + Ct^{-5/2}\big(\|\cd\|_{H^m}\|K_n \otimes T\|_{H^m} + \|\dk\|_{H^{m+1}}\|T\|_{H^m}\big)\\
        &\leq C_n \langle \ln t \rangle t^{-m-2+\varepsilon} \sum_{a=0}^m t^a \|\sn^a T\|_{L^2} + C\|K_n \otimes T\|_{H^m} + C\|T\|_{H^m},
    \end{split}
    \]
    where we have used Lemma~\ref{product sobolev estimate}. The $K_n \otimes T$ term can now be estimated similarly as in the previous case to obtain the result.

    For the second estimate,
    \[
    [\lie_{\p_t},\sn^m]T = [\lie_{\p_t},\sn] \sn^{m-1}T + \cdots + \sn^{m-a-1}( [\lie_{\p_t},\sn] \sn^aT ) + \cdots + \sn^{m-1}( [\lie_{\p_t},\sn]T ),
    \]
    implying
    \[
    \begin{split}
        \|[\lie_{\p_t},\sn^m]T\|_{L^2} &\leq C_n \langle \ln t \rangle \sum_{a=0}^{m-1} \left( t^{-m+a-1+\varepsilon} \sum_{b=0}^{m-a-1} t^b \|\sn^{a+b} T\|_{L^2} \right)\\
        &\leq C_n \langle \ln t \rangle t^{-m-1+\varepsilon} \sum_{a=0}^{m-1} t^a \|\sn^a T\|_{L^2},
    \end{split}
    \]
    which is what we wanted to prove.
\end{proof}

\begin{lemma} \label{estimate for difference of hessians}
    Let $T_1$ and $T_2$ be tensors on $\Sigma$ of the same type, and let $\delta T := T_1 - T_2$. Then
    \[
    \begin{split}
        &\big\|\sn^m\big( \sn^2 T_1^\sharp - (\sn^{(n)})^2 T_2^\sharp \big)\big\|_{L^2}\\
        &\hspace{2cm}\leq \|\sn^{m+2}\delta T\|_{L^2} + C\sum_{a+b=m} \| \sn^a\deh^{-1} \otimes \sn^b (\sn^{(n)} \sn^{(n)} T_2)\|_{L^2}\\
        &\hspace{2cm}\quad +C\sum_{a+b=m+1} \|\sn^a \cd \otimes \sn^b T_2\|_{L^2} + C\sum_{a+b=m} \| \sn^a \cd \otimes \sn^b (\sn^{(n)} T_2)\|_{L^2},
    \end{split}
    \]
    where $C$ depends only on $m$, and $\sharp$ is applied to $\sn^2$ and $(\sn^{(n)})^2$ according to Definition~\ref{raising an index}. In particular, if $T_1$ and $T_2$ are scalars, then the third term on the right-hand side of the inequality vanishes.
\end{lemma}

\begin{proof}
    For a tensor $T$, we use the notation $\cd T(X) = \cd_X T = \sn_X T - \sn^{(n)}_X T$. Then $\cd T$ is a tensor with one more degree of covariance than $T$, which consists of a sum of contractions of tensor products of $\cd$ and $T$. For $X \in \mathfrak X(\Sigma)$ and $\alpha \in \Omega^1(\Sigma)$, we have
    \begin{align*}
        &\big( \sn^2 T_1^\sharp - (\sn^{(n)})^2 T_2^\sharp \big)(X,\alpha)\\
        &= h^{-1}(\alpha,\omega^i) \sn^2_{X,e_i} T_1 - (h_n)^{-1}(\alpha,\omega^i)(\sn^{(n)})^2_{X,e_i} T_2\\
        &= \deh^{-1}(\alpha,\omega^i)(\sn^{(n)})^2_{X,e_i}T_2 + h^{-1}(\alpha,\omega^i)\big( \sn^2_{X,e_i} T_1 - (\sn^{(n)})^2_{X,e_i} T_2 \big)\\
        &= \deh^{-1}(\alpha,\omega^i)(\sn^{(n)})^2_{X,e_i}T_2 + h^{-1}(\alpha,\omega^i)\big( \sn^2_{X,e_i} \delta T + \sn_X (\cd T_2)(e_i) + \cd_X( \sn^{(n)} T_2 )(e_i) \big).
    \end{align*}
    Note that, if $T_1$ and $T_2$ are scalars, then the second term inside the parentheses in the last line vanishes. The result follows. 
\end{proof}

\subsection{Estimates for polynomial terms}

When computing the energy estimates for the system \eqref{the system}, it will repeatedly happen that we need to estimate terms which are polynomial in the $\delta$ variables, the approximate solution and $\cd$. In order to facilitate these calculations, we explain how to estimate the polynomial terms here.

\begin{table}[ht]
\centering
\captionsetup{width=13cm}
\begin{tabular}{ | c | c | c | }
 \hline
 Object & \makecell{Number of derivatives\\ controlled in $L^\infty$} & Bound in $W^{m,\infty}$ norm \\ 
 \hline\hline
 $K_n$ & $s-1$ & $( C + C_n\langle \ln t \rangle t^\varepsilon ) t^{-m-1}$ \\
 \hline
 $\lie_{\p_t} K_n$ & $s-1$ & $( C + C_n\langle \ln t \rangle t^\varepsilon ) t^{-m-2}$\\
 \hline
 $\sn^{(n)} K_n$ & $s-1$ & $C_n\langle \ln t \rangle t^{-m-2+\varepsilon}$\\
 \hline
 $\sn^{(n)}\sn^{(n)} K_n$ & $s-1$ &$C_n\langle \ln t \rangle t^{-m-3+\varepsilon}$\\
 \hline
 $d\s_n$ & $s-1$ & $C_n\langle \ln t \rangle t^{-m+\varepsilon}$\\
 \hline
 $\p_t\s_n$ & $s$ & $( C + C_n t^\varepsilon ) t^{-m-1}$\\
 \hline
 $\sn^{(n)} d\s_n$ & $s-1$ & $C_n\langle \ln t \rangle t^{-m-2+\varepsilon}$\\
 \hline
 $\sn^{(n)}\sn^{(n)} d\s_n$ & $s-1$ & $C_n\langle \ln t \rangle t^{-m-3+\varepsilon}$\\
 \hline
 $\sn^{(n)} d(\p_t\s_n)$ & $s-1$ & $C_n t^{-m-3+\varepsilon}$\\
 \hline
 $\deh, \deh^{-1}, \dt, \dep$ & $s-1$ & $C$\\
 \hline
 $\dk, \p_t\dep, \cd$ & $s-2$ & $C$\\
 \hline
 $\lie_{\p_t} \dk$ & $s-3$ & $C$\\
 \hline
\end{tabular}
\caption{This table summarizes the relevant information about the objects of interest, which we obtained in the previous subsection, see Lemmas~\ref{sup norm bootstrap estimates}, \ref{difference of connections bound}, \ref{kn estimates} and \ref{phin estimates}. The second column refers to the number of $\sn$ derivatives for which we have $L^\infty$ estimates. In the third column, the value of $m$ is assumed to be less or equal than the corresponding value in the second column.}
\label{table}
\end{table}

Assume that we want to estimate a term of the form
\[
\sn^{a_1} T_1 \otimes \cdots \otimes \sn^{a_r} T_r
\]
in $L^2$, where $T_i$ for $1 \leq i \leq r$ denotes any of the objects listed in Table~\ref{table}, and there is at least one $i$ such that $T_i$ is one of the $\delta$ objects or $\cd$. We assume that it is possible to single out one index $1 \leq \xi \leq r$, such that $T_\xi$ is one of the $\delta$ objects or $\cd$, in such a way that for $i \neq \xi$, $a_i$ is less or equal than the value in the second column of Table~\ref{table} which corresponds to the object $T_i$. Moreover, we assume that $a_\xi \leq s$ if $T_\xi = \cd$ and, otherwise, $a_\xi$ is less or equal than the number of derivatives of $T_\xi$ appearing in the energy $\se_s$. 

The idea for estimating the term of interest is as follows. We want to estimate every factor, except for $T_\xi$, in $L^\infty$, while $\|\sn^{a_\xi} T_\xi\|_{L^2}$ is estimated in terms of the energy. Note that $\cd$ does not appear in the energy, so in the case $T_\xi = \cd$, we use Lemma~\ref{difference of connections bound} first to estimate $\cd$ in terms of $\deh$ and $\deh^{-1}$. The important information to keep track of is then: the total number of spatial derivatives present, this includes the instances of $\sn^{(n)}$ and $d$ which may occur in the $T_i$ themselves; the total number of time derivatives present; and the total number of times $K_n$ appears. We now introduce appropriate counters for these quantities,
\begin{align*}
    \zeta &:= \text{total number of spatial derivatives,}\\
    \tau &:= \text{total number of time derivatives},\\
    \kappa &:= \text{total number of appearances of $K_n$}.
\end{align*}
Additionally, define $\lambda$ by
\[
\lambda := 
\begin{cases}
    1, \quad \text{if} \; T_\xi = \p_t \dep, \dep, \deh, \deh^{-1};\\
    0, \quad  \text{otherwise}.
\end{cases}
\]
Then, using the information in Table~\ref{table} and the definition of $\se_s$, we see that
\begin{equation} \label{polynomial estimate}
    t^{\zeta + \tau + \kappa} \|\sn^{a_1} T_1 \otimes \cdots \otimes \sn^{a_r} T_r\|_{L^2} \leq \big(C+C_n\langle \ln t \rangle t^\varepsilon\big) t^\lambda \se_s^{1/2}(t).
\end{equation}
This takes care of most of the polynomial terms that we will come across below, but there are two special cases. We show that these special cases can, in fact, be estimated in the same way as in \eqref{polynomial estimate}. The first special case is
\[
\sn^s T_1 \otimes \sn^{a_2} T_2,
\]
where $T_1$ is one of $K_n$, $d\s_n$, $\sn^{(n)} d\s_n$ or $d(\p_t\s_n)$; $T_2$ is one of the $\delta$ objects or $\cd$; and $a_2$ is either $0$ or $1$. For this type of term, we need to use Lemma~\ref{technical estimate 2}. Introduce $\lambda$ and the counters $\kappa$ and $\tau$ as above. Then, letting $b$ denote the number of spatial derivatives in $T_1$ itself and $\tau_2$ the number of time derivatives in $T_2$, we see that
\[
\begin{split}
    &t^{s+b+a_2+\kappa+\tau} \|\sn^s T_1 \otimes \sn^{a_2} T_2\|_{L^2}\\
    &\leq Ct^{s+b+a_2+\kappa+\tau}\Big[ \big( \|\sn^{s-1}(\sn^{(n)}T_1)\|_{L^\infty} + \|T_1\|_{W^{s-1,\infty}} \big)\|\sn^{a_2}T_2\|_{L^2} + \|T_1\|_{L^\infty}\|\sn^{a_2}T_2\|_{H^2} \Big]\\
    &\leq \big( C + C_n \langle \ln t \rangle t^\varepsilon \big)( t^{a_2+\tau_2}\|\sn^{a_2}T_2\|_{L^2} + t^{s+a_2+\tau_2}\|\sn^{a_2}T_2\|_{H^2} )\\
    &\leq \big(C + C_n \langle \ln t \rangle t^\varepsilon\big)t^\lambda\se_s^{1/2}(t),
\end{split}
\]
which is consistent with \eqref{polynomial estimate}. The second special case is 
\[
\sn^{s+1}K_n \otimes T_2,
\]
where $T_2$ is either $\dk$, $\deh$ or $\deh^{-1}$. Once again we need to use Lemma~\ref{technical estimate 2},
\[
\begin{split}
    t^{s+2}\|\sn^{s+1}K_n \otimes T_2\|_{L^2} &\leq Ct^{s+2}\Big[ \big( \|\sn^{s-1}( \sn^{(n)}\sn^{(n)}K_n )\|_{L^\infty} + \|\sn^{(n)}K_n\|_{W^{s-1,\infty}} \big)\|T_2\|_{L^2}\\
    &\quad \hspace{1.5cm} + \|\sn^{(n)}K_n\|_{L^\infty}\|T_2\|_{H^2} + \|K_n\|_{W^{s-1,\infty}} \|T_2\|_{H^2} \Big]\\
    &\leq \big(C+C_n\langle \ln t \rangle t^\varepsilon\big)( \|T_2\|_{L^2} + t^2\|T_2\|_{H^2} )\\
    &\leq \big(C+C_n\langle \ln t \rangle t^\varepsilon\big) t^\lambda \se_s^{1/2}(t),
\end{split}
\]
which is again consistent with \eqref{polynomial estimate}.

\subsection{Energy estimates} \label{energy estimates}

We begin this subsection by obtaining energy estimates for general transport and wave equations. 

\begin{proposition} \label{transport estimate}
    Let $T$ be a one parameter family of tensors on $\Sigma$, then
    \[
    \frac{d}{dt} \|T\|_{L^2}^2 \leq \frac{1}{t}(C + C_nt^\varepsilon) \|T\|_{L^2}^2 + 2\int_{\Sigma_t} h(\lie_{\p_t} T,T) \mu.
    \]
    In particular,
    \[
    \frac{d}{dt} \|T\|_{L^2}^2 \leq t\|\lie_{\p_t} T\|_{L^2}^2 + \frac{1}{t}( C + C_nt^\varepsilon ) \|T\|_{L^2}^2.
    \]
\end{proposition}

\begin{proof}
    We need to estimate the time derivative of $|T|_h^2$ for a one parameter family of tensors $T$ on $\Sigma$. First note that the dual metric $h^{-1}$ satisfies the equation
    \[
    \lie_{\p_t} h^{-1}(\alpha,\beta) = - h^{-1}(K(\,\cdot\,,\alpha), \beta) -h^{-1}( \alpha, K(\,\cdot\,,\beta) ),
    \]
    for $\alpha, \beta \in \Omega^1(\Sigma)$. Fix $r \in [t_0,t_b]$, let $\{E_i\}$ be an orthonormal frame for $h(r)$ with dual frame $\{\theta^i\}$ and extend it to $[t_0,t_b] \times \Sigma$ by requiring that $[\p_t,E_i] = 0$. As usual, for simplicity, we work with a time dependent $(1,1)$-tensor $T$. Then $T = T_i^k \theta^i \otimes E_k$ and
    \[
    |T|_h^2 = h^{-1}(\theta^i,\theta^k) h(E_\ell,E_m) T_i^\ell T_k^m,
    \]
    implying
    \[
    \begin{split}
        \p_t |T|_h^2 &= \big( -h^{-1}(K(\,\cdot\,,\theta^i),\theta^k) - h^{-1}( \theta^i, K(\,\cdot\,,\theta^k) ) \big)h(E_\ell,E_m)T_i^\ell T_k^m\\
        &\quad + h^{-1}(\theta^i,\theta^k)\big( h(K(E_\ell),E_m) + h(E_\ell,K(E_m)) \big)T_i^\ell T_k^m + 2h(\lie_{\p_t} T,T).
    \end{split}
    \]
    If we evaluate this expression at $t = r$, we obtain
    \[
    \begin{split}
        \p_t |T|_h^2(r) &= \sum_{i,k,\ell}( -K_r(E_k,\theta^i) - K_r(E_i,\theta^k) ) T_i^\ell(r) T_k^\ell(r)\\
        &\quad + \sum_{i,\ell,m}( K_r(E_\ell,\theta^m) + K_r(E_m,\theta^\ell) ) T_i^\ell(r) T_i^m(r) + 2h(\lie_{\p_t}T,T)(r).
    \end{split}
    \]
    Now, after a few applications of the Cauchy-Schwarz inequality, we conclude that there is a constant $C$ such that
    \[
    \p_t |T|_h^2(r) \leq C|K|_h(r) |T|_h^2(r) + 2h(\lie_{\p_t}T,T)(r).
    \]
    Since $r$ was arbitrary, the estimate holds for all $t \in [t_0,t_b]$.

    Now for the $L^2$ norm. Note that $\p_t \sqrt{\det h} = (\tr K)\sqrt{\det h}$. Hence
    \[
    \begin{split}
        \frac{d}{dt} \|T\|_{L^2}^2 = \frac{d}{dt} \int_{\Sigma_t} |T|_h^2 \sqrt{\det h} \mathring \mu &= \int_{\Sigma_t} \Big( \p_t|T|_h^2 + (\tr K)|T|_h^2 \Big) \mu\\
        &\leq \frac{1}{t}(C + C_nt^\varepsilon) \int_{\Sigma_t} |T|_h^2 \mu + 2\int_{\Sigma_t} h(\lie_{\p_t}T,T)\mu,
    \end{split}
    \]
    thus finishing the proof.
\end{proof}

\begin{proposition} \label{wave estimate}
    Let $T$ and $F$ be one parameter families of tensors on $\Sigma$ such that \[-\lie_{\p_t}^2 T + \Delta_h T = F,\] then
    \[
    \frac{d}{dt} \mathbb E[T] \leq \frac{1}{t}( C + C_nt^\varepsilon )\mathbb E[T] + t\|F\|_{L^2}^2.
    \]
\end{proposition}

\begin{proof}
    We compute,
    \[
    \begin{split}
        \frac{d}{dt} \|\lie_{\p_t}T\|_{L^2}^2 &\leq \frac{1}{t}(C + C_nt^{\varepsilon}) \|\lie_{\p_t}T\|_{L^2}^2 + 2\int_{\Sigma_t} h(\lie_{\p_t}^2 T, \lie_{\p_t}T)\mu\\
        &= \frac{1}{t}(C + C_nt^{\varepsilon}) \|\lie_{\p_t}T\|_{L^2}^2 + 2\int_{\Sigma_t} \Big( h(\Delta_h T, \lie_{\p_t}T) - h(F,\lie_{\p_t}T) \Big) \mu\\
        &= \frac{1}{t}(C + C_nt^{\varepsilon}) \|\lie_{\p_t}T\|_{L^2}^2 - 2\int_{\Sigma_t} h(\sn T,\sn \lie_{\p_t}T) \mu - 2\int_{\Sigma_t} h(F, \lie_{\p_t}T) \mu.
    \end{split}
    \]
    Furthermore,
    \[
        \frac{d}{dt} \|\sn T\|_{L^2}^2 \leq \frac{1}{t}(C + C_nt^\varepsilon)\|\sn T\|_{L^2}^2 + 2\int_{\Sigma_t} h(\lie_{\p_t} \sn T, \sn T) \mu.
    \]
    Putting together these two estimates yields
    \[
    \begin{split}
        &\frac{d}{dt}\big( \|\lie_{\p_t}T\|_{L^2}^2 + \|\sn T\|_{L^2}^2 \big)\\
        &\leq \frac{1}{t}(C + C_nt^\varepsilon) \big(\|\lie_{\p_t} T\|_{L^2}^2 + \|\sn T\|_{L^2}^2\big) + 2\int_{\Sigma_t} h([\lie_{\p_t},\sn]T, \sn T)\mu - 2\int_{\Sigma_t} h(F, \lie_{\p_t}T) \mu\\
        &\leq \frac{1}{t}(C + C_nt^\varepsilon) \big(\|\lie_{\p_t} T\|_{L^2}^2 + \|\sn T\|_{L^2}^2\big) + t\|[\lie_{\p_t},\sn]T\|_{L^2}^2 + t\|F\|_{L^2}^2\\
        &\leq \frac{1}{t}(C + C_nt^\varepsilon)\mathbb E[T] + t\|F\|_{L^2}^2,
    \end{split}
    \]
    where we have used the commutator estimate in the last inequality. Finally,
    \[
    \begin{split}
        \frac{d}{dt} \big( t^{-2}\|T\|_{L^2}^2 \big) &\leq -2t^{-3}\|T\|_{L^2}^2 + t^{-2}\left( t\|\lie_{\p_t}T\|_{L^2}^2 + \frac{1}{t}(C + C_nt^\varepsilon)\|T\|_{L^2}^2 \right)\\
        &\leq \frac{1}{t}(C + C_nt^\varepsilon)\mathbb E[T],
    \end{split}
    \]
    which together with the previous estimate yields the result.
\end{proof}

Now we move on to the energy estimates for $\dk$ and $\dep$. Note that, as a consequence of \eqref{the system}, they satisfy the equations
\begin{subequations} \label{wave equations for delta objects}
\begin{align} 
\begin{split}
    -\lie_{\p_t}^2 \dk + \Delta_h \dk &= F_1(K) - F_1(K_n) + F_2(K) - F_2(K_n) + F_3(\s) - F_3(\s_n)\\
    &\quad  + \sn^2 \theta^\sharp - (\sn^{(n)})^2 \theta_n^\sharp + (\Delta_{h_n} - \Delta_{h}) K_n + I_{K_n}, 
\end{split}\label{delta k equation}\\
\begin{split}
    -\p_t^2 \dep + \Delta_h \dep &= (\Delta_{h_n} - \Delta_h) \s_n + \theta \p_t \s - \theta_n \p_t \s_n\\
    &\quad + V' \circ \s - V' \circ \s_n + (V' \circ \s_n - \Box_{g_n} \s_n), \label{delta phi equation}
\end{split}
\end{align}
\end{subequations}
where $I_{K_n}$ was introduced in Lemma~\ref{estimates for inhomogeneous terms}.

\begin{proposition} \label{energy estimate for delta k and delta phi}
    Given a positive integer $N$, there is an $n_{N,s}$ such that for $n \geq n_{N,s}$, 
    \[
    \frac{d}{dt}\left(\sum_{m=0}^{s-1} t^{2(m+1)}\mathbb E_m[\dk] + \sum_{m=0}^s t^{2m}\mathbb E_m[\dep] \right) \leq \frac{1}{t}\big(C + C_n\langle \ln t \rangle t^\varepsilon\big) \se_s(t) + C_nt^{2N+2s+1}.
    \]
\end{proposition}

\begin{proof}
    Applying $\sn^m$ to \eqref{wave equations for delta objects} yields
    \begin{align*}
    \begin{split}
        -\lie_{\p_t}^2 \sn^m \dk + \Delta_h \sn^m \dk &= \sn^m\big( F_1(K) - F_1(K_n) \big) + \sn^m\big( F_2(K) - F_2(K_n ) \big)\\
        &\quad + \sn^m\big( F_3(\s) - F_3(\s_n) \big) + \sn^m I_{K_n} + \sn^m\big( \sn^2 \theta^\sharp - (\sn^{(n)})^2\theta_n^\sharp \big)\\
        &\quad + \sn^m (\Delta_{h_n} - \Delta_{h})K_n - [\lie_{\p_t}^2,\sn^m]\dk + [\Delta_h,\sn^m]\dk,
    \end{split}\\
    \begin{split}
        -\p_t^2 \sn^m \dep + \Delta_h \sn^m \dep &= \sn^m\big((\Delta_{h_n} - \Delta_h) \s_n\big) + \sn^m(\theta \p_t \s - \theta_n \p_t \s_n) \\
        &\quad + \sn^m (V' \circ \s - V' \circ \s_n) + \sn^m(V' \circ \s_n - \Box_{g_n} \s_n)\\
        &\quad - [\lie_{\p_t}^2,\sn^m]\dep + [\Delta_h,\sn^m]\dep.
    \end{split}
    \end{align*}
    We want to apply Proposition~\ref{wave estimate} to these equations. For that purpose, we now proceed to estimate all the terms on the right-hand side. We mainly focus on the equation for $\dk$, since the one for $\dep$ is similar but simpler.

    \paragraph{Step 1: Estimating $(\Delta_h - \Delta_{h_n})K_n$.} By Lemma~\ref{estimate for difference of hessians} with $T_1 = T_2 = K_n$,
    \[
    \begin{split}
        \|\sn^m(\Delta_h - \Delta_{h_n})K_n\|_{L^2} &\leq C\sum_{a+b=m} \big(\|\sn^a \deh^{-1} \otimes \sn^b ( \sn^{(n)} \sn^{(n)} K_n)\|_{L^2}\\
        &\quad + \|\sn^a \cd \otimes \sn^b ( \sn^{(n)} K_n )\|_{L^2}\big) + C\sum_{a+b=m+1} \|\sn^a \cd \otimes \sn^b K_n \|_{L^2}.
    \end{split}
    \]
    Applying \eqref{polynomial estimate}, we obtain
    \[
    t^{m+1}\|\sn^m(\Delta_h - \Delta_{h_n})K_n\|_{L^2} \leq \frac{1}{t}\big(C + C_n \langle \ln t \rangle t^\varepsilon\big) \se_s(s)^{1/2}.
    \]
    \paragraph{Step 2: Estimating $\sn^2 \theta^\sharp - (\sn^{(n)})^2 \theta_n^\sharp$.} From Lemma~\ref{estimate for difference of hessians},
    \[
    \begin{split}
        \|\sn^m(\sn^2 \theta^\sharp - (\sn^{(n)})^2 \theta_n^\sharp)\|_{L^2} &\leq \|\sn^{m+2} \dt\|_{L^2} + C\sum_{a+b=m} \| \sn^a \deh^{-1} \otimes \sn^b (\sn^{(n)} d\theta_n)\|_{L^2}\\
        &\quad + C\sum_{a+b=m} \|\sn^a \cd \otimes \sn^b d\theta_n \|_{L^2}. 
    \end{split}
    \]
    Then \eqref{polynomial estimate} yields
    \[
    t^{m+1}\|\sn^m(\sn^2 \theta^\sharp - (\sn^{(n)})^2 \theta_n^\sharp)\|_{L^2} \leq \frac{1}{t}\big(C + C_n \langle \ln t \rangle t^\varepsilon\big)\se_s(t)^{1/2}.
    \]

    \paragraph{Step 3: Estimating the $F_1$ terms.} Since this term consists of a sum of contractions of $K \otimes K \otimes K - K_n \otimes K_n \otimes K_n$, then
    \[
    \begin{split}
        \big\|\sn^m\big( F_1(K) - F_1(K_n) \big)\big\|_{L^2} &\leq C \sum_{a+b+c=m} \big(\|\sn^a K_n \otimes \sn^b K_n \otimes \sn^c \dk \|_{L^2}\\
        &\quad + \|\sn^a K_n \otimes \sn^b \dk \otimes \sn^c \dk \|_{L^2} + \|\sn^a \dk \otimes \sn^b \dk \otimes \sn^c \dk \|_{L^2}\big).
    \end{split}
    \]
    From \eqref{polynomial estimate}, we conclude that 
    \[
    t^{m+1}\big\|\sn^m\big( F_1(K) - F_1(K_n) \big)\big\|_{L^2} \leq \frac{1}{t}\big(C + C_n \langle \ln t \rangle t^\varepsilon\big)\se_s(t)^{1/2}.
    \]

    \paragraph{Step 4: Estimating the $F_2$ terms.} We have
    \[
    \begin{split}
        &\big\|\sn^m\big(F_2(K) - F_2(K_n)\big)\big\|_{L^2}\\
        &\leq C\sum_{a+b=m}\big( \|\sn^a \lie_{\p_t} K_n \otimes \sn^b \dk\|_{L^2} + \|\sn^a \lie_{\p_t} \dk \otimes \sn^b K_n\|_{L^2} + \|\sn^a \lie_{\p_t} \dk \otimes \sn^b \dk\|_{L^2}\big). 
    \end{split}
    \]
    \eqref{polynomial estimate} then implies
    \[
    t^{m+1}\big\|\sn^m\big(F_2(K) - F_2(K_n)\big)\big\|_{L^2} \leq \frac{1}{t}\big(C+C_n \langle \ln t \rangle t^\varepsilon\big) \se_s(t)^{1/2}.
    \]

    \paragraph{Step 5: Estimating the $F_3$ terms.} First consider $H$. We have
    \[
    \begin{split}
        &\big\|\sn^m\big(H(\s) - H(\s_n)\big)\big\|_{L^2}\\
        &\leq C\big( \|\sn^m(d\s \otimes d\s \otimes \dk)\|_{L^2} + \|\sn^m(d\s \otimes d\dep \otimes K_n)\|_{L^2}\\
        &\quad + \|\sn^m(d\dep \otimes d\s_n \otimes K_n)\|_{L^2} + \|\sn^m( \deh^{-1} \otimes d\s_n \otimes d\s_n \otimes K_n )\|_{L^2}\big),
    \end{split}
    \]
    hence, substituting $\s = \s_n + \dep$ and applying \eqref{polynomial estimate}, we obtain
    \[
    t^{m+1}\big\|\sn^m\big(H(\s) - H(\s_n)\big)\big\|_{L^2} \leq \frac{1}{t}\big(C+C_n \langle \ln t \rangle t^\varepsilon\big) \se_s(t)^{1/2}.
    \]
    Now for $\p_t \s \sn^2 \s^\sharp - \p_t\s_n(\sn^{(n)})^2 \s_n^\sharp$,
    \[
    \p_t \s \sn^2 \s^\sharp - \p_t\s_n(\sn^{(n)})^2 \s_n^\sharp = \p_t \dep \sn^2 \s_n^\sharp + \p_t \dep \sn^2\dep^\sharp + \p_t \s_n\big( \sn^2\s^\sharp - (\sn^{(n)})^2 \s_n^\sharp \big).
    \]
    For the third term on the right-hand side, first we extract $\p_t\s_n$ in $L^\infty$ and apply Lemma~\ref{estimate for difference of hessians}. Then we can apply \eqref{polynomial estimate} to the result to get
    \[
    t^{m+1}\big\|\sn^m\big(\p_t \s \sn^2 \s^\sharp - \p_t\s_n(\sn^{(n)})^2 \s_n^\sharp\big)\big\|_{L^2} \leq \frac{1}{t}\big(C+C_n \langle \ln t \rangle t^\varepsilon\big) \se_s(t)^{1/2}.
    \]
    We continue with the potential terms. Consider $(V \circ \s)K - (V \circ \s_n)K_n$, then
    \[
    (V \circ \s)K - (V \circ \s_n)K_n = (V \circ \s)\dk + \int_0^1 V'\big(r\s + (1-r)\s_n\big)dr\, \dep K_n.
    \]
    It follows from \eqref{polynomial estimate} that
    \[
    t^{m+1}\big\|\sn^m\big( (V \circ \s)K - (V \circ \s_n)K_n \big)\big\|_{L^2} \leq \frac{1}{t}\big(C + C_n \langle \ln t \rangle t^{\varepsilon}\big) \se_s(t)^{1/2}.
    \]
    The other terms involving the potential can be treated in a similar way, so altogether we obtain
    \[
    t^{m+1}\big\|\sn^m\big( F_3(\s) - F_3(\s_n) \big)\big\|_{L^2} \leq \frac{1}{t}\big( C + C_n \langle \ln t \rangle t^\varepsilon \big) \se_s(t)^{1/2}.
    \]

    \paragraph{Step 6: Estimating the commutators.} For $[\lie_{\p_t}^2, \sn^m]\dk$, we have
    \[
    [\lie_{\p_t}^2, \sn^m]\dk = [\lie_{\p_t},\sn^m]\lie_{\p_t}\dk + \lie_{\p_t}\big( [\lie_{\p_t},\sn^m]\dk \big).
    \]
    For the first term, by Lemma~\ref{commutator estimates},
    \[
    t^{m+1}\| [\lie_{\p_t},\sn^m]\lie_{\p_t}\dk \|_{L^2} \leq C_nt^\varepsilon \sum_{a=0}^{m-1} t^a \|\sn^a \lie_{\p_t} \dk\|_{L^2} \leq C_n \langle \ln t \rangle t^{-1+\varepsilon} \se_s(t)^{1/2}.
    \]
    For the second term, recall that $[\lie_{\p_t},\sn^m]\dk = \sum_{a+b=m-1} \sn^a\big( [\lie_{\p_t},\sn]\sn^b \dk \big)$. Hence, by Lemmas~\ref{commutator formula} and \ref{commutator formula 2}, it is enough to estimate $\lie_{\p_t}$ of $\sum_{a+b=m-1} h \otimes h^{-1}\otimes \sn^{a+1} K \otimes \sn^b \dk$. After commuting $\lie_{\p_t}$ with $\sn^{a+1}$ and $\sn^b$, and using Lemmas~\ref{commutator formula} and \ref{commutator formula 2} again, we obtain
    \[
    \begin{split}
        \big\|\lie_{\p_t}\big( [\lie_{\p_t},\sn^m] \dk\big)\big\|_{L^2} &\leq C\sum_{a+b+c=m-1} \|\sn^{a+1} K \otimes \sn^b K \otimes \sn^c \dk\|_{L^2}\\
        &\quad + C\sum_{a+b=m-1} \big( \|\sn^{a+1} \lie_{\p_t} K \otimes \sn^b \dk\|_{L^2} + \|\sn^{a+1} K \otimes \sn^b \lie_{\p_t} \dk\|_{L^2}\big).
    \end{split}
    \]
    So we can substitute $K = K_n + \dk$ and apply \eqref{polynomial estimate} to obtain
    \[
    \|[\lie_{\p_t}^2, \sn^m]\dk\|_{L^2} \leq \frac{1}{t}\big( C + C_n \langle \ln t \rangle t^\varepsilon \big) \se_s(t)^{1/2}.
    \]
    For the other commutator, note that $[\Delta_h, \sn^m]\dk$ consists of a sum of contractions and $h$-traces of $\sum_{a+b=m} \sn^a( \bar R \otimes \sn^b \dk )$. Moreover, in dimension 3 the curvature tensor is completely determined by the Ricci tensor; recall \eqref{curvature in terms of ricci}. Hence, by Lemma~\ref{estimates for the curvature} and Sobolev embedding,
    \[
    \begin{split}
        &t^{m+1}\|[\Delta_h,\sn^m]\dk\|_{L^2}\\
        &\hspace{2cm} \leq Ct^{m+1}\sum_{a+b=m} \big(\|\sn^a \sric_n\|_{L^\infty} \|\sn^b \dk\|_{L^2} + \|\sn^a(\sric - \sric_n) \otimes \sn^b \dk\|_{L^2}\big)\\
        &\hspace{2cm} \leq C_nt^{-1+\varepsilon} \se_s(t)^{1/2} + Ct^{m+1}\sum_{\substack{a+b=m\\ a \geq s-2}} t^{-5/2} \|\sn^a(\sric - \sric_n)\|_{L^2} \|\sn^b \dk\|_{H^2}\\
        &\hspace{2cm} \leq C_nt^{-1+\varepsilon} \se_s(t)^{1/2}.
    \end{split}
    \]

    \paragraph{Step 7: Putting everything together.} By Lemma~\ref{estimates for inhomogeneous terms}, there is an $n$ large enough such that
    \begin{equation} \label{wave estimate for delta k}
        t^{m+1}\|-\lie_{\p_t}^2 \sn^m \dk + \Delta_h \sn^m \dk\|_{L^2} \leq \frac{1}{t}\big(C + C_n \langle \ln t \rangle t^\varepsilon\big) \se_s(t)^{1/2} + C_nt^{N+s+1}.
    \end{equation}
    Now we can apply Proposition~\ref{wave estimate} to obtain
    \[
    \begin{split}
        &\frac{d}{dt}\big( t^{2(m+1)}\mathbb E_m[\dk] \big)\\
        &\hspace{1cm}\leq \frac{2(m+1)}{t} t^{2(m+1)}\mathbb E_m[\dk]\\
        &\hspace{1cm}\quad + t^{2(m+1)}\bigg( \frac{1}{t}\big(C + C_n\langle \ln t \rangle t^\varepsilon\big)\mathbb E_m[\dk] + t\|-\lie_{\p_t} \sn^m \dk + \Delta_h \sn^m \dk\|_{L^2}^2  \bigg)\\
        &\hspace{1cm}\leq \frac{1}{t}\big(C + C_n\langle \ln t \rangle t^\varepsilon\big)\se_s(t) + C_nt^{2N+2s+3},
    \end{split}
    \]
    where we have used Lemma~\ref{commutator estimates} to estimate 
    \[
    t^{m+1}\| \lie_{\p_t} \sn^m \dk\|_{L^2} \leq t^{m+1}\big( \|\sn^m \lie_{\p_t}\dk\|_{L^2} + \|[\lie_{\p_t},\sn^m]\dk\|_{L^2} \big) \leq \big(C + C_n\langle \ln t \rangle t^\varepsilon\big)\se_s(t)^{1/2}.
    \]

    \paragraph{Step 8: Estimate for $\dep$.} The main difference with the case for $\dk$, is that now $m \leq s$. However, $\dep$ is a scalar. So the first covariant derivative becomes the differential $d$. This means that there is one derivative less to worry about in the commutators. Therefore, all the terms can be estimated similarly to the case for $\dk$, by repeatedly using \eqref{polynomial estimate}, so we omit the details. By Theorem~\ref{approximate solutions}, there is then an $n$ large enough such that for $m \leq s$,
    \[
    \frac{d}{dt}\big( t^{2m}\mathbb E_m[\dep] \big) \leq \frac{1}{t}\big(C + C_n\langle \ln t \rangle t^\varepsilon\big)\se_s(t)^{1/2} + C_nt^{2N+2s+1},
    \]
    thus finishing the proof.
\end{proof}

We continue with the estimates for $\dt$, $\deh$ and $\deh^{-1}$. As a consequence of \eqref{the system}, they satisfy the equations
\begin{subequations} \label{transport equations for delta objects}
\begin{align}
    \begin{split}
        \p_t \dt &= -\tr(\dk^2) - 2\tr(K_n \circ \dk) - (\p_t \dep)^2\\
        &\quad - 2\p_t \s_n \p_t \dep + V \circ \s - V \circ \s_n - I_{\theta_n},
    \end{split}\label{delta theta equation}\\
    \begin{split}
        \lie_{\p_t} \deh(X,Y) &= h(\dk(X),Y) + h(X,\dk(Y))\\
        &\quad + \deh(K_n(X),Y) + \deh(X,K_n(Y)),
    \end{split}\label{delta h equation}\\
    \begin{split}
        \lie_{\p_t} \deh^{-1}(\alpha,\beta) &= -h^{-1}(\dk(\,\cdot\,,\alpha),\beta) - h^{-1}(\alpha,\dk(\,\cdot\,,\beta))\\
        &\quad - \deh^{-1}(K_n(\,\cdot\,,\alpha),\beta) - \deh^{-1}(\alpha,K_n(\,\cdot\,,\beta)), \label{delta h dual equation}
    \end{split}
\end{align}
\end{subequations}
for $X,Y \in \mfx(\Sigma)$ and $\alpha, \beta \in \Omega^1(\Sigma)$. We begin by obtaining control of the less than top order derivatives.

\begin{proposition} \label{energy estimate for delta theta and delta h}
    For every positive integer $N$, there is an $n_{N,s}$ large enough such that, for $n \geq n_{N,s}$ and $m \leq s$,
    \[
    \begin{split}
        &\frac{d}{dt} \left( t^{2m} \|\sn^m \dt\|_{L^2}^2 + t^{2(m-1)}\big( \|\sn^m \deh\|_{L^2}^2 + \|\sn^m \deh^{-1}\|_{L^2}^2 \big) \right)\\
        &\hspace{6cm} \leq \frac{1}{t}\big(C + C_n \langle \ln t \rangle t^\varepsilon\big) \se_s(t) + C_n t^{2N+2s+1}.
    \end{split}
    \]
\end{proposition}

\begin{proof}
    We focus on $\dt$, since the case for $\deh$ and $\deh^{-1}$ is similar but simpler. From \eqref{delta theta equation}, we have
    \[
    \begin{split}
        \lie_{\p_t} \sn^m \dt &= -\sn^m \tr(\dk^2) - 2\sn^m \tr(K_n \circ \dk) - 2\sn^m (\p_t \s_n \p_t \dep) - \sn^m( \p_t \dep )^2\\
        &\quad + \sn^m( V \circ \s - V \circ \s_n ) - \sn^m I_{\theta_n} + [\lie_{\p_t},\sn^m]\dt.
    \end{split}
    \]
    We estimate the terms on the right-hand side by applying \eqref{polynomial estimate} and Lemma~\ref{commutator estimates}. From Lemma~\ref{estimates for inhomogeneous terms}, we conclude that there is an $n$ large enough such that
    \[
    t^m\|\lie_{\p_t} \sn^m \dt\|_{L^2} \leq \frac{1}{t}\big(C + C_n \langle \ln t \rangle t^\varepsilon\big) \se_s(t)^{1/2} + C_nt^{N+s}.
    \]
    The estimate for $\dt$ then follows from Proposition~\ref{transport estimate}.
\end{proof}

Finally, we obtain estimates for the modified top order quantities. As we shall see later, these actually give us control for all top order derivatives of $\dt$, $\deh$ and $\deh^{-1}$; see Lemma~\ref{control of energy with modified energy} below.

\begin{lemma} \label{modified energy in terms of the energy}
    The following estimate holds,
    \[
    t^{2(s+1)}\widetilde{\mathbb E}_{s+1}[\dt] + t^{2s}\big( \widetilde{\mathbb E}_{s+1}[\deh] + \widetilde{\mathbb E}_{s+1}[\deh^{-1}] \big) \leq \big(C + C_n\langle \ln t \rangle t^\varepsilon\big)\se_s(t).
    \]
\end{lemma}

\begin{proof}
    By the definitions of $\widetilde{\mathbb E}_{s+1}[\dt]$, $\widetilde{\mathbb E}_{s+1}[\deh]$ and $\widetilde{\mathbb E}_{s+1}[\deh^{-1}]$,
    \[
    \begin{split}
        \widetilde{\mathbb E}_{s+1}[\dt] &\leq C\big( \|\sn^{s+1} \dt\|_{L^2}^2 + \|\dk + K_n\|_{L^\infty}^2 \|\sn^{s-1} \lie_{\p_t} \dk\|_{L^2}^2\\
        &\quad + \|\dk + K_n\|_{L^\infty}^2 \|[\lie_{\p_t},\sn^{s-1}] \dk\|_{L^2}^2 + \|\p_t(\dep + \s_n)\|_{L^\infty}^2 \|\sn^{s+1}\dep\|_{L^2}^2 \big),\\
        \widetilde{\mathbb E}_{s+1}[\deh] &\leq C\big( \|\sn^{s+1}\deh\|_{L^2}^2 + (1+\|\deh\|_{L^\infty}^2)(\|\sn^{s-1} \lie_{\p_t} \dk\|_{L^2}^2 + \|[\lie_{\p_t},\sn^{s-1}]\dk\|_{L^2}^2) \big),\\
        \widetilde{\mathbb E}_{s+1}[\deh^{-1}] &\leq C\big( \|\sn^{s+1}\deh^{-1}\|_{L^2}^2 + (1+\|\deh^{-1}\|_{L^\infty}^2)(\|\sn^{s-1} \lie_{\p_t} \dk\|_{L^2}^2 + \|[\lie_{\p_t},\sn^{s-1}]\dk\|_{L^2}^2) \big).
    \end{split}
    \]
    The result follows by Lemma~\ref{commutator estimates}.
\end{proof}

\begin{proposition} \label{estimates for modified top order quantities}
    For every positive integer $N$, there is an $n_{N,s}$ such that, for $n \geq n_{N,s}$,
    \[
    \frac{d}{dt} \left( t^{2(s+1)} \widetilde{\mathbb E}_{s+1}[\dt] + t^{2s}\big( \widetilde{\mathbb E}_{s+1}[\deh] + \widetilde{\mathbb E}_{s+1}[\deh^{-1}] \big) \right) \leq \frac{1}{t}\big( C + C_n \langle \ln t \rangle t^\varepsilon \big) \se_s(t) + C_nt^{2N + 2s + 1}.
    \]
\end{proposition}

\begin{proof}
    We begin with $\dt$. The idea is to use Proposition~\ref{transport estimate}. To that end, we compute,
    \[
    \lie_{\p_t} \Delta_h \sn^{s-1} \dt = -2\tr_s\big( (\dk + K_n) \circ \Delta_h \sn^{s-1} \dk \big) - 2\p_t(\dep + \s_n) \Delta_h \sn^{s-1} \p_t \dep + \cdots
    \]
    where $\cdots$ denotes terms that can be estimated by applying \eqref{polynomial estimate} and Lemma~\ref{commutator estimates}. Note that the problem with estimating the right-hand side of this equality is that $\Delta_h \sn^{s-1} \dk$ and $\Delta_h \sn^{s-1} \p_t \dep$ cannot be estimated in terms of $\se_s$. However,
    \[
    \begin{split}
        \lie_{\p_t}\big( (\dk + K_n) \circ \lie_{\p_t} \sn^{s-1} \dk \big) &= \lie_{\p_t}(\dk + K_n) \circ \lie_{\p_t}\sn^{s-1} \dk + (\dk + K_n) \circ \Delta_h \sn^{s-1} \dk\\
        &\quad + (\dk + K_n) \circ (\lie_{\p_t}^2 - \Delta_h) \sn^{s-1} \dk,
    \end{split}
    \]
    and moreover
    \[
    \begin{split}
        \lie_{\p_t}\big( \p_t(\dep + \s_n) \Delta_h \sn^{s-1} \dep \big) &= \p_t^2(\dep + \s_n) \Delta_h \sn^{s-1} \dep + \p_t(\dep + \s_n) \Delta_h \sn^{s-1} \p_t \dep\\
        &\quad + \p_t(\dep + \s_n)[\lie_{\p_t},\Delta_h \sn^{s-1}] \dep.
    \end{split}
    \]
    Putting everything together, we see that the problematic terms cancel out and we obtain
    \[
    \begin{split}
        &\big\|\lie_{\p_t}\big( \Delta_h \sn^{s-1} \dt + 2\tr_s\big( (\dk + K_n) \circ \lie_{\p_t} \sn^{s-1} \dk \big) + 2\p_t(\dep + \s_n) \Delta_h \sn^{s-1} \dep\big)\big\|_{L^2}\\
        &\hspace{0.5cm} \leq 2\|\lie_{\p_t}(\dk + K_n)\|_{L^\infty} \|\lie_{\p_t} \sn^{s-1} \dk\|_{L^2} + 2\|\dk + K_n\|_{L^\infty}  \|(\lie_{\p_t}^2 - \Delta_h)\sn^{s-1}\dk\|_{L^2}\\
        &\hspace{0.5cm}\quad + 2\|\p_t^2(\dep + \s_n)\|_{L^\infty} \|\Delta_h \sn^{s-1}\dep\|_{L^2} + 2\|\p_t(\dep + \s_n)\|_{L^\infty} \|[\lie_{\p_t},\Delta_h \sn^{s-1}] \dep\|_{L^2} + \cdots.
    \end{split}
    \]
    After using the evolution equations \eqref{delta phi equation} and \eqref{approximate scalar field equation} satisfied by $\dep$ and $\s_n$ to substitute the second time derivatives, we can use \eqref{wave estimate for delta k}, \eqref{polynomial estimate}, Lemma~\ref{commutator estimates} and Lemma~\ref{estimates for inhomogeneous terms} to estimate the right-hand side of this inequality. We conclude that 
    \[
    \begin{split}
        &t^{s+1}\|\lie_{\p_t}\big( \Delta_h \sn^{s-1} \dt + 2\tr_s\big( (\dk + K_n) \circ \lie_{\p_t} \sn^{s-1} \dk \big) + 2\p_t(\dep + \s_n) \Delta_h \sn^{s-1} \dep\big)\|_{L^2}\\
        &\hspace{8cm} \leq \frac{1}{t}\big(C + C_n\langle \ln t \rangle t^\varepsilon\big)\se_s(t)^{1/2} + C_nt^{N+s}.
    \end{split}
    \]
    The estimate for $\widetilde{\mathbb E}_{s+1}[\dt]$ now follows from Proposition~\ref{transport estimate} and Lemma~\ref{modified energy in terms of the energy}.

    We continue with $\deh$. We compute
    \[
    \begin{split}
        \lie_{\p_t} \Delta_h \sn^{s-1} \deh(X,Y) &= h(\Delta_h \sn^{s-1} \dk(X),Y) + h(X,\Delta_h \sn^{s-1}\dk(Y))\\
        &\quad + [\lie_{\p_t},\Delta_h \sn^{s-1}]\deh(X,Y) + \cdots,
    \end{split}
    \]
    where $\cdots$ denotes terms that can be estimated by applying \eqref{polynomial estimate}. Note that the issue with estimating the commutator $[\lie_{\p_t},\Delta_h \sn^{s-1}]\deh$, is that it contains the term $\Delta_h \sn^{s-2}([\lie_{\p_t},\sn]\deh)$, which contains $s+1$ derivatives of $K$. Since $\lie_{\p_t}\sn = \cs[\sn K]$, by Lemma~\ref{commutator formula}, we see that
    \[
    \Delta_h \sn^{s-1}([\lie_{\p_t},\sn]\deh) = -\deh(\cs[\Delta_h \sn^{s-1} K],\,\cdot\,) - \deh(\,\cdot\,,\cs[\Delta_h \sn^{s-1} K]) + \cdots,
    \]
    where $\cdots$ denotes terms which contain up to $s$ derivatives of $K$, hence can be treated with \eqref{polynomial estimate}. Thus
    \[
    \begin{split}
        \lie_{\p_t} \Delta_h \sn^{s-1} \deh &= h(\Delta_h \sn^{s-1} \dk,\,\cdot\,) + h(\,\cdot\,,\Delta_h \sn^{s-1}\dk)\\
        &\quad -\deh(\cs[\Delta_h \sn^{s-1} \dk],\,\cdot\,) - \deh(\,\cdot\,,\cs[\Delta_h \sn^{s-1}\dk]) + \cdots, 
    \end{split}
    \]
    where the terms in $\cdots$ can be estimated with \eqref{polynomial estimate} and Lemma~\ref{commutator estimates}. We keep going,
    \[
    \begin{split}
        &\lie_{\p_t}\big(h(\lie_{\p_t} \sn^{s-1} \dk,\,\cdot\,) + h(\,\cdot\,,\lie_{\p_t} \sn^{s-1} \dk)\big)\\
        &\hspace{0.5cm} = \lie_{\p_t}h( \lie_{\p_t}\sn^{s-1} \dk,\,\cdot\, ) + \lie_{\p_t}h( \,\cdot\,,\lie_{\p_t}\sn^{s-1} \dk ) + h(\Delta_h \sn^{s-1}\dk,\,\cdot\,) + h(\,\cdot\,,\Delta_h \sn^{s-1}\dk)\\
        &\hspace{0.5cm} \quad + h((\lie_{\p_t}^2 - \Delta_h) \sn^{s-1}\dk,\,\cdot\,) + h(\,\cdot\,,(\lie_{\p_t}^2 - \Delta_h) \sn^{s-1}\dk).
    \end{split}
    \]
    For the remaining terms, note that $\lie_{\p_t} \cs[T] = \cs[\lie_{\p_t}T] + \cf(K \otimes T)$, where $\cf(K \otimes T)$ is a sum of contractions of $K \otimes T$ with some indices raised and lowered with $h$. Hence
    \[
    \begin{split}
        &\lie_{\p_t}\big(\deh( \cs[\lie_{\p_t} \sn^{s-1}\dk],\,\cdot\, ) + \deh(\,\cdot\,, \cs[\lie_{\p_t} \sn^{s-1}\dk] )\big)\\
        &\hspace{3cm} = \lie_{\p_t}\deh( \cs[\lie_{\p_t} \sn^{s-1}\dk],\,\cdot\, ) + \lie_{\p_t}\deh( \,\cdot\,,\cs[\lie_{\p_t} \sn^{s-1}\dk] )\\
        &\hspace{3cm}\quad + \deh(\cs[ (\lie_{\p_t}^2 - \Delta_h)\sn^{s-1}\dk ],\,\cdot\,) + \deh(\,\cdot\,,\cs[ (\lie_{\p_t}^2 - \Delta_h)\sn^{s-1}\dk ])\\
        &\hspace{3cm}\quad + \deh(\cs[  \Delta_h\sn^{s-1}\dk ],\,\cdot\,) + \deh(\,\cdot\,,\cs[  \Delta_h\sn^{s-1}\dk ])\\
        &\hspace{3cm}\quad + \deh( \cf(K \otimes \lie_{\p_t} \sn^{s-1}\dk),\,\cdot\, ) + \deh(\,\cdot\,, \cf(K \otimes \lie_{\p_t} \sn^{s-1}\dk) ).
    \end{split}
    \]
    Note that again, after putting everything together, the problematic terms cancel out. We can then use \eqref{polynomial estimate}, \eqref{wave estimate for delta k}, Lemma~\ref{commutator estimates} and Proposition~\ref{transport estimate} as before to obtain the result. The proof for $\deh^{-1}$ is similar.
\end{proof}

We can summarize what we have done so far in the following estimate for the modified energy.

\begin{proposition}
    For every positive integer $N$ there is an $n_{N,s}$ such that for $n \geq n_{N,s}$,
    \[
    \frac{d}{dt} \widetilde \se_s(t) \leq \frac{1}{t}\big(C + C_n\langle \ln t \rangle t^\varepsilon\big) \se_s(t) + C_nt^{2N + 2s + 1}.
    \]
\end{proposition}

\subsection{The bootstrap improvement result} \label{controlling the energy with the modified energy}

Now we verify that, in fact, the modified energy controls the main energy. After that, we can finally conclude the proofs of Theorem~\ref{bootstrap theorem} and Corollary~\ref{global existence corollary}.

\begin{lemma} \label{elliptic estimates}
    If $T$ is a tensor on $\Sigma$, then
    \[
    \|\sn^2T\|_{L^2}^2 \leq 2\|\Delta_h T\|_{L^2}^2 + C_nt^{-2+\varepsilon} \|\sn T\|_{L^2}^2 + C_nt^{-4+2\varepsilon} \|T\|_{L^2}^2.
    \]
\end{lemma}

\begin{proof}
    For $X, Y \in \mfx(\Sigma)$, we use the notation $\bar R(X,Y)T = \sn^2_{X,Y}T - \sn^2_{Y,X}T$. We compute,
    \[
    \begin{split}
        \|\Delta_h T\|_{L^2}^2 &= \int_{\Sigma_t} h(\Delta_h T, \Delta_h T)\mu = -\int_{\Sigma_t} h(\sn T, \sn \Delta_h T) \mu\\
        &= -\int_{\Sigma_t} h^{ik} h^{\ell m} h(\sn_{e_i} T,\sn^3_{e_k,e_\ell,e_m} T) \mu\\
        &= -\int_{\Sigma_t} h^{ik}h^{\ell m} h(\sn_{e_i}T, \sn^3_{e_\ell,e_k,e_m}T ) \mu - \int_{\Sigma_t} h^{ik} h^{\ell m} h(\sn_{e_i}T,\bar R(e_k,e_\ell) \sn T(e_m))\mu\\
        &= \int_{\Sigma_t} h^{ik}h^{\ell m} h(\sn^2_{e_\ell,e_i} T, \sn^2_{e_k,e_m} T) \mu - \int_{\Sigma_t} h^{\ell m} h(\sn T, \bar R(\,\cdot\,,e_\ell)\sn T(e_m))\mu\\
        &= \int_{\Sigma_t} h(\sn^2 T, \sn^2 T) \mu + \int_{\Sigma_t} h^{ik} h^{\ell m} h(\sn^2_{e_\ell,e_i} T, \bar R(e_k,e_m)T)\mu\\
        &\quad - \int_{\Sigma_t} h^{\ell m} h(\sn T, \bar R(\,\cdot\,,e_\ell)\sn T(e_m))\mu\\
        &= \|\sn^2 T\|_{L^2}^2 - \int_{\Sigma_t} h(\sn^2 T, \bar R(\,\cdot\,,\,\cdot\,)T)\mu - \int_{\Sigma_t} h^{\ell m} h(\sn T, \bar R(\,\cdot\,,e_\ell) \sn T(e_m))\mu.
    \end{split}
    \]
    Hence
    \[
    \begin{split}
        \|\sn^2 T\|_{L^2}^2 &\leq \|\Delta_h T\|_{L^2}^2 + C\|\bar R\|_{L^\infty} \|\sn^2 T\|_{L^2} \|T\|_{L^2} + C\|\bar R\|_{L^\infty} \|\sn T\|_{L^2}^2\\
        &\leq \|\Delta_h T\|_{L^2}^2 + C_nt^{-2+\varepsilon} \|\sn^2 T\|_{L^2} \|T\|_{L^2} + C_nt^{-2+\varepsilon} \|\sn T\|_{L^2}^2,
    \end{split}
    \]
    where we have used Lemma~\ref{estimates for the curvature} and that $\bar R$ is completely determined by $\sric$. The result follows by Young's inequality. 
\end{proof}

\begin{lemma} \label{control of energy with modified energy}
    We have
    \[
    \se_s(t) \leq (C + C_n t^\varepsilon) \widetilde \se_s(t).
    \]
\end{lemma}

\begin{proof}
    First note that in terms of time derivatives, the energy contains $\|\sn^m \lie_{\p_t} \dk\|_{L^2}$, whereas the modified energy contains $\|\lie_{\p_t} \sn^m \dk\|_{L^2}$. We estimate, by Lemma~\ref{commutator estimates},
    \[
    \begin{split}
        t^{2(m+1)} \|\sn^m \lie_{\p_t} \dk\|_{L^2}^2 &\leq 2t^{2(m+1)} \|\lie_{\p_t} \sn^m \dk\|_{L^2}^2 + 2t^{2(m+1)} \|[\lie_{\p_t},\sn^m]\dk\|_{L^2}^2\\
        &\leq 2t^{2(m+1)} \mathbb E_m[\dk] + C_n\langle \ln t \rangle t^{2\varepsilon} \sum_{a=0}^{m-1} t^{2a}\|\sn^a \dk\|_{L^2}^2\\
        &\leq (C + C_n t^\varepsilon) \widetilde \se_s(t),
    \end{split}
    \]
    and similarly for $\dep$. Now we need to verify that $\|\sn^{s+1} \dt\|_{L^2}$ is indeed controlled by the corresponding modified top order quantity. Indeed, by Lemma~\ref{elliptic estimates},
    \[
    \begin{split}
        t^{2(s+1)}\|\sn^{s+1} \dt\|_{L^2}^2 &\leq 2t^{2(s+1)} \|\Delta_h \sn^{s-1} \dt\|_{L^2}^2 + C_nt^{2s+\varepsilon} \|\sn^s \dt\|_{L^2}^2\\
        &\quad + C_nt^{2(s-1) + 2\varepsilon} \|\sn^{s-1} \dt\|_{L^2}^2\\
        &\leq Ct^{2(s+1)}\Big( \widetilde{\mathbb E}_{s+1}[\dt] + \big\|\tr_s\big( (\dk + K_n) \circ \lie_{\p_t} \sn^{s-1} \dk \big)\big\|_{L^2}^2\\
        &\quad + \|\p_t(\dep + \s_n) \Delta_h \sn^{s-1} \dep\|_{L^2}^2  \Big) + C_nt^\varepsilon \widetilde \se_s(t)\\
        &\leq (C + C_nt^\varepsilon) \widetilde \se_s(t).
    \end{split}
    \]
    The estimates for $\deh$ and $\deh^{-1}$ are similar.
\end{proof}

\begin{proof}[Proof of Theorem \ref{bootstrap theorem}]
    Since $\widetilde \se_s(t_0) = 0$ and since, for $n$ large enough,
    \[
    \frac{d}{dt} \left( t^{-2N-2s} \widetilde \se_s(t) \right) \leq \frac{-2N-2s}{t} t^{-2N-2s} \widetilde \se_s(t) + t^{-2N-2s} \frac{1}{t}\big(C + C_n \langle \ln t \rangle t^\varepsilon\big) \se_s(t) + C_nt,
    \]
    it follows that
    \[
    t^{-2N-2s} \widetilde \se_s(t) \leq -\int_{t_0}^t \frac{2N+2s}{r} r^{-2N-2s} \widetilde \se_s(r)dr + \int_{t_0}^t \frac{C + C_n\langle \ln r \rangle r^\varepsilon}{r} r^{-2N-2s} \se_s(r)dr + C_nt^2.
    \]
    Now Lemma~\ref{control of energy with modified energy} yields
    \[
    \begin{split}
        \frac{t^{-2N-2s}}{C+C_nt^\varepsilon} \se_s(t) &\leq -\int_{t_0}^t \frac{2N+2s}{r(C + C_nr^\varepsilon)} r^{-2N-2s} \se_s(r)dr\\
        &\quad + \int_{t_0}^t \frac{C + C_n\langle \ln r \rangle r^\varepsilon}{r} r^{-2N-2s} \se_s(r)dr + C_nt^2.
    \end{split}
    \]
    At this point we fix the constants. Choose $N$ such that $N \geq 2C^2$. Next, choose $n_{N,s}$ such that the estimate above holds for all $n \geq n_{N,s}$. Now fix $n$ and choose $T_{N,s,n}$ such that $C_n \langle \ln t \rangle t^\varepsilon \leq C$ for all $t \in (0,T_{N,s,n}]$, then
    \[
    \frac{t^{-2N-2s}}{2C} \se_s(t) \leq \frac{-N}{C} \int_{t_0}^t r^{-2N-2s-1}\se_s(r)dr + 2C\int_{t_0}^t r^{-2N-2s-1} \se_s(t)dr + C_nt^2.
    \]
    Note that by our choice of $N$, the total contribution of the two integrals is non-positive. Therefore, by choosing $T_{N,s,n}$ smaller if necessary,
    \[
    \se_s(t) \leq 2C C_nt^{2N+2s+2} \leq t^{2N+2s}.
    \]
    In particular, this improves on the bootstrap assumptions~\eqref{second bootstrap assumption} and \eqref{third bootstrap assumption}.

    It remains to show that \eqref{first bootstrap assumption} is also improved. We have
    \[
    |(\bar\h - \bar\h_n)(e_i,e_k)| = t^{-p_i-p_k}|\deh(e_i,e_k)| \leq C|\deh|_h \leq Ct^{-5/2} \|\deh\|_{H^2} \leq Ct^{N+s-1-5/2}.
    \]
    Thus, by choosing $N$ larger if necessary, we conclude that
    \[
    |\bar\h(e_i,e_k) - \delta_{ik}| \leq |(\bar\h - \bar\h_n)(e_i,e_k)| + |(\bar\h_n - \ho)(e_i,e_k)| \leq C_n\langle \ln t \rangle^2 t^{2\varepsilon + |p_i-p_k|},
    \]
    which improves over \eqref{first bootstrap assumption} after taking $T_{N,s,n}$ smaller if necessary.
\end{proof}

\begin{proof}[Proof of Corollary~\ref{global existence corollary}]
Define the set $\mathcal{A}$ as the set of $t \in [t_0,T_{N,s,n}]$ such that the solution given by Lemma~\ref{local existence} extends to $[t_0,t]$ and satisfies the bootstrap assumptions \eqref{bootstrap assumptions} there. By choice of initial data, if $t_0$ is small enough, then $\mathcal{A}$ is non-empty. Moreover, $\mathcal{A}$ is connected by definition. We first show that it is open. If $t \in \mathcal{A}$, we can solve \eqref{the system} by setting as initial data at $t$ the one induced by the solution; see \cite[Corollary 4]{ringstrom_local_2024}. We thus obtain that there is a $\delta > 0$ such that the solution extends to $[t_0,t+\delta)$. Moreover, by Theorem~\ref{bootstrap theorem}, since \eqref{main energy estimate} improves on the bootstrap assumptions, by taking $\delta$ smaller if necessary, we can ensure that the bootstrap assumptions hold on $[t_0,t+\delta)$. Hence $\mathcal{A}$ is open. It remains to show that it is closed. Assume that $t$ is in the closure of $\mathcal{A}$, then the bootstrap assumptions hold in all of $[t_0,t)$. Thus there is a uniform bound on the corresponding Sobolev norms of the solution on $[t_0,t)$. Moreover, since $s \geq 5$, by Sobolev embedding there is a uniform bound on the $C^3$ norm of the solution on $[t_0,t)$. Therefore, by \cite[Corollary 4]{ringstrom_local_2024}, the solution can be extended to $[t_0,t]$. Note that, by continuity, the bootstrap assumptions still hold on $[t_0,t]$. We conclude that $\mathcal{A}$ is closed, which means $\mathcal{A} = [t_0,T_{N,s,n}]$. Finally, Theorem~\ref{bootstrap theorem} implies that \eqref{main energy estimate} holds in all of $[t_0,T_{N,s,n}]$.
\end{proof}

\subsection{Construction of the solution} \label{finishing the construction}

In this subsection we construct the solution of Theorem~\ref{global existence}. We begin by obtaining some uniform estimates which are required for the proof, after which the construction of the solution follows. Then we finish this subsection by proving that, for the constructed solution, $K$ is in fact the Weingarten map of the $\Sigma_t$ hypersurfaces. 

\begin{lemma} \label{uniform bounds for the error terms}
    Let $T_{N,s,n}$ be as in Theorem~\ref{bootstrap theorem} and $0 < t_1 < t_2 \leq T_{N,s,n}$. Let $t_0 \in (0,t_1]$ and $(h,K,\theta,\s)$ the solution to \eqref{the system} on $[t_0,T_{N,s,n}]$ given by Corollary~\ref{global existence corollary}. Then for $N$ large enough,
    \[
    \sum_{m+r \leq 4} \big(|D^m \lie_{\p_t}^r \deh|_\ho + |D^m \lie_{\p_t}^r \deh^{-1}|_\ho + |D^m \p_t^r \dep|_\ho\big) + \sum_{m+r\leq3} \big(|D^m \lie_{\p_t}^r \dk|_\ho + |D^m \p_t^r \dt|_\ho\big) \leq 1
    \]
    for all $t \in [t_1,t_2]$.
\end{lemma}

\begin{proof}
    When there are no time derivatives, or only one time derivative for $\dk$ and $\dep$, it follows from Lemma~\ref{estimates in tensor components} and Sobolev embedding, after taking $N$ large enough and $T_{N,s,n}$ smaller if necessary. The estimates for the remaining time derivatives can then be deduced directly from the equations~\eqref{wave equations for delta objects} and \eqref{transport equations for delta objects}.
\end{proof}

\begin{proposition} \label{global existence proposition}
    For every sufficiently large positive integer $N$, there is an $n_{N,s}$ such that for every $n \geq n_{N,s}$ there is a $T_{N,s,n} > 0$ such that the following holds. There is a $C^3 \times C^2 \times C^2 \times C^3$ solution $(h,K,\theta,\varphi)$ to \eqref{the system} on $(0,T_{N,s,n}] \times \Sigma$ satisfying the estimate
    \begin{equation} \label{main main energy estimate 2}
        \se_s(t) \leq t^{2N + 2s}
    \end{equation}
    for $t \in (0,T_{N,s,n}]$.
\end{proposition}

\begin{proof}
    Let $N$, $s$ and $n$ be such that Lemma~\ref{uniform bounds for the error terms} holds. Take a decreasing sequence $t_i \to 0$ such that ${t_i \in (0,T_{N,s,n})}$ and corresponding solutions $(\check h_i,\check K_i,\check \theta_i,\check \s_i)$ with initial data at $t_i$ given by $(g_n,\s_n)$ as in Lemma~\ref{local existence}. We want to apply the Arzelà-Ascoli theorem to this sequence. To that end, we define the Riemannian metric $\hat g$ on $(0,T_{N,s,n}] \times \Sigma$ by
    \[
    \hat g := dt \otimes dt + \ho.
    \]
    Denote by $d$ the Riemanninan distance associated with $\hat g$. 

    Let $0 < t_1 < t_2 \leq T_{N,s,n}$. By Lemma~\ref{uniform bounds for the error terms}, we know that the sequence $(\check h_i,\check K_i,\check \theta_i,\check \s_i)$, after some $i$, is uniformly bounded in $[t_1,t_2] \times \Sigma$ in $C^3 \times C^2 \times C^2 \times C^3$. We need to show that the derivatives up to these orders are uniformly equicontinuous. Consider $u \in C^1([t_1,t_2] \times \Sigma)$. Let $p, q \in [t_1,t_2] \times \Sigma$ with $p = (t,x)$ and $q = (s,y)$. Let $\bar \gamma : [0,1] \to \Sigma$ be a minimizing geodesic from $x$ to $y$ with respect to $\ho$, and define $\gamma: [0,1] \to [t_1,t_2] \times \Sigma$ by $\gamma(r) := ( rs + (1-r)t, \bar \gamma(r) )$. Then $\gamma$ is a minimizing geodesic with respect to $\hat g$, joining $p$ to $q$. If we assume that $|du|_{\hat g} \leq C$, then
    \[
    |u(p) - u(q)| \leq \int_0^1 | du_{\gamma(r)}(\gamma'(r)) |dr \leq C \int_0^1 |\gamma'(r)|_{\hat g}dr = C d(p,q).
    \]
    This means that the control of one additional derivative provided by Lemma~\ref{uniform bounds for the error terms} ensures that the sequence is uniformly equicontinuous in $C^3 \times C^2 \times C^2 \times C^3$.

    Define $(\widetilde h_I, \widetilde K_I, \widetilde \theta_I, \widetilde \s_I)$ in $[t_I,T_{N,s,n}] \times \Sigma$, for $I \geq 1$, inductively as follows. First, define $(\widetilde h_1, \widetilde K_1, \widetilde \theta_1, \widetilde \s_1)$ in $[t_1,T_{N,s,n}] \times \Sigma$ as the $C^3 \times C^2 \times C^2 \times C^3$ limit of a subsequence of $(\check h_i,\check K_i,\check \theta_i,\check \s_i)$. Next, assume we are given $(\widetilde h_I, \widetilde K_I, \widetilde \theta_I, \widetilde \s_I)$ in $[t_I,T_{N,s,n}] \times \Sigma$ as the $C^3 \times C^2 \times C^2 \times C^3$ limit of a subsequence $(\check h_{i_k},\check K_{i_k},\check \theta_{i_k},\check \s_{i_k})$ for $i_k \geq I$. Then we can define $(\widetilde h_{I+1}, \widetilde K_{I+1}, \widetilde \theta_{I+1}, \widetilde \s_{I+1})$ in $[t_{I+1},T_{N,s,n}] \times \Sigma$ as the $C^3 \times C^2 \times C^2 \times C^3$ limit of a subsequence of $(\check h_{i_k},\check K_{i_k},\check \theta_{i_k},\check \s_{i_k})$ for $i_k \geq I+1$. Note that this construction ensures that $(\widetilde h_I, \widetilde K_I, \widetilde \theta_I, \widetilde \s_I)$ and $(\widetilde h_J, \widetilde K_J, \widetilde \theta_J, \widetilde \s_J)$ agree on $[t_{\min\{I,J\}},T_{N,s,n}] \times \Sigma$. Moreover, $(\widetilde h_I, \widetilde K_I, \widetilde \theta_I, \widetilde \s_I)$ solves \eqref{the system} in $[t_I,T_{N,s,n}] \times \Sigma$ for all $I$. Now we can define $(h,K,\theta,\s)$ as follows. Let $(t,x) \in (0,T_{N,s,n}] \times \Sigma$ and $I$ such that $t \in [t_I,T_{N,s,n}]$, then define $h(t,x) := \widetilde h_I(t,x)$. Note that $h$ is well defined. We define $K$, $\theta$ and $\s$ similarly.

    It remains to show that \eqref{main main energy estimate 2} holds for $(h,K,\theta,\s)$. Note that for each $t$, there is a subsequence $i_k$ such that $(\check h_{i_k}(t),\check h^{-1}_{i_k}(t),\check K_{i_k}(t),\lie_{\p_t} \check K_{i_k}(t),\check \theta_{i_k}(t),\check \s_{i_k}(t),\p_t \check \s_{i_k}(t))$ has a weak limit satisfying \eqref{main main energy estimate 2}. But by Sobolev embedding, the weak limit has to coincide with $(h,K,\theta,\s)$. 
\end{proof}

\begin{proposition} \label{k is the weingarten map}
    Let $(h,K,\theta,\s)$ be a solution to \eqref{the system} as in Proposition~\ref{global existence proposition}. Then $K$ is the Weingarten map of the $\Sigma_t$ hypersurfaces with respect to $g := -dt \otimes dt + h$.
\end{proposition}

\begin{proof}
    Our goal is to show that the antisymmetric part of $K$ with respect to $h$ vanishes. For that purpose, we deduce an equation for it. Define
    \[
    A(X,Y) := h(K(X),Y) - h(X,K(Y))
    \]
    for $X,Y \in \mfx(\Sigma)$. Then
    \begin{equation} \label{time derivative of A}
        \lie_{\p_t} A(X,Y) = A(K(X),Y) + A(X,K(Y)) + h(\lie_{\p_t}K(X),Y) - h(X,\lie_{\p_t}K(Y)),
    \end{equation}
    implying
    \[
    \begin{split}
        \lie_{\p_t}^2 A(X,Y) &= \lie_{\p_t}A(K(X),Y) + \lie_{\p_t} A(X,K(Y)) + A(\lie_{\p_t}K(X),Y) + A(X,\lie_{\p_t}K(Y))\\
        &\quad + h(K \circ \lie_{\p_t}K(X),Y) + h(\lie_{\p_t}K(X),K(Y)) - h(K(X),\lie_{\p_t}K(Y))\\
        &\quad  - h(X, K \circ \lie_{\p_t}K(Y)) + h(\lie_{\p_t}^2K(X),Y) - h(X,\lie_{\p_t}^2K(Y)).
    \end{split}
    \]
    Denote by $\mathfrak G$ terms that consist of contractions of tensor products of $A$, $\lie_{\p_t}K$, $h$ and $h^{-1}$; by $\mathfrak H$ terms that consist of contractions of tensor products of $\lie_{\p_t}A$, $K$, $h$ and $h^{-1}$; by $\mathfrak I$ terms that consist of contractions of tensor products of $A$, $K$, $h$ and $h^{-1}$; and finally, by $\Phi$ terms consisting of contractions of tensor products of $A$, $d\s$ and $h^{-1}$. We will use the evolution equation \eqref{main equation for k} for $K$ to substitute $\lie_{\p_t}^2K$ and use the types of terms just defined to describe the result. We have
    \[
    \begin{split}
        &h(K \circ \lie_{\p_t}K(X),Y) + h(\lie_{\p_t}K(X),K(Y)) - h(K(X),\lie_{\p_t}K(Y))\\
        &- h(X,K \circ \lie_{\p_t}K(Y)) - h(F_2(K)(X),Y) + h(X,F_2(K)(Y))\\
        &\hspace{2cm}= 3h(\lie_{\p_t}K \circ K(X),Y) + 3h(K \circ \lie_{\p_t}K(X),Y) - A(\lie_{\p_t}K(X),Y)\\
        &\hspace{2cm}\quad -3h(X,\lie_{\p_t}K \circ K(Y)) - 3h(X,K \circ \lie_{\p_t}K(Y)) - A(X,\lie_{\p_t}K(Y)) + \cdots,
    \end{split}
    \]
    where $\cdots$ are the terms which correspond to the $\tr K$ terms in $F_2(K)$. Moreover,
    \[
    \begin{split}
        &h((\lie_{\p_t}K \circ K + K \circ \lie_{\p_t}K)(X),Y) - h(X,(\lie_{\p_t}K \circ K + K \circ \lie_{\p_t}K)(Y))\\
        &= \p_t\big( h(K^2(X),Y) - h(X,K^2(Y)) \big) - h(K^3(X),Y) + h(X,K^3(Y))\\
        &\quad - h(K^2(X),K(Y)) + h(K(X),K^2(Y))\\
        &= \p_t\big( A(K(X),Y) + A(X,K(Y)) \big) - A(K^2(X),Y) - A(X,K^2(Y)) - 2A(K(X),K(Y))\\
        &= \mathfrak G + \mathfrak H + \mathfrak I.
    \end{split}
    \]
    Looking at the rest of the terms in $F_2(K)$ and $F_1(K)$, we see that we can put them directly in the required form, or we just need to use \eqref{time derivative of A}.
    
    Turning our attention to the scalar field terms, we just need to check the following,
    \[
    \begin{split}
        &h(d\s(K(X)) \sn \s,Y) - h(X,d\s(K(Y))\sn \s) + h(d\s(X)K(\sn \s),Y) - h(X,d\s(Y)K(\sn \s))\\
        &= d\s(Y)h(K(X),\sn \s) - d\s(X)h(\sn \s,K(Y)) + d\s(X)h(K(\sn \s),Y) - d\s(Y)h(X,K(\sn \s))\\
        &= d\s(Y)A(X,\sn \s) + d\s(X)A(\sn \s,Y).
    \end{split}
    \]
    Finally, since
    \[
    \Delta_h A(X,Y) = h(\Delta_h K(X),Y) - h(X,\Delta_h K(Y)),
    \]
    we conclude that
    \[
    -\lie_{\p_t}^2 A + \Delta_h A = \mathfrak G + \mathfrak H + \mathfrak I + \Phi + (V \circ \s) A.
    \]

    Now consider $K_i$ as in the proof of Proposition~\ref{global existence proposition} and its antisymmetric part $A_i$. Then, by choice of initial data, $A_i(t_i) = \lie_{\p_t} A_i(t_i) = 0$. But, as we saw above, $A_i$ solves a \emph{homogeneous} wave equation. Thus $A_i \equiv 0$. Since this is true for all $i$, then $A \equiv 0$. 
\end{proof}

\subsection{The constructed solution solves Einstein's equations} \label{the constructed solution solves einsteins eqs}

In this subsection we show that the solution to \eqref{the system} given by Proposition~\ref{global existence proposition}, is in fact a solution to the Einstein--nonlinear scalar field equations with potential $V$, after taking the parameter $N$ large enough. This finishes the proof of Theorem~\ref{global existence}.

\begin{proposition} \label{theta is the mean curvature}
    For $N$ sufficiently large, the solution $(h,K,\theta,\s)$ to \eqref{the system} given by Proposition~\ref{global existence proposition}, satisfies $\theta = \tr K$. In particular $E(\p_t,\p_t) = 0$ and $\Box_g \s = V' \circ \s$.
\end{proposition}

\begin{proof}
    Note that
    \[
    \tr F_1(K) = 2(\tr K) \tr K^2, \quad \tr F_2(K) = 2\tr(\lie_{\p_t} K \circ K) + 2(\tr K)\p_t \tr K
    \]
    and $\tr H(\s) = 0$. Hence, after taking the trace of Equation~\eqref{main equation for k},
    \[
    \begin{split}
        -\p_t^2 \tr K + \Delta_h \tr K &= \Delta_h \theta + 2(\tr K)\tr K^2 + \p_t(\tr K^2) + 2(\tr K) \p_t \tr K\\
        &\quad - 2(V \circ \s)\tr K + 2(\p_t \s) \Delta_h \s - 3(V' \circ \s)\p_t \s. 
    \end{split}
    \]
    Since $\theta$ satisfies \eqref{main equation for theta} and $\s$ satisfies \eqref{main equation for phi}, we see that
    \[
    -\p_t^2(\tr K - \theta) + \Delta_h(\tr K - \theta) = 2(\tr K) \p_t(\tr K - \theta) - 2(\p_t \s)^2(\tr K - \theta).
    \]
    Now we perform energy estimates with this wave equation,
    \[
    \begin{split}
        \frac{d}{dt}\mathbb E[\tr K - \theta] &\leq \frac{1}{t}(C + C_nt^\varepsilon) \mathbb E[\tr K - \theta] + 2t\|(\p_t \s)^2(\tr K - \theta) - (\tr K)\p_t(\tr K - \theta)\|_{L^2}^2\\
        &\leq \frac{1}{t}(C + C_nt^\varepsilon) \mathbb E[\tr K - \theta].
    \end{split}
    \]
    Moreover, since $\tr K_n = \theta_n$,
    \[
    \|\tr K - \theta\|_{H^1}^2 + \|\p_t(\tr K - \theta)\|_{L^2}^2 \leq Ct^{2N + 2s -2},
    \]
    where we have estimated $\p_t(\theta - \theta_n)$ directly from \eqref{delta theta equation}. Therefore, by Grönwall's inequality, we obtain that there are constants $C$ and $C'$ such that
    \[
    \mathbb E[\tr K - \theta](t) \leq C\left( \frac{t}{t_0} \right)^{C'} t_0^{2N + 2s - 4}
    \]
    for $t 
    \geq t_0$. Thus, by taking $N$ large enough and letting $t_0 \to 0$, we conclude that $\theta = \tr K$. 
\end{proof}

\begin{proposition} \label{vanishing of the einstein scalar field tensor}
    Consider the solution $(h,K,\theta,\s)$ to \eqref{the system} given by Proposition~\ref{global existence proposition} and let $g := -dt \otimes dt + h$. Then, for $N$ large enough, $(g,\s)$ is a solution to the Einstein--nonlinear scalar field equations with potential $V$.
\end{proposition}

\begin{proof}
    By Proposition~\ref{theta is the mean curvature}, if $N$ is large enough, we already have $\Box_g \s = V' \circ \s$ and ${E(\p_t,\p_t) = 0}$. Hence we only need to ensure that $\ce$ and $\cm$ also vanish. We will derive a system of equations for $\ce$ and $\cm$. Since $K$ satisfies \eqref{main equation for k}, then $\ce$ satisfies
    \[
    \lie_{\p_t} \ce = 3 \ce \circ K + K \circ \ce - 2\tr( \ce \circ K )I - 2(\tr K) \ce + (\tr \ce)\big( K - (\tr K) I  \big) + \big( \lie_{\cm^{\sharp}} h \big)^{\sharp}.
    \]
    Moreover, since $E(\p_t,\p_t) = 0$ and $\Box_g \s = V' \circ \s$, by Lemma~\ref{evolution equations for the constraint equations} $\cm$ satisfies
    \[
    \lie_{\p_t} \cm = -\theta \cm + \diver_h \ce - \frac{1}{2} d(\tr \ce).
    \]
    The idea is to take $\lie_{\p_t}$ of this equation. For tensors $T_1$ and $T_2$, denote by $F(T_1,T_2)$ terms which consist of contractions of tensor products of $T_1$ and $T_2$. We have
    \[
    \lie_{\p_t} \Big( \diver_h \ce - \frac{1}{2} d(\tr \ce) \Big) = \diver_h \lie_{\p_t} \ce - \frac{1}{2} d(\tr \lie_{\p_t} \ce) + ([\lie_{\p_t},\sn] \ce)(e_k,\,\cdot\,,\omega^k).
    \]
    Furthermore, $\tr \lie_{\p_t}\ce = 2\diver_h \cm + F(K,\ce)$, implying
    \[
    \diver_h \lie_{\p_t} \ce - \frac{1}{2} d(\tr \lie_{\p_t}\ce) = \Delta_h \cm + \sric(\,\cdot\,,\cm^{\sharp}) + F(\sn K,\ce) + F(K,\sn \ce).
    \]
    We see that $\cm$ and $\ce$ satisfy a system of the form
    \begin{subequations} \label{equations for einstein scalar field tensor}
        \begin{align}
            \lie_{\p_t} \ce &= \big(\lie_{\cm^\sharp}h\big)^\sharp + F(K,\ce)\label{equation for spatial components of ce}\\
            -\lie_{\p_t}^2 \cm + \Delta_h \cm &= \theta \lie_{\p_t}\cm + (\p_t \theta) \cm - \sric(\,\cdot\,,\cm^\sharp) + F(\sn K,\ce) + F(K,\sn \ce).\label{wave eq for the momentum constraint}
        \end{align}
    \end{subequations}
    Note that \eqref{equations for einstein scalar field tensor} presents a potential loss of derivatives, due to \eqref{equation for spatial components of ce} having terms with one derivative of $\cm$ on the right-hand side. However, this can be dealt with by introducing a modified energy, similarly as for \eqref{the system}. Define the energy
    \[
    \mfe(t) := \sum_{m=0}^1 t^{2(m+1)} \|\sn^m \lie_{\p_t}\cm\|_{L^2}^2 + \sum_{m=0}^2 t^{2m} \|\sn^m \cm\|_{L^2}^2 + \sum_{m=0}^2 t^{2m}\|\sn^m \ce\|_{L^2}^2.
    \]
    In order to define the modified energy, we make the following observation. Note that, since $(\lie_{\cm^\sharp} h)^\sharp = \sn \cm^\sharp + h^{\ell m} \sn_{e_\ell} \cm(\,\cdot\,)e_m$, then
    \[
    \begin{split}
        \lie_{\p_t} \Delta_h \ce &= \Delta_h\big( \sn \cm^\sharp + h^{\ell m} \sn_{e_\ell} \cm(\,\cdot\,)e_m \big) + \cdots\\
        &= \sn \Delta_h\cm^\sharp + h^{\ell m} \sn_{e_\ell} \Delta_h\cm(\,\cdot\,)e_m + \cdots\\
        &= \lie_{\p_t} \big( \sn (\lie_{\p_t}\cm)^\sharp + h^{\ell m} \sn_{e_\ell} (\lie_{\p_t}\cm)(\,\cdot\,)e_m \big)\\
        &\quad + \sn (-\lie_{\p_t}^2\cm + \Delta_h\cm)^\sharp + h^{\ell m} \sn_{e_\ell} (-\lie_{\p_t}^2\cm + \Delta_h\cm)(\,\cdot\,)e_m + \cdots,
    \end{split}
    \]
    where $\cdots$ denotes terms that can be estimated in terms of the energy $\mfe$. Hence, by \eqref{wave eq for the momentum constraint}, we see that
    \[
    \lie_{\p_t} \big( \Delta_h \ce - \sn (\lie_{\p_t}\cm)^\sharp - h^{\ell m} \sn_{e_\ell} (\lie_{\p_t}\cm)(\,\cdot\,)e_m ) \big)
    \]
    is an object that can be estimated in terms of the energy $\mfe$. This motivates us to define the modified energy by
    \[
    \begin{split}
        \widetilde \mfe(t) &:= \sum_{m=0}^1 t^{2(m+1)}\mathbb E_m[\cm] + \sum_{m=0}^1 t^{2m}\|\sn^m \ce\|_{L^2}^2\\
        &\quad + t^4\|\Delta_h \ce - \sn (\lie_{\p_t}\cm)^\sharp - h^{\ell m} \sn_{e_\ell} (\lie_{\p_t}\cm)(\,\cdot\,)e_m\|_{L^2}^2.
    \end{split}
    \]
    Moving on to the energy estimates, by Lemmas~\ref{elliptic estimates} and \ref{commutator estimates}, we see that
    \[
    \mfe(t) \leq \big(C + C_n\langle \ln t \rangle t^\varepsilon\big)\widetilde \mfe(t).
    \]
    Moreover, by using the information in Table~\ref{table} and Lemma~\ref{estimates for the curvature} to control the coefficients of \eqref{equations for einstein scalar field tensor}, and by Propositions~\ref{transport estimate} and \ref{wave estimate},
    \[
    \frac{d}{dt} \widetilde \mfe(t) \leq \frac{1}{t}\big(C + C_n \langle \ln t \rangle t^\varepsilon\big)\mfe(t).
    \]
    We omit the details since everything works similarly as in the energy estimates for \eqref{the system}, but the current situation is much simpler since \eqref{equations for einstein scalar field tensor} is linear. As in the proof of Theorem~\ref{bootstrap theorem}, we can choose $T_{N,s,n}$ smaller if necessary to ensure that $C_n \langle \ln t \rangle t^\varepsilon \leq C$ for all $t \in (0,T_{N,s,n}]$, hence
    \[
    \begin{split}
        \frac{d}{dt}\Big(t^{-4C^2}\widetilde \mfe(t)\Big) &\leq \frac{-4C^2}{t} t^{-4C^2} \frac{\mfe(t)}{C + C_n\langle \ln t \rangle t^\varepsilon} + t^{-4C^2} \frac{1}{t}\big(C + C_n\langle \ln t \rangle t^\varepsilon\big)\mfe(t) \leq 0.
    \end{split}
    \]
    This implies
    \[
    \frac{t^{-4C^2}}{C + C_n\langle \ln t \rangle t^\varepsilon} \mfe(t) \leq t^{-4C^2} \widetilde \mfe(t) \leq t_0^{-4C^2} \widetilde \mfe(t_0)
    \]
    for $t \geq t_0$. So we are finished if we can ensure that $\widetilde \mfe(t)$ can be made to decay as $t \to 0$ as an sufficiently large power of $t$. By \eqref{main main energy estimate 2}, it follows that the relevant norms of $\ce - \ce_n$ and $\cm - \cm_n$ can be made to decay as fast as we want by choosing $N$ large enough. Moreover, by Theorem~\ref{approximate solutions}, the same is true for $\ce$ and $\cm$ by choosing $n$ large enough.
\end{proof}

\subsection{The constructed solution is smooth} \label{the constructed solution is smooth}

In this subsection we prove that the solution constructed in Theorem~\ref{global existence} is actually smooth. The proof makes use of the following uniqueness statement for solutions to \eqref{the system}.

\begin{proposition} \label{first uniqueness result}
    Let $(\Sigma,\ho,\Ko,\phio,\psio)$ be initial data on the singularity and let $V$ be an admissible potential. There is a sufficiently large positive integer $M$, depending only on the initial data and the potential, such that the following holds. Let $(g,\s)$ and $(\widetilde g, \widetilde \s)$, where $g = -dt \otimes dt + h$ and $\widetilde g = -dt \otimes dt + \widetilde h$, be $C^3 \times C^3$ solutions to the Einstein--nonlinear scalar field equations with potential $V$ on $(0,T] \times \Sigma$. Suppose that there are constants $C$ and $\delta>0$ such that
    \[
    \begin{split}
        \sum_{m=0}^2 \Big( |D^m( \bar\h - \ho )|_\ho + |D^m( \overline{\widetilde{\h}} - \ho )|_\ho \Big) &\leq Ct^\delta,\\
        \sum_{\substack{m+r = 2\\ r \leq 1}} t^r\Big( |D^m \lie_{\p_t}^r( tK - \Ko )|_\ho + |D^m \lie_{\p_t}^r( t\widetilde{K} - \Ko )|_\ho \Big) &\leq Ct^\delta,\\
        \sum_{m=0}^2 \Big(|D^m( \bar\Phi - \phio )|_\ho + |D^m( \overline{\widetilde{\Phi}} - \phio )|_\ho \Big) &\leq Ct^\delta,\\
        \sum_{m=0}^1 \Big(|D^m( \bar\Psi - \psio )|_\ho + |D^m( \overline{\widetilde{\Psi}} - \psio )|_\ho \Big) &\leq Ct^\delta,
    \end{split}
    \]
    for some constants $C$ and $\delta > 0$. Moreover, assume that for $i \neq k$,
    \[
    \begin{split}
    \sum_{m=0}^2 \left(\big|D^m\big( \bar\h(e_i,e_k) \big)\big|_\ho(x) + \big|D^m\big( \overline{\widetilde{\h}}(e_i,e_k) \big)\big|_\ho(x) \right) &\leq Ct^{\delta + (p_i+p_k-2p_1)(x)},\\
    \sum_{m=0}^2 \left(\big|D^m\big( \bar\h(e_i,e_k) \big)\big|_\ho(y) + \big|D^m\big( \overline{\widetilde{\h}}(e_i,e_k) \big)\big|_\ho(y) \right) &\leq Ct^{\delta + |p_i - p_k|(y)},\\
    \sum_{\substack{m+r = 2\\ r \leq 1}} t^r\left(\big|D^m \p_t^r\big( tK(e_i,\omega^k) \big)\big|_\ho + \big|D^m \p_t^r\big( t\widetilde K(e_i,\omega^k) \big)\big|_\ho \right) &\leq Ct^\delta \min\{ 1, t^{2(p_i - p_k)} \},
    \end{split}
    \]
    for $x \in D_+$ and $y \in D_-$. If there is a constant $C$ such that
    \[
    \begin{split}
      \sum_{m=0}^3 |D^m(h - \widetilde h)|_\ho + \sum_{m=0}^2 |D^m \lie_{\p_t}(h - \widetilde h)|_\ho &\leq Ct^M,\\
      \sum_{m=0}^2 |D^m(\s - \widetilde \s)|_\ho + \sum_{m=0}^1 |D^m \p_t(\s - \widetilde \s)|_\ho &\leq Ct^M,
    \end{split}
    \]
    then $(g,\s) = (\widetilde g, \widetilde \s)$.
\end{proposition}

\begin{proof}
    Denote by $C'$ constants depending only on the initial data and the potential, and by $C''$ constants depending also on the constant $C$ in the statement of the proposition. For an interval $[t_1,t_2] \subset (0,T]$, uniqueness follows from \cite[Lemma~24]{ringstrom_local_2024}. So we only need to prove uniqueness on an interval $(0,T']$ for $T' \leq T$. That being the case, there is no loss of generality in taking $T$ smaller if necessary to ensure that $C''\langle \ln t \rangle t^\delta \leq 1$ for all $t \in (0,T]$. Note that our assumptions imply that Lemma~\ref{estimates in tensor components} holds with $h_n$ replaced by $h$ and $m \leq 2$. Hence, similarly as in Lemmas~\ref{kn estimates} and \ref{phin estimates}, we have
    \begin{align*}
        \sum_{m=0}^2 t^m\big(\|\sn^m K\|_{L^\infty} + \|\sn^m \widetilde K\|_{L^\infty}\big) &\leq \frac{1}{t}\big(C' + C''\langle \ln t \rangle t^\delta\big),\\
        \sum_{m=0}^1 t^{m+1}\big( \|\sn^m \lie_{\p_t} K\|_{L^\infty} + \|\sn^m \lie_{\p_t} \widetilde K\|_{L^\infty} \big) &\leq \frac{1}{t}\big(C' + C''\langle \ln t \rangle t^\delta\big)\\
        \sum_{m=1}^2t^m\big(\|\sn^m \s\|_{L^\infty} + \|\sn^m \widetilde \s\|_{L^\infty}\big) &\leq C''\langle \ln t \rangle t^\delta,\\
        \sum_{m=0}^1 t^m\big(\|\sn^m \p_t \s\|_{L^\infty} + \|\sn^m \p_t \widetilde\s\|_{L^\infty}\big) &\leq \frac{1}{t}\big(C' + C''t^\delta\big),
    \end{align*}
    where the $L^\infty$ norms and $\sn$ are taken with respect to $h$. Define the variables
    \[
    \begin{split}
        \deh := h - \widetilde h, \quad \deh^{-1} := h^{-1} - \widetilde h^{-1}, \quad \dk := K - \widetilde K, \quad \dt := \theta - \widetilde \theta, \quad \dep := \s - \widetilde \s.
    \end{split}
    \]
    Given $M'$, there is $M$ large enough such that 
    \[
    \|\deh\|_{H^2} + \|\deh^{-1}\|_{H^2} + \|\dk\|_{H^1} + \|\dt\|_{H^2} + \|\dep\|_{H^2} + \|\p_t \dep\|_{H^1} \leq C''t^{M'}.
    \]
    Given $M''$, we can choose $M$ larger if necessary such that
    \[
    \|\ric_h - \ric_{\widetilde h}\|_{L^2} \leq C''t^{M''}.
    \]
    Then, from Equation~\eqref{evolution equation for k} for $K$ and $\widetilde K$,
    \[
    \|\lie_{\p_t} \dk\|_{L^2} \leq C'\max\{t^{M''},t^{M'-2}\}.
    \]
    Note that the $\delta$ variables solve the equations \eqref{wave equations for delta objects} and \eqref{transport equations for delta objects} with $h_n$ replaced by $\widetilde h$, etc. Moreover, the inhomogeneous terms vanish since $(\widetilde g, \widetilde \s)$ solves Einstein's equations. Hence, if we define the energy
    \[
    \begin{split}
        \se(t) &:= t^2 \| \lie_{\p_t} \dk \|_{L^2}^2 + \sum_{m = 0}^1 t^{2m} \| \sn^m \dk \|_{L^2}^2 + \sum_{m = 0}^2 t^{2m} \| \sn^m \dt \|_{L^2}^2\\
        &\quad + \sum_{m = 0}^1 t^{2m} \| \sn^m \p_t \dep \|_{L^2}^2 + \sum_{m = 0}^2 t^{2(m-1)} \| \sn^m \dep \|_{L^2}^2\\
        &\quad + \sum_{m = 0}^2 t^{2(m-1)} \big( \| \sn^m \deh \|_{L^2}^2 + \| \sn^m \deh^{-1} \|_{L^2}^2 \big),
    \end{split}
    \]
    we can run the energy estimates exactly as in the proof of Theorem~\ref{bootstrap theorem}, including the introduction of a modified energy, to obtain
    \[
    t^{-N} \se(t) \leq C't_0^{-N} \se(t_0)
    \]
    for $t \geq t_0$ and $N$ large enough, depending only on the initial data and the potential. By choosing $M$ larger if necessary, we can ensure that $t_0^{-N} \se(t_0) \to 0$ as $t_0 \to 0$, implying $\se(t) \equiv 0$. We conclude that $g = \widetilde g$ and $\s = \widetilde \s$ in $(0,T]$.
\end{proof}

\begin{lemma} \label{asymptotics for the actual solution}
    Let $(g,\s)$, with $g = -dt \otimes dt + h$ be the solution to the Einstein--nonlinear scalar field equations with potential $V$ given by Theorem~\ref{global existence}. If $N$ is large enough, then the following estimates hold,
    \begin{align*}
        \big|D^m\big((\bar\h - \ho)(e_i,e_k)\big)\big|_\ho(x) &\leq C\langle \ln t \rangle^{m+2} t^{2\varepsilon + (p_i + p_k - 2p_1)(x)}, &m \leq s-1;\\
        \big|D^m\big((\bar\h - \ho)(e_i,e_k)\big)\big|_\ho(y) &\leq C\langle \ln t \rangle^{m+2} t^{2\varepsilon + |p_i - p_k|(y)}, &m \leq s-1;\\
        \big|D^m \big( (tK - \Ko)(e_i,\omega^k)\big)\big|_\ho &\leq C\langle \ln t \rangle^{m+2} t^{2\varepsilon} \min\{1,t^{2(p_i - p_k)}\},  &m \leq s-2;\\
        t\big|D^m \p_t\big( (tK - \Ko)(e_i,\omega^k)\big)\big|_\ho &\leq C\langle \ln t \rangle^{m+2} t^{2\varepsilon} \min\{1,t^{2(p_i - p_k)}\}, &m \leq s-3;\\
        |D^m(\bar\Psi - \psio)|_\ho &\leq C\langle \ln t \rangle^{m+2}t^{2\varepsilon}, &m \leq s-2;\\
        |D^m(\bar\Phi - \phio)|_\ho &\leq C\langle \ln t \rangle^{m+3}t^{2\varepsilon}, &m \leq s-1;
    \end{align*}
    for $t \in (0,T_{N,s,n}]$, $x \in D_+$ and $y \in D_-$, where the constant $C$ depends only on the initial data, the potential and $s$.
\end{lemma}

\begin{proof}
    We only prove the estimate for $\bar\h$, since the rest are similar. Let $|\alpha| = m \leq s-1$, then
    \[
    \begin{split}
        e_\alpha (\bar\h - \bar\h_n)(e_i,e_k) &= e_\alpha( t^{-p_i-p_k}(h-h_n)_{ik} )\\
        &= \sum (-\ln t)^r e_{\beta_1}(p_i+p_k) \cdots e_{\beta_r}(p_i+p_k) t^{-p_i-p_k} e_\gamma(h-h_n)_{ik},
    \end{split}
    \]
    where the sum is over appropriate multiindices such that $|\beta_1| + \cdots + |\beta_r| + |\gamma| = m$. We can use Lemmas~\ref{estimates in tensor components}, \ref{comparison of h and hn norms} and Sobolev embedding to estimate the right-hand side, which yields
    \[
    |e_\alpha(\bar\h - \bar\h_n)(e_i,e_k)| \leq C\langle \ln t \rangle^{2m} t^{m(p_1-1+\varepsilon)-5/2} \|\deh\|_{H^{m+2}}.
    \]
    Since $\bar\h_n$ satisfies the desired estimate, by \eqref{main main energy estimate}, it is enough to take $N$ larger if necessary to ensure that the conclusion holds. 
\end{proof}

\begin{theorem} \label{there is a smooth solution}
    There are $N$ and $n$ large enough such that the following holds. Let $(g,\s)$, with $g = -dt \otimes dt + h$, be the solution to the Einstein--nonlinear scalar field equations with potential $V$ constructed in Theorem~\ref{global existence} with $s = 5$ and $N$ and $n$ as above. Then $(g,\s)$ is smooth. Moreover, for every non-negative integer $m$, there are constants $C_m$ such that the estimates
    \[
    \begin{split}
        \big|D^m\big((\bar\h - \ho)(e_i,e_k)\big)\big|_\ho(x) &\leq C_m\langle \ln t \rangle^{m+2} t^{2\varepsilon + (p_i + p_k - 2p_1)(x)},\\
        \big|D^m\big((\bar\h - \ho)(e_i,e_k)\big)\big|_\ho(y) &\leq C_m \langle \ln t \rangle^{m+2} t^{2\varepsilon + |p_i - p_k|(y)},\\
        \big|D^m \big((tK - \Ko)(e_i,\omega^k)\big)\big|_\ho + t\big|D^m \p_t\big((tK - \Ko)(e_i,\omega^k)\big)\big|_\ho &\leq C_m \langle \ln t \rangle^{m+2} t^{2\varepsilon} \min\{1,t^{2(p_i - p_k)}\},\\
        |D^m(\bar\Psi - \psio)|_\ho &\leq C_m \langle \ln t \rangle^{m+2} t^{2\varepsilon},\\
        |D^m(\bar\Phi - \phio)|_\ho &\leq C_m \langle \ln t \rangle^{m+3} t^{2\varepsilon}
    \end{split}
    \]
    hold for $t \in (0,T_{N,5,n}]$, $x \in D_+$ and $y \in D_-$.
\end{theorem}

\begin{proof}
    Let $M$ be as in Proposition~\ref{first uniqueness result}, then there is an $n_0$ large enough such that if $n , n' \geq n_0$, then
    \[
    \begin{split}
        \sum_{m=0}^3 |D^m(h_n - h_{n'})|_\ho + \sum_{m=0}^2 |D^m \lie_{\p_t}(h_n - h_{n'})|_\ho &\leq Ct^M,\\
        \sum_{m=0}^2 |D^m(\s_n - \s_{n'})|_\ho + \sum_{m=0}^1 |D^m \p_t(\s_n - \s_{n'})|_\ho &\leq Ct^M.
    \end{split}
    \]
    Moreover, for $N_0$ large enough in Theorem~\ref{global existence}, if $n \geq n_{N_0,s}$, then \eqref{main main energy estimate} implies
    \[
    \begin{split}
        \sum_{m=0}^3 |D^m(h - h_n)|_\ho + \sum_{m=0}^2|D^m \lie_{\p_t}(h - h_n)|_\ho &\leq Ct^M,\\
        \sum_{m=0}^2 |D^m(\s - \s_n)|_\ho + \sum_{m=0}^1 |D^m \p_t(\s - \s_n)|_\ho &\leq Ct^M,
    \end{split}
    \]
    for $t \in (0,T_{N_0,s,n}]$. Now denote by $(g,\s)$ the solution obtained for $s = 5$, an appropriate $N \geq N_0$ and $n \geq \max\{n_0,n_{N,5}\}$ on $(0,T = T_{N,5,n}]$. Let $s \geq 5$ and $\widetilde N \geq N_0$ large enough, and denote by $(\widetilde g, \widetilde \s)$ the corresponding solution obtained for $ \widetilde n \geq \max\{n_0,n_{\widetilde N,s}\}$ on $(0, \widetilde T = T_{\widetilde N,s,\widetilde n}]$. Then, by Lemma~\ref{asymptotics for the actual solution}, the hypotheses of Proposition~\ref{first uniqueness result} are satisfied in $(0,\min\{T, \widetilde T\}]$, implying $(g,\s) = (\widetilde g, \widetilde \s)$ there. That is, $(h,K,\lie_{\p_t}K,\s,\p_t \s)$ is $H^{s+1} \times H^s \times H^{s-1} \times H^{s+1} \times H^s$ for all $t \in (0,\min\{T, \widetilde T\}]$. But then, by \cite[Lemma~27]{ringstrom_local_2024}, since $(g,\s)$ remains $C^3 \times C^3$ in all of $(0,T]$, it follows that $(g,\s)$ has the improved regularity on the original interval $(0,T]$. Since $s$ was arbitrary, then $(g,\s)$ is smooth for all $t \in (0,T]$. Regularity in time now follows from the evolution equations \eqref{the system}. Finally, the estimates follow by applying Lemma~\ref{asymptotics for the actual solution}, after extending the energy estimate \eqref{main main energy estimate}, for each $s$, to all of $(0,T]$.
\end{proof}

\subsection{Asymptotics for the expansion normalized induced metric} \label{asymptotics for the expansion normalized metric}

The solution $(g,\s)$ given in Theorem~\ref{there is a smooth solution} satisfies the estimates 
\[
|D^m(t\theta - 1)|_\ho \leq C_m\langle \ln t \rangle^{m+2}t^{2\varepsilon}.
\]
This implies that the required convergence of $\K$, $\Phi$ and $\Psi$ to the initial data on the singularity in Theorem~\ref{main existence theorem} is satisfied by $(g,\s)$. It remains to show that $\h$ also converges, which is the purpose of the present subsection. The issue lies in controlling the eigenspaces of $\K$ as $t \to 0$.

\begin{proposition} \label{construction of the frame}
    Let $(\Sigma,\ho,\Ko,\phio,\psio)$ be initial data on the singularity and suppose we have a Lorentzian metric $g = -dt \otimes dt + h$ on $(0,T] \times \Sigma$. Assume that, for every non-negative integer $m$, there are constants $C_m$ and $\delta>0$ such that 
    \[
    |D^m(\K - \Ko)|_\ho \leq C_mt^\delta.
    \]
    Then, by taking $T$ smaller if necessary, there are constants $C_m$ such that the following holds. The eigenvalues of $\K$ are everywhere distinct. Let $q_1 < q_2 < q_3$ denote the eigenvalues of $\K$. Then there is a frame $\{\widetilde e_i\}$ on $\Sigma$, satisfying $\K(\widetilde e_i) = q_i \widetilde e_i$ and $|\widetilde e_i|_\ho = 1$, such that
    \[
    |D^m(q_i - p_i)|_\ho + |D^m(\widetilde e_i - e_i)|_\ho \leq C_m t^\delta.
    \]
    Moreover, if for $i \neq k$,
    \[
    \big|D^m\big( (\K - \Ko)(e_i,\omega^k) \big)\big|_\ho \leq C_mt^\delta \min\{1,t^{2(p_i-p_k)}\},
    \]
    then
    \[
    |D^m( \omega^k(\widetilde e_i) - \delta_i^k )|_\ho \leq C_m t^\delta \min\{1,t^{2(p_i - p_k)}\}.
     \]
\end{proposition}

\begin{proof}
    By \cite[Equation~(3.6), p. 192]{stewart_matrix_1990}, there are continuous parametrizations $q_1 \leq q_2 \leq q_3$ of the eigenvalues of $\K$. Moreover, there is a constant $C$, depending only on the initial data, such that
    \[
    |q_i - p_i| \leq Ct^\delta.
    \]
    Since the $p_i$ are distinct, by taking $T$ smaller if necessary, we can ensure that the $q_i$ are everywhere distinct. 

    Next, we obtain estimates for the spatial derivatives of $q_i$. Let $f(x,\lambda) := \det(\K_x - \lambda I)$ for $x \in (0,T] \times \Sigma$. That is, $f$ is the characteristic polynomial of $\K$. Note that $f$ is smooth. Moreover, if $i$, $k$ and $\ell$ are distinct,
    \[
    \p_\lambda f(x,q_i(x)) = -( q_k(x) - q_i(x) )( q_\ell(x) - q_i(x) ) \neq 0.
    \]
    Thus, by the implicit function theorem, $q_i$ is smooth at $x$. Since $x$ was arbitrary, then $q_i$ is smooth. Furthermore,
    \begin{equation} \label{derivatives of the eigenvalues}
        e_a(q_i) = \frac{e_a\big( \det(\K - \lambda I) \big)\big|_{\lambda = q_i}}{(q_k - q_i)(q_\ell - q_i)}.
    \end{equation}
    Focusing on the numerator, we see that
    \[
    \begin{split}
        e_a\big( \det(\K - \lambda I) \big)\big|_{\lambda = q_i} &= e_a \big(\K(e_i,\omega^i)\big)( \K(e_k,\omega^k) - q_i )( \K(e_\ell,\omega^\ell) - q_i )\\
        &\quad + (\K(e_i,\omega^i) - q_i) e_a\big(\K(e_k,\omega^k)\big)( \K(e_\ell,\omega^\ell) - q_i )\\
        &\quad + (\K(e_i,\omega^i) - q_i)(\K(e_k,\omega^k) - q_i)e_a\big(\K(e_\ell,\omega^\ell)\big) + O(t^\delta)\\
        &= e_a \big(\K(e_i,\omega^i)\big)( \K(e_k,\omega^k) - q_i )( \K(e_\ell,\omega^\ell) - q_i ) + O(t^\delta)\\
        &= e_a\big(\K(e_i,\omega^i)\big)(q_k - q_i)(q_\ell - q_i) + O(t^\delta),
    \end{split}
    \]
    where we have used that (no summation on $a$)
    \[
    |\K(e_a,\omega^a) - q_a| \leq |\K(e_a,\omega^a) - p_a| + |p_a - q_a| \leq Ct^\delta.
    \]
    By taking $T$ smaller if necessary, there is an $\eta > 0$, depending only on the initial data, such that for $a \neq b$, $|q_a - q_b| \geq \eta$. It now follows from \eqref{derivatives of the eigenvalues} that
    \[
    |e_a(q_i - p_i)| \leq Ct^\delta.
    \]
    The estimates for the higher order spatial derivatives of $q_i-p_i$ follow similarly by iteratively differentiating \eqref{derivatives of the eigenvalues}.

    We now turn our attention to the eigenspaces. Fix $x_0 \in (0,T] \times \Sigma$. Let $\Gamma$ be a simple closed positively oriented curve in the resolvent set $\rho(\K_{x_0}) \subset \mathbb{C}$ of $\K_{x_0}$, enclosing $q_j(x_0)$ but no other eigenvalues of $\K_{x_0}$. By continuity, there is a neighborhood $U$ of $x_0$ such that $q_k(x)$ does not touch $\Gamma$ for all $k$ and all $x \in U$, in particular, $q_j(x)$ are enclosed in $\Gamma$. Then the eigenprojection of $\K_x$ associated with the eigenvalue $q_j(x)$ is given by
    \begin{equation} \label{formula for eigenprojections}
        P_j(x) = -\frac{1}{2\pi i} \int_\Gamma R(z,x)dz,
    \end{equation}
    where $R(z,x) := (\K_x - zI)^{-1}$, for $z \in \rho(\K_x)$, denotes the resolvent of $\K_x$ at $z$ (see \cite[Equation~(5.25) in Chapter~I]{kato_perturbation_1995}). Note that for $z \in \rho(\K_x)$, $R(z,x)$ is smooth as a $(1,1)$-tensor on $U$ (in the complexified $(1,1)$-tensor bundle). We conclude that the eigenprojections $P_j$ are smooth $(1,1)$-tensors on $U$. Since $x_0$ was arbitrary, then $P_j$ are smooth on $(0,T] \times \Sigma$. Similarly, the $P_j$ are continuous on $[0,T] \times \Sigma$. 

    Now we define the frame $\{\widetilde e_j\}$ by
    \[
    \widetilde e_j(t) := \frac{P_j(t)(e_j)}{|P_j(t)(e_j)|_\ho}.
    \]
    Note that, by taking $T$ smaller if necessary, $\{\widetilde e_j\}$ is a well defined smooth frame on $\Sigma$ for all $t \in (0,T]$. Indeed, since $P_j(0)(e_j) = e_j$, we can ensure that $P_j(t)(e_j) \neq 0$ for all small enough $t$. In particular, $\widetilde e_j \to e_j$ as $t \to 0$.

    In order to obtain estimates, we need first to obtain estimates for the resolvent. Similarly as above, for $x_0 \in \Sigma$ let $\Gamma$ be a simple closed positively oriented curve in $\rho(\Ko_{x_0})$, enclosing $p_j(x_0)$ but no other eigenvalues of $\Ko_{x_0}$. Then there is a neighborhood $U_{x_0} \subset \Sigma$ and a $t_{x_0} > 0$, such that $q_k(t,x)$ does not touch $\Gamma$ for all $k$ and all $(t,x) \in [0,t_{x_0}] \times U_{x_0}$. Let $\Ro(z) := (\Ko - zI)^{-1}$ denote the resolvent of $\Ko$ at $z$. Then
    \[
    \K - zI = ( I + (\K - \Ko)\Ro(z) )(\Ko - zI).
    \]
    Now take $t_{x_0}$ small enough so that $|(\K - \Ko)\Ro(z)|_\ho \leq \alpha < 1$ for all $(t,x) \in [0,t_{x_0}] \times U_{x_0}$ and all $z \in \Gamma$. Then, if we omit the spatial variable, we have
    \[
    R(z,t) = \Ro(z)( I + (\K_t - \Ko)\Ro(z) )^{-1} = \sum_{n=0}^\infty \Ro(z)[ (\Ko - \K_t)\Ro(z) ]^n.
    \]
    Therefore,
    \begin{equation} \label{resolvent equation}
    \begin{split}
        R(z,t) - \Ro(z) &= \Ro(z)(\Ko - \K_t)\Ro(z) \sum_{n=0}^\infty [(\Ko - \K_t) \Ro(z)]^n\\
        &= \Ro(z)(\Ko - \K_t)\Ro(z) (I - (\Ko - \K_t) \Ro(z))^{-1}.
    \end{split}
    \end{equation}
    This means that there is a constant $C(x_0)$, such that 
    \[
    |R(z,t) - \Ro(z)|_\ho \leq C(x_0)t^\delta
    \]
    for all $(t,x) \in [0,t_{x_0}] \times U_{x_0}$ and all $z \in \Gamma$. Now from \eqref{formula for eigenprojections}, if we denote $\mathring{P}_j = P_j(0)$, we get
    \[
    |P_j(t) - \mathring{P}_j|_\ho \leq \left| \frac{1}{2\pi i} \int_\Gamma R(z,t) - \Ro(z)dz \right|_\ho \leq C(x_0)t^\delta
    \]
    for all $(t,x) \in [0,t_{x_0}] \times U_{x_0}$. In order to estimate the spatial derivatives note that, by uniform convergence, all the factors on the right-hand side of \eqref{resolvent equation} have bounded derivatives of all orders. Thus,
    \[
    |D^m[R(z,t) - \Ro(z)]|_\ho \leq C_m(x_0)t^\delta.
    \]
    Implying
    \[
    |D^m[P_j(t) - \mathring{P}_j]|_\ho \leq C_m(x_0)t^\delta.
    \]
    By compactness of $\Sigma$, we can now find constants $C_m$ (independent of $x_0$) and a small enough $T > 0$ such that the estimates above hold, with $C(x_0)$ replaced by $C$, in all of $(0,T] \times \Sigma$. This implies the desired result for $\widetilde e_j$, except for the off-diagonal improvements.
    
    We continue with the off-diagonal improvements. For that purpose, first note that
    \[
    \Ro(z)(e_j,\omega^k) = \delta_{j}^k(p_j - z)^{-1}.
    \]
    Hence
    \[
    \begin{split}
        &\Ro(z)[ (\Ko - \K)\Ro(z) ]^n(e_j,\omega^k)\\
        &= \sum_{\ell_1,\ldots,\ell_{n-1}} \frac{(\Ko - \K)(e_j,\omega^{\ell_1}) (\Ko - \K)(e_{\ell_1},\omega^{\ell_2}) \cdots (\Ko - \K)(e_{\ell_{n-2}},\omega^{\ell_{n-1}}) (\Ko - \K)(e_{\ell_{n-1}},\omega^k)}{(p_j - z)(p_{\ell_1} - z) \cdots (p_{\ell_{n-1}} - z)(p_k - z)}.  
    \end{split}
    \]
    To estimate the numerator, we make the following observation. For $j$, $k$ and $\ell$ distinct,
    \[
    \begin{split}
        &|(\Ko - \K)(e_j,\omega^m)(\Ko - \K)(e_m,\omega^k)|\\
        &\hspace{3cm}\leq B^2t^{2\delta} \sum_m \min\{1,t^{2(p_j - p_m)}\} \min\{1,t^{2(p_m - p_k)}\}\\
        &\hspace{3cm}\leq B^2t^{2\delta}\big( \min\{1,t^{2(p_j - p_k)}\} + \min\{1,t^{2(p_j - p_\ell)}\} \min\{1, t^{2(p_\ell - p_k)}\} \big)\\
        &\hspace{3cm}\leq B^2t^{2\delta} \min\{1,t^{2(p_j - p_k)}\},
    \end{split}
    \]
    where $B$ denotes the constant which comes from estimating $\Ko - \K$. So if we work in a neighborhood $U$ of a point $x_0 \in \Sigma$ with an appropriate $\Gamma$ enclosing $p_j(x_0)$ as before, then there is a constant $C(x_0)$ such that
    \[
    |\Ro(z)[ (\Ko - \K)\Ro(z) ]^n(e_j,\omega^k)| \leq C(x_0)^{n+1}B^nt^{n\delta} \min\{1,t^{2(p_j - p_k)}\}. 
    \]
    This implies
    \[
    \begin{split}
        \big|\big(R(z,t) - \Ro(z)\big)(e_j,\omega^k)\big| &\leq \sum_{n=1}^\infty C(x_0)^{n+1}B^nt^{n\delta} \min\{1,t^{2(p_j - p_k)}\}\\
        &= \frac{C(x_0)^2B}{1-C(x_0)Bt^\delta} t^\delta \min\{1,t^{2(p_j - p_k)}\} 
    \end{split}
    \]
    for all $t < t_{x_0}$ small enough. Now, by \eqref{formula for eigenprojections} and compactness of $\Sigma$, we can find a constant $C$ (independent of $x_0$) such that, by taking $T$ smaller if necessary, we have
    \[
    |(P_j(t) - \mathring{P}_j)(e_k,\omega^\ell)| \leq Ct^\delta \min\{1,t^{2(p_k - p_\ell)}\}. 
    \]
    The estimates for the derivatives follow similarly as before. The result for $\widetilde e_j$ now follows from its definition.
\end{proof}

Before proceeding with the proof of Theorem~\ref{main existence theorem}, we need a way to estimate expressions involving the mean curvature raised to powers which depend linearly on the eigenvalues of $\K$.

\begin{lemma} \label{powers of theta}
    Let $(\Sigma,\ho,\Ko,\phio,\psio)$ be initial data on the singularity and let $g = -dt \otimes dt + h$ be a Lorentzian metric on $(0,T] \times \Sigma$. Suppose that there are positive constants $\delta$, $\sigma$ and $C_m$, for every non-negative integer $m$, such that
    \[
    |D^m(t\theta - 1)|_\ho \leq C_mt^\sigma, \qquad |D^m(q_i-p_i)|_\ho \leq C_mt^\delta,
    \]
    where $q_1\leq q_2\leq q_3$ denote the eigenvalues of $\K$. Let $L(q_1,q_2,q_3)$ be a linear function of the $q_i$ and  $b \in \R$. Then, by taking $T$ smaller if necessary, there are constants $C_m$ such that
    \[
    |D^m(\theta^{L(q_1,q_2,q_3) + b})|_\ho \leq C_mt^{-L(p_1,p_2,p_3)-b}.
    \]
\end{lemma}

\begin{proof}
    First note that our assumptions on $\theta$ imply that, after taking $T$ smaller if necessary, there is a positive constant $C$ such that $t\theta \geq C$. For the case with no derivatives, if $L(q_1,q_2,q_3) = a_1q_1 + a_2q_2 + a_3q_3$ with $a_i \in \R$, we have
    \[
    \theta^{L(q_1,q_2,q_3) + b} = (t\theta)^{L(q_1,q_2,q_3) + b} e^{[a_1(p_1-q_1) + a_2(p_2-q_2) + a_3(p_3-q_3)]\ln t} t^{-L(p_1,p_2,p_3) - b},
    \]
    from which the result follows. For the derivatives, we begin by estimating the derivatives of $\ln \theta$. If $\alpha$ is a multiindex with $|\alpha| \leq m$, then
    \[
    e_\alpha(\ln\theta) = \sum \frac{(-1)^{r+1}(r-1)!}{\theta^r}(e_{\beta_1}\theta) \cdots (e_{\beta_r}\theta),
    \]
    where the sum is over appropriate multiindices $\beta_i$ such that $|\beta_1| + \cdots + |\beta_r| = |\alpha|$. Then, by the lower bound on $|t\theta|$,
    \[
    |e_\alpha(\ln\theta)| \leq C_m \sum |e_{\beta_1}(t\theta)|\cdots |e_{\beta_r}(t\theta)| \leq C_m.
    \]
    Moving on, note that
    \[
    e_\alpha(\theta^{L(q_1,q_2,q_3) + b}) = \sum e_{\beta_1}\big( (L(q_1,q_2,q_3) + b)\ln\theta\big) \cdots e_{\beta_r}\big( (L(q_1,q_2,q_3) + b)\ln\theta\big) \theta^{L(q_1,q_2,q_3)+b},
    \]
    which together with the already obtained estimates yields the desired conclusion.
\end{proof}

\begin{proof}[Proof of Theorem~\ref{main existence theorem}]
    Consider the solution given by Theorem~\ref{there is a smooth solution}. Note that the estimates for $K$ imply
    \[
    |D^m(t\theta - 1)|_\ho \leq C_m\langle \ln t \rangle^{m+2} t^{2\varepsilon},
    \]
    hence we immediately obtain the convergence estimates for $\K$, $\Psi$ and $\Phi$ from the ones for $K$, $\bar \Psi$ and $\bar \Phi$. The only thing that remains to prove is the convergence of $\h$. 
    
    We begin by applying Proposition~\ref{construction of the frame} to obtain a frame $\{\widetilde e_i\}$ of eigenvectors of $\K$, such that $\K(\widetilde e_i) = q_i \widetilde e_i$, satisfying the estimates
    \[
    \begin{split}
        |D^m(q_i - p_i)|_\ho &\leq C_m t^{3\varepsilon/2},\\
        |D^m(\omega^k(\widetilde e_i) - \delta_i^k)|_\ho &\leq C_m t^{3\varepsilon/2} \min\{1,t^{2(p_i-p_k)}\}.
    \end{split}
    \]
    Note that we have applied Proposition~\ref{construction of the frame} with $\delta = 3\varepsilon/2$, so that we have used a factor of $t^{\varepsilon/2}$ to bound the powers of $\langle \ln t \rangle$ in the estimates for $\K$. Let $E_i := |\widetilde e_i|_{\h}^{-1} \widetilde e_i$. We want to show that the frame $\{E_i\}$ satisfies estimates similar to those satisfied by $\{\widetilde e_i\}$, from which the desired result for $\h$ will follow. Note that
    \[
    |\widetilde e_i|_{\h}^2 = \theta^{2q_i} \omega^k(\widetilde e_i) \omega^\ell(\widetilde e_i) h_{k \ell},
    \]
    hence
    \[
    \begin{split}
        |\widetilde e_i|_\h^2 - 1 &= \theta^{2q_i} \omega^i(\widetilde e_i)^2 h_{ii} - 1 + \sum_{(k,\ell) \neq (i,i)} \theta^{2q_i} \omega^k(\widetilde e_i) \omega^\ell(\widetilde e_i) h_{k \ell}\\
        &= (t\theta)^{2q_i} - 1 + (t\theta)^{2q_i}( t^{2(p_i-q_i)} - 1 ) + \theta^{2q_i}t^{2p_i}( t^{-2p_i}h_{ii} - 1 )\\
        &\quad + \theta^{2q_i}h_{ii}\big( \omega^i(\widetilde e_i)^2 - 1 \big) + \sum_{(k,\ell) \neq (i,i)} \theta^{2q_i} \omega^k(\widetilde e_i) \omega^\ell(\widetilde e_i) h_{k \ell}.
    \end{split}
    \]
    Now we can use the first order Taylor expansion of $e^x$ and $\ln(1+x)$ around $x = 0$ on $(t\theta)^{2q_i} = e^{2q_i\ln(t\theta)}$ and $t^{2(p_i-q_i)} = e^{2(p_i-q_i)\ln t}$, in addition to Lemma~\ref{powers of theta}, to obtain 
    \[
    |D^m( |\widetilde e_i|_\h^2 - 1 )|_\ho \leq C_mt^\varepsilon.
    \]
    Note that here we have used a factor of $t^{\varepsilon/2}$ coming from the estimate for $q_i - p_i$ to bound the $\ln t$ factor that arises from $t^{2(p_i-q_i)}$. But we need an estimate for $|\widetilde e_i|_\h$ without the exponent. First note that 
    \[
    \big||\widetilde e_i|_\h - 1\big| \leq \big||\widetilde e_i|_\h - 1\big| \big||\widetilde e_i|_\h + 1\big| = \big||\widetilde e_i|_\h^2 - 1\big| \leq Ct^\varepsilon.
    \]
    For the derivatives, note that if $|\alpha| \geq 1$, then
    \[
    2|\widetilde e_i|_{\h} e_\alpha( |\widetilde e_i|_{\h} ) = e_\alpha( |\widetilde e_i|_{\h}^2 ) - \sum e_\beta(|\widetilde e_i|_{\h}) e_\gamma(|\widetilde e_i|_{\h})
    \]
    where the sum is over appropriate multiindices $\beta$ and $\gamma$ such that $|\beta|, |\gamma| < |\alpha|$. By taking $T$ smaller if necessary, so that $|\widetilde e_i|_{\h}$ is bounded from below by a positive constant, we can estimate the derivatives of $|\widetilde e_i|_{\h}$ inductively from the expression above. We conclude that
    \[
    |D^m(|\widetilde e_i|_\h - 1)|_\ho \leq C_mt^\varepsilon.
    \]
    Now for $E_i$, we write
    \[
    E_i - e_i = \frac{1}{|\widetilde e_i|_\h} \big( \widetilde e_i - e_i + ( 1 - |\widetilde e_i|_\h ) e_i \big)
    \]
    which yields
    \[
    |D^m( \omega^k(E_i) - \delta_i^k )|_\ho \leq C_mt^\varepsilon \min\{1,t^{2(p_i-p_k)}\}.
    \]

    Finally, for the metric $\h$, we consider the dual frame $\{\eta^i\}$ of $\{E_i\}$. Since the matrix with components $\eta^k(e_i)$ is the inverse of the matrix with components $\omega^k(E_i)$, it can be computed that
    \[
    |D^m(\eta^k(e_i) - \delta_i^k)| \leq C_mt^{\varepsilon} \min\{1,t^{2(p_i - p_k)}\}.
    \]
    Since $\h = \sum_i \eta^i \otimes \eta^i$, the desired estimate for $\h$ follows.
\end{proof}

\section{Detailed asymptotics and uniqueness of developments} \label{detailed asymptotics and uniqueness}

Throughout the proof of Theorem~\ref{main existence theorem}, we obtain much more detailed asymptotic information about the constructed solutions than what the statement of the theorem says. The purpose of this section is to show that for a locally Gaussian development of initial data on the singularity, the more detailed asymptotic information is a consequence of Einstein's equations. Furthermore, as a consequence of the improved asymptotics, the uniqueness statement of Theorem~\ref{main uniqueness theorem} follows.

\subsection{Asymptotics for the mean curvature} \label{asymptotics for the mean curvature section}

We begin by relating the mean curvature $\theta$ with the time coordinate $t$ and by constructing an appropriate frame of eigenvectors of $\K$. This is the purpose of the following lemma.

\begin{lemma} \label{asymptotics for the mean curvature}
    Let $(\Sigma,\ho,\Ko,\phio,\psio)$ be initial data on the singularity, $V$ an admissible potential and $(M,g,\s)$ a locally Gaussian development of the data. By taking $T$ smaller if necessary, there are constants $C_m$, for every non-negative integer $m$, such that the following holds. Define $\sigma := \min\{\varepsilon,\delta/2\}$. We have
    \[
    |D^m( t\theta - 1 )|_\ho + |D^m( \ln \theta + \ln t )|_\ho \leq C_mt^{2\sigma}.
    \]
    In particular,
    \[
    |D^m(\theta^{-2}V\circ\s)|_\ho \leq C_m\langle \ln t \rangle^m t^{2\varepsilon_V}.
    \]
    The eigenvalues of $\K$ are everywhere distinct. Denote by $q_1<q_2<q_3$ the eigenvalues of $\K$. There is a frame $\{E_i\}$ which is orthonormal with respect to $\h$, with dual frame $\{\eta^i\}$, such that $\K(E_i) = q_iE_i$ and
    \[
    |D^m(q_i-p_i)|_\ho + |D^m(E_i-e_i)|_\ho + |D^m(\eta^i-\omega^i)|_\ho \leq C_m t^\delta.
    \]
\end{lemma}

\begin{proof}
    The statements about the eigenvalues of $\K$ follow from Propostion~\ref{construction of the frame}. Now we construct the frame. Consider the frame $\{\widetilde e_i\}$ as in Proposition~\ref{construction of the frame}. Note that the convergence estimates for $\h$ imply $|D^m( |\widetilde e_i|_\h^2 - 1 )|_\ho \leq C_mt^\delta$. Now define $E_i := |\widetilde e_i|_\h^{-1} \widetilde e_i$. Similarly as in the proof of Theorem~\ref{main existence theorem}, we see that
    \[
    |D^m(E_i - e_i)|_\ho + |D^m(\eta^i - \omega^i)|_\ho \leq C_mt^\delta.
    \]

    We move on to the mean curvature. Recall from \eqref{the system} that $\theta$ satisfies the equation
    \[
    \p_t \theta = -\tr K^2 - (\p_t \s)^2 + V \circ \s,
    \]
    implying
    \begin{equation} \label{equation for 1/theta}
    \p_t(\theta^{-1}) = \tr \K^2 + \Psi^2 - \theta^{-2} V \circ \s.
    \end{equation}
    Since $\tr\Ko = 1$, then $\tr \Ko^2 \geq 1/3$. This implies $|\psio| = \sqrt{1-\tr \Ko^2} \leq \sqrt{2/3}$. Therefore,
    \[
    |\theta^{-2} V \circ \s| \leq C\theta^{-2} e^{a|\Psi|\ln \theta + a|\Phi|} \leq C\theta^{-2\varepsilon_V + a|\Psi - \psio|}.
    \]
    It follows that $\theta^{-2} V \circ \s$ converges to zero as $t \to 0$. This in turn yields that $\p_t(\theta^{-1})$ converges to $1$ as $t \to 0$. Hence, by taking $T$ smaller if necessary, we can ensure that $1/2 \leq \p_t(\theta^{-1}) \leq 2$. Upon integrating this inequality from $0$ to $t$, recalling that $\theta \to \infty$ as $t \to 0$, we obtain that 
    \begin{equation} \label{first relation between theta and t}
        \frac{1}{2}t \leq \theta^{-1} \leq 2t, \qquad |\ln \theta + \ln t| \leq \ln 2.
    \end{equation}
    Now we can go back to estimating $\theta^{-2} V \circ \s$ using \eqref{first relation between theta and t} to obtain 
    \[
    |\theta^{-2}V \circ \s| \leq Ct^{2\varepsilon_V}.
    \]
    Moving on, we integrate \eqref{equation for 1/theta} from $0$ to $t$ to obtain
    \begin{equation} \label{formula for 1/theta}
        \theta^{-1} = t + \int_0^t ( \tr \K^2 + \Psi^2 - 1 - \theta^{-2}V \circ \s )(r)dr.
    \end{equation}
    Using our assumptions and the estimate for the potential term, we see that
    \[
    |\theta^{-1} - t| \leq Ct^{1+\min\{\delta,2\varepsilon_V\}},
    \]
    which implies the result for $t \theta - 1$ with no derivatives. The estimate for $\ln \theta + \ln t = \ln(t\theta)$ now follows from a Taylor approximation. 
    
    However, obtaining estimates for the derivatives requires a bit more work. Define the logarithmic volume density $\varrho$ by the condition $\mu = e^\varrho \mathring \mu$, where $\mu$ and $\mathring{\mu}$ are the volume forms of $h$ and $\ho$ respectively. By using the form of the metric $h = \sum_i \theta^{-2q_i} \eta^i \otimes \eta^i$, we see that 
    \[
    e^\varrho = \frac{1}{\theta} \eta^1 \wedge \eta^2 \wedge \eta^3(e_1,e_2,e_3) = \frac{1}{\theta} \det \eta^i(e_k),
    \]
    implying
    \[
    \varrho + \ln \theta = \ln\big( \det \eta^i(e_k) \big).
    \]
    From here, using the asymptotics of the dual frame $\{\eta^i\}$, it can be deduced that
    \begin{equation} \label{estimate for rho plus log theta}
        |D^m(\varrho + \ln \theta)|_\ho \leq C_mt^\delta.
    \end{equation}

    Our goal now is to use these estimates to obtain estimates for $D^m \theta^{-1}$ and $D^m \ln \theta$. This needs to be done inductively, so we start with only one derivative. Note that $\varrho$ satisfies $\p_t \varrho = \theta$, implying
    \[
    \frac{1}{\theta} \p_t e_i \varrho = -e_i \varrho + e_i( \varrho + \ln \theta ).
    \]
    Now let $\gamma$ be an integral curve of $\theta^{-1}\p_t$ such that $\varrho \circ \gamma(s) = s$. By introducing an integrating factor, we can write the equation above as follows,
    \[
    \frac{d}{ds}\big( e^s(e_i \varrho) \circ \gamma(s) \big) = e^s e_i( \varrho + \ln \theta ) \circ \gamma(s).
    \]
    After integrating from $s_1$ to $s_2$, and using $\eqref{estimate for rho plus log theta}$ and \eqref{first relation between theta and t}, we obtain
    \[
    |e^{s_1} (e_i \varrho) \circ \gamma(s_1)| \leq |e^{s_2}(e_i \varrho) \circ \gamma(s_2)| + Ce^{(1+\delta)s_2}.
    \]
    Now we let $\gamma(s_2) = (T,x)$ and $\gamma(s_1) = (t,x)$, then
    \[
    |e_i \varrho|(t,x) \leq Ce^{-\varrho(t,x)} \leq C t^{-1}
    \]
    where we have used \eqref{estimate for rho plus log theta} and \eqref{first relation between theta and t} again, and $C$ depends on $T$ but not on $x$. Note that \eqref{estimate for rho plus log theta} now implies $|e_i \ln \theta| \leq Ct^{-1}$. Unfortunately, this is still not enough, but we can now estimate $e_i(\theta^{-2} V \circ \s)$ as follows. We have
    \[
    e_i( \theta^{-2} V \circ \s) = -2(e_i \ln \theta) \theta^{-2} V \circ \s + \theta^{-2} (V' \circ \s) \big( e_i \Phi - (\ln \theta) e_i \Psi - (e_i \ln \theta) \Psi \big),
    \]
    implying
    \[
    |e_i( \theta^{-2}V \circ \s )| \leq Ct^{-1+2\varepsilon_V}.
    \]
    Now we can take $e_i$ of \eqref{formula for 1/theta} to obtain $|e_i (\theta^{-1})| \leq Ct^{\min\{1+\delta,2\varepsilon_V\}}$. But then we can write ${e_i \ln \theta = -\theta e_i (\theta^{-1})}$, which gives the improvement
    \[
    |e_i \ln \theta| \leq Ct^{-1+\min\{1+\delta,2\varepsilon_V\}}.
    \]
    If $1+\delta \leq 2\varepsilon_V$, we are done. Otherwise, we have $|e_i \ln \theta| \leq Ct^{-1+2\varepsilon_V}$. But then, we can go back to estimating $e_i( \theta^{-2} V \circ \s )$ and repeat the process to obtain further improvements. There is a positive integer $n$ such that $2n\varepsilon_V \geq 1+\delta$. Hence, after performing the improvement process $n$ times, we obtain
    \[
    |e_i( \theta^{-2} V \circ \s )| \leq C\langle \ln t \rangle t^{2\varepsilon_V}, \qquad |e_i(\theta^{-1})| \leq Ct^{1+\delta}, \qquad |e_i \ln \theta| \leq Ct^\delta.
    \]
    Note that at this point, no further improvements are possible. For higher order derivatives, similar arguments can be made by applying successively more derivatives to $\p_t \varrho = \theta$ and to \eqref{formula for 1/theta}. We conclude that
    \[
    |D^m( \theta^{-1} - t )|_\ho \leq C_mt^{1+\delta}, \qquad |D^m( \ln \theta + \ln t )|_\ho \leq C_mt^\delta.
    \]
    This implies the result for $t\theta - 1$.
\end{proof}

\subsection{Detailed asymptotics for the frame of eigenvectors} \label{asymptotics for the frame of eigenvectors}

In order to finish the proof of Theorem~\ref{asymptotics of the frame}, we now need to obtain the off-diagonal improvements for the estimates of the frame of eigenvectors $\{E_i\}$. For that purpose, we first obtain evolution equations for the eigenvalues and eigenvectors, after which we proceed with the proof of the theorem. 

\begin{lemma} \label{equation of the eigenframe}
    Let $(-dt \otimes dt + h,\s)$ be a solution to the Einstein--nonlinear scalar field equations with potential $V$. Let $\{X_i\}$ be an orthonormal frame with respect to $h$ such that $\K(X_i) = q_i X_i$, then
    \[
    \begin{split}
        \frac{1}{\theta} \p_t q_i &= \theta^{-2} \big( \bar S - |d\s|_h^2 - 3V\circ \s \big)q_i - \theta^{-2} \big( \sric(X_i,X_i) - (X_i\s)^2 - V \circ \s \big),\\
        \frac{1}{\theta}[\p_t,X_i] &= -q_i X_i + \sum_{k \neq i}\frac{\theta^{-2}}{q_k - q_i}\big(\sric - d\s \otimes d\s \big)(X_i,X_k)X_k.
    \end{split}
    \]
\end{lemma}

\begin{proof}
    Since $X_i$ is orthonormal, we have
    \[
    -h([\p_t,X_i],X_k) - h(X_i,[\p_t,X_k]) = \lie_{\p_t} h(X_i,X_k) = 2\theta q_i \delta_{ik}.
    \]
    In particular,
    \[
    h([\p_t,X_i],X_i) = -\theta q_i. 
    \]
    Now we apply $\lie_{\p_t}$ to the equation $K(X_i) = \theta q_i X_i$ to obtain
    \begin{equation} \label{derivative of eigenvector formula}
        \lie_{\p_t}K(X_i) + \sum_\ell \theta q_\ell h([\p_t,X_i],X_\ell)X_\ell = \p_t(\theta q_i)X_i + \theta q_i[\p_t,X_i].
    \end{equation}
    If we multiply this equation by $X_i$ using $h$, we see that
    \[
    h(\lie_{\p_t}K(X_i),X_i) = \theta \p_t q_i + q_i \p_t \theta.
    \]
    By substituting $\p_t \theta$ and $\lie_{\p_t}K$ with Equation~\eqref{evolution equation for k}, we obtain the desired equation for $q_i$. If $i \neq k$, we now take the scalar product of \eqref{derivative of eigenvector formula} with $X_k$ with respect to $h$, which yields
    \[
    h([\p_t,X_i],X_k) = \frac{1}{\theta(q_i - q_k)} h(\lie_{\p_t}K(X_i),X_k). 
    \]
    Summing up,
    \[
    \frac{1}{\theta}[\p_t,X_i] = \frac{1}{\theta}\sum_k h([\p_t,X_i],X_k)X_k = -q_i X_i + \sum_{k \neq i} \frac{\theta^{-2}}{q_i - q_k} h(\lie_{\p_t}K(X_i),X_k)X_k.
    \]
    The equation for $X_i$ follows by using Equation~\eqref{evolution equation for k} to substitute $\lie_{\p_t}K$.
\end{proof}

\begin{lemma} \label{detailed asymptotics for the frame lemma}
    Let $(\Sigma,\ho,\Ko,\phio,\psio)$ be initial data on the singularity, $V$ an admissible potential and $(M,g,\s)$ a locally Gaussian development of the data. Consider the frame $\{E_i\}$ of eigenvectors of $\K$, with dual frame $\{\eta^i\}$, given by Lemma~\ref{asymptotics for the mean curvature}. Define $\sigma := \min\{\varepsilon, \delta/2\}$. Then there are constants $C_m$, for every non-negative integer $m$, such that
    \begin{align*}
        \big|D^m\big(\omega^k(E_i)\big)\big|_\ho(x) + \big|D^m\big(\eta^k(e_i)\big)\big|_\ho(x) &\leq C_m\langle \ln t \rangle^{m+2}t^{2\varepsilon + 2(p_i-p_1)(x)}, \quad &i \neq k;\\
        \big|D^m\big(\omega^k(E_i)\big)\big|_\ho(y) + \big|D^m\big(\eta^k(e_i)\big)\big|_\ho(y) &\leq C_m\langle \ln t \rangle^{m+2}t^{2\varepsilon + 2(p_i-p_k)(y)}, \quad &i > k;
    \end{align*}
    for $x \in D_+$ and $y \in D_-$. In particular, if the $\lambda_{ik}^\ell$ are defined by $[E_i,E_k] = \lambda_{ik}^\ell E_\ell$, we have
    \[
    |D^m(\lambda_{23}^1)|_\ho(y) \leq C_m\langle \ln t \rangle^{m+2}t^{2\varepsilon + 2(p_2-p_1)(y)}.
    \]
    Moreover, if $i \neq k$, the following estimates hold for $\sric$,
    \[
    \begin{split}
        \big|D^m\big(\theta^{-2}\sric^\sharp(E_i,\eta^i)\big)\big|_\ho &\leq C_m\langle \ln t \rangle^{m+2}t^{2\sigma},\\
        \big|D^m\big( \theta^{-2}\sric^\sharp( E_i, \eta^k ) \big)\big|_\ho(x) &\leq C_m\langle \ln t \rangle^{m+2} t^{2\varepsilon + 2(p_i - p_1)(x)},\\
        \big|D^m\big( \theta^{-2}\sric^\sharp( E_i, \eta^k ) \big)\big|_\ho(y) &\leq  C_m\langle \ln t \rangle^{m+2} t^{2\varepsilon} \min\{ 1, t^{2(p_i - p_k)(y)} \},
    \end{split}
    \]
    where there is no summation over $i$ in the first inequality.
\end{lemma}

\begin{proof}
    For this proof, let $i$, $k$ and $\ell$ denote fixed indices, so that there is no summation over any of them when repeated; also, $x$ shall always denote an element of $D_+$ and $y$ an element of $D_-$. Let $X_i := \theta^{q_i} E_i$. By Lemma~\ref{equation of the eigenframe}, we then see that $q_i$ and $E_i$ satisfy the equations
    \begin{subequations} \label{equations for eigenvalues and frame}
    \begin{align}
        \begin{split}
            \frac{1}{\theta}\p_t q_i &= \theta^{-2}\big( \bar S - |d\s|_h^2 - 3V\circ\s \big)q_i\\
            &\quad - \theta^{-2}\sric^\sharp(E_i,\eta^i) + \theta^{2(q_i-1)}(E_i\s)^2 + \theta^{-2}V\circ\s,\label{equation for the eigenvalues}
        \end{split}\\
        \begin{split}
            \frac{1}{\theta}[\p_t,E_i] &= \left( \theta^{-2}\big( \bar S - |d\s|_h^2 - 3V\circ\s \big)q_i - \frac{\ln \theta}{\theta} \p_t q_i \right) E_i\\
            &\quad + \sum_{m \neq i} \frac{\theta^{-2}}{q_m - q_i} \big(\sric^\sharp - d\s \otimes \sn\s\big)(E_i, \eta^m) E_m.\label{eqution for the expansion normalized frame}
        \end{split}
    \end{align}
    \end{subequations}
    Our aim is to use \eqref{eqution for the expansion normalized frame} in order to improve on the estimates that we already have for $E_i$. For that purpose, it is thus of interest to estimate the objects related with the scalar field, the components of $\sric^\sharp$ in terms of the frame $\{E_i\}$, and to obtain a decay estimate for $\theta^{-1}\p_t q_i$. 
    
    Regarding the scalar field, note that $d\s \otimes \sn \s(E_i,\eta^k) = E_i(\s) E_k(\s) \theta^{2q_k}$. Hence, by Lemma~\ref{powers of theta},
    \[
    \begin{split}
        \big|D^m\big( \theta^{-2} d\s \otimes \sn \s(E_i, \eta^k) \big)\big|_\ho(x) &\leq C_m\langle \ln t \rangle^2 t^{2\varepsilon + 2(p_i-p_1)(x)},\\
        \big|D^m\big( \theta^{-2} d\s \otimes \sn \s(E_i, \eta^k) \big)\big|_\ho(y) &\leq C_m\langle \ln t \rangle^2 t^{2\varepsilon} \min\{1,t^{2(p_i - p_k)(y)}\},
    \end{split}
    \]
    which is what we want.
    
    Turning our attention to $\sric$, we intend to apply Lemma~\ref{general estimate for ricci lemma}, so we verify that its conditions are met with the frame $\{E_i\}$. Recall that the $\lambda_{ik}^\ell$ are defined by $[E_i, E_k] = \lambda_{ik}^\ell E_\ell$. Then
    \begin{equation} \label{formula for the lambdas}
        \lambda_{ik}^\ell = \omega^a(E_i) \omega^b(E_k) \eta^\ell(e_m) \gamma_{ab}^m + \omega^a(E_i) e_a\big( \omega^b( E_k ) \big) \eta^\ell(e_b) - \omega^b(E_k) e_b\big( \omega^a(E_i) \big) \eta^\ell(e_a),
    \end{equation}
    implying
    \[
    |D^m( \lambda_{ik}^\ell - \gamma_{ik}^\ell )|_\ho \leq C_mt^\delta.
    \]
    Note that, in particular, 
    \begin{equation} \label{estimate for lambda}
        |D^m(\lambda_{23}^1)|_\ho(y) \leq C_mt^\delta. 
    \end{equation}
    Denote $\Gamma_{ik}^\ell = \eta^\ell(\sn_{E_i} E_k)$. Note that in terms of the frame $\{E_i\}$, we have $h = \sum_i \theta^{-2q_i} \eta^i \otimes \eta^i$. Hence we can use the estimates for $\theta$, the $q_i$ and the $E_i$, together with Lemma~\ref{powers of theta}, to conclude that for $|\alpha| \leq m$,
    \[
    |e_\alpha \Gamma_{ii}^\ell| \leq C_m t^{2(p_i - p_\ell)}, \qquad |e_\alpha \Gamma_{ik}^k| + |e_\alpha \Gamma_{ik}^i| \leq C_m.
    \]
    Moreover, for $i$, $k$ and $\ell$ distinct,
    \[
    2\Gamma_{ik}^\ell = \theta^{2q_\ell}( -\lambda_{k \ell}^i \theta^{-2q_i} - \lambda_{i\ell}^k \theta^{-2q_k}  + \lambda_{ik}^\ell \theta^{-2q_\ell}),
    \]
    implying
    \begin{equation} \label{christoffel symbols in the frame}
        |e_\alpha \Gamma_{ik}^\ell|(x) \leq C_mt^{2(p_1 - p_\ell)(x)}, \qquad |e_\alpha \Gamma_{ik}^\ell|(y) \leq C_mt^{2(p_1 - p_\ell)(y) + \min\{\delta,2(p_2 - p_1)(y)\}}.
    \end{equation}
    We are now in a position to apply Lemma~\ref{general estimate for ricci lemma} (note that the last inequality in \eqref{hypothesis on the connection} holds since $h$ is diagonal in terms of the frame $\{E_i\}$). 
    
    Next, we obtain the required estimate for $\theta^{-1}\p_t q_i$. To that end, we look at the diagonal components of $\sric^\sharp$. By Lemma~\ref{general estimate for ricci lemma},  
    \[
    \big|D^m\big( \sric^\sharp(E_i,\eta^i) + \Lambda_{ik\ell} \big)\big|_\ho \leq C_m \langle \ln t \rangle^{m+2} t^{-2+2\varepsilon},
    \]
    where
    \[
    \Lambda_{ik\ell} = \theta^{2q_\ell} \Gamma_{i\ell}^k \Gamma_{\ell k}^i + \theta^{2q_k} \Gamma_{ik}^\ell \Gamma_{k\ell}^i + \theta^{2q_\ell} \lambda_{i\ell}^k \Gamma_{k\ell}^i + \theta^{2q_k} \lambda_{ik}^\ell \Gamma_{\ell k}^i
    \]
    for $i$, $k$ and $\ell$ distinct. If we look at the terms in $\Lambda_{ik\ell}$, we notice that in $D_+$ they are bounded by the same expression as the right-hand side of the inequality, implying
    \[
    \big|D^m\big( \theta^{-2}\sric^\sharp(E_i,\eta^i)\big)\big|_\ho(x) \leq C_m\langle \ln t \rangle^{m+2} t^{2\varepsilon}.
    \]
    In $D_-$ we need to look more carefully. We have
    \begin{align*}
    \begin{split}
        \theta^{2q_\ell} \Gamma_{i\ell}^k \Gamma_{\ell k}^i &= \frac{1}{4} \theta^2( -\lambda_{\ell k}^i \theta^{-2q_i} - \lambda_{ik}^\ell \theta^{-2q_\ell} + \lambda_{i\ell}^k \theta^{-2q_k} )( -\lambda_{ki}^\ell \theta^{-2q_\ell} - \lambda_{\ell i}^k \theta^{-2q_k} + \lambda_{\ell k}^i \theta^{-2q_i} )\\
        &= \pm\frac{1}{4}(\lambda_{23}^1)^2 \theta^{2-4q_1} + \cdots,
    \end{split}\\
    \begin{split}
        \theta^{2q_k} \Gamma_{ik}^\ell \Gamma_{k\ell}^i &=\frac{1}{4} \theta^2( -\lambda_{k\ell}^i \theta^{-2q_i} - \lambda_{i\ell}^k \theta^{-2q_k} + \lambda_{ik}^\ell \theta^{-2q_\ell} )( -\lambda_{\ell i}^k \theta^{-2q_k} - \lambda_{k i}^\ell \theta^{-2q_\ell} + \lambda_{k\ell}^i \theta^{-2q_i} )\\
        &= \pm\frac{1}{4}(\lambda_{23}^1)^2 \theta^{2-4q_1} + \cdots,
    \end{split}\\
    \theta^{2q_\ell} \lambda_{i\ell}^k \Gamma_{k\ell}^i &= \frac{1}{2} (\lambda_{i\ell}^k)^2 \theta^{2-4q_k} + \cdots,\\
    \theta^{2q_k} \lambda_{ik}^\ell \Gamma_{\ell k}^i &= \frac{1}{2} (\lambda_{ik}^\ell)^2 \theta^{2-4q_\ell} + \cdots,
    \end{align*}
    where $\cdots$ stands for terms which decay after multiplication by $\theta^{-2}$. Consequently,
    \begin{equation} \label{diagonal components of ricci}
        \big|D^m\big( \theta^{-2}\sric^\sharp(E_i,\eta^i) \pm \textstyle\frac{1}{2}(\lambda_{23}^1)^2 \theta^{-4q_1} \big)\big|_\ho(y) \leq C_m\langle \ln t \rangle^{m+2} t^{2\varepsilon},
    \end{equation}
    where we have $-$ for $i=1$ and $+$ for $i=2,3$. This implies that the scalar curvature satisfies the estimate
    \begin{equation} \label{scalar curvature estimate 1}
        \big|D^m\big( \theta^{-2} \bar S + \textstyle\frac{1}{2}(\lambda_{23}^1)^2 \theta^{-4q_1} \big)\big|_\ho(y) \leq C_m\langle \ln t \rangle^{m+2} t^{2\varepsilon};
    \end{equation}
    cf. \cite[Equation~(5), p. 6]{ringstrom_geometry_2021}. On the other hand, the Hamiltonian constraint reads
    \[
    \theta^{-2}\bar S = \tr\K^2 + \Psi^2 - 1 + \theta^{-2}|d\s|_h^2 + 2\theta^{-2}V \circ \s.
    \]
    Implying
    \begin{equation} \label{scalar curvature estimate 2}
        |D^m( \theta^{-2} \bar S )|_\ho \leq C_m \langle \ln t \rangle^{m+2} t^{2\sigma}.
    \end{equation}
    By putting \eqref{scalar curvature estimate 1} and $\eqref{scalar curvature estimate 2}$ together, we see that
    \begin{equation} \label{estimate for lambda squared}
        \big|D^m\big((\lambda_{23}^1)^2\big)\big|_\ho(y) \leq C_m \langle \ln t \rangle^{m+2}t^{-4p_1(y) + 2\sigma}.
    \end{equation}
    Note that, unfortunately, we cannot use this estimate to improve on \eqref{estimate for lambda}. But now we can go back to \eqref{diagonal components of ricci}, which yields
    \[
    \big|D^m\big(\theta^{-2}\sric^\sharp(E_i,\eta^i)\big)\big|_\ho(y) \leq C_m \langle \ln t \rangle^{m+2}t^{2\sigma}.
    \]
    Finally, we can go back to \eqref{equation for the eigenvalues} to conclude that
    \[
    |D^m(\theta^{-1} \p_t q_i)|_\ho \leq C_m\langle \ln t \rangle^{m+2}t^{2\sigma}.
    \]
    We remark that the importance of this estimate is that now we know that the expression in front of $E_i$, on the right-hand side of \eqref{eqution for the expansion normalized frame}, decays as a positive power of $t$.
    
    We are now ready to improve on the estimates for $E_i$. We start by looking at the off-diagonal components of $\sric^\sharp$. From Lemma~\ref{general estimate for ricci lemma}, it follows that, for $i$, $k$ and $\ell$ distinct,
    \begin{equation*}
    \begin{split}
        \big|D^m\big( \theta^{-2} \sric^\sharp(E_i, \eta^k) \big)\big|_\ho &\leq C_m \sum_{|\alpha| \leq m+1} t^{2(1-p_\ell)}( |e_\alpha \Gamma_{i \ell}^k| + |e_\alpha \Gamma_{\ell i}^k| + |e_\alpha \lambda_{i\ell}^k| )\\
        &\quad + C_m\sum_{|\alpha|\leq m} t^{2(1-p_k)} |e_\alpha \Gamma_{ik}^\ell|\\
        &\quad + C_m\langle \ln t \rangle^{m+2} t^{2(1-p_3)} \min\{1,t^{2(p_i - p_k)}\}.
    \end{split}
    \end{equation*}
    In $D_+$, we immediately obtain what we want,
    \[
    \begin{split}
        \big|D^m\big( \theta^{-2}\sric^\sharp( E_i, \eta^k ) \big)\big|_\ho(x) &\leq C_m \langle \ln t \rangle^{m+2} t^{2\varepsilon + 2(p_i - p_1)(x)}\\
        &\leq  C_m \langle \ln t \rangle^{m+2} t^{2\varepsilon} \min\{ 1, t^{2(p_i - p_k)(x)} \}.
    \end{split}
    \]
    However, in $D_-$ we do not necessarily get the desired estimates right away. In this case, we have
    \begin{equation*}
    \begin{split}
        \big|D^m\big( \theta^{-2} \sric^\sharp(E_1, \eta^2) \big)\big|_\ho(y) &\leq C_m\langle \ln t \rangle^{m+2}t^{2\varepsilon + 2(p_1-p_2)(y) + \min\{\delta,2(p_2-p_1)(y)\}},\\
        \big|D^m\big( \theta^{-2} \sric^\sharp(E_1, \eta^3) \big)\big|_\ho(y) &\leq C_m\langle \ln t \rangle^{m+2}t^{2\varepsilon + 2(p_1-p_2)(y) + \min\{\delta,2(p_2-p_1)(y)\}},\\
        \big|D^m\big( \theta^{-2} \sric^\sharp(E_2, \eta^3) \big)\big|_\ho(y) &\leq C_m\langle \ln t \rangle^{m+2}t^{2\varepsilon},\\
        \big|D^m\big( \theta^{-2} \sric^\sharp(E_2, \eta^1) \big)\big|_\ho(y) &\leq C_m\langle \ln t \rangle^{m+2}t^{2\varepsilon + \min\{\delta,2(p_2-p_1)(y)\}},\\
        \big|D^m\big( \theta^{-2} \sric^\sharp(E_3, \eta^1) \big)\big|_\ho(y) &\leq C_m\langle \ln t \rangle^{m+2}t^{2\varepsilon + 2(p_3 - p_2)(y) + \min\{\delta,2(p_2-p_1)(y)\}},\\
        \big|D^m\big( \theta^{-2} \sric^\sharp(E_3, \eta^2) \big)\big|_\ho(y) &\leq C_m\langle \ln t \rangle^{m+2}t^{2\varepsilon + 2(p_3 - p_2)(y)}.
    \end{split}
    \end{equation*}
    Note that if $\delta \geq 2(p_2 - p_1)(y)$, we obtain the desired estimates for $\sric$ at the point $y$. Otherwise, the estimates where one of the indices is $1$ require further improvement. Coming back to \eqref{eqution for the expansion normalized frame}, if $\tau = -\ln t$, we see that for $i$, $k$ and $\ell$ distinct,
    \begin{equation} \label{equation for the components of the frame}
    \begin{split}
        \p_\tau \omega^k(E_i) = A_i \omega^k(E_i) - e^{-\tau}\theta &\bigg( \frac{\theta^{-2}}{q_k - q_i}( \sric - d\s \otimes d\s )^\sharp(E_i,\eta^k)\omega^k(E_k)\\
        &+ \frac{\theta^{-2}}{q_\ell - q_i} ( \sric - d\s \otimes d\s )^\sharp(E_i, \eta^\ell) \omega^k(E_\ell) \bigg),
    \end{split}
    \end{equation}
    where $A_i$ and all of their derivatives decay exponentially in $\tau$. Thus each of the $\omega^k(E_i)$ satisfies an equation as in Lemma~\ref{linear ode lemma}. Note that, since we already know the $\omega^k(E_i)$ to decay exponentially in $\tau$, \eqref{equation for the components of the frame} in addition to the estimates for $\sric^\sharp$ and $d\s \otimes \sn \s$ in $D_+$, immediately imply the desired improvements for $\omega^k(E_i)$ in $D_+$. On the other hand, in $D_-$ we obtain the following improvements for the estimates when $i > k$,
    \begin{equation} \label{improvements on the frame}
    \begin{split}
        \big|D^m\big( \omega^1(E_2) \big)\big|_\ho(y) &\leq C_m\langle \ln t \rangle^{m+2}t^{2\varepsilon + \min\{\delta,2(p_2-p_1)(y)\}},\\
        \big|D^m\big( \omega^1(E_3) \big)\big|_\ho(y) &\leq C_m\langle \ln t \rangle^{m+2}t^{2\varepsilon + 2(p_3 - p_2)(y) + \min\{\delta,2(p_2-p_1)(y)\}},\\
        \big|D^m\big( \omega^2(E_3) \big)\big|_\ho(y) &\leq C_m\langle \ln t \rangle^{m+2}t^{2\varepsilon + 2(p_3 - p_2)(y)}.
    \end{split}
    \end{equation}
    Once again, note that if $\delta \geq 2(p_2 - p_1)(y)$ we obtain the desired estimates at $y$. Turning our attention to the dual frame $\{\eta^i\}$, for $i$, $k$ and $\ell$ distinct, we have
    \[
    \eta^k(e_i) = \frac{\pm 1}{\det( \omega^b(E_a) )} \left( \omega^k( E_i) \omega^\ell(E_\ell) - \omega^k(E_\ell) \omega^\ell(E_i) \right),
    \]
    hence the improvements in \eqref{improvements on the frame} translate to the dual frame. We can now go back to estimating $\lambda_{23}^1$ from $\eqref{formula for the lambdas}$ to obtain $|D^m(\lambda_{23}^1)|_\ho(y) \leq C_m\langle \ln t \rangle^{m+2}t^{2\varepsilon + \min\{\delta,2(p_2-p_1)(y)\}}$, which is an improvement on \eqref{estimate for lambda}. In general, there are going to be points $y \in D_-$ with $\delta < 2(p_2-p_1)(y)$, so the improvements we have obtained are not good enough. On the other hand, they allow us to start an iterative procedure to obtain further improvements.
    
    To set up the iteration, define the sets 
    \[
    B_n := \{ y \in D_- : 2n\varepsilon + \delta < 2(p_2-p_1)(y) \}
    \]
    for $n$ a non-negative integer, and set $B_{-1} := D_-$. Now fix a positive integer $n$ and make the inductive assumption that the conclusions of the lemma hold for $y \in B_{n-2} \setminus B_{n-1}$, while the estimates
    \[
    \begin{split}
        |D^m (\lambda_{23}^1)|_\ho(y) &\leq C_m\langle \ln t \rangle^{m+2} t^{2n\varepsilon + \delta},\\
        \big|D^m\big( \omega^1(E_2) \big)\big|_\ho(y) + \big|D^m\big( \eta^1(e_2) \big)\big|_\ho(y) &\leq C_m\langle \ln t \rangle^{m+2}t^{2n\varepsilon + \delta},\\
        \big|D^m\big( \omega^1(E_3) \big)\big|_\ho(y) + \big|D^m\big( \eta^1(e_3) \big)\big|_\ho(y) &\leq C_m\langle \ln t \rangle^{m+2}t^{2n\varepsilon + \delta + 2(p_3 - p_2)(y)},
    \end{split}
    \]
    hold for $y \in B_{n-1}$ (note that we already know the inductive assumption to hold for $n = 1$).  Then for $i$, $k$ and $\ell$ distinct, $\Gamma_{ik}^\ell$ satisfies
    \begin{equation*}
        |e_\alpha \Gamma_{ik}^\ell|(y) \leq C_m\langle \ln t \rangle^{m+2}t^{2(p_1 - p_\ell)(y) + \min\{2n\varepsilon + \delta,2(p_2-p_1)(y)\}}, \quad y \in B_{n-1},
    \end{equation*}
    instead of the second inequality in \eqref{christoffel symbols in the frame}. It follows that
    \begin{equation*}
    \begin{split}
        \big|D^m\big( \theta^{-2} \sric^\sharp(E_1, \eta^2) \big)\big|_\ho(y) &\leq C_m\langle \ln t \rangle^{m+2}t^{2\varepsilon + 2(p_1-p_2) + \min\{2n\varepsilon + \delta,2(p_2-p_1)\}},\\
        \big|D^m\big( \theta^{-2} \sric^\sharp(E_1, \eta^3) \big)\big|_\ho(y) &\leq C_m\langle \ln t \rangle^{m+2}t^{2\varepsilon + 2(p_1-p_2) + \min\{2n\varepsilon + \delta,2(p_2-p_1)\}},\\
        \big|D^m\big( \theta^{-2} \sric^\sharp(E_2, \eta^1) \big)\big|_\ho(y) &\leq C_m\langle \ln t \rangle^{m+2}t^{2\varepsilon + \min\{2n\varepsilon + \delta,2(p_2-p_1)(y)\}},\\
        \big|D^m\big( \theta^{-2} \sric^\sharp(E_3, \eta^1) \big)\big|_\ho(y) &\leq C_m\langle \ln t \rangle^{m+2}t^{2\varepsilon + 2(p_3 - p_2)(y) + \min\{2n\varepsilon + \delta,2(p_2-p_1)(y)\}},
    \end{split}
    \end{equation*}
    for $y \in B_{n-1}$. Going back to \eqref{equation for the components of the frame}, we can estimate $\omega^1(E_2)$ and $\omega^1(E_3)$ again, and then estimate $\eta^1(e_2)$ and $\eta^1(e_3)$ once more, to obtain
    \[
    \begin{split}
        \big|D^m\big( \omega^1(E_2) \big)\big|_\ho(y) + \big|D^m\big( \eta^1(e_2) \big)\big|_\ho(y) &\leq C_m\langle \ln t \rangle^{m+2}t^{2\varepsilon + \min\{2n\varepsilon + \delta,2(p_2-p_1)(y)\}},\\
        \big|D^m\big( \omega^1(E_3) \big)\big|_\ho(y) + \big|D^m\big( \eta^1(e_3) \big)\big|_\ho(y) &\leq C_m\langle \ln t \rangle^{m+2}t^{2\varepsilon + 2(p_3 - p_2)(y) + \min\{2n\varepsilon + \delta,2(p_2-p_1)(y)\}},
    \end{split}
    \]
    for $y \in B_{n-1}$. Finally, we can estimate $\lambda_{23}^1$ again from \eqref{formula for the lambdas}, which leads to
    \[
    |D^m( \lambda_{23}^1 )|_\ho(y) \leq C_m\langle \ln t \rangle^{m+2} t^{2\varepsilon + \min\{2n\varepsilon + \delta,2(p_2-p_1)(y)\}}, \quad y \in D_{n-1}.
    \]
    But then the inductive assumption holds with $n$ replaced by $n+1$. There is a positive integer $N$ such that $B_N = \varnothing$. The lemma follows.
\end{proof}

\subsection{Proofs of Theorems \ref{asymptotics of the frame} and \ref{main uniqueness theorem}} \label{proofs of detailed asymptotics and uniqueness}

\begin{proof}[Proof of Theorem~\ref{asymptotics of the frame}]
    The statements for the mean curvature, the eigenvalues and the frame of eigenvectors follow from Lemmas~\ref{asymptotics for the mean curvature} and \ref{detailed asymptotics for the frame lemma}. For $\bar\h$ note that
    \[
    \bar\h(e_i,e_k) = t^{-p_i - p_k}h(e_i,e_k) = t^{-p_i - p_k} \sum_\ell \theta^{-2q_\ell}\eta^\ell(e_i)\eta^\ell(e_k). 
    \]
    If $i$, $k$ and $\ell$ are distinct, we have
    \[
    \begin{split}
        \bar \h(e_i,e_i) - 1 &= (t\theta)^{-2q_i} - 1 + (t\theta)^{-2q_i}( t^{2(q_i-p_i)} - 1 ) + t^{-2p_i}\theta^{-2q_i} \big( \eta^i(e_i)^2 - 1 \big)\\
        &\quad + t^{-2p_i}\big( \theta^{-2q_k} \eta^k(e_i)^2 + \theta^{-2q_\ell} \eta^\ell(e_i)^2 \big).
    \end{split}
    \]
    So that, after a suitable Taylor approximation of $(t\theta)^{-2q_i}$ and $t^{2(q_i-p_i)}$, we get
    \[
    \big|D^m\big( \bar\h(e_i,e_i) - 1 \big)\big|_\ho \leq C_mt^\sigma.
    \]
    Moreover,
    \[
    \bar\h(e_i,e_k) = t^{-p_i-p_k}\big( \theta^{-2q_i} \eta^i(e_i)\eta^i(e_k) + \theta^{-2q_k} \eta^k(e_i)\eta^k(e_k) + \theta^{-2q_\ell} \eta^\ell(e_i)\eta^\ell(e_k) \big).
    \]
    Hence
    \[
    \begin{split}
        \big|D^m\big(\bar\h(e_i,e_k)\big)\big|_\ho &\leq C_m \langle \ln t \rangle^{m} \sum_{a = 0}^m \Big( t^{p_i-p_k} \big|D^a\big( \eta^i(e_k) \big)\big|_\ho + t^{p_k-p_i} \big|D^a\big( \eta^k(e_i) \big)\big|_\ho \Big)\\
        &\quad + C_m \langle \ln t \rangle^{m}\sum_{a+b \leq m} t^{p_\ell-p_i} \big|D^a\big( \eta^\ell(e_i) \big)\big|_\ho \cdot t^{p_\ell-p_k} \big|D^b\big( \eta^\ell(e_k) \big)\big|_\ho.\\ 
    \end{split}
    \]
    The result for $\bar\h$ follows. For $tK$, note that
    \[
    tK(e_i,\omega^k) = t \eta^\ell(e_i) \omega^k(E_m) K(E_\ell,\eta^m) = t\theta \sum_\ell q_\ell \eta^\ell(e_i) \omega^k(E_\ell). 
    \]
    Then, if $i$, $k$ and $\ell$ are distinct,
    \[
    \begin{split}
        tK(e_i,\omega^i) - p_i &= q_i - p_i + q_i\big( \eta^i(e_i) - 1 \big) + q_i \eta^i(e_i)\big( \omega^i(E_i) - 1 \big) + (t\theta - 1)q_i\eta^i(e_i)\omega^i(E_i)\\
        &\quad + t\theta\big( q_k\eta^k(e_i)\omega^i(E_k) + q_\ell\eta^\ell(e_i)\omega^i(e_\ell) \big),\\
        tK(e_i,\omega^k) &= t\theta \big( q_i\eta^i(e_i)\omega^k(E_i) + q_k\eta^k(e_i)\omega^k(E_k) + q_\ell\eta^\ell(e_i)\omega^k(E_\ell) \big).
    \end{split}
    \]
    The result for $tK$ follows.

    The only thing that remains to prove is the result for the Kretschmann scalar. The argument is similar to those used in \cite{oude_groeniger_formation_2023,fournodavlos_stable_2023}. We consider the orthonormal frame $\{\p_t,X_1,X_2,X_3\}$ for $g$, where $X_i = \theta^{q_i}E_i$. Then
    \[
        |R|_g^2 = \sum_{i,k,\ell,m} \Big( R(X_i,X_k,X_\ell,X_m)^2 - 4R(\p_t,X_i,X_k,X_\ell)^2 + 4R(\p_t,X_i,X_k,\p_t)^2 \Big).
    \]
    We look at each of the terms separately. By the Gauss equation,
    \[
    R(X_i,X_k,X_\ell,X_m) = \bar R(X_i,X_k,X_\ell,X_m) - k(X_i,X_\ell)k(X_k,X_m) + k(X_i,X_m)k(X_k,X_\ell).
    \]
    Now note that we can use \eqref{curvature in terms of ricci} to obtain
    \[
    t^2|D^m[\bar R(X_i,X_k,X_\ell,X_m)]|_\ho \leq C_mt^\sigma,
    \]
    and moreover
    \[
    \begin{split}
        &\sum_{i,k,\ell,m} \Big( -k(X_i,X_\ell)k(X_k,X_m) + k(X_i,X_m)k(X_k,X_\ell) \Big)^2\\
        &= 2\sum_{i,k,\ell,m} \Big( k(X_i,X_\ell)^2k(X_k,X_m)^2 - k(X_i,X_\ell)k(X_k,X_m)k(X_i,X_m)k(X_k,X_\ell)\Big)\\
        &= 2\theta^4 \Bigg( \sum_{i,k} q_i^2 q_k^2 - \sum_i q_i^4  \Bigg).
    \end{split}
    \]
    Hence
    \[
    t^4\big|D^m\big[ \textstyle \sum_{i,k,\ell,m} R(X_i,X_k,X_\ell,X_m)^2 - 4\theta^4 \textstyle \sum_{i<k} q_i^2 q_k^2 \big]\big|_\ho \leq C_mt^\sigma.
    \]
    Moving on,
    \[
    \begin{split}
        R(\p_t,X_i,X_k,\p_t) &= -\ric(X_i,X_k) + \sum_{\ell} R(X_\ell,X_i,X_k,X_\ell)\\
        &= - \ric(X_i,X_k) + \sric(X_i,X_k) -\sum_{\ell} k(X_\ell,X_k)k(X_i,X_\ell) + \theta k(X_i,X_k)\\
        &= -X_i(\s)X_k(\s) - (V \circ \s)\delta_{ik} + \sric(X_i,X_k) + \theta^2 q_i(1-q_i)\delta_{ik},
    \end{split}
    \]
    which yields
    \[
    t^4\big| D^m \big[ \textstyle \sum_{i,k} R(\p_t,X_i,X_k,\p_t)^2 - \theta^4 \textstyle \sum_i q_i^2(1-q_i)^2 \big] \big|_\ho \leq C_mt^\sigma.
    \]
    Now for the last term, by the Codazzi equation,
    \[
    \begin{split}
        R(\p_t,X_i,X_k,X_\ell) &= \sn_{X_k} k(X_\ell,X_i) - \sn_{X_\ell} k(X_k,X_i)\\
        &= \theta^{q_i+q_k+q_\ell} \big( E_k( \theta^{1-2q_\ell} q_\ell \delta_{\ell i}  ) -E_\ell( \theta^{1-2q_k} q_k \delta_{ki} ) + \lambda_{\ell k}^i \theta^{1-2q_i} q_i + \lambda_{\ell i}^k \theta^{1-2q_k} q_k \big).
    \end{split}
    \]
    Therefore, by using the decay estimate for $\lambda_{23}^1$ in $D_-$, we see that
    \[
    t^2|D^m[ R(\p_t,X_i,X_k,X_\ell) ]|_\ho \leq C_mt^\sigma.
    \]
    We can now put all the estimates together to obtain the desired result for $|R|_g^2$.
\end{proof}

\begin{proof}[Proof of Corollary~\ref{estimates for time derivatives}]
    From \eqref{evolution equation for k}, \eqref{evolution equation for phi}, \eqref{equations for eigenvalues and frame}, and the fact that the matrix with components $\eta^k(e_i)$ is the inverse of the matrix with components $\omega^k(E_i)$, it follows that
    \begin{align*}
        t\p_t(t\theta) &= t\theta(1-t\theta) + t^2( -\bar S + |d\s|_h^2 + 3V \circ \s ),\\
        t\p_t \bar\Psi &= (1 - t\theta)\bar\Psi + t^2( \Delta_h\s - V' \circ\s ),\\
        t\p_t q_i &= (t\theta)^{-1} t^2( \bar S - |d\s|_h^2 - 3V \circ \s )q_i\\
        &\quad + (t\theta)^{-1} t^2\big( -\sric^\sharp + d\s \otimes \sn\s \big)(e_k,\omega^\ell)\omega^k(E_i)\eta^i(E_\ell) + (t\theta)^{-1} t^2V\circ\s,\\
        t\p_t \omega^k(E_i) &= \Big( (t\theta)^{-1}t^2(\bar S - |d\s|_h^2 - 3V\circ\s)q_i - (\ln\theta)t\p_t q_i \Big) \omega^k(E_i)\\
        &\quad + (t\theta)^{-1} \sum_{m \neq i} \frac{t^2}{q_m-q_i}(\sric^\sharp - d\s \otimes \sn\s)(e_a,\omega^b) \omega^a(E_i) \eta^m(e_b) \omega^k(E_m),\\
        t\p_t \eta^k(e_i) &= -\eta^a(e_i) \eta^k(e_b) t\p_t \omega^b(E_a).
    \end{align*}
    The result follows from Lemmas~\ref{decay of ricci}, \ref{asymptotics for the mean curvature} and \ref{detailed asymptotics for the frame lemma} by repeatedly applying $t\p_t$ to the above equations, estimating the resulting right-hand side, and then estimating the $t\p_t$ derivatives of $\bar\h$ and $tK$ from their expressions in terms of the $E_i$ and the $\eta^i$ as in the proof of Theorem~\ref{asymptotics of the frame}. 
\end{proof}

Now we are ready to prove uniqueness of solutions. This relies on the following result, which comes from the fact that given asymptotics for $\bar\h$, $tK$, $\bar\Phi$ and $\bar\Psi$ up to a high enough regularity, the hypotheses of Proposition~\ref{first uniqueness result} are satisfied.

\begin{proposition} \label{second uniqueness result}
    Let $(\Sigma,\ho,\Ko,\phio,\psio)$ be initial data on the singularity and let $V$ be an admissible potential. Let $(g = -dt \otimes dt + h,\s)$ and $(\widetilde g = -dt \otimes dt + \widetilde h, \widetilde \s)$ on $(0,T] \times \Sigma$, be $C^{A+1} \times C^{A+1}$ solutions to the Einstein--nonlinear scalar field equations with potential $V$ such that
    \[
    \begin{split}
        \sum_{m=0}^A \Big( |D^m( \bar\h - \ho )|_\ho + |D^m( \overline{\widetilde{\h}} - \ho )|_\ho + |D^m( tK - \Ko )|_\ho + |D^m( t\widetilde{K} - \Ko )|_\ho \Big) &\leq Ct^\delta,\\
        \sum_{m=0}^A \left(|D^m( \bar\Phi - \phio )|_\ho + |D^m( \overline{\widetilde{\Phi}} - \phio )|_\ho + |D^m( \bar\Psi - \psio )|_\ho + |D^m( \overline{\widetilde{\Psi}} - \psio )|_\ho \right) &\leq Ct^\delta,
    \end{split}
    \]
    for some constants $C$ and $\delta > 0$. Moreover, assume that for $i \neq k$,
    \[
    \begin{split}
    \sum_{m=0}^A \left(\big|D^m\big( \bar\h(e_i,e_k) \big)\big|_\ho(x) + \big|D^m\big( \overline{\widetilde{\h}}(e_i,e_k) \big)\big|_\ho(x) \right) &\leq Ct^{\delta + (p_i+p_k-2p_1)(x)},\\
    \sum_{m=0}^A \left(\big|D^m\big( \bar\h(e_i,e_k) \big)\big|_\ho(y) + \big|D^m\big( \overline{\widetilde{\h}}(e_i,e_k) \big)\big|_\ho(y) \right) &\leq Ct^{\delta + |p_i - p_k|(y)},\\
    \sum_{m=0}^A \left(\big|D^m\big( tK(e_i,\omega^k) \big)\big|_\ho(x) + \big|D^m\big( t\widetilde K(e_i,\omega^k) \big)\big|_\ho(x) \right) &\leq Ct^{\delta + 2(p_i-p_1)(x)},\\
    \sum_{m=0}^A \left(\big|D^m\big( tK(e_i,\omega^k) \big)\big|_\ho(y) + \big|D^m\big( t\widetilde K(e_i,\omega^k) \big)\big|_\ho(y) \right) &\leq Ct^\delta \min\{ 1, t^{2(p_i - p_k)(y)} \},
    \end{split}
    \]
    for $x \in D_+$, $y \in D_-$. If $A$ is large enough (depending only on the initial data, the potential and $\delta$), then $(g,\s) = (\widetilde g, \widetilde\s)$.
\end{proposition}

\begin{proof}
    Let $M$ be as in Proposition~\ref{first uniqueness result}. We claim that there is an $A$ large enough such that 
    \[
    \begin{split}
      \sum_{m=0}^3 |D^m(h - h_n)|_\ho + \sum_{m=0}^2 |D^m \lie_{\p_t}(h - h_n)|_\ho &\leq Ct^M,\\
      \sum_{m=0}^2 |D^m(\s - \s_n)|_\ho + \sum_{m=0}^1 |D^m \p_t(\s - \s_n)|_\ho &\leq Ct^M,
    \end{split}
    \]
    for $n$ large enough, where $(h_n,\s_n)$ is an approximate solution as in Theorem~\ref{approximate solutions}; and similarly for $(\widetilde g, \widetilde \s)$.

    To prove the claim, let $\tau = -\ln t$ and define the variables
    \[
    \dt := e^{-\tau}(\theta - \theta_n), \quad \dk := e^{-\tau}(K - K_n), \quad \delta \bar\h := \bar\h - \bar\h_n, \quad \delta \bar\Psi := \bar\Psi - \bar\Psi_n, \quad \delta \bar\Phi := \bar\Phi - \bar\Phi_n.
    \]
    As a consequence of \eqref{adm equations}, \eqref{k first order equation}, \eqref{approximate metric equation} and \eqref{approximate scalar field equation}, these variables satisfy the following system of equations,
    \begin{subequations} \label{equations for uniqueness} 
    \begin{align} 
    \begin{split}
        \p_\tau (e^{-\tau}\dt) &= (e^{-\tau} \theta + e^{-\tau} \theta_n - 2) e^{-\tau} \dt + e^{-3\tau}(\bar S - \bar S_n)\\
        &\quad - e^{-3\tau}( |d\s|_h^2 - |d\s_n|_{h_n}^2 ) - 3e^{-3\tau}( V \circ \s - V \circ \s_n ) + e^{-3\tau} \tr\ce_n,
    \end{split}\label{improvemets for theta}\\
    \begin{split}
        \lie_{\p_\tau} \dk &= (e^{-\tau}\theta_n - 1) \dk + \dt e^{-\tau}K - e^{-2\tau}( d\s \otimes \sn \s - d\s_n \otimes \sn \s_n )\\
        &\quad - e^{-2\tau}( V \circ \s - V \circ \s_n )I + e^{-2\tau}( \sric^\sharp - \sric_n^\sharp ) +e^{-2\tau} \ce_n,
    \end{split}\label{improvements for k}\\
    \begin{split}
        \lie_{\p_\tau} \delta \bar\h(X,Y) &= \delta \bar\h( e^{-\tau \Ko} \circ (\Ko - e^{-\tau}K) \circ e^{\tau \Ko}(X),Y )\\
        &\quad + \delta \bar\h( X, e^{-\tau \Ko} \circ (\Ko - e^{-\tau}K) \circ e^{\tau \Ko}(Y) ) \\
        &\quad -\bar\h_n( e^{-\tau \Ko} \circ \dk \circ e^{\tau \Ko}(X),Y ) - \bar\h_n( X, e^{-\tau \Ko} \circ \dk \circ e^{\tau \Ko}(Y) )\\
        &\quad +\bar\h_n( e^{-\tau \Ko} \circ e^{-\tau}(\bar K_n - K_n) \circ e^{\tau \Ko}(X),Y )\\
        &\quad + \bar\h_n( X, e^{-\tau \Ko} \circ e^{-\tau}(\bar K_n - K_n) \circ e^{\tau \Ko}(Y) ),
    \end{split}\label{improvements for h}\\
    \begin{split}
        \p_\tau \delta \bar\Psi &= (e^{-\tau} \theta - 1) \delta \bar\Psi + \dt \bar\Psi_n - e^{-2\tau}(\Delta_h \s - \Delta_{h_n} \s_n)\\
        &\quad + e^{-2\tau}( V' \circ \s - V' \circ \s_n ) + e^{-2\tau}( V' \circ \s_n - \Box_{g_n}\s_n ),
    \end{split}\label{improvements for psi}\\
    \p_\tau \delta \bar\Phi &= \tau \p_\tau \delta \bar\Psi,\label{improvements for phi}
    \end{align}
    \end{subequations}
    for $X,Y \in \mfx(\Sigma)$. Let $\sigma := \min\{\varepsilon,\delta\}$. Note that our assumptions and Theorem~\ref{approximate solutions} ensure that the following holds,
    \begin{align*}
        \sum_{m=0}^A\Big(|D^m \dt|_\ho + |D^m\dk|_\ho + |D^m\delta\bar\h|_\ho + |D^m\delta\bar\Psi|_\ho + |D^m\delta\bar\Psi|_\ho\Big) \leq Ct^\sigma,\\
    \end{align*}
    along with the corresponding off-diagonal improvements for $\delta\bar\h$ and $\dk$. Furthermore, by Theorem~\ref{approximate solutions} and Lemma~\ref{error in the approximate weingarten map}, we can take $n$ large enough such that $\ce_n$, $\Box_{g_n}\s_n - V' \circ \s_n$ and $\bar K_n - K_n$ decay as an arbitrarily large power of $t$. Hence each equation in \eqref{equations for uniqueness} is of the type considered in Lemma~\ref{linear ode lemma}. We can thus use Lemma~\ref{linear ode lemma} to successively improve on the estimates for the variables. We illustrate the idea by going through the first iteration.

    By Corollary~\ref{estimates of differences}, we have
    \[
    \sum_{m=0}^{A-2}\Big(t^2|D^m(\sric^\sharp - \sric_n^\sharp)|_\ho + t^2|D^m(d\s \otimes \sn\s - d\s_n \otimes \sn\s_n)|_\ho\Big) \leq Ct^{2\sigma},
    \]
    along with the corresponding off-diagonal improvements, and
    \[
    \begin{split}
        \sum_{m=0}^{A-2} t^2|D^m( \Delta_h \s - \Delta_{h_n}\s_n )|_\ho \leq Ct^{2\sigma},\\
        \sum_{m=0}^{A-2}\Big(t^2|D^m( V \circ \s - V \circ \s_n)|_\ho + t^2|D^m( V' \circ \s - V' \circ \s_n)|_\ho \Big) \leq Ct^{2\sigma}.
    \end{split}
    \]
    Hence, from \eqref{improvemets for theta}, it follows that
    \[
    \sum_{m=0}^{A-2} |D^m \dt|_\ho \leq Ct^{2\sigma}.
    \]
    Now we can use \eqref{improvements for k} and \eqref{improvements for psi} to obtain
    \[
    \sum_{m=0}^{A-2} \Big( |D^m \dk|_\ho + |D^m \delta\bar\Psi|_\ho \Big) \leq Ct^{2\sigma},
    \]
    along with the corresponding off-diagonal improvements for $\dk$. Finally, \eqref{improvements for h} and \eqref{improvements for phi} imply
    \[
    \sum_{m=0}^{A-2} |D^m \delta\bar\h|_\ho \leq Ct^{2\sigma}, \qquad \sum_{m=0}^{A-2} |D^m \delta\bar\Phi|_\ho \leq C\langle \ln t \rangle t^{2\sigma},
    \]
    along with the corresponding off-diagonal improvements for $\delta\bar\h$. Note that we have obtained improvements on the estimates for all the variables. On the other hand, we have lost two derivatives in the process as a consequence of applying Lemma~\ref{estimates of differences}. Nonetheless, by iterating this process we can ensure that the claim holds, with $A$ depending only on $M$ and $\sigma$.

    In order to apply Proposition~\ref{first uniqueness result} it only remains to show that the required estimates for $\lie_{\p_t} K$ hold. But these follow from the estimates for $K$, Lemma~\ref{decay of ricci} and the evolution equation \eqref{evolution equation for k} for $K$. The result follows.
\end{proof}

\begin{proof}[Proof of Theorem~\ref{main uniqueness theorem}]
    By Theorem~\ref{asymptotics of the frame}, we see that there is a sufficiently small $T>0$ such that the assumptions of Proposition~\ref{second uniqueness result} are satisfied by $(F_1^*g_1,\s_1 \circ F_1)$ and $(F_2^* g_2, \s_2 \circ F_2)$ in $(0,T] \times \Sigma$. The result follows.
\end{proof}

Using Theorem~\ref{main uniqueness theorem}, we can justify why it is reasonable to assume that the frame $\{e_i\}$ of eigenvectors of $\Ko$ is global.

\begin{remark} \label{about the global frame}
    Let $(\Sigma,\ho,\Ko,\phio,\psio)$ be initial data on the singularity and $V$ be an admissible potential. If $\Ko$ does not have a global frame of eigenvectors, there is a finite covering space $\widetilde \Sigma$, with covering map $\pi:\widetilde\Sigma \to \Sigma$, such that $\pi^*\Ko$ has a global frame of eigenvectors; see \cite[Lemma~A.1]{ringstrom_wave_2021}. We can then pull back the initial data to $\widetilde\Sigma$ and by Theorem~\ref{main existence theorem}, we obtain a corresponding locally Gaussian development, say $((0,T] \times \widetilde\Sigma,\widetilde g = -dt \otimes dt + \widetilde h, \widetilde\s)$. The idea is to take an appropriate quotient of the development to obtain a development of the original initial data. To that end, let $\Gamma$ denote the group of deck transformations of $\pi$ and define the map $\overline{\pi}:(0,T] \times \widetilde\Sigma \to (0,T] \times \Sigma$ by $\overline{\pi}(t,x) := (t,\pi(x))$. Clearly $\overline{\pi}$ is a covering map. Moreover, if $\bar\gamma$ is a deck transformation of $\overline{\pi}$, then $\bar\gamma(t,x) = (t,\gamma(x))$ for some $\gamma \in \Gamma$. Since each $\gamma \in \Gamma$ preserves the pulled back initial data, then $((0,T] \times \widetilde\Sigma, \widetilde g, \widetilde\s)$ and $((0,T] \times \widetilde\Sigma, \bar\gamma^*\widetilde g, \widetilde\s \circ \bar\gamma)$ are both locally Gaussian developments of the same initial data. Thus, by Theorem~\ref{main uniqueness theorem} and taking $T$ smaller if necessary, $\bar\gamma^*\widetilde g = \widetilde g$ and $\widetilde\s \circ \bar\gamma = \widetilde\s$. Consequently, there is a unique Lorentzian metric $g$ and a unique function $\s$ on $(0,T] \times \Sigma$ such that $\overline{\pi}$ is a local isometry and $\s \circ \overline{\pi} = \widetilde\s$; see \cite[Corollary~12, p. 191]{oneill_semi-riemannian_1983}. Then $((0,T] \times \Sigma, g, \s)$ is the desired locally Gaussian development of the original initial data.
\end{remark}

\begin{appendices}

\section{Conventions} \label{appendix}

\paragraph{Notation for constants.} Throughout the paper we use $C$, $C_m$, etc., to denote positive constants whose value may change from line to line. Moreover, unless otherwise stated, they are only allowed to depend on the initial data on the singularity and the potential. 

\paragraph{Norms of tensors.}

Let $(M,g)$ be a semi-Riemannian manifold. We begin by extending the metric to arbitrary tensors.

\begin{definition}
    The metric $g$ can be extended to tensors as follows. Let $X_i, Y_i \in \mathfrak X(M)$ and $\omega^k, \alpha^k \in \Omega^1(M)$ for $i = 1,\ldots,q$ and $k = 1,\ldots,r$; then $g$ can be defined for simple tensors by the formula
    \[
    \begin{split}
        &g( \omega^1 \otimes \cdots \otimes \omega^r \otimes X_1 \otimes \cdots \otimes X_q, \alpha^1 \otimes \cdots \otimes \alpha^r \otimes Y_1 \otimes \cdots \otimes Y_q )\\
        &\hspace{4cm} := g^{-1}(\omega^1,\alpha^1) \cdots g^{-1}(\omega^r,\alpha^r) g(X_1,Y_1) \cdots g(X_q,Y_q),
    \end{split}
    \]
    and we extend it to arbitrary $(q,r)$-tensors by bilinearity. Then define the norm of a $(q,r)$-tensor $T$ by
    \[
    |T|_g := \sqrt{|g(T,T)|}.
    \]
\end{definition}

Now let $(M,g)$ be a closed Riemannian manifold. We define the corresponding $L^p$, Sobolev and $C^k$ norms as follows.

\begin{definition}
    Let $\mu$ be the volume form of $g$, $T$ be a tensor and $1 \leq p < \infty$, then the $L^p(M,g)$ and $L^\infty(M,g)$ norms of $T$ are defined by
    \[
    \|T\|_{L^p(M,g)} := \bigg( \int_M |T|_g^p \mu \bigg)^{1/p}, \qquad \|T\|_{L^\infty(M,g)} := \sup_M |T|_g.
    \]
    Also, for $1 \leq p \leq \infty$, define the Sobolev $W^{s,p}(M,g)$ norm of $T$ by
    \[
    \|T\|_{W^{s,p}(M,g)} := \sum_{m=0}^s \|\n^m T\|_{L^p(M,g)},
    \]
    where $\n$ is the Levi-Civita connection of $g$ and $\n^m T$ denotes the $m$-fold covariant differential of $T$. Moreover, denote $\|\cdot\|_{H^s(M,g)} := \|\cdot\|_{W^{s,2}(M,g)}$. Finally, define the $C^k(M)$ norm by
    \[
    \|T\|_{C^k(M)} := \sum_{m=0}^k \sup_M |\n^m T|_g.
    \]
\end{definition}

\begin{definition} \label{multiindex definition}
    A multiindex $\alpha$ of order $m$ is a tuple $\alpha = (\alpha_1,\ldots,\alpha_m)$, such that if ${i = 1,\ldots,m}$ then $\alpha_i \in \{1,\ldots,n\}$; where $n = \dim M$. Let $\{e_i\}$ be a frame on $M$ and $u \in C^\infty(M)$, then we use the notation $e_\alpha u := e_{\alpha_1} \cdots e_{\alpha_m} u$ and $|\alpha| := m$. Note that our notation for multiindices differs from the usual one. This is because the frame $\{e_i\}$ does not, in general, commute.
\end{definition}

Often we estimate objects of the form $|\n^m T|_g$ by estimating derivatives of components of $T$ in terms of an orthonormal frame. This is justified by the following observation.

\begin{remark}
    Let $\{e_i\}$ be a global orthonormal frame on $M$ and $T$ be a $(q,r)$-tensor. Then there is a constant $C$, independent of $T$, such that
    \[
    C^{-1} \sum_{k=0}^m |\n^k T|_g \leq \sum |e_\alpha T_{i_1 \cdots i_r}^{k_1 \cdots k_q}| \leq C \sum_{k=0}^m |\n^k T|_g,
    \]
    where the sum in the middle is over all indices and every $\alpha$ such that $|\alpha| \leq m$.
\end{remark}

\paragraph{Normal Lie derivatives.} 

Throughout the paper, we make use of metrics of the form $g = -dt \otimes dt + h$ on $I \times \Sigma$ where $I$ is an interval, $\Sigma$ is a closed manifold, and the hypersurfaces $\Sigma_t := \{t\} \times \Sigma$ are spacelike with induced metric $h_t$ and future pointing unit normal $\p_t$. In this setting, it is convenient to introduce a notion of normal derivative for tensors which are defined on each $\Sigma_t$, like for instance the Weingarten map $K$. 

First a comment regarding the regularity. Let $T$ be a one parameter family of $(q,r)$-tensors on $\Sigma$, for $t \in I$. We say that $T$ is smooth if the function $T(X_1,\ldots,X_r,\omega^1,\ldots,\omega^q)$ is smooth as a function from $I \times \Sigma$ to $\R$, for all $X_1,\ldots,X_r \in \mathfrak X(\Sigma)$ and all $\omega^1,\ldots,\omega^q \in \Omega^1(\Sigma)$. 

Let $X$ be a smooth one parameter family of vector fields on $\Sigma$. We can equivalently think of $X$ as a vector field on $I \times \Sigma$, such that $X_t$ is tangent to the hypersurface $\Sigma_t$ for each $t \in I$. Then we can consider $[\p_t,X]$. Note that for each $t \in I$, the vector field $[\p_t,X]$ is tangent to $\Sigma_t$, hence we can think of it as a one parameter family of vector fields on $\Sigma$. Given this observation, it makes sense to make the following definition.

\begin{definition} \label{normal lie derivative}
    Consider a smooth one parameter family of $(q,r)$-tensors $T$ on $\Sigma$, for $t \in I$. Define
    \[
    \begin{split}
        (\lie_{\p_t} T)(X_1,\ldots,X_r,\omega^1,\ldots,\omega^q) &:= \p_t \big( T(X_1,\ldots,X_r,\omega^1,\ldots,\omega^q) \big)\\
        &\quad - \sum_i T(X_1,\ldots,[\p_t,X_i],\ldots,X_r,\omega^1,\ldots,\omega^q)\\
        &\quad - \sum_k T(X_1,\ldots,X_r,\omega^1,\ldots,\lie_{\p_t} \omega^k,\ldots,\omega^q),
    \end{split}
    \]
    where $X_1,\ldots,X_r$ and $\omega^1,\ldots,\omega^q$ are smooth one parameter families of vector fields and one forms on $\Sigma$ respectively. 
\end{definition}

Note that the same formula defines $\lie_{\p_t}$ for a one form, while only making reference to $[\p_t,X]$. Hence $\lie_{\p_t} T$ is a well defined smooth one parameter family of $(q,r)$-tensors on $\Sigma$.

\paragraph{Raising indices of tensors.}

Here we clarify our conventions regarding our use of the notation $\sharp$.

\begin{definition} \label{raising an index}
    Let $(M,g)$ be a semi-Riemannian manifold and let $T$ be a 2-covariant tensor. Then $T^\sharp$ is the $(1,1)$-tensor defined by
    \[
    g(T^\sharp(X),Y) := T(X,Y),
    \]
    for all $X, Y \in \mathfrak X(M)$.
\end{definition}

\end{appendices}

\printbibliography[heading=bibintoc]

\end{document}